\documentclass[12pt]{article}
\usepackage[T1]{fontenc}
\usepackage[margin=1in]{geometry} 
\usepackage[onehalfspacing]{setspace} 
\usepackage{graphicx}

\usepackage{microtype, lmodern}

\usepackage{natbib}
\usepackage{multibib}
\newcites{appendix}{References}
\usepackage[usenames, dvipsnames]{xcolor}
\usepackage[colorlinks=true]{hyperref}
\hypersetup{linkcolor=blue, citecolor=blue}

\usepackage{amsmath, amssymb, amsthm, thmtools, dsfont, enumerate, placeins, accents, xurl, bbm, subcaption, multicol, multirow}
\allowdisplaybreaks 

\usepackage{titlesec}

\theoremstyle{plain} % italic text
\newtheorem{theorem}{Theorem}
\newtheorem{corollary}[theorem]{Corollary}

\newtheorem{assumption}[theorem]{Assumption}
\newtheorem{lemma}[theorem]{Lemma}
\newtheorem{proposition}[theorem]{Proposition}

\newcommand{\R}{\mathbb{R}} % R 

\begin{document}

\title{Optimal Policy Choices Under Uncertainty\thanks{I am especially grateful to Isaiah Andrews and Nathan Hendren for their guidance and advice. I thank Amy Finkelstein, Anna Mikusheva, and participants in the MIT econometrics lunch seminar and second-year paper course for their helpful feedback. This material is based upon work supported by the National Science Foundation Graduate Research Fellowship Program under Grant No. 2141064. Any opinions, findings, and conclusions or recommendations expressed in this material are those of the author and do not necessarily reflect the views of the National Science Foundation. All errors are my own.}}
\author{Sarah Moon\thanks{Department of Economics, Massachusetts Institute of Technology, Cambridge, MA 02139, USA. Email: \href{mailto:sarahmn@mit.edu}{sarahmn@mit.edu}} }
\date{\today}

\maketitle

\begin{abstract}
Policymakers often face the decision of how to allocate resources across many different policies using noisy estimates of policy impacts. This paper develops a framework for locally optimal policy choices under statistical uncertainty. I show that posterior mean benefits and net costs are sufficient statistics for an oracle planner who knows the distribution of policy impacts. Since this distribution is unknown, I propose an empirical Bayes approach to estimate posterior means and approximate the oracle. I derive rates of convergence to the oracle's decision and show that, unlike empirical Bayes, plug-in methods can fail to converge. In an application to 127 policies, empirical Bayes rules have positive estimated local welfare effects, while the plug-in rule has negative estimated effects.
\end{abstract}

% Keywords: Empirical Bayes, policy choice, sufficient statistic, social welfare

\newpage

\section{Introduction}

Policymakers often rely on empirical evidence to decide how to allocate resources across many different policies, including food stamp expansions, job training programs, tax rate changes, and more. However, empirical estimates of policy benefits and costs are often noisy, and it's unclear how resources should be allocated when faced with this statistical uncertainty. For example, one policy might have a larger estimated welfare gain but also be estimated with substantial noise, while another might have a smaller estimated welfare gain but be estimated more precisely. Should policymakers favor the policy with the larger estimated welfare gain or the policy with less statistical uncertainty? Without a principled approach to managing this tradeoff, policymakers risk directing resources toward policies that only appear attractive due to estimation error. In this paper I address the question: How should policymakers make simultaneous changes to many different policies while accounting for statistical uncertainty about policy impacts? 

This paper develops a framework for optimal policy choices that explicitly accounts for statistical uncertainty. I begin with a social planner who must choose upfront spending on a given menu of policies to maximize social welfare---defined as the weighted sum of individual money-metric utilities---subject to a future budget constraint. I incorporate statistical uncertainty by treating both the true policy benefits and costs and their empirical estimates as random. To establish a welfare benchmark, I consider an oracle social planner who knows the distribution of true policy benefits and costs and maximizes expected social welfare. Standard arguments show that this problem is equivalent to one in which, for (almost) every realization of the empirical estimates, the oracle planner maximizes posterior expected social welfare using the distribution of true benefits and costs as the prior.

I then characterize the information required for optimal policy choices. Since empirical estimates are only informative about the welfare and budget impacts of changes close to the existing policy regime, I focus on a local analysis of the planner's problem, restricting attention to small changes in spending. I show that optimal local policy choices depend on the posterior expected gradient of the net welfare impact. This result delivers a simple sufficient statistic, in the spirit of the public economics literature, which a researcher can report to the planner: the posterior mean benefit and net cost of each policy. Given these posterior means, the planner can solve for optimal local policy without additional information about the underlying distribution of policy impacts.

In practice, the distribution of true policy benefits and costs is unknown, making the oracle planner's decision infeasible because posterior means are unknown. I propose an empirical Bayes approach to approximate the oracle planner by estimating the most likely prior from the data and using it to construct estimates of posterior mean benefits and costs. The key idea is that the planner observes noisy sample estimates of benefits and costs for many policies. Taken together, these sample estimates are informative about the underlying distribution of true policy benefits and costs. I show that policy choices based on empirical Bayes estimates of posterior means approximate the oracle planner arbitrarily well as the number of policies grows.

Formally, I assume the observed estimates for policy benefits and costs are conditionally Gaussian and centered at their true values, consistent with widely-used hypothesis tests and standard error calculations. I consider two different approaches to modeling the distribution of true policy benefits and costs. In the first, I assume the distribution of the true benefits and costs lies within a correctly specified parametric class. I estimate the parameters of this distribution and use them to construct estimates of posterior mean benefits and net costs. In the second, I assume the true benefits and costs follow a multivariate location-scale model that depends on policy type, with residuals distributed according to a flexible prior unknown to the planner. I estimate this unknown prior nonparametrically and combine it with location-scale parameter estimates to yield estimates of posterior mean benefits and net costs. I define the empirical Bayes local spending rule as the local spending rule that solves the planner's local problem using estimated posterior means.

Existing results in the empirical Bayes literature suggest that the estimated posterior mean benefits and costs converge in mean squared error to the true posterior mean benefits and costs. But for the planner, what matters is not mean squared error but the performance of the empirical Bayes local spending rule in the planner's local problem. My main theoretical contribution is to derive rates of convergence for two measures of the gap between the empirical Bayes approximation of the local problem and the oracle planner's local problem that are uniformly valid over a large class of data generating processes and policy environments. I also show that a sample plug-in approach, which solves the local problem with raw point estimates of benefit and cost in place of posterior means, fails to converge in important cases. Together, these results demonstrate that simply plugging in raw sample estimates can lead to suboptimal decisions, while the empirical Bayes approach, which adjusts for sampling uncertainty through posterior mean shrinkage, asymptotically matches the oracle planner's optimal performance.

Finally, I illustrate the proposed method in an empirical application to 127 policies studied by \cite{hendren2020unified} and \cite{hahn2024welfare}. The application demonstrates the implications of empirical Bayes shrinkage of benefit and cost estimates for optimal policy choice. The central finding is that across several different local problem specifications, I find that the empirical Bayes approach is estimated to improve welfare, while the sample plug-in approach is estimated to reduce welfare. These empirical results are consistent with the theoretical performance results established earlier in the paper and illustrate that carefully accounting for statistical uncertainty can qualitatively change the welfare implications of policy choices.

\paragraph{Related Literature}
The decision problem in this paper builds upon the well-studied question in public economics of how to allocate spending across government policies to maximize welfare. I follow the large literature in public economics that models a planner who chooses local changes in spending to maximize welfare; recent examples include \cite{hendren2020unified}, \cite{finkelstein2020welfare}, and \cite{bergstrom2024optimal}. 
I adopt a sufficient statistics approach to solving the local welfare maximization problem, in the spirit of an extensive literature in public economics that derives low-dimensional statistics that are sufficient to draw welfare conclusions \citep[and numerous others]{gruber1997consumption, feldstein1999tax, saez2001using, chetty2008moral, chetty2009sufficient, schmieder2016effects, kleven2021sufficient}.
I depart from this literature by allowing for statistical uncertainty about policy impacts, then carefully setting up a tractable decision problem under statistical uncertainty that can be solved with sample estimates.

This paper also builds on a broad literature in statistical decision theory, which dates back to \cite{wald1949statistical} and more recently the seminal paper of \cite{manski2004statistical}, and the more recent literature on Empirical Welfare Maximization (EWM) \citep[and more]{kitagawa2018should, athey2021policy, mbakop2021model, sun2024empirical}. The EWM literature typically studies how to optimally target a given policy using sample data (e.g., selecting the optimal group of individuals by their observable characteristics to receive a treatment). 
I instead study how to optimally make changes to many policies using noisy empirical estimates. 

One paper in the EWM literature that considers statistical uncertainty when making many policy changes is \cite{chern2025policy}. 
They propose a policy rule that explicitly trades off between the size of estimated welfare and the estimation uncertainty of welfare, which, as also discussed in \cite{andrews2025certified}, is equivalent to maximizing worst-case performance over a confidence set for welfare. 
While I consider the same question of how to make changes to a set of policies based on noisy sample estimates of policy impacts, my approach is different and developed independently. 
In particular, I consider a planner who maximizes expected welfare, knowing policy impacts are drawn from some unknown distribution. I show that estimation uncertainty matters because it informs Bayesian updating even though the planner has no direct preference over estimation error (i.e., the planner is risk neutral in welfare space). Moreover, the empirical Bayes approach I propose allows for shrinkage by pooling together information across policies and attains bounds on expected regret as the number of policies grows, even if there is non-vanishing uncertainty for each policy estimate. In contrast, \cite{chern2025policy} provide high probability bounds on regret that are attained as estimation error goes to zero. 

The relative changes to upfront spending from the optimal local change to spending induce an implicit ranking of policies. Rankings with statistical uncertainty have been studied in the econometrics literature; \cite{mogstad2024inference} propose a frequentist approach to inference on the ranks themselves, while \cite{andrews2024inference} perform inference on the highest ranked outcome. Several papers have proposed empirical Bayes approaches to ranking under a decision-theoretic framework, including \cite{gu2023invidious} and \cite{kline2024discrimination}. Instead of focusing on discrete rankings, in my paper I allow for decisions to vary continuously in magnitude and direction, capturing the idea that policymakers choose not only whether to change a policy but by how much.

While empirical Bayes methods are typically used in economics for denoising or ranking, in this paper I apply them to a policymaking decision problem. Another paper that proposes empirical Bayes methods in a policymaking setting is \cite{yamin2025poverty}, who studies the problem of how to allocate cash transfers to minimize poverty using noisy measures of income. That paper shows in theory and in simulations that a nonparametric empirical Bayes approach to allocating transfers outperforms a sample plug-in approach. 
Similarly, in this paper I prove that an empirical Bayes approach to decision-making can perform better than a sample plug-in approach in a related but distinct policymaking setting.

The literature in Bayesian statistical decision theory is deeply related to the large literature on empirical Bayes methods, which dates back to the seminal work of \cite{robbins1956empirical} and has since been expanded by many researchers in various fields \citep[and numerous others]{jiang2009general, efron2012large, koenker2014convex, jiang2020general, gu2023invidious, soloff2024multivariate, chen2022empirical}. The nonparametric empirical Bayes approach in this paper builds upon the nonparametric maximum likelihood approach, pioneered by \cite{kiefer1956consistency}, for multivariate, heteroscedastic empirical Bayes, as developed by \cite{soloff2024multivariate}. I extend results on mean squared error risk bounds in this literature to allow for a multivariate location-scale family of distributions and derive rates of convergence directly for the social welfare measure of interest. 

\paragraph{Outline} The rest of the paper proceeds as follows. In Section \ref{sec2}, I characterize the optimal local change to spending for the oracle planner under statistical uncertainty, showing that posterior mean benefits and net costs are sufficient statistics. In Section \ref{sec3}, I show that replacing these posterior means with raw sample estimates can lead to suboptimal decisions. I then propose an empirical Bayes approach to estimate posterior means that recovers the oracle planner's performance asymptotically. Section \ref{sec4} presents simulations and Section \ref{sec5} applies the method to 127 policies studied by \cite{hendren2020unified} and \cite{hahn2024welfare}. Section \ref{sec6} concludes.

\section{Social Welfare Optimization} \label{sec2}

\subsection{Setup and Local Approximation} \label{sec_2.1}

Consider a social planner who can adjust a finite number of policies $j = 1, \dots, J$ through changes to upfront spending on each policy. Let $s_j$ denote the change in upfront spending on policy $j$ and let $s = (s_1, \dots, s_J)$ collect upfront spending changes on all policies. The planner wants to maximize social welfare subject to a budget constraint. Let $W(s)$ denote social welfare after upfront spending changes $s$ and $G(s)$ denote the present discounted value of the planner's long-run budget after upfront spending changes $s$. Then $W(0)$ and $G(0)$ are the welfare and budget, respectively, of the current policy regime, where throughout the paper I use $0$ to denote the zero vector. Note that $G(\cdot) > 0$ means that the planner is spending money in the long run, and $G(\cdot) < 0$ means the planner is bringing in money in the long run.

Without statistical uncertainty, the planner knows the exact welfare and the long-run budget after any given spending change and can therefore exactly enforce the budget constraint. In practice, only estimates of the welfare and budget impact of policy changes are available at the time of decision-making. I instead assume that the planner chooses $s$ knowing that, in some future period after the true welfare and budget impacts are realized, the budget constraint will be closed with certainty through a budget-closing policy. Let $\mu$, assumed to be known, denote the welfare impact per unit increase in budget from closing the budget. Then closing a budget shortfall of size $G$ yields a welfare impact of $-\mu G$. Note that by assuming $\mu$ to be known, I rule out the ability to close the budget constraint with any of the policies that are being learned about from sample estimates at the time of decision-making.

I first analyze the planner's decision problem without statistical uncertainty and defer the case with uncertainty to Section \ref{sec_2.3}. 
The planner chooses changes to spending $s$ to maximize the \textit{net welfare impact} of $s$, that is, the sum of the direct welfare impact $W(s)$ and the indirect welfare impact from closing the budget $-\mu G(s)$:
\begin{align*}
    \max_s W(s) - \mu G(s).
\end{align*}
Under this formulation of the planner's objective, an alternative interpretation for $\mu$ is the marginal value of relaxing the budget constraint from the Lagrangian formulation of maximizing welfare subject to closing the budget constraint.

Global maximization of welfare is straightforward when the net welfare impact $W(s) - \mu G(s)$ of every possible spending change $s$ is known. In this case, the global problem can be solved using standard optimization methods. If instead the planner observes sample estimates of $W(s) - \mu G(s)$ across possible spending changes, the planner can use standard stochastic optimization methods that account for sampling uncertainty. Alternatively, the planner could take a Bayesian approach, placing a prior on the net welfare impact function, updating it using the sample estimates, and choosing the spending change that maximizes posterior expected welfare, as in \cite{kasy2018optimal}.

In practice, however, available empirical estimates are often estimates of local changes to observed policies, rather than the net welfare impact of every possible spending change. Global maximization therefore entails extrapolating empirical estimates beyond the observed policy regime, which in turn requires strong assumptions about how net welfare impact varies with $s$. To avoid such extrapolation, I consider the optimal \textit{local} change to policy spending: that is, starting from the current policy regime, what is the best ``small'' change to policy spending? This local approach to welfare analysis is often taken in the sufficient statistics literature in public economics to address a similar tension between credibility of assumptions and ability to make welfare statements.\footnote{See \cite{chetty2009sufficient} and \cite{kleven2021sufficient} for recent reviews of the sufficient statistics approach in public economics.}

To formalize this idea, let $w(s) \equiv W(s) - \mu G(s)$ denote the net welfare impact of spending change $s$ and suppose $w(s)$ is continuously differentiable in a neighborhood of $s = 0$.
The meaning of a ``local'' change is not unique in a multidimensional policy space: it depends on the geometry of the neighborhood around the current policy regime. For example, one could define local changes as those for which the total absolute change in spending across policies is small, or as those for which the absolute change in spending on each policy is small. I represent this geometry by a set $V \subseteq \R^J$ with $0 \in V$, a compact \textit{consideration set} of feasible directions of change in spending. The set $V$ may also encode real-world constraints faced by the planner, such as political constraints on changing spending on certain policies. 

For each scale $t>0$, define $tV = \{tv:v \in V\}$, the corresponding neighborhood around the current policy regime at scale $t$. 
Restricting the global problem to spending changes $s \in tV$ gives
\begin{align*}
    \max_{s \in tV} \frac{w(s) - w(0)}{t},
\end{align*}
where dividing by $t$ makes the objective comparable across different spending scales $t$. 

As $t \to 0$, this sequence of restricted global problems converges to
\begin{align*}
    \lim_{t \to 0} \max_{s \in tV} \frac{w(s) - w(0)}{t} = \max_{v \in V} \langle \nabla w, v \rangle,
\end{align*}
where $\langle \cdot, \cdot \rangle$ denotes the dot product on $\mathbb{R}^J$ and $\nabla w$ denotes the gradient of $w$ at $0$.
I define the planner's \textit{local problem} to be this $t \to 0$ limit of the restricted global problem. 

Note that a first-order Taylor expansion of $w(s)$ around zero gives $w(s) \approx w(0) + \langle \nabla w, s \rangle$ for small $s$. This formalizes the idea that the local problem approximates the global problem for small deviations from the status quo. 
The objective of the local problem is the instantaneous \textit{rate of increase} in $w(s)$ along direction $v$,
\begin{align*}
    \langle \nabla w, v \rangle = \frac{\partial}{\partial t} w(0 + tv) \big\vert_{t = 0}.
\end{align*}
An \textit{optimal local spending rule} is any direction $v \in V$ that maximizes this rate of increase. The optimal local spending rule may not be unique, but all optimal local spending rules result in the same maximal rate of increase in net welfare impact. Going forward, I use a fixed deterministic tie-breaking rule to select a unique maximizer, which I refer to as the optimal local spending rule.

Given any consideration set $V$, knowing the gradient $\nabla w$ is enough to characterize both the optimal local spending rule and the maximal rate of increase in net welfare impact. The sufficiency of the gradient to determine the optimal local spending rule is in the spirit of the sufficient statistics approach to welfare analysis, which provides low-dimensional statistics that are sufficient to make certain statements about welfare effects in various economic models.

Certain forms of the consideration set $V$ yield simple closed-form solutions. For example, if $V$ is equal to an $L^p$ unit ball for $p \in [1,\infty]$, $V = \mathcal B_p \equiv \{v \in \mathbb{R}^J: \Vert v \Vert_p \leq 1\}$, the dual norm gives 
\begin{align*}
    \max_{v \in \mathcal B_p} \langle \nabla w, v \rangle &= \left\Vert \nabla w \right\Vert_{\frac{p}{p-1}},
\end{align*}
where I use the usual conventions that $\frac{p}{p-1} = \infty$ when $p=1$ and $\frac{p}{p-1} = 1$ when $p = \infty$. The choice of $p$ corresponds to different notions of a local change in spending. When $p=1$, local changes are those with a small total absolute change in spending across policies. When $p=\infty$, local changes are those for which the absolute change in spending on each policy is small. Intermediate values of $p$, like $p=2$, interpolate between these two definitions.
In Section \ref{sec3} I provide theoretical results for consideration sets $V \subseteq \mathcal B_p$.

\subsection{Notation} \label{sec_2.2}

I use the notation of \cite{hendren2020unified} to re-express the gradient $\nabla w$ in terms of more familiar economic objects.
Taking social welfare to be the weighted sum of individuals' money-metric utilities, the change in social welfare due to a marginal change in upfront spending on policy $j$ can be written as
\begin{align*}
    \frac{\partial W}{\partial s_j}\bigg\vert_{s = 0} &= \eta_j WTP_j.
\end{align*}
Here $\eta_j$ is the average social marginal utility of income for policy $j$, which is the increase in social welfare from giving \$1 on average to the individuals impacted by policy $j$. I will also refer to $\eta_j$ as the welfare weight for policy $j$. The term $WTP_j$ is the sum of individuals' marginal willingness to pay for policy $j$, which I will call the \textit{benefit} of policy $j$.

The change in long-run budget due to a marginal change in upfront spending on policy $j$ is 
\begin{align*}
    \frac{\partial G}{\partial s_j} \bigg \vert_{s=0} &= G_j,
\end{align*}
which I will call the \textit{net cost} of policy $j$. 
With this notation, the gradient of the net welfare impact at zero spending is
\begin{align*}
    \nabla w &= \begin{pmatrix} \eta_1 WTP_1 - \mu G_1 \\ \vdots \\ \eta_J WTP_J - \mu G_J \end{pmatrix}.
\end{align*}
Note that each component of $\nabla w$ is a partial derivative, evaluated holding spending on all other policies fixed at the current policy regime. Interaction effects generated by simultaneously changing multiple policies enter only through higher-order terms and therefore drop out of the local problem. This does not rule out interactions at the status quo, which can affect the partial derivatives.

As discussed in \cite{finkelstein2020welfare}, estimates of $WTP_j$ and $G_j$ are available for many different policy changes. The Policy Impacts Library provides such estimates for over 149 policies in the United States. To ensure estimates of benefit and net cost are comparable across different policies, in this paper I normalize the size of a marginal change in upfront spending on policy $j$ to be one monetary unit of program cost. In practice this means that I divide estimates of the benefit and net cost of policies by the program cost. 

In this paper I assume that $\eta_j$ is known ex ante by the planner for each policy $j = 1,\dots,J$. This means that the uncertainty to be introduced in Section \ref{sec_2.3} about the direct welfare impact of each policy comes from uncertainty about the benefit of each policy to its recipients, rather than uncertainty about welfare weights $\eta_j$. I make this assumption because empirical estimates in the literature primarily quantify the benefit and net cost of different policy changes, while $\eta_j$ captures in part the preferences of the planner, which are not as easily estimated.\footnote{There do exist methods to back out the social marginal utility of income across the income distribution, albeit without taking into account statistical uncertainty \citep{bourguignon2012tax, hendren2020measuring}. However, for those approaches to be valid the planner must find the current tax schedule optimal, which is at odds with the premise of this paper that the planner wants to make changes to the current policy regime.}

\subsection{Adding Statistical Uncertainty} \label{sec_2.3}

In practice the true welfare and budget impacts from changes to spending are unknown to the planner. Instead, the planner observes sample estimates of the benefits and net costs, denote $\{(\widehat{WTP}_j, \widehat{G}_j)\}_{j=1}^J$, together with their covariance matrices $\{\Sigma_j\}_{j=1}^J$ from empirical studies of $J$ different policy changes. In this paper I model both the true welfare and budget impacts and the sample estimates as random. In particular, I first assume that the true impacts $\{(WTP_j, G_j)\}_{j=1}^J$ are jointly drawn from some distribution $F$. In this section I assume $F$ is known by the planner; in Section \ref{sec3} I present an empirical Bayes approach to proceed when $F$ is not known.

Motivated by the central limit theorem, I model the sample estimates as independent across policies and conditionally Gaussian with known covariance matrices:\footnote{As discussed in Section \ref{sec_2.2}, I normalize the sample estimates of benefit and net cost by program cost to ensure they are comparable across different policies. If program costs are observed without statistical uncertainty, the Gaussian distribution approximation motivated by the central limit theorem is reasonable. Most of the statements in Section \ref{sec2} are still valid without exact Gaussianity. See \cite{chen2026normal} for a discussion of how the empirical Bayes results derived in Section \ref{sec3} would change when the data is not exactly Gaussian.} 
\begin{equation}\label{eq:likelihood1}
    \begin{pmatrix} \widehat{WTP}_j \\ \widehat{G}_j \end{pmatrix} \Bigg\vert \begin{pmatrix} {WTP}_j \\ {G}_j \end{pmatrix}, \Sigma_j \overset{\text{ind.}}{\sim} N \left(\begin{pmatrix} {WTP}_j \\ {G}_j \end{pmatrix}, \Sigma_j \right).
\end{equation}
While unbiasedness and exact Gaussianity are not without loss, approximate normality and unbiasedness are already implicit in the reported standard errors and hypothesis tests in the literature. By conditioning on $\Sigma_j$ I take it to be fixed. In practice, following the empirical Bayes literature \citep[and others]{walters2024empirical, soloff2024multivariate, chen2022empirical}, I will use consistent covariance matrix estimates from the empirical studies for $\Sigma_j$. I leave the problem of dealing with estimated covariance matrices and error in the normal approximation to future work.

I again analyze the planner's local problem. Consider the oracle social planner, who knows both the distribution of true policy impacts $F$ and the sample estimate model of \eqref{eq:likelihood1}. In line with the aggregation theorem of \cite{harsanyi1955cardinal}, the planner evaluates any candidate local direction of change in spending by the expectation of the planner's local objective.
Here the expectation averages over both the randomness of the true policy impacts and the randomness of sample estimate noise.
In this sense, the planner is risk neutral in welfare space, where here risk comes from uncertainty about the welfare and budget impacts. 

For notational simplicity let $Y_j \equiv (\widehat{WTP}_j, \widehat{G}_j)'$. The planner solves for the optimal local spending rule with respect to consideration set $V$ by solving
\begin{align*}
    \max_{v: \mathcal{Y} \to V} E \left[ \langle \nabla w, v(Y_{1:J}) \rangle \right],
\end{align*}
where $\mathcal{Y} = \R^{2J}$ is the space of sample estimates and $Y_{1:J}$ collects sample estimates $Y_j$ for all $J$ policies as a vector in $\mathcal{Y}$. I write the spending rule $v$ as a map from $\mathcal{Y}$ to $V$ to emphasize that it can depend on the realization of the sample estimates because the planner observes sample estimates before choosing spending.

Note that by the law of iterated expectations,
\begin{align*}
    E \left[ \langle \nabla w, v(Y_{1:J}) \rangle \right] = E \left[ E \left[ \left\langle \nabla w, v(Y_{1:J}) \right\rangle \vert Y_{1:J} \right] \right].
\end{align*}
Then the optimal local spending rule is the Bayes solution, which is the rule that maps each realization of the sample estimates $Y_{1:J} \in \mathcal{Y}$ to the vector $v \in V$ that (locally) maximizes the posterior expected net welfare impact given the realization of the sample estimates. The optimality of the Bayes solution for a decision-maker who maximizes such an expected objective is a standard result in Bayesian inference \citep[see, e.g., Theorem 2.3.2 of][]{robert2007bayesian}. 
Thus, using $\pi$ as shorthand for the posterior distribution of the true impacts $\{({WTP}_j, {G}_j)\}_{j=1}^J$ given sample estimates $Y_{1:J}$, the optimal local spending rule maps each $Y_{1:J} \in \mathcal{Y}$ to the vector $v \in V$ that solves the following problem: 
\begin{align}
   \max_{v \in V} E \left[ \left\langle \nabla w, v \right\rangle \vert Y_{1:J} \right] = \max_{v \in V} E_\pi \left[\langle \nabla w, v \rangle \right] = \max_{v \in V} \langle E_\pi \left[\nabla w \right], v \rangle, \label{eq:objective}
\end{align}
where the final equality follows from linearity of the expectation operator.

\paragraph{Planner sufficient statistic}
The sufficient statistic for the planner to determine the optimal local spending rule is now the posterior expected gradient,
\begin{align}
    E_\pi \left[ \nabla w \right] &= \begin{pmatrix}
        \eta_1 E_\pi[WTP_1] - \mu E_\pi[G_1] \\ \vdots \\ \eta_J E_\pi[WTP_J] - \mu E_\pi[G_J] 
    \end{pmatrix}, \label{eq:post_grad}
\end{align}
in the sense that given $E_\pi[\nabla w]$, the planner who knows consideration set $V$ can solve for the optimal local spending rule.
Because $\mu$ and the $\eta_j$'s are assumed known by the planner, the planner only additionally needs to know the posterior mean benefit and net cost of each policy to construct the posterior expected gradient. 
As discussed in Section \ref{sec_2.2}, the local problem depends only on the impacts of individual marginal policy changes rather than the joint impacts of changing all policies simultaneously. Consequently, the local approach is directly compatible with the many empirical estimates of individual policy changes. 

\paragraph{Researcher sufficient statistic}
A researcher who has compiled benefit and net cost estimates for a set of policies and would like to assist the planner in making optimal policy choices may not know the parameters of the oracle planner's local problem like $\mu$, welfare weights $\eta_j$, or the consideration set $V$. 
Crucially, the sufficiency of posterior mean benefits and net costs for the planner to construct the posterior expected gradient means that the researcher only needs to report to the planner the set of posterior mean benefits and net costs,
\begin{align*}
    \left\{ (E_\pi[WTP_j], E_\pi[G_j]) \right\}_{j=1}^J.
\end{align*}
The planner---who knows $\mu$, the $\eta_j$'s, and $V$---can then use the posterior mean benefits and net costs to solve for the optimal local spending rule as described above.

I emphasize that posterior means depend on both sample estimates and uncertainty, where uncertainty arises from sample estimate noise and from the randomness of true policy impacts. Intuitively, posterior means shrink sample estimates towards the distribution of the true policy impacts, with more shrinkage for sample estimates that are noisier relative to the dispersion of the true policy impacts. This can be seen in, for example, Tweedie's formula for the posterior mean \citep{robbins1956empirical}. Thus for the researcher to be able to calculate posterior means, at a minimum they must observe sample estimates together with some measure of sampling noise, like standard errors. In fact, even if the distribution of the true policy impacts were known, without observing standard errors, it is not possible to do posterior mean shrinkage and thus optimal policy.

In practice, the researcher does not know the distribution of the true policy impacts and so cannot calculate posterior mean benefits and net costs. In Section \ref{sec3} I propose that the researcher instead report feasible empirical Bayes estimates of the posterior means to the planner, and show that the planner can approximate the optimal decision well using the empirical Bayes estimates in place of posterior mean benefits and net costs.

\paragraph{Sufficient statistic for direction}
For a wide range of policy-relevant consideration sets $V$, the researcher who does not know the planner's parameters $\mu$, $\{\eta_j\}$, or $V$ can report to the planner a simpler statistic, namely the ratios of posterior mean benefit to posterior mean net cost and the signs of the posterior mean net costs (or benefits) for all policies. 
For the planner, who knows $\mu$, $\{\eta_j\}$, and $V$, those ratios and signs are sufficient to know whether the optimal local spending rule increases or decreases upfront spending on each policy for a subset of possible consideration sets $V$.
In particular, suppose $V = c\mathcal{B}_p$ for any $p \geq 1$ and scalar $c > 0$, where recall $\mathcal{B}_p$ denotes the $L^p$ unit ball. Let $v^*$ denote the optimal local spending rule, where for $L^p$ ball consideration sets I adopt the convention that $v_j^*=0$ whenever the $j$th component of the posterior expected gradient is zero. Then if either $E_\pi[G_j] > 0$ and $\eta_j > 0$ or $E_\pi[G_j] < 0$ and $\eta_j < 0$,
\begin{equation*}
    \frac{E_\pi[WTP_j]}{E_\pi[G_j]} \geq \frac{\mu}{\eta_j} ~\Rightarrow~ v^*_j \geq 0, \qquad \frac{E_\pi[WTP_j]}{E_\pi[G_j]} \leq \frac{\mu}{\eta_j} ~\Rightarrow~ v^*_j \leq 0.
\end{equation*}
If instead either $E_\pi[G_j] > 0$ and $\eta_j < 0$ or $E_\pi[G_j] < 0$ and $\eta_j > 0$, then
\begin{equation*}
    \frac{E_\pi[WTP_j]}{E_\pi[G_j]} \leq \frac{\mu}{\eta_j} ~\Rightarrow~ v^*_j \geq 0, \qquad \frac{E_\pi[WTP_j]}{E_\pi[G_j]} \geq \frac{\mu}{\eta_j} ~\Rightarrow~ v^*_j \leq 0.
\end{equation*}
Because $\mu$ and the $\eta_j$s are known by the planner, given a consideration set $V$ that is an $L^p$ ball, the ratio $E_\pi[WTP_j]/E_\pi[G_j]$ together with the sign of $E_\pi[G_j]$ (or the sign of $E_\pi[WTP_j]$) is a sufficient statistic that a researcher can report to the planner for whether the optimal spending rule increases or decreases spending on policy $j$.

\subsection{Upfront Versus Net Spending} \label{sec_2.4}

Throughout this paper I assume that the planner chooses changes to upfront spending $s_j$ on each policy $j$. One could instead imagine that the planner chooses \textit{net} spending on policies, equivalently the change in budget due to policy changes, which takes into account fiscal externalities in addition to upfront spending. 
To understand how the problem with net spending as the choice variable is different, let $p_j$ denote the change in net spending on each policy $j$, which I collect into a vector $p = (p_1, \dots, p_J)$. For policy changes that are local to zero, I can approximate $p_j$ by $s_jG_j$ for each policy $j$. 

The optimal local change to net spending can be summarized by the gradient of the net welfare impact with respect to the choice variable $p$ at zero. In the absence of statistical uncertainty, the problem with net spending as the choice variable is locally a reparameterization of the problem with upfront spending as the choice variable, using $p_j = s_jG_j$. So by the chain rule, restricting to $G_j \neq 0$, the gradient of net welfare impact with respect to $p$ at zero is
\begin{align*}
    \begin{pmatrix}
        \eta_1 \frac{WTP_1}{G_1} - \mu \\ \vdots \\ \eta_J \frac{WTP_J}{G_J} - \mu
    \end{pmatrix}.
\end{align*}
This formulation could be appealing because the gradient depends on the ratio of $WTP_j$ and $G_j$ for each policy $j$, which is exactly the marginal value of public funds (MVPF) for each policy $j$, as discussed in \cite{hendren2020unified}. 

With statistical uncertainty, there are several issues with taking net spending to be the choice variable. First, when net costs are noisily measured, choosing a target level of net spending is infeasible for the planner. This is because net spending $p_j$ depends in part on net cost $G_j$, which is unknown due to statistical uncertainty at the time of decision-making. 
Additionally, a small change in net spending may not correspond to a local policy change when the policy has net cost $G_j$ close to zero. Crucially, \cite{hendren2020unified} estimate that some policies ``pay for themselves'' and thus have net costs that are close to or equal to zero. For such policies, a small change in net spending means a large change in upfront spending, that is, a large policy change. For such large changes, existing empirical estimates seem unlikely to provide useful guidance.

A further difficulty under uncertainty arises due to the irregular behavior of expectations of ratios. Locally the planner's optimal choice of net spending is summarized by the posterior expected gradient, which involves terms that are a posterior expectation of a ratio of noisy parameters, $E_\pi \left[ \frac{WTP_j}{G_j} \right]$. These expected ratios can be statistically ill-behaved, that is, the posterior expectation may be infinite or undefined. To provide intuition for why, note that if a random variable $X$ has positive and right-continuous density at 0, $E \left[\frac{1}{X} \right]$ is either infinite or does not exist. So when net spending is the choice variable for a set of policies, some of which pay for themselves, the expected gradient is likely to not be well-defined. In contrast, when the planner chooses upfront spending, the gradient is a linear combination of noisy parameters and so the expected gradient is well-defined.

These issues highlight that the formulation of the planner's problem becomes delicate in settings with statistical uncertainty and that working with upfront spending as the planner's choice variable yields better behavior than using net spending as the choice variable.

\section{Empirical Bayes} \label{sec3}

In the previous section I considered an oracle planner who forms a posterior over policy impacts $\{(WTP_j, G_j)\}_{j=1}^{J}$, taking as their prior the correctly specified distribution of true policy impacts. In practice, however, it may be the case that the planner does not know the true distribution of those policy impacts, and so is not able to construct posterior means and derive the optimal local spending rule as before. In this section I propose an empirical Bayes approach to approximate the optimal but infeasible local spending rule of the oracle planner. The empirical Bayes approach uses an estimate of the distribution of policy impacts, together with sample estimates and their standard errors, to produce shrunk posterior mean estimates.

To estimate the distribution of policy impacts, I assume $(WTP_j, G_j)$ are drawn independently from an unknown prior that varies with policy type, and maintain the assumption that observed sample estimates $\{(\widehat{WTP}_j, \widehat{G}_j)\}_{j=1}^{J}$ are conditionally Gaussian and unbiased. Despite the unbiasedness of the sample estimates for the true policy impacts, I show in Section \ref{sec_3.2} that a sample plug-in approach---which solves for the optimal local spending rule by replacing posterior means with sample estimates---can perform poorly. This result is valid for any number of policies $J$, fixing the sampling noise of the sample estimates, and so is not resolved as the number of policies for which I observe sample estimates grows, unless the sample estimates grow arbitrarily precise. Thus, rather than simply plugging sample estimates into the planner's local problem, in this paper I develop an empirical Bayes approach to obtain asymptotically optimal decisions in the planner's local problem. 

I propose two different approaches to estimate posterior means. In the first approach (Section \ref{sec_3.3.1}), I assume the true distribution of policy impacts lies in a parametric class and estimate this distribution parametrically. 
In the second approach (Section \ref{sec_3.3.2}), I do not impose a parametric form for the true distribution of policy impacts and use an adaptation of modern empirical Bayes methods to nonparametrically estimate the distribution. These estimated distributions yield posterior mean estimates of benefit and net cost, which can be plugged into \eqref{eq:post_grad} to obtain estimates of the posterior expected gradient. The empirical Bayes local spending rule, either parametric or nonparametric, then solves the local problem of \eqref{eq:objective} with the respective estimated posterior expected gradient. In Section \ref{sec_3.5} I provide theoretical convergence results showing that the empirical Bayes approaches approximate the oracle planner's local problem arbitrarily well as the number of policies grows, with the nonparametric approach coming at the expense of slower rates for certain policy environments.

\subsection{Model}\label{sec_3.1}

I assume the true values of benefit and cost $(WTP_j, G_j)$ are independent random vectors across policies $j$. 
Benefits and net costs may systematically differ based on type of policy; for example, \cite{hendren2020unified} find that policies targeting children have systematically higher returns than policies targeting adults.
To account for this I assume that $(WTP_j, G_j)$ is drawn from a distribution that depends on characteristics of the type of policy $j$, which I denote $X_j$. 

The planner observes estimates $(\widehat{WTP}_j, \widehat{G}_j)$ of the benefit and net cost of each policy change $j$ from empirical studies, together with their covariance matrix $\Sigma_j$. Recall from \eqref{eq:likelihood1} that I assume estimates are independent and conditionally Gaussian. Going forward, I extend this assumption to also hold conditional on policy-type characteristics $X_j$.

I assume the following model for $(WTP_j, G_j)$ given $X_j = t$:
\begin{equation}\label{eq:likelihood_general}
    \begin{pmatrix} {WTP}_j \\ {G}_j \end{pmatrix} \Bigg\vert \, (X_j = t) \overset{\text{ind.}}{\sim} F_0\left(\cdot \mid t \right).
\end{equation}

By conditioning on $X_j$ and $\Sigma_j$ I take them to be known and fixed. In this paper I maintain that policies are independent conditional on policy-type characteristics and that policy-type characteristics are correctly specified. The independence assumption may be plausible in the local problem that is considered in this paper, where $WTP_j$ and $G_j$ describe the impact of a marginal change in policy $j$, holding spending on all other policies fixed. Furthermore, conditioning on policy-type characteristics allows for systematic similarities among related policies.
I leave the problem of dealing with dependent policies and with misspecified policy-type characteristics to future work.

For the theoretical results in this paper I impose assumptions on the data-generating process that I argue are economically reasonable. 
I first impose the following assumption, which assumes the tails of the distribution of true policy benefits and net costs are sufficiently thin. This assumption is reasonable if one believes that, for example, no single policy change has an impact on welfare or budget per unit of upfront spending as large as GDP.
\begin{assumption}\label{ass:compact}
    There exists a constant $S_0 > 0$ such that for all $p \geq 1$, $(E_{\theta \sim F_0}\left[\Vert \theta \Vert_2^p \mid X_j \right])^{1/p} \leq S_0 \sqrt{p}$ uniformly over the support of $X_j$.
\end{assumption}
I additionally impose an assumption on the social planner's preferences, which uniformly bounds the marginal welfare impact of closing the budget constraint $\mu$ and the welfare weight $\eta_j$ for each policy $j$ away from infinity. This assumption is reasonable if one thinks the planner's preferences are represented by finite Pareto weights for each individual in society.
\begin{assumption}\label{ass:eta}
    For all $j$, $\eta_j$ is uniformly bounded away from infinity, $\left\vert \eta_j \right\vert \leq M < \infty$, and $\mu < \infty$.
\end{assumption}

Finally, I impose the following regularity assumption:
\begin{assumption}\label{ass:bdd_sig}
    For all $j$ there exist constants $\underline{b}, \overline{b} > 0$ such that $\underline{b}I_2 \preceq \Sigma_j \preceq \overline{b}I_2$. \footnote{Recall that for square matrices $A$ and $B$, $A \preceq B$ means $B-A$ is positive semi-definite. Throughout this paper I use the notation $I_k$ to denote the $k \times k$ identity matrix.}
\end{assumption}
Assumption \ref{ass:bdd_sig} requires that the eigenvalues of the sample covariance matrices $\Sigma_j$ are uniformly bounded away from zero and infinity.

\subsection{Performance of Sample Plug-In Rule}\label{sec_3.2}

Under the model \eqref{eq:likelihood1} and \eqref{eq:likelihood_general}, the sample estimates $\widehat{WTP}_j$ and $\widehat{G}_j$ are unbiased for the true benefit and net cost $WTP_j$ and $G_j$ for each $j$. One might think that a sample plug-in approach---which solves the planner's local problem using a gradient constructed from using raw sample estimates in place of posterior means in \eqref{eq:post_grad}---would perform well because of this unbiasedness. In fact, for a fixed number of policies $J$, the sample plug-in local objective is consistent if the empirical estimates $(\widehat{WTP}_j,\widehat{G}_j)$ are consistent for all policies $j$. Usually consistency of the empirical estimates holds when the sample size of each empirical study that produces the estimates grows to infinity.

In this paper I consider a different thought experiment where the sample size of each empirical study is fixed and the number of policies $J$ grows. This scenario better approximates the problem of a policymaker who has access to noisy empirical estimates for many policies, but with limited precision for each policy. In Proposition \ref{prop:plug_in} below I derive lower bounds on two different measures of the gap between the oracle planner's local problem and the sample plug-in local problem. The lower bounds do not converge to zero as the number of policies $J$ increases, implying that the sample plug-in local spending rule can perform poorly relative to the oracle planner's optimal local spending rule. 

To analyze the local spending rule across different numbers of policies $J$, I need a sequence of consideration sets $\{V_J\}_{J \in \mathbb{N}}$. I also require suitable normalization of the planner's local objective so that the local problem is comparable across different $J$. Because the bounds I provide will be for sequences such that for all $J$, $V_J \subseteq \mathcal{B}_p \subseteq \mathbb{R}^J$ for some given $p \geq 1$, I normalize by the order of the largest possible expected local problem objective among directions in $\mathcal{B}_p \subseteq \mathbb{R}^J$. The following lemma characterizes the normalization factor.

\begin{lemma}\label{lem:order}
Suppose Assumptions \ref{ass:compact} and \ref{ass:eta} hold. Then
\begin{equation*}
    E\left[ \max_{v \in \mathcal{B}_p} \langle \nabla w, v \rangle \right] =\begin{cases}
        O\left(\sqrt{\log J}\right) & p = 1 \\
        O \left(J^{\frac{p-1}{p}} \min\left\{\sqrt{\frac{p}{p-1}}, \sqrt{\log J}\right\} \right) & 1 < p < \infty \\
        O(J) & p = \infty
    \end{cases}.
\end{equation*}
\end{lemma}
The proof of this and all subsequent results is available in Supplemental Appendix \ref{app:proofs}.

Let $\widehat{\nabla w}$ denote the sample plug-in gradient, constructed by plugging the sample estimates into the posterior expected gradient of \eqref{eq:post_grad}. Let $\hat v_J : \mathcal{Y} \to V_J$ denote the sample plug-in local spending rule, which, for every value of the sample estimates $Y_{1:J} \in \mathcal{Y}$, solves $\max_{v \in V} \langle \widehat{\nabla w}, v \rangle$.
Let $N_p$ denote the normalization factor given $p$ derived in Lemma \ref{lem:order}, that is, $N_p = \sqrt{\log J}$ for $p=1$, $N_p = J^{\frac{p-1}{p}} \min\left\{\sqrt{\frac{p}{p-1}}, \sqrt{\log J}\right\}$ for $p \in (1,\infty)$, and $N_p = J$ for $p = \infty$. 
In what follows all expectation and probability statements are conditional on $\Sigma_{1:J}$ and $X_{1:J}$, which I omit when unambiguous. 

In Proposition \ref{prop:plug_in} below I derive lower bounds on two different expressions. The first object,
\begin{align*}
    \frac{1}{N_p} \max_{v: \mathcal{Y} \to V_J} E \left[ \left\vert E [ \langle \nabla w, v(Y_{1:J}) \rangle - \langle \widehat{\nabla w}, v(Y_{1:J}) \rangle \vert Y_{1:J}] \right\vert \right],
\end{align*}
measures whether the objective of the sample plug-in local problem is uniformly close to the true local objective.
The second object,
\begin{align*}
    \frac{1}{N_p} \max_{v: \mathcal{Y} \to V_J} E \left[ E [ \langle \nabla w, v(Y_{1:J}) \rangle - \langle \nabla w, \hat{v}_J(Y_{1:J}) \rangle \vert Y_{1:J} ] \right],
\end{align*} 
measures how close to optimal the sample plug-in local spending rule is, as given by the difference between the local welfare improvement from the sample plug-in local spending rule and the maximal local welfare improvement.

\begin{proposition}
\label{prop:plug_in}
    Suppose Assumptions \ref{ass:compact}, \ref{ass:eta}, and \ref{ass:bdd_sig} hold. Let $K_p$ denote a positive constant that may depend on $p$ but does not depend on $J$. For any $p \in [1,\infty]$ let $V_J = \mathcal B_p \subseteq \R^J$ for each $J$. 
    \begin{enumerate}
        \item Objective of local problem:
    
    For each $p \in [1,\infty]$ there exist $F_0$, $\{\eta_j\}_{j=1}^J$, $\{\Sigma_j\}_{j=1}^J$, and $\mu$ such that
    \begin{align*}
        \frac{1}{N_p} \max_{v: \mathcal{Y} \to V_J} E \left[ \left\vert E [ \langle \nabla w, v(Y_{1:J}) \rangle - \langle \widehat{\nabla w}, v(Y_{1:J}) \rangle \vert Y_{1:J}] \right\vert \right] &\geq K_p ~ \text{for all } J \geq 4.
    \end{align*}
        \item Local spending rule:
        
    For each $p \in [1,\infty]$ there exist $F_0$, $\{\eta_j\}_{j=1}^J$, $\{\Sigma_j\}_{j=1}^J$, and $\mu$ such that
    \begin{align*}
        \frac{1}{N_p} \max_{v: \mathcal{Y} \to V_J} E \left[ E [ \langle \nabla w, v(Y_{1:J}) \rangle - \langle \nabla w, \hat v_J(Y_{1:J}) \rangle \vert Y_{1:J} ] \right] &\geq K_p ~ \text{for all } J \geq 4.
    \end{align*}
    \end{enumerate}
\end{proposition}
The first result holds because it is possible to construct a data-generating process such that each component of the sample plug-in gradient falls outside of the support of each component of the posterior expected gradient with sufficiently high probability. The proof intuition of the second result depends on $p$. For $p > 1$, note that if the $j$th component of the posterior expected gradient is strictly positive but the $j$th component of the sample plug-in gradient is strictly negative (or vice versa), the sample plug-in local spending rule $\hat v_J$ changes spending in the opposite direction as the optimal local spending rule on policy $j$. The result then follows because it is possible to construct a data-generating process such that these sign mistakes happen with sufficiently high probability. For $p=1$, the sample plug-in rule spends on the policy with the largest sample plug-in gradient component in magnitude, while the optimal rule spends on the policy with the largest posterior expected gradient component in magnitude. The result then follows because it is possible to construct a data-generating process such that mistakes in the chosen policy happen with high enough probability.

This result shows that in settings with limited information about the effect of each policy, using the sample plug-in rule can lead to suboptimal decisions. In contrast, I next show that an empirical Bayes approach can asymptotically match the performance of the oracle planner as the number of policies grows.

\subsection{Empirical Bayes Posterior Mean Estimation}\label{sec_3.3}

Motivated by the lack of convergence of the sample plug-in approach as the number of policies $J$ grows large, I instead propose an empirical Bayes approach. The empirical Bayes approach first estimates the unknown distribution of the true policy impacts and then approximates the oracle planner's local problem with estimates of the posterior means, obtained using the estimated distribution. I propose two different empirical Bayes approaches to estimate posterior means. In the parametric approach, I assume the unknown distribution lies in a parametric class and estimate the distribution parametrically. In the nonparametric approach, I assume that the location and scale of the distribution depends parametrically on policy-type characteristics $X_j$ but do not impose any additional parametric restrictions, similar to the conditional location-scale model proposed in \cite{chen2022empirical}.

\subsubsection{Parametric Empirical Bayes} \label{sec_3.3.1}

I assume there exists a parametric family $\{f_\beta : \beta \in B\}$ such that each true prior $F_0$ lies in the parametric family. In particular, I assume
\begin{equation}\label{eq:likelihood_param}
    F_0\left( \cdot \mid X_j \right) \text{ has density } f_{\beta_0} \left(\cdot \mid X_j \right) \in \{f_\beta\left(\cdot \mid X_j \right) : \beta \in B\}.
\end{equation}
Assumption \ref{ass:compact} requires that the true prior is sub-Gaussian.

To estimate posterior means under the parametric model of \eqref{eq:likelihood1}, \eqref{eq:likelihood_general}, and \eqref{eq:likelihood_param}, I must first estimate the unknown parameters in the model, $\beta_0$. For notational simplicity recall $Y_j \equiv (\widehat{WTP}_j, \widehat{G}_j)'$ and let $\theta_j \equiv (WTP_j, G_j)'$. I can estimate $\beta_0$ using standard conditional maximum likelihood estimation (MLE). The MLE $\widehat\beta_0$ is the estimate of $\beta_0$ that maximizes the log-marginal likelihood of the data $Y_j$, conditional on $X_j$. Specifically,
\begin{align*}
    \widehat\beta &\in \underset{\beta \in B}{\text{arg}\max} \frac{1}{J} \sum_{j=1}^J \log \int \varphi_{\Sigma_j}(Y_j - \theta) f_\beta \left(\theta \mid X_j \right) d\theta.
\end{align*}

The oracle posterior means of benefit and net cost, which I denote $WTP_j^*$ and $G_j^*$, are
\begin{align*}
    \begin{pmatrix} WTP_j^* \\ G_j^* \end{pmatrix} &= E_{\beta_0}\left[ \begin{pmatrix} WTP_j \\ G_j \end{pmatrix} \middle\vert \widehat{WTP}_j, \widehat{G}_j, X_j, \Sigma_j \right].
\end{align*}
The expectation with respect to $\beta_0$ emphasizes that these are the posterior means under the true parameter $\beta_0$. The parametric empirical Bayes posterior mean estimates, which I denote $\widehat{WTP}_j^*$ and $\widehat{G}_j^*$, are plug-in versions of the oracle posterior means, using estimate $\widehat{\beta}$:
\begin{align*}
    \begin{pmatrix} \widehat{WTP}_j^* \\ \widehat{G}_j^* \end{pmatrix} &= E_{\widehat \beta_0}\left[ \begin{pmatrix} WTP_j \\ G_j \end{pmatrix} \middle\vert \widehat{WTP}_j, \widehat{G}_j, X_j, \Sigma_j \right].
\end{align*}
Here the subscript emphasizes that these are posterior means computed as if the true parameter was the estimated parameter $\widehat \beta$. Posterior mean estimates for parametric distributions can be computed from Tweedie's formula.

\subsubsection{Nonparametric Empirical Bayes} \label{sec_3.3.2}

For a distribution $F_0$ normalized to have zero mean and identity covariance matrix, I assume the following model for $(WTP_j, G_j)$:
\begin{equation}\label{eq:likelihood2}
    \begin{pmatrix} {WTP}_j \\ {G}_j \end{pmatrix} = a(X_j; \alpha_0) + B(X_j; \Omega_0)^{1/2} \tau_j, \quad \tau_j \big\vert X_j, \Sigma_j \overset{\text{i.i.d.}}{\sim} F_0.
\end{equation}
Here the residuals $\tau_j \in \R^2$ are distributed according to common prior distribution $F_0$ that does not vary with $X_j$. Assumption \ref{ass:compact} requires that $F_0$ is sub-Gaussian. The function $a(X_j; \alpha_0)$, which is a parametric function of $X_j$ with parameter $\alpha_0$ that returns a vector in $\R^2$, shifts this distribution. The function $B(X_j; \Omega_0)$, which is a parametric function of $X_j$ with parameter $\Omega_0$ that returns a symmetric and positive definite scale matrix in $\R^{2 \times 2}$, scales this distribution. This model is similar to the conditional location-scale model proposed in \cite{chen2022empirical}, albeit for parametric location and scale functions.

To estimate posterior means under the model of \eqref{eq:likelihood1} and \eqref{eq:likelihood2}, I must first estimate the unknown parameters in the model, which are the prior $F_0$, the location parameter $\alpha_0$, and the scale parameter $\Omega_0$. For notational simplicity recall $Y_j \equiv (\widehat{WTP}_j, \widehat{G}_j)'$ and let $\theta_j \equiv (WTP_j, G_j)'$. Let $E_J[ \cdot | X_j]$ and $Var_J[\cdot | X_j]$ denote the sample mean and sample variance, respectively, conditional on $X_j$.
To estimate the location and scale parameters, notice that for each policy type $t$,
\begin{equation*}
    \begin{gathered}
        a(X_j; \alpha_0) = E[\theta_j \vert X_j ] = E[Y_j \vert X_j ], \\
        B(X_j; \Omega_0) = Var ( \theta_j \vert X_j ) = Var ( Y_j \vert X_j ) - E [ \Sigma_j \vert X_j ].
    \end{gathered}
\end{equation*}
Thus location estimate $\widehat\alpha$ and scale estimate $\widehat\Omega$ can be constructed using, for example, minimum-distance estimation:
\begin{equation}\label{eq:location_scale_ests}
    \begin{gathered}
        (\widehat\alpha,\widehat\Omega) \in \underset{\alpha,\Omega}{\text{arg}\min} ~ \left( \frac{1}{J} \sum_{j=1}^J g_j(\alpha,\Omega)\right)^\prime W \left( \frac{1}{J} \sum_{j=1}^J g_j(\alpha,\Omega)\right) , \\
        g_j(\alpha,\Omega) \equiv q(X_j) \otimes \begin{pmatrix} Y_j - a(X_j;\alpha) \\ vec \left( \left( Y_j - a(X_j; \alpha)\right) \left( Y_j - a(X_j; \alpha)\right)' - \Sigma_j - B(X_j; \Omega)\right) \end{pmatrix}
    \end{gathered}
\end{equation}
for some positive-definite weighting matrix $W$ and instrument vector $q(X_j)$, where $\otimes$ denotes the Kronecker product.
In practice $B(X_j;\widehat\Omega)$ may not be positive definite due to estimation error, although it will always be symmetric. In the empirical illustration I truncate the negative eigenvalues of each $B(X_j;\widehat\Omega)$ away from zero to produce positive definite scale matrix estimates.

To estimate the unknown prior $F_0$ and obtain posterior mean estimates, I will first transform the model to remove the location-scale transformation. The model of \eqref{eq:likelihood1} and \eqref{eq:likelihood2} is equivalent to
\begin{equation}\label{eq:likelihood_transf}
\begin{gathered}
    Z_j \big\vert \tau_j, X_j, \Sigma_j \overset{\text{ind.}}{\sim} N \left(\tau_j, \Psi_j \right), \quad \tau_j \big\vert X_j,\Sigma_j \overset{\text{i.i.d.}}{\sim} F_0, \qquad j = 1,\dots,J, \\
    Z_j \equiv B(X_j;\Omega_0)^{-1/2} \left(Y_j - a(X_j;\alpha_0) \right), \quad \Psi_j \equiv B(X_j;\Omega_0)^{-1/2}\Sigma_jB(X_j;\Omega_0)^{-1/2},
\end{gathered}
\end{equation}
which is an example of the multivariate heteroscedastic empirical Bayes model studied by \cite{soloff2024multivariate}. 
I therefore implement the nonparametric maximum likelihood estimation (NPMLE) method from \cite{soloff2024multivariate} to estimate $F_0$, but replacing the unknown $\alpha_0$ and $\Omega_0$ with estimates $\widehat\alpha$ and $\widehat\Omega$.
The NPMLE $\widehat{F}_J$ is the estimate of $F_0$ that maximizes the log-likelihood of the transformed data $\widehat{Z}_j$ under the following model:
\begin{equation}\label{eq:likelihood_transf_est}
\begin{gathered}
    \widehat{Z}_j \big\vert \tau_j,X_j,\Sigma_j \overset{\text{ind.}}{\sim} N(\tau_j, \widehat\Psi_j), \quad \tau_j|X_j,\Sigma_j \overset{\text{i.i.d.}}{\sim} F_0,  \\
    \widehat{Z}_j \equiv B(X_j;\widehat{\Omega})^{-1/2} \left(Y_j - a(X_j;\widehat\alpha) \right), \quad \widehat{\Psi}_j \equiv B(X_j;\widehat{\Omega})^{-1/2}\Sigma_j B(X_j;\widehat{\Omega})^{-1/2}.
\end{gathered}
\end{equation}
Specifically, for $\mathcal{P}(\mathbb{R}^2)$ the set of all probability distributions supported on $\mathbb{R}^2$ and $\varphi_{A}(\cdot)$ the density of a Gaussian random vector with mean 0 and covariance matrix $A$,
\begin{align*}
    \widehat{F}_J \in \underset{F \in \mathcal{P}(\R^2)}{\text{arg}\max} ~ \frac{1}{J} \sum_{j=1}^J \log \int \varphi_{\widehat\Psi_j}\left( \widehat{Z}_j - \theta \right) dF(\theta).
\end{align*}
In practice I approximate the above maximization problem by replacing $\mathcal{P}(\mathbb{R}^2)$ with the collection of distributions supported on a finite grid \citep{koenker2014convex}.
\cite{soloff2024multivariate} provide a Python package \texttt{npeb} to implement the NPMLE estimation procedure.

The oracle posterior means of benefit and net cost for policy $j$, which I denote $WTP_j^*$ and $G_j^*$, are
\begin{align*}
    \begin{pmatrix} WTP_j^* \\ G_j^* \end{pmatrix} &= a(X_j;\alpha_0) + B(X_j;\Omega_0)^{1/2} E_{F_0, \alpha_0, \Omega_0}\left[ \tau_j \middle\vert \widehat{WTP}_j, \widehat{G}_j, X_j, \Sigma_j \right].
\end{align*}
The expectation with respect to $F_0, \alpha_0$, and $\Omega_0$ emphasizes that these are the posterior means under the true common prior $F_0$ and true location-scale parameters $\alpha_0, \Omega_0$. 
The nonparametric empirical Bayes posterior mean estimates, which I denote $\widehat{WTP}_j^*$ and $\widehat{G}_j^*$, are plug-in versions of the oracle posterior means, using estimates $\widehat{F}_J$, $\widehat\alpha$, and $\widehat\Omega$:
\begin{align*}
    \begin{pmatrix} \widehat{WTP}_j^* \\ \widehat{G}_j^* \end{pmatrix} &= a(X_j;\widehat\alpha) + B(X_j;\widehat\Omega)^{1/2} E_{\widehat F_J, \widehat\alpha, \widehat\Omega}\left[ \tau_j \middle\vert \widehat{WTP}_j, \widehat{G}_j, X_j, \Sigma_j \right].
\end{align*}
Here the subscript emphasizes that these are posterior means computed as if the true parameters were the estimated prior $\widehat F_J$ and estimated location-scale parameters $\widehat{\alpha}, \widehat\Omega$.
Posterior mean estimates $E_{\widehat F_J, \widehat\alpha, \widehat\Omega}\left[ \tau_j \middle\vert \widehat{WTP}_j, \widehat{G}_j, X_j, \Sigma_j \right]$ can easily be calculated with the Python package \texttt{npeb}.

\subsection{Oracle and Empirical Bayes Rules}\label{sec_3.4}
Using either the parametrically or the nonparametrically estimated posterior means from the previous subsection, I construct the empirical Bayes local spending rule to approximate the oracle planner's local spending rule.
The oracle planner solves their local problem given a consideration set $V$:
\begin{align*}
    \max_{v \in V} \langle \nabla w^*, v \rangle, \qquad \nabla w^* \equiv \begin{pmatrix}
        \eta_1 WTP_1^* - \mu G_1^* \\ \vdots \\ \eta_J WTP_J^* - \mu G_J^*
    \end{pmatrix}.
\end{align*}
The \textit{empirical Bayes local spending rule} solves the oracle planner's local problem but plugs in empirical Bayes posterior mean estimates for the oracle posterior means. The \textit{parametric empirical Bayes local spending rule} plugs in parametric empirical Bayes posterior mean estimates, while the \textit{nonparametric empirical Bayes local spending rule} plugs in nonparametric empirical Bayes posterior mean estimates. Specifically, the empirical Bayes local spending rule $\hat v^*: \mathcal{Y} \to V$ solves, for every possible value of the sample estimates $Y_{1:J} \in \mathcal{Y}$,
\begin{align}
    \max_{v \in V} \langle \widehat{\nabla w}^*, v \rangle, \qquad 
    \widehat{\nabla w}^* = \begin{pmatrix}
        \eta_1 \widehat{WTP}_1^* - \mu \widehat{G}_1^* \\ \vdots \\ \eta_J \widehat{WTP}_J^* - \mu \widehat{G}_J^*
    \end{pmatrix}. \label{eq:eb_est_grad}
\end{align}

As discussed at the end of Section \ref{sec_2.3}, the empirical Bayes posterior mean estimates are a sufficient statistic for the oracle planner who knows $\mu$, welfare weights $\eta_j$, and consideration set $V$ to solve for the empirical Bayes local spending rule. Moreover, by the reasoning at the end of Section \ref{sec_2.3}, if the planner's consideration set is an $L^p$ ball, reporting the ratios of shrunk benefits to shrunk net costs together with the signs of the shrunk net costs suffices for the planner to know whether the empirical Bayes local spending rule increases or decreases spending on each policy.

\subsection{Performance of Empirical Bayes Local Spending Rule}\label{sec_3.5}

How well does the empirical Bayes local spending rule perform relative to the optimal but infeasible local spending rule of an oracle planner? Results in existing studies, like \cite{soloff2024multivariate}, suggest that the posterior mean estimates $\widehat{WTP}_j^*$ and $\widehat{G}_j^*$ will approximate the oracle posterior means $WTP_j^*$ and $G_j^*$ well on average over all policies $j$ under the mean squared error criterion as the number of policies $J$ grows. However, the criterion for a well-performing local spending rule is not the minimization of mean squared error but the maximization of net welfare impact achieved by solving the planner's local problem.

In this section I derive rates of convergence for the gap between the oracle planner's local problem and the empirical Bayes approximation of the local problem. These rates characterize the asymptotic performance of the empirical Bayes local spending rule relative to the oracle local spending rule as the number of policies grows large. They are valid over a large class of data-generating processes and consideration sets. I find that the parametric approach converges to the oracle planner's local problem at a $1/\sqrt{J}$ rate, while the nonparametric approach avoids some of these parametric restrictions at the cost of a slower rate of convergence for certain consideration sets.

\subsubsection{Parametric Empirical Bayes Local Spending Rule}

I first derive rates of convergence for the parametric empirical Bayes local spending rule. To do so, I need to impose additional assumptions on the estimators used in Section \ref{sec_3.3.1} to obtain the parametric empirical Bayes local spending rule.

\begin{assumption}\label{ass:param}
    \begin{enumerate}
        \item Parameter space $B$ is compact and convex, and both $\beta_0 \in B$ and $\widehat\beta \in B$. Furthermore, there exists a constant $C_1$ such that $E[\Vert \widehat\beta-\beta_0 \Vert_2^2] \leq C_1/J$.
        \item Prior $f_\beta$ is continuously differentiable in $\beta$, and 
        \begin{align*}
            \int \left(1 + \Vert \theta \Vert \right) \sup_{\beta \in B} \left\Vert \nabla_\beta f_\beta(\theta \vert X_j) \right\Vert d\theta < \infty
        \end{align*}
        uniformly over the support of $X_j$.
        \item Prior $f_\beta$ is sub-Gaussian uniformly over the support of $X_j$ and over $\beta \in B$. That is, there exists a constant $C_2$ such that for all $\beta \in B$, $(E_\beta [\Vert \theta_j \Vert_2^p \mid X_j])^{1/p} \leq C_2 \sqrt{p}$ uniformly over the support of $X_j$.
        \item There exists a constant $C_3$ such that
            \begin{align*}
                \sup_{\beta \in B} E_\beta \left[ \left( \frac{p_{\beta_0,j}(Y_j|X_j)}{p_{\beta,j}(Y_j|X_j)} \right)^2 \mid X_j \right] &\leq C_3
            \end{align*}
            uniformly over the support of $X_j$, where $p_{\beta,j}(y|X_j) = \int \varphi_{\Sigma_j}(y-\theta) f_\beta(\theta \vert X_j) d\theta$ denotes the marginal density of $Y_j$ given $X_j$.
        \item There exists some $k_s \in (0,2]$ and constant $C_4$ such that for all $p \geq 1$ and for all $\beta \in B$, $(E_\beta [\Vert \nabla_\beta \log f_{\beta}(\theta_j \vert X_j) \Vert_2^p] \vert X_j])^{1/p} \leq C_4 p^{1/k_s}$ uniformly over the support of $X_j$.
        \item There exist constants $C_5$ and $k_m \in (0,2]$ such that for all $q \geq 1$ and $\beta,\beta' \in B$,
        \begin{align*}
            \left( E_{\beta_0} \left[ \Vert \nabla_\beta m_j(Y_j,\beta) - \nabla_\beta m_j(Y_j,\beta') \Vert_{op}^q \vert X_j \right] \right)^{1/q} &\leq C_5 q^{1/k_m} \Vert \beta - \beta' \Vert_2
        \end{align*}
        uniformly over $j$ and the support of $X_j$, where $m_j(y,\beta) = E_\beta[\theta_j|Y_j = y,X_j]$.
    \end{enumerate}
\end{assumption}

Assumption \ref{ass:param} imposes regularity assumptions on the parametric family and the parameter estimate. Assumption \ref{ass:param}(1) requires that both the true and estimated parameter live in compact parameter space $B$ and that the estimator has mean squared error of order $1/J$, which is satisfied by standard conditional MLE under regularity conditions. Assumption \ref{ass:param}(2) imposes regularity conditions on the prior density $f_\beta$ to allow passing the derivative through the integrals defining the posterior mean. Assumption \ref{ass:param}(3) requires that $f_\beta$ is uniformly sub-Gaussian over the parameter space $B$ and the support of $X_j$. Assumption \ref{ass:param}(4) requires the marginal distribution of $Y_j$ under any candidate parameter $\beta$ be sufficiently close to the distribution under the true parameter $\beta_0$. Assumption \ref{ass:param}(5) imposes that the prior score has sufficiently thin tails, uniformly over the parameter space $B$ and the support of $X_j$. While this assumption is non-standard, in Supplemental Appendix \ref{app:param} I show that it holds for many commonly used parametric classes, including many exponential family distributions. Finally, Assumption \ref{ass:param}(6) requires that the posterior mean be sufficiently smooth in the parameter $\beta$.

The above assumptions, in addition to Assumptions \ref{ass:compact}, \ref{ass:eta}, and \ref{ass:bdd_sig} from earlier, specify a class of prior distributions, location and scale estimators, and planner preference parameters that are governed by a set of hyperparameters, $\mathcal H = (S_0, M, \mu, \underline{b}, \overline b, B, C_1, C_2, C_3, k_s, C_4, k_m, C_5)$. The following rates are uniform over data-generating processes for a given $\mathcal H$. In what follows, I use the notation $x \lesssim_{\mathcal H} y$ to mean there exists some positive constant $C_{\mathcal H}$ that depends only on $\mathcal H$ such that $x \leq C_{\mathcal H} y$. I also define constant $k$ by $\frac{1}{k} = \max \left\{ \frac{1}{2} + \frac{1}{k_s}, \frac{1}{k_m} \right\}$.

Given a sequence $\{V_J\}_{J \in \mathbb{N}}$ I define for each $J$ the empirical Bayes local spending rule $\hat{v}_J^*: \mathcal{Y} \to V_J$ as the solution to, for every possible value of the sample estimates $Y_{1:J} \in \mathcal{Y}$, $\max_{v \in V_J} \langle \widehat{\nabla w}^*, v \rangle$, breaking ties with a fixed deterministic tie-breaking rule.

\begin{theorem}\label{thm:regret_rates_param}
    Suppose that Assumptions \ref{ass:compact}, \ref{ass:eta}, \ref{ass:bdd_sig}, and \ref{ass:param} hold under the parametric model of \eqref{eq:likelihood1}, \eqref{eq:likelihood_general}, and \eqref{eq:likelihood_param}; that for some $p \geq 1$ it holds that $V_J \subseteq \mathcal B_p \subseteq \R^J$ for each $J$; and that $J \geq 3$. Further suppose posterior means are estimated parametrically, following Section \ref{sec_3.3.1}.
    
    \begin{enumerate}
        \item Objective of local problem:
    \begin{multline*}
    \frac{1}{N_p} \max_{v: \mathcal{Y} \to V_J} E \left[ \left\vert E_{\beta_0} [ \langle \nabla w, v(Y_{1:J}) \rangle - \langle \widehat{\nabla w}^*, v(Y_{1:J}) \rangle \vert Y_{1:J}] \right\vert \right] \\
    \lesssim_{\mathcal H} 
    \begin{cases} 
        J^{-\frac{1}{2}}(\log J)^{\frac{1}{k}-\frac{1}{2}} & p = 1 \\
        J^{-\frac{1}{2}}\left(\frac{p}{p-1}\right)^{\frac{1}{k}} \frac{1}{\sqrt{\min\{\frac{p}{p-1}, \log J\}}} & 1 < p < \infty \\
        J^{-\frac{1}{2}} & p = \infty
    \end{cases}.
    \end{multline*}
        \item Local spending rule:
    \begin{multline*}
    \frac{1}{N_p} \max_{v: \mathcal{Y} \to V_J} E \left[ E_{\beta_0} [ \langle \nabla w, v(Y_{1:J}) \rangle - \langle \nabla w, \hat{v}_J^*(Y_{1:J}) \rangle \vert Y_{1:J} ] \right] \\
    \lesssim_{\mathcal H} 
    \begin{cases} 
        J^{-\frac{1}{2}}(\log J)^{\frac{1}{k}-\frac{1}{2}} & p = 1 \\
        J^{-\frac{1}{2}}\left(\frac{p}{p-1}\right)^{\frac{1}{k}} \frac{1}{\sqrt{\min\{\frac{p}{p-1}, \log J\}}} & 1 < p < \infty \\
        J^{-\frac{1}{2}} & p = \infty
    \end{cases}.
    \end{multline*}
    \end{enumerate}
\end{theorem}

To prove this theorem, I bound the left-hand side of results 1 and 2 above by a function of the average $p/(p-1)$ norm error of the empirical Bayes posterior mean estimates $\widehat{WTP}_j^*$ and $\widehat{G}_j^*$. I then derive an upper bound rate for this average $p/(p-1)$ norm error under the assumptions made earlier on the distribution of the data.

The theorem shows that under the (correctly specified) parametric model, the parametric empirical Bayes approach approximates the oracle planner well for any $p \geq 1$. The parametric empirical Bayes approach approximates the oracle planner well in two different ways: the objective of the empirical Bayes local problem uniformly approaches the true local objective in expectation at a $1/\sqrt{J}$ rate up to log factors (result 1), and the local welfare improvement from the empirical Bayes local spending rule approaches the maximal local welfare improvement in expectation at a $1/\sqrt{J}$ rate up to log factors (result 2).

\subsubsection{Nonparametric Empirical Bayes Local Spending Rule}

Next I derive rates of convergence for the nonparametric empirical Bayes local spending rule. To do so, I need to impose additional assumptions on the estimators used in Section \ref{sec_3.3.2} to obtain the nonparametric empirical Bayes local spending rule.

\begin{assumption}\label{ass:est_var}
    \begin{enumerate}
        \item There exists a constant $\overline a > 0$ such that $\left\Vert a(X_j;\alpha) \right\Vert_\infty \leq \overline{a}$ uniformly over $\alpha$ and the support of $X_j$. There also exist constants $\underline{c}, \overline{c} > 0$ such that $\underline{c}I_2 \preceq B(X_j;\Omega) \preceq \overline{c}I_2$ uniformly over $\Omega$ and the support of $X_j$.
        \item The estimated location and scale functions satisfy $Pr\left((a(\cdot; \widehat\alpha), B(\cdot;\widehat\Omega)^{1/2}) \in \mathcal C \right) = 1$, where $\mathcal{C}$ is a class of function pairs with metric entropy bound $\log N(\varepsilon, \mathcal C, \Vert \cdot \Vert_\infty) \leq C_{\mathcal C} \log \left( \frac{C_{\mathcal C}}{\varepsilon} \right)$ for some constant $C_{\mathcal C}$. I define $\left\Vert (a,B^{1/2}) \right\Vert_\infty \equiv \sup_x \max \left\{ \left\Vert a(x) \right\Vert_\infty, \left\Vert B(x)^{1/2} \right\Vert_{op} \right\}$.
        \item There exist constants $\underline{k},\overline{k} > 0$ such that for all $j = 1,\dots, J$, $\underline{k}I_2 \preceq \Psi_j \preceq \overline{k}I_2$. 
        \item There exist constants $C_1, C_2 > 0$ such that for all $J$,
    \begin{align*}
        P \left( \Vert \widehat{\chi} - \chi_0 \Vert_J > C_1 \sqrt{\frac{\log J}{J}} \right) &\leq \frac{C_2}{J^2},
    \end{align*}
    where for $\chi = (\alpha, \Omega)$ I define $$\Vert \chi - \tilde\chi \Vert_J \equiv \max_{1\leq j \leq J} \max \left\{\Vert a(X_j;\alpha) - a(X_j;\tilde\alpha) \Vert_\infty, \Vert B(X_j;\Omega)^{1/2} - B(X_j;\tilde\Omega)^{1/2} \Vert_{op} \right\}.\footnote{In Supplemental Appendix \ref{app:est_rates} I show that the estimators of \eqref{eq:location_scale_ests} satisfy this estimation rate under additional regularity assumptions.}$$
    \end{enumerate} 
\end{assumption}

\begin{assumption}\label{ass:good_approx}
    Estimated prior $\widehat{F}_J$ satisfies 
    \begin{align*}
        \frac{1}{J} \sum_{j=1}^J \psi_j(Z_j, \widehat\alpha, \widehat\Omega, \widehat{F}_J) \geq \sup_{F} \frac{1}{J} \sum_{j=1}^J \psi_j(Z_j, \widehat\alpha, \widehat\Omega, F) - \kappa_J
    \end{align*}
    for tolerance $\kappa_J = \frac{3}{J} \log \left( \frac{J}{(2\pi e)^{1/3}} \right)$, where 
    \begin{align*}
        \psi_j(Z_j, \widehat{\alpha}, \widehat\Omega, F) &\equiv \log \left(\int \varphi_{\widehat\Psi_j}\left(\widehat{Z}_j-\tau\right)dF(\tau) \right),\\
        \varphi_{\widehat\Psi_j}(x) &= \exp \left(-\frac{1}{2}x^T \widehat\Psi_j^{-1} x \right).
    \end{align*}
    Moreover, for some $C_3 <\infty$, $\widehat F_J$ is supported on $\left\{ \tau \in \R^2: \Vert \tau \Vert_2 \leq C_3 \max_{1 \leq j \leq J}\{ \Vert \widehat Z_j \Vert_2, 1\} \right\}$.
\end{assumption}

Assumption \ref{ass:est_var} requires that the location and scale are uniformly bounded and live in a class with controlled metric entropy, 
that the eigenvalues of the normalized sample estimate variances are uniformly bounded, and that the location and scale estimators perform well. Assumption \ref{ass:good_approx} requires that the prior estimate is an approximate maximizer of the log-likelihood of the residualized data $\widehat{Z}_j$. These are regularity assumptions that are similar to those used in the literature, with Assumption \ref{ass:est_var} similar to assumptions in \cite{chen2022empirical} and Assumption \ref{ass:good_approx} satisfied by the NPMLE estimator proposed by \cite{soloff2024multivariate} with appropriate choice of discretization rate\footnote{See Proposition 6 and Section 4.1.1 of \cite{soloff2024multivariate} for a discussion of how to choose the discretization rate for the support of an approximate NPMLE so that statistical requirements like Assumption \ref{ass:good_approx} are satisfied.}. In the statement of the theorem I will assume, among other things, that $J \geq 7$, which is sufficient for $\kappa_J$ to be positive. 

The above assumptions, in addition to Assumptions \ref{ass:compact}, \ref{ass:eta}, and \ref{ass:bdd_sig} from earlier, specify a class of prior distributions, location and scale estimators, and planner preference parameters that are governed by a set of hyperparameters, $\mathcal H = (S_0, M, \mu, \overline a, \underline{b}, \overline b, \underline{c}, \overline c, \underline{k}, \overline{k}, C_{\mathcal C}, C_1, C_2, C_3)$. The following rates are uniform over data-generating processes for a given $\mathcal H$. In what follows, I use the notation $x \lesssim_{\mathcal H} y$ to mean there exists some positive constant $C_{\mathcal H}$ that depends only on $\mathcal H$ such that $x \leq C_{\mathcal H} y$.

Given a sequence $\{V_J\}_{J \in \mathbb{N}}$ I define for each $J$ the empirical Bayes local spending rule $\hat{v}_J^*: \mathcal{Y} \to V_J$ as the solution to, for every possible value of the sample estimates $Y_{1:J} \in \mathcal{Y}$, $\max_{v \in V_J} \langle \widehat{\nabla w}^*, v \rangle$, breaking ties with a fixed deterministic tie-breaking rule.

\begin{theorem}\label{thm:regret_rates}
    Suppose that Assumptions \ref{ass:compact}, \ref{ass:eta}, \ref{ass:bdd_sig}, \ref{ass:est_var}, and \ref{ass:good_approx} hold under the nonparametric model of \eqref{eq:likelihood1} and \eqref{eq:likelihood2}; that for some $p \geq 1$ it holds that $V_J \subseteq \mathcal B_p \subseteq \R^J$ for each $J$; and that $J \geq \max\{\frac{5}{\underline{k}}, 7\}$. Further suppose posterior means are estimated nonparametrically, following Section \ref{sec_3.3.2}.

    \begin{enumerate}
        \item Objective of local problem: If $p \in [1,2)$, 
    \begin{align*}
    \frac{1}{N_p} \max_{v: \mathcal{Y} \to V_J} E \left[ \left\vert E_{F_0,\alpha_0,\Omega_0} [ \langle \nabla w, v(Y_{1:J}) \rangle - \langle \widehat{\nabla w}^*, v(Y_{1:J}) \rangle \vert Y_{1:J}] \right\vert \right] &\lesssim_{\mathcal H} J^{-\frac{p-1}{p}}(\log J)^3
    \end{align*}
    and if $p \in [2, \infty]$, 
    \begin{align*}
    \frac{1}{N_p} \max_{v: \mathcal{Y} \to V_J} E \left[ \left\vert E_{F_0,\alpha_0,\Omega_0} [ \langle \nabla w, v(Y_{1:J}) \rangle - \langle \widehat{\nabla w}^*, v(Y_{1:J}) \rangle \vert Y_{1:J}] \right\vert \right] &\lesssim_{\mathcal H}  J^{-\frac{1}{2}}(\log J)^3.
    \end{align*}
        \item Local spending rule: If $p \in [1,2)$, 
        \begin{align*}
    \frac{1}{N_p} \max_{v: \mathcal{Y} \to V_J} E \left[ E_{F_0,\alpha_0,\Omega_0} [ \langle \nabla w, v(Y_{1:J}) \rangle - \langle \nabla w, \hat{v}_J^*(Y_{1:J}) \rangle \vert Y_{1:J} ] \right] &\lesssim_{\mathcal H}  J^{-\frac{p-1}{p}}(\log J)^3
    \end{align*}
    and if $p \in [2, \infty]$, 
    \begin{align*}
    \frac{1}{N_p} \max_{v: \mathcal{Y} \to V_J} E \left[ E_{F_0,\alpha_0,\Omega_0} [ \langle \nabla w, v(Y_{1:J}) \rangle - \langle \nabla w, \hat{v}_J^*(Y_{1:J}) \rangle \vert Y_{1:J} ] \right] &\lesssim_{\mathcal H} J^{-\frac{1}{2}}(\log J)^3.
    \end{align*}
    \end{enumerate}
\end{theorem}

The theorem shows that with a large number of policies, the empirical Bayes approach approximates the oracle planner well as long as the sequence of consideration sets $V_J$ can be written as a subset of the unit $L^p$ ball for some $p$ strictly greater than 1. The empirical Bayes approach approximates the oracle planner well in two different ways: the objective of the empirical Bayes local problem uniformly approaches the true local objective in expectation (result 1), and the local welfare improvement from the empirical Bayes local spending rule approaches the maximal local welfare improvement in expectation (result 2).

To prove this theorem, I bound the left-hand side of results 1 and 2 above by a function of the mean squared error of the empirical Bayes posterior mean estimates $\widehat{WTP}_j^*$ and $\widehat{G}_j^*$. I then derive an upper bound rate for the mean squared error, extending the proof of Theorem 1 in \cite{chen2022empirical} to the multivariate setting for the location-scale model with parametric estimation of location and scale. %The proof of this mean squared error result, available in Supplemental Appendix \ref{app:mse_proof}, may be of independent interest.

Note that for $p=1$ the upper bound rates go to infinity with $J$, so the theorem does not speak to how well the nonparametric empirical Bayes approach performs when $p=1$. The intuition for why nonparametric empirical Bayes can perform poorly when $p=1$ is that when $V_J = \mathcal{B}_1$, the optimal local spending rule only spends on the single policy with the largest posterior expected rate of increase in net welfare impact, while the empirical Bayes local spending rule spends on the single policy with the largest empirical Bayes estimated rate of increase in net welfare impact. However, nonparametric empirical Bayes ensures performance guarantees on average across all policies but not necessarily for any individual policy \citep[see, for example, Chapter 1.3 of][]{efron2012large}.

\section{Simulations} \label{sec4}

In this section I present simulations calibrated to the empirical illustration in Section \ref{sec5} to demonstrate the performance of the empirical Bayes method proposed in the previous section. These simulations are designed to study how the parametric and the nonparametric empirical Bayes methods perform relative to the sample plug-in approach at various sample sizes (numbers of policies). 

\subsection{Simulation Design} \label{sec_4.1}
As in the empirical illustration, I take the policy type characteristic $X_j$ to be a discrete random variable taking on three values: climate, child, and adult. For each sample size I randomly assign each policy to one of these three types such that the relative frequencies of policy types match those in the empirical illustration. Conditional on policy type $X_j = t$ I draw the true benefit and net cost according to 
\begin{align*}
    \begin{pmatrix} WTP_j \\ G_j \end{pmatrix} &= \alpha_t + \Omega_t^{1/2} \tau_j, \qquad \tau_j \overset{\text{i.i.d.}}{\sim} N(0,I_2),
\end{align*}
where I set $\alpha_t$ and $\Omega_t$ equal to the estimated type-specific mean and covariance matrix from the parametric empirical Bayes approach in the empirical illustration. 

To reproduce the heteroskedasticity in the empirical estimates, for each simulated policy I randomly draw a sampling covariance matrix $\Sigma_j$, with replacement, from the covariance matrices observed for policies of the same type. I then draw the sample estimate
\begin{align*}
    \begin{pmatrix} \widehat{WTP}_j \\ \widehat{G}_j \end{pmatrix} \bigg\vert \begin{pmatrix} WTP_j \\ G_j \end{pmatrix}, \Sigma_j \overset{\text{ind.}}{\sim} N \left( \begin{pmatrix} WTP_j \\ G_j \end{pmatrix}, \Sigma_j \right).
\end{align*}
In the simulation I set all $\eta_j = 1$ and $\mu = 1$. I show results for consideration sets $V = \mathcal{B}_1, \mathcal{B}_2$, and $\mathcal{B}_\infty$  (that is, $p \in \{1,2,\infty\}$), and for sample sizes $J \in \{30, 60, 100, 150, 250, 400, 600\}$. I use 500 Monte Carlo replications at each sample size. 

For the parametric empirical Bayes procedure, I proceed under the correctly specified parametric model and estimate posterior means as in the empirical illustration, described in Supplemental Appendix \ref{app:data}. For the nonparametric empirical Bayes procedure, I estimate posterior means as in the empirical illustration, described in Supplemental Appendix \ref{app:data}. The sample plug-in procedure uses the observed sample estimates in place of the posterior means.

I evaluate the procedures using the two standardized criteria described in Section \ref{sec3}. The first criterion measures how well the estimated objective of the local problem uniformly approximates the oracle local objective. The second criterion measures how close to optimal the estimated local spending rule is relative to the oracle.

\subsection{Results}\label{sec_4.2}

\begin{figure}[!thbp]
\caption{Simulation results}
\label{fig:sims}
\centering
\begin{subfigure}[t]{0.49\textwidth}
  \centering
  \caption{Criterion 1, $p=1$}
  \includegraphics[width=\linewidth]{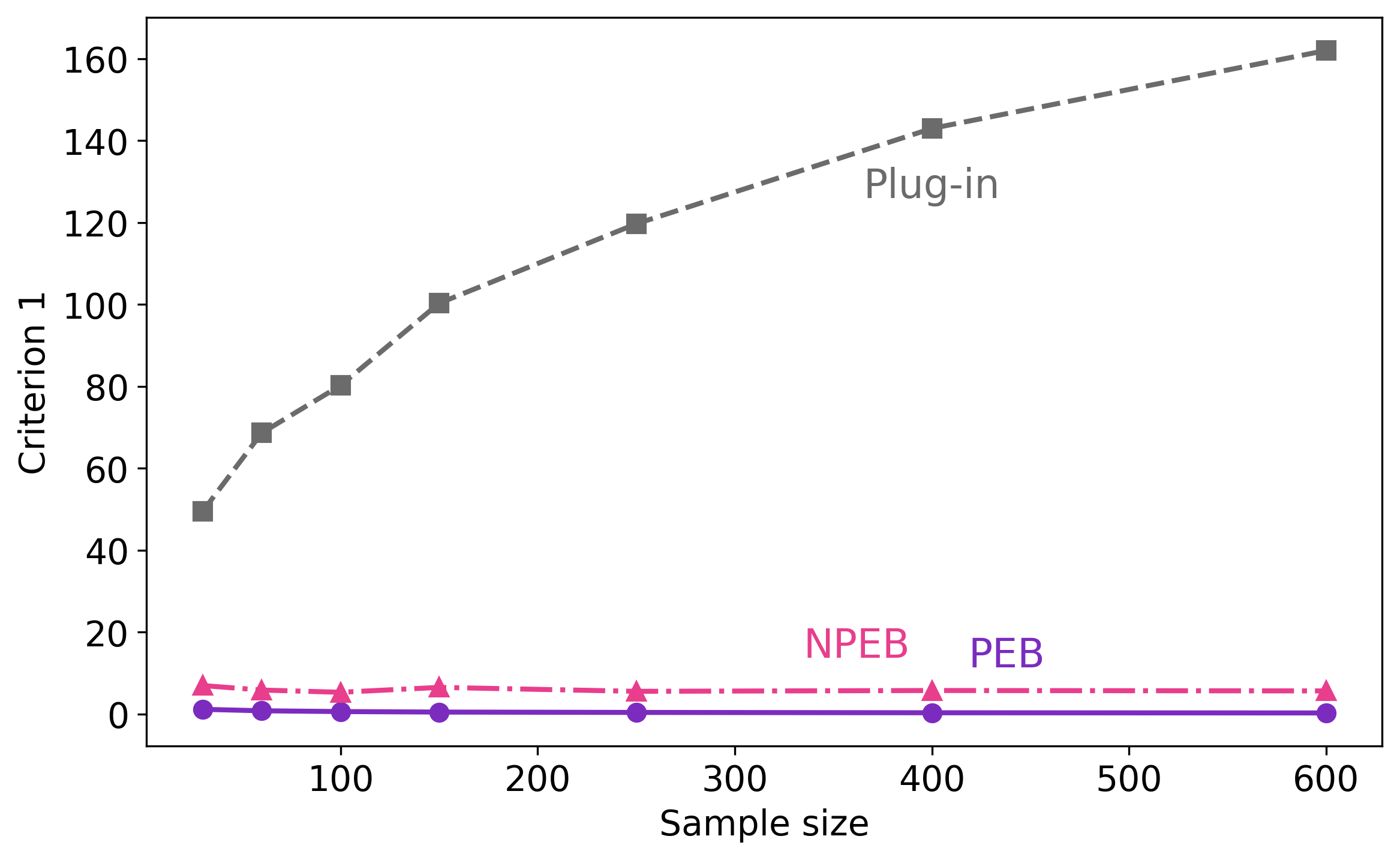}
\end{subfigure}
\hfill
\begin{subfigure}[t]{0.49\textwidth}
  \centering
  \caption{Criterion 2, $p=1$}
  \includegraphics[width=\linewidth]{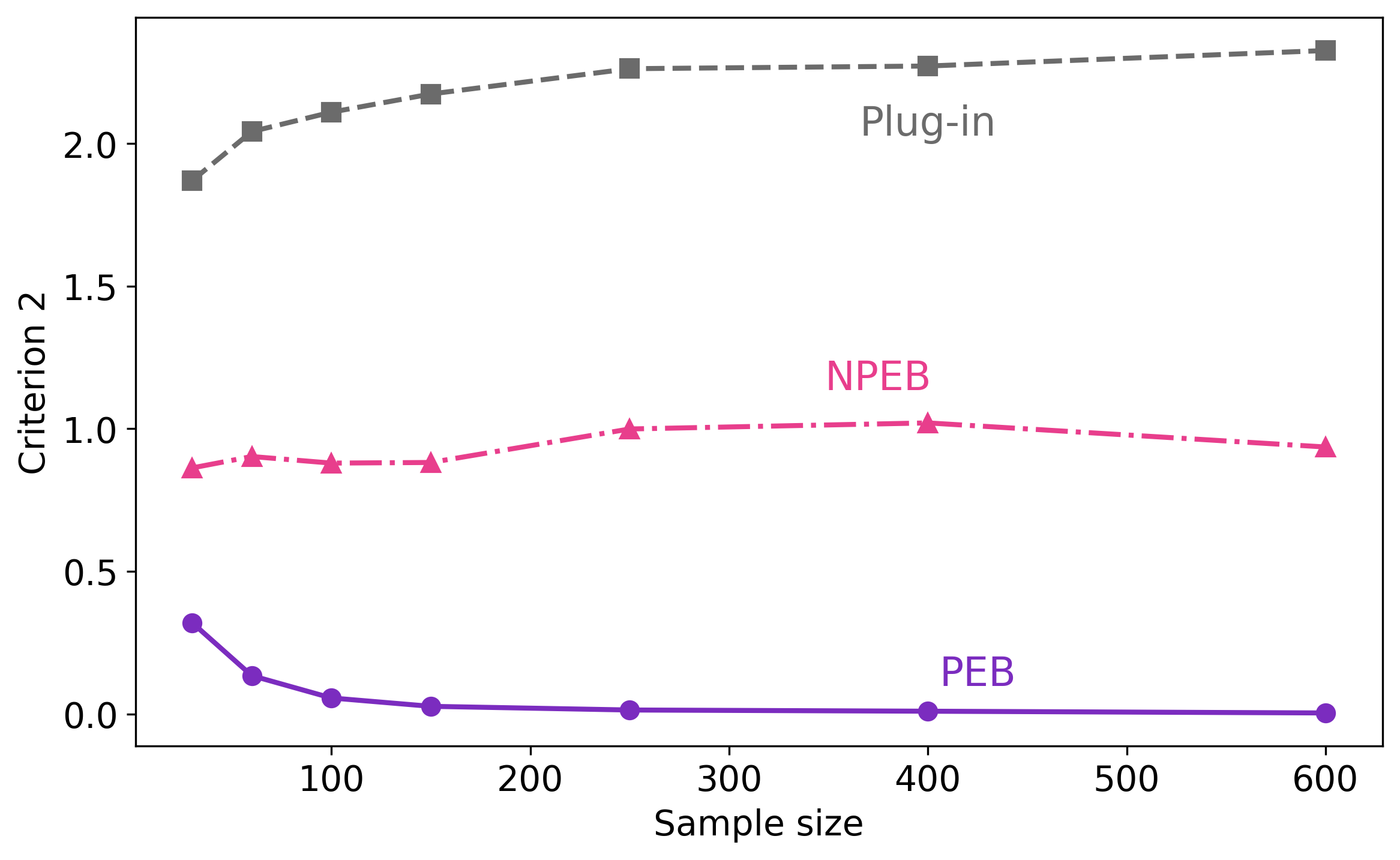}
\end{subfigure}
\begin{subfigure}[t]{0.49\textwidth}
  \centering
  \caption{Criterion 1, $p=2$}
  \includegraphics[width=\linewidth]{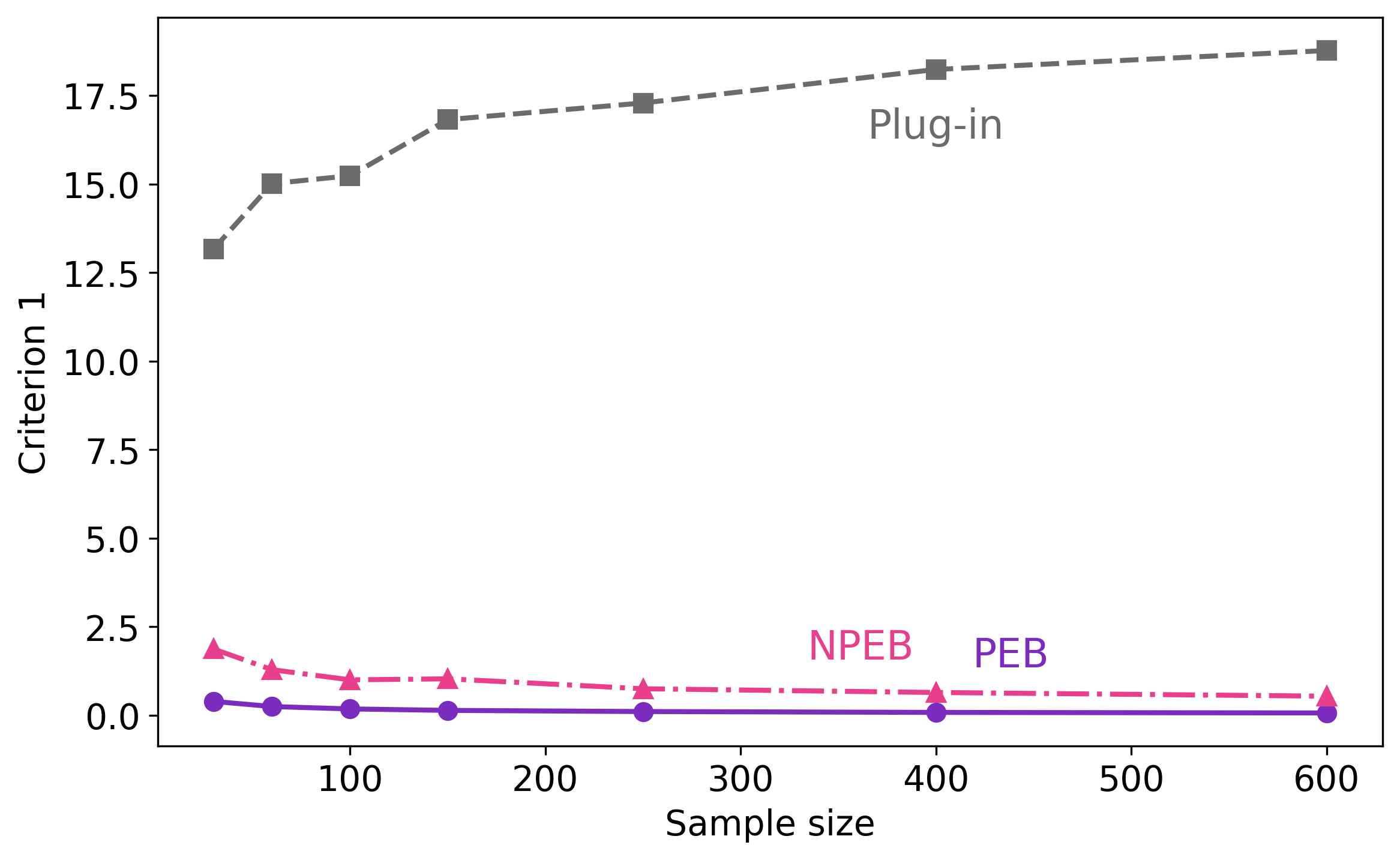}
\end{subfigure}
\hfill
\begin{subfigure}[t]{0.49\textwidth}
  \centering
  \caption{Criterion 2, $p=2$}
  \includegraphics[width=\linewidth]{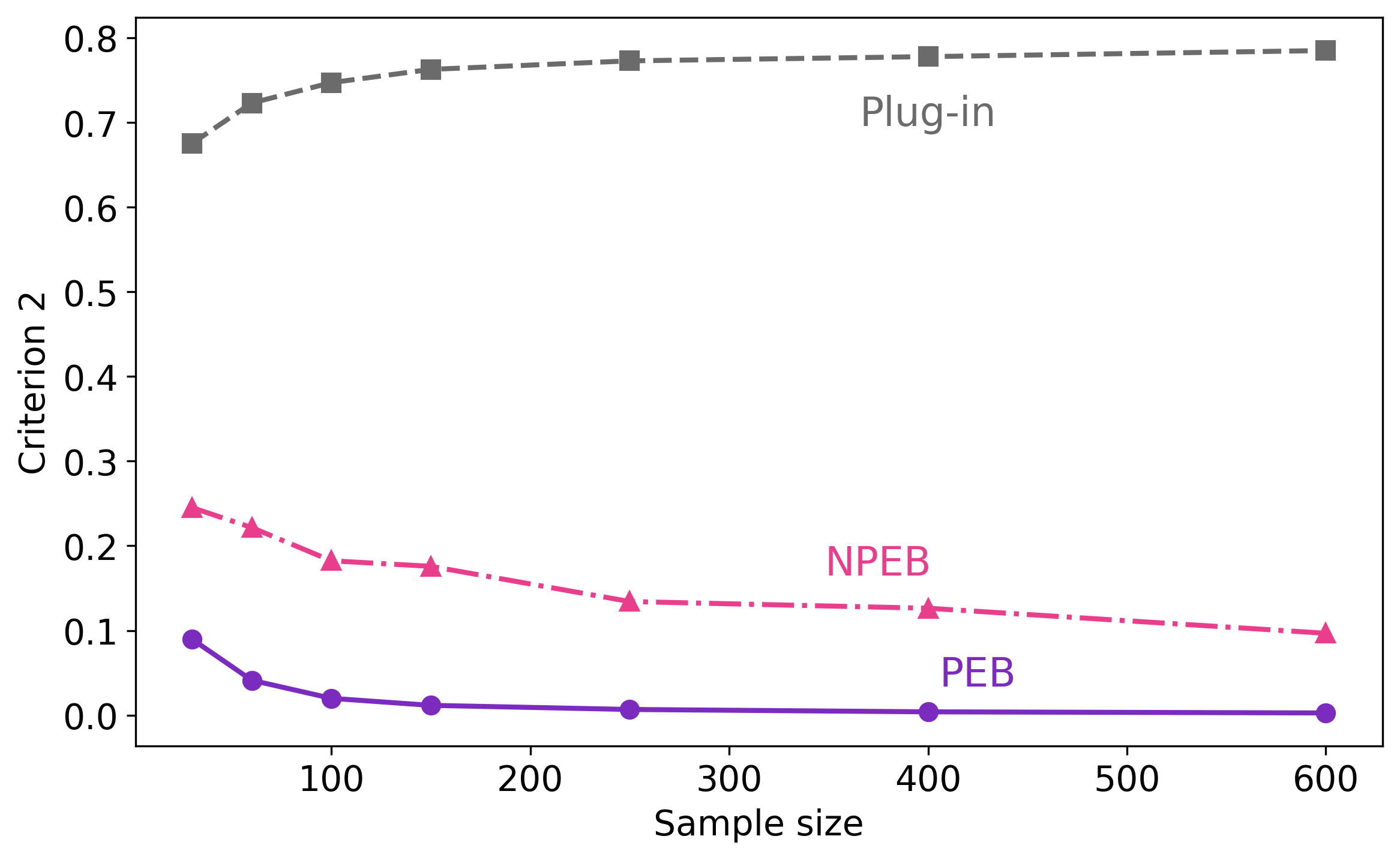}
\end{subfigure}
\begin{subfigure}[t]{0.49\textwidth}
  \centering
  \caption{Criterion 1, $p=\infty$}
  \includegraphics[width=\linewidth]{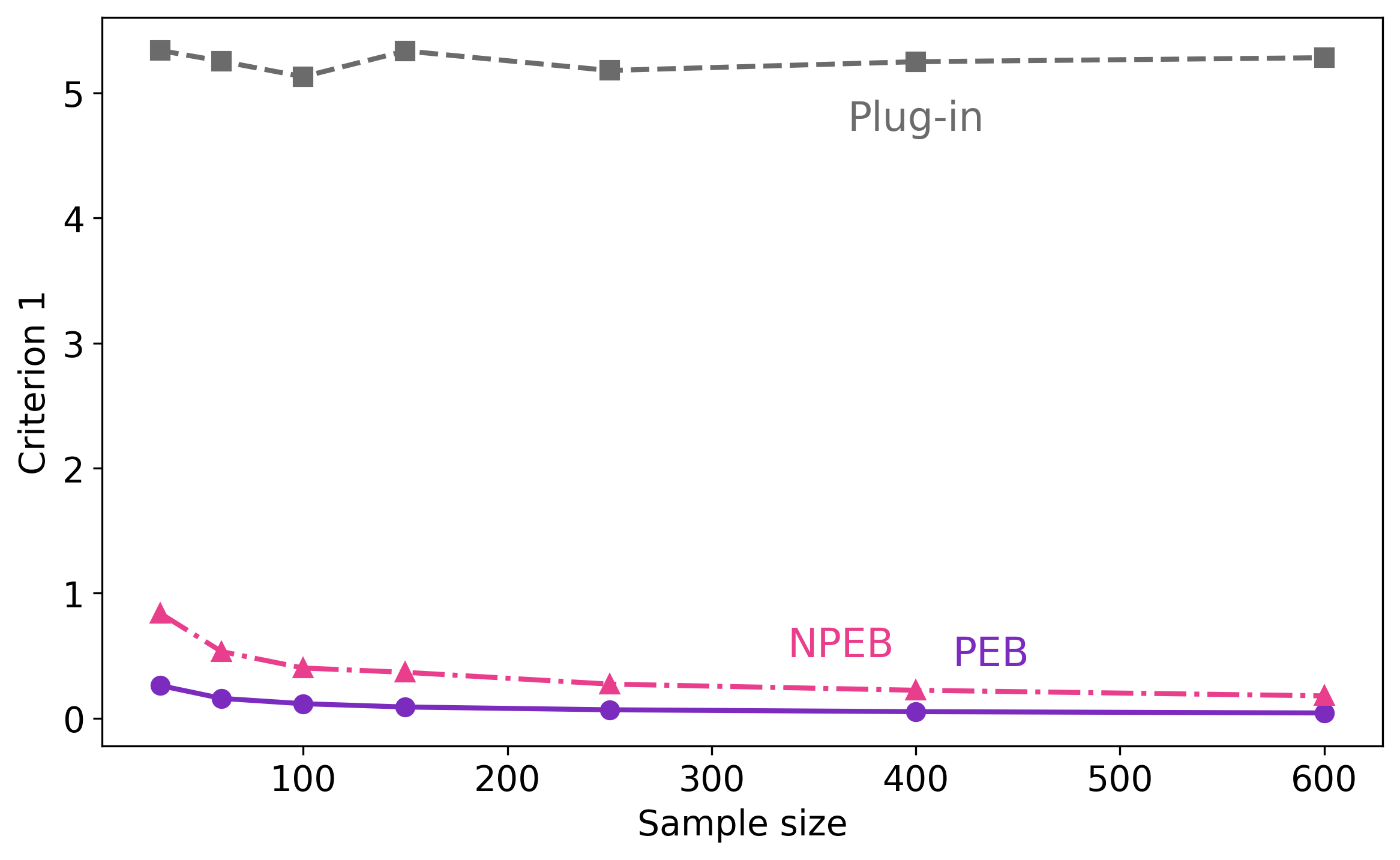}
\end{subfigure}
\hfill
\begin{subfigure}[t]{0.49\textwidth}
  \centering
  \caption{Criterion 2, $p=\infty$}
  \includegraphics[width=\linewidth]{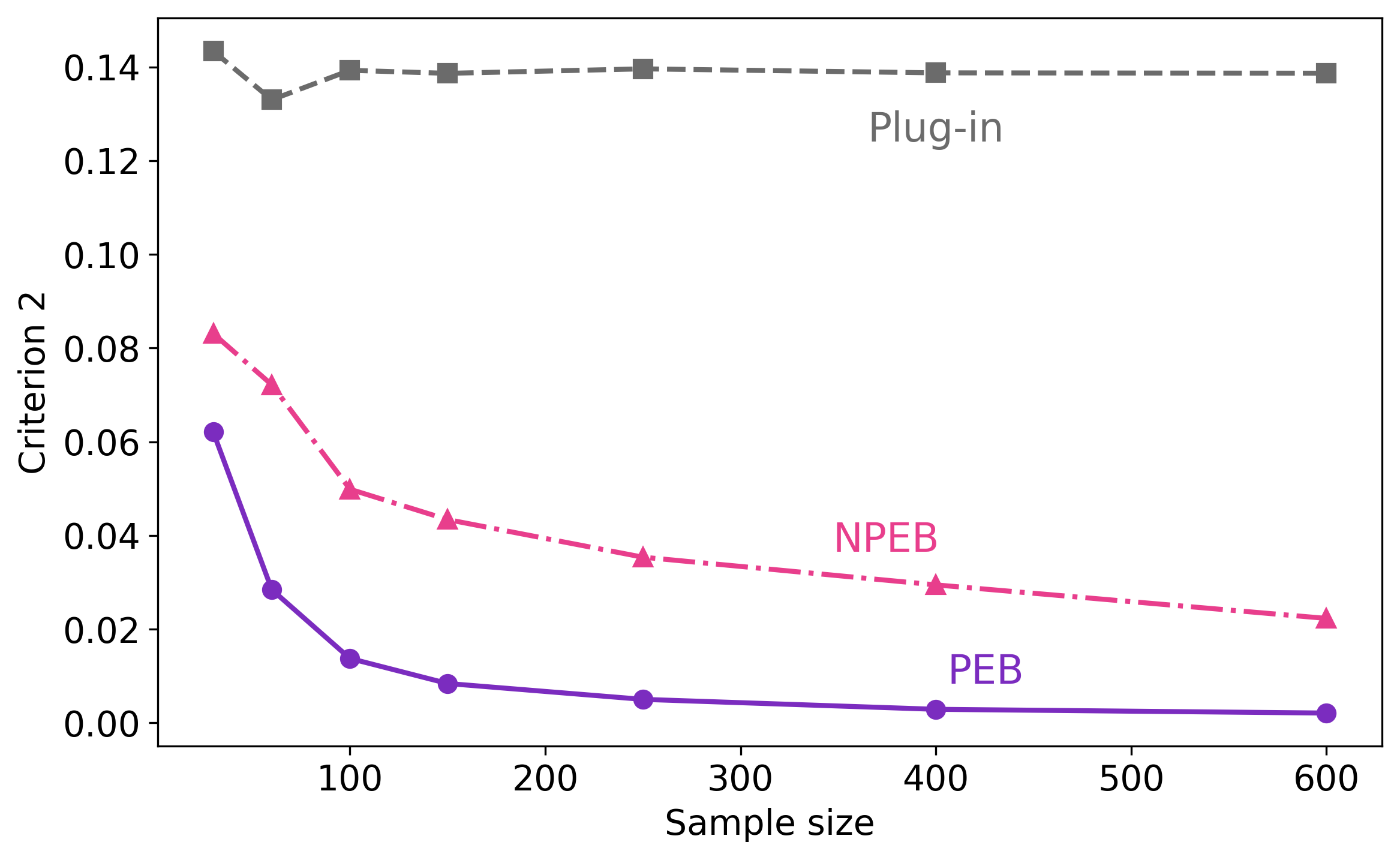}
\end{subfigure}
\begin{minipage}{\textwidth}
    \fontsize{9}{2}\linespread{1}\selectfont
    \footnotesize{
    \textit{Notes}: Panels (a), (c), and (e) plot simulation results for criterion 1 and panels (b), (d), and (f) plot simulation results for criterion 2, where the criteria are described in Section \ref{sec_4.1}. Panels (a) and (b) show results under consideration set $V = \mathcal{B}_1$, panels (c) and (d) show results under consideration set $V = \mathcal{B}_2$, and panels (e) and (f) show results under consideration set $V = \mathcal{B}_\infty$. The sample plug-in procedure is plotted as a gray dashed line with square markers, the parametric empirical Bayes procedure is plotted as a purple solid line with circular markers, and the nonparametric empirical Bayes procedure is plotted as a pink dash-dot line with triangular markers. I take $\eta_j = 1$ for all $j$ and $\mu = 1$ for all results in this figure.
    }
    \end{minipage}
\end{figure}

Figure \ref{fig:sims} plots the average of the two criteria over the Monte Carlo replications against the sample size. Three main patterns emerge. First, the sample plug-in procedure does not perform well for any of the consideration sets or criteria, consistent with the theoretical results of Proposition \ref{prop:plug_in}. Second, the parametric empirical Bayes procedure performs well for all consideration sets and both criteria. These results are consistent with Theorem \ref{thm:regret_rates_param}. Finally, the nonparametric empirical Bayes procedure appears to converge as the sample size increases for $p=2$ and $p=\infty$, but not for $p=1$. This pattern is consistent with Theorem \ref{thm:regret_rates}, which provides convergence guarantees for the nonparametric procedure when $p>1$.

\section{Empirical Illustration} \label{sec5}

In this section I apply the empirical Bayes method proposed in the previous section to estimate optimal local spending rules for making many policy changes at once. I find that empirical Bayes shrinkage can have major consequences for welfare relative to a sample plug-in approach, which solves the local problem with raw point estimates of benefit and net cost instead of empirical Bayes posterior mean estimates. In particular, I estimate that the empirical Bayes approach results in increases to welfare relative to a sample plug-in approach.

I use a sample of benefit and net cost estimates for 127 different policies compiled by \cite{hendren2020unified} and \cite{hahn2024welfare}. I restrict attention to non-international policies in the papers' baseline samples, consistent with a single policymaker operating in a common policy environment. I include all such policies for which the papers report confidence intervals for both benefit and net cost. Motivated by the finding in \cite{hendren2020unified} that policies targeting children differ systematically from those targeting adults, and by the focus of \cite{hahn2024welfare} on climate policy, I define the policy type characteristic $X_j$ to take one of three values: climate, child, and adult. The resulting sample contains 52 climate policies, 32 child policies, and 43 adult policies.

As discussed at the end of Section \ref{sec_2.3}, a researcher who has compiled sample estimates of benefit and net cost can report to the planner the posterior mean benefits and net costs, which are sufficient for the planner to solve for the optimal local spending rule. In practice, the researcher must estimate the distribution of true policy impacts and the posterior means. In Section \ref{sec_5.1} I walk through how to calculate empirical Bayes posterior mean estimates following the parametric empirical Bayes method of Section \ref{sec_3.3.1} and the nonparametric empirical Bayes method of Section \ref{sec_3.3.2}. For the parametric empirical Bayes method, I proceed as if the true prior were Gaussian. In Section \ref{sec_5.2} I compare posterior mean estimates to sample estimates of benefit and net cost and discuss the consequences for optimal policy choice. I then present numerical welfare estimates showing that in this illustration, the empirical Bayes approaches yield welfare gains relative to the sample plug-in approach.

Further details about the data and the implementation are in Supplemental Appendix \ref{app:data}.

\subsection{Calculation Walk-through} \label{sec_5.1}

Before displaying results for all policies in the sample, I first walk through how to obtain empirical Bayes estimates for two different policies. The first is the Kalamazoo Promise Scholarship, a Michigan college scholarship program, for which \cite{hendren2020unified} estimate a program cost-normalized WTP of 1.929 with estimated variance 0.378, and a program-cost normalized net cost of 0.998 with estimated variance 0.014. The estimated covariance between normalized WTP and normalized net cost is $-0.070$. The second policy is the Hope and Lifetime Learners Tax Credits, for which \cite{hendren2020unified} estimate a program cost-normalized WTP of $-42.82$ with estimated variance 1242.88, and a program cost-normalized net cost of $4.86$ with estimated variance 31.31. The estimated covariance between normalized WTP and normalized net cost is $-166.63$.

I obtain empirical Bayes estimates of benefits and net costs using both the parametric approach and the nonparametric approach, following the procedures proposed in Sections \ref{sec_3.3.1} and \ref{sec_3.3.2}, respectively. For the parametric approach I assume the true benefits and costs are drawn from a Gaussian prior, with location and scale that varies with policy type. I then obtain parametric empirical Bayes (PEB) posterior mean estimates using standard Normal-Normal shrinkage formulas under the Gaussian prior. For the nonparametric approach, I obtain location and scale estimates for each policy type and residualize the sample estimates against the location and scale estimates. I then perform nonparametric empirical Bayes (NPEB) shrinkage on the residuals, using the Python package \texttt{npeb} from \cite{soloff2024multivariate}. Finally I scale and shift the shrunk residuals back to obtain NPEB posterior mean estimates.

Intuitively, empirical Bayes posterior mean estimates shrink sample estimates towards the estimated distribution of the true policy impacts. The amount of shrinkage is decreasing in the precision of the sample estimate and increasing in the extremeness of the sample estimate relative to the distribution of the true policy impacts. This pattern is clear for these two policies. For the Kalamazoo Promise Scholarship, the sample estimates are relatively precise and the empirical Bayes posterior means (PEB: WTP of 1.968, net cost of 0.990; NPEB: WTP of 1.918, net cost of 1.005) are close to the sample estimates. In contrast, for the noisily estimated Hope and Lifetime Learners Tax Credits, the empirical Bayes posterior means (PEB: WTP of 0.94, net cost of 0.86; NPEB: WTP of 1.12, net cost of 1.14) are quite different from the sample estimates.

Posterior mean benefits and net costs are sufficient for the planner to solve their local problem, as discussed at the end of Section \ref{sec_2.3}. A researcher who knows the set of $J$ policies but is unsure about the planner's parameters (consideration set $V$, average social welfare weights $\eta_j$, or marginal welfare impact of closing the budget constraint $\mu$) can report empirical Bayes posterior mean estimates, either parametric or nonparametric, to the planner. Given those estimates, the planner can solve for either the parametric or the nonparametric empirical Bayes local spending rule.

Furthermore, as discussed at the end of Section \ref{sec_2.3}, reporting the ratios of posterior mean benefit to posterior mean net cost together with the signs of posterior mean net cost is sufficient for the planner to know whether to locally increase or decrease spending on each policy for a range of consideration sets.
For the sake of illustration, assume that the consideration set $V$ is an $L^p$ ball and that $\eta_j = 1$ for all policies $j$. Then the planner should increase spending on a policy if that policy's ratio of posterior mean benefit to posterior mean net cost is greater than $\mu$ and the posterior mean net cost is positive, or if the ratio is less than $\mu$ and the posterior mean net cost is negative. 

To understand what this means in practice, recall that $\mu$ is the marginal welfare impact of the budget-closing policy, meaning that $-\mu$ is the marginal welfare impact of closing a budget deficit. It may be reasonable to assume $\mu$ is positive, since closing a budget deficit, whether through tax increases or spending cuts, generally reduces welfare. Furthermore, $\mu = 1$ means the welfare impact of closing the budget is one to one with the size of the budget, while $\mu > 1$ indicates that closing the budget is especially distortionary or costly for welfare, which may be the case for many budget-closing policies.

For the Kalamazoo Promise Scholarship, the set of $\mu$ for which the planner increases spending is similar whether the researcher reports the ratio of sample benefit to sample net cost (any $\mu$ less than 1.932) or the ratio of empirical Bayes posterior mean benefit to empirical Bayes posterior mean net cost (PEB: any $\mu$ less than 1.988, NPEB: any $\mu$ less than 1.909). In contrast, for the noisily estimated Hope and Lifetime Learners Tax Credits, the set of $\mu$ for which the planner increases spending differs drastically: for the ratio of sample estimates the range of $\mu$ may be unreasonable in many settings (any $\mu$ less than $-8.81$), while for the ratio of empirical Bayes estimates, the range of $\mu$ may be reasonable for many settings (PEB: any $\mu$ less than 1.09, NPEB: any $\mu$ less than 0.98). This contrast highlights how accounting for statistical uncertainty through empirical Bayes shrinkage can overturn the spending recommendations implied by raw sample estimates.

\subsection{Results} \label{sec_5.2}

I first obtain PEB and NPEB posterior mean estimates of benefit and net cost as described above for all policies in the sample. 
In Figure \ref{fig:shrunk_vs_sample} I plot shrunk empirical Bayes estimates versus sample estimates for both WTP and net cost. As previously discussed, a researcher who does not know the parameters of the planner's decision problem, namely $\mu, \eta_1, \dots, \eta_J,$ or $V$, only needs to report the shrunk posterior mean estimates of Figure \ref{fig:shrunk_vs_sample} to the planner. The planner, who knows the parameter values, can solve for the empirical Bayes local spending rule using the posterior mean estimates. In each panel I label the six policies with the greatest amount of empirical Bayes shrinkage. Generally, these are the policies with the largest estimated sampling uncertainty, though this also includes policies with extreme estimates relative to the estimated prior, which mechanically shrink more. 

\begin{figure}[!thbp]
\caption{Shrunk versus sample estimates}
\label{fig:shrunk_vs_sample}
\centering
\begin{subfigure}[t]{0.49\textwidth}
  \centering
  \caption{NPEB WTP}
  \includegraphics[width=\linewidth]{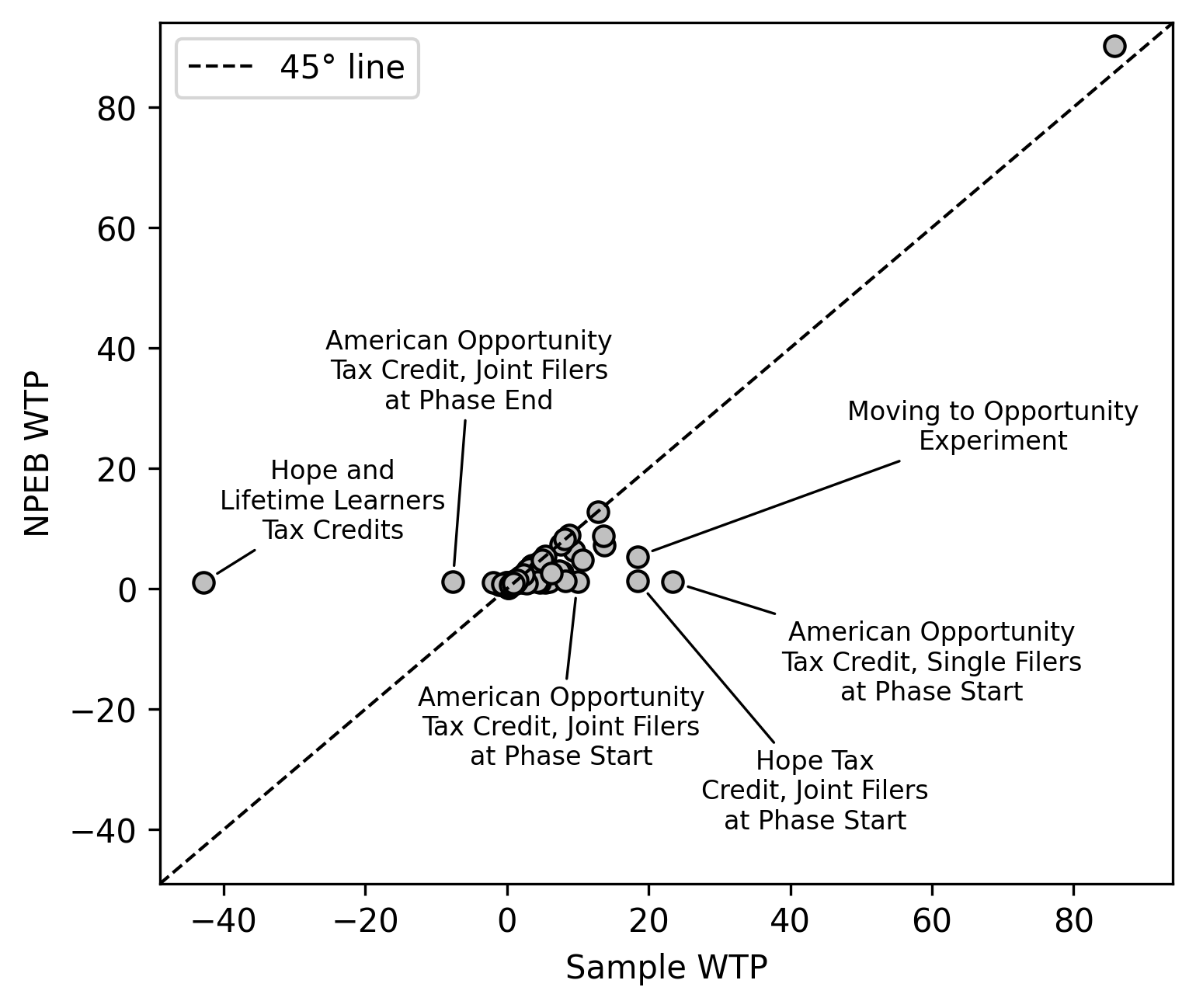}
\end{subfigure}
\begin{subfigure}[t]{0.49\textwidth}
  \centering
  \caption{NPEB net cost}
  \includegraphics[width=\linewidth]{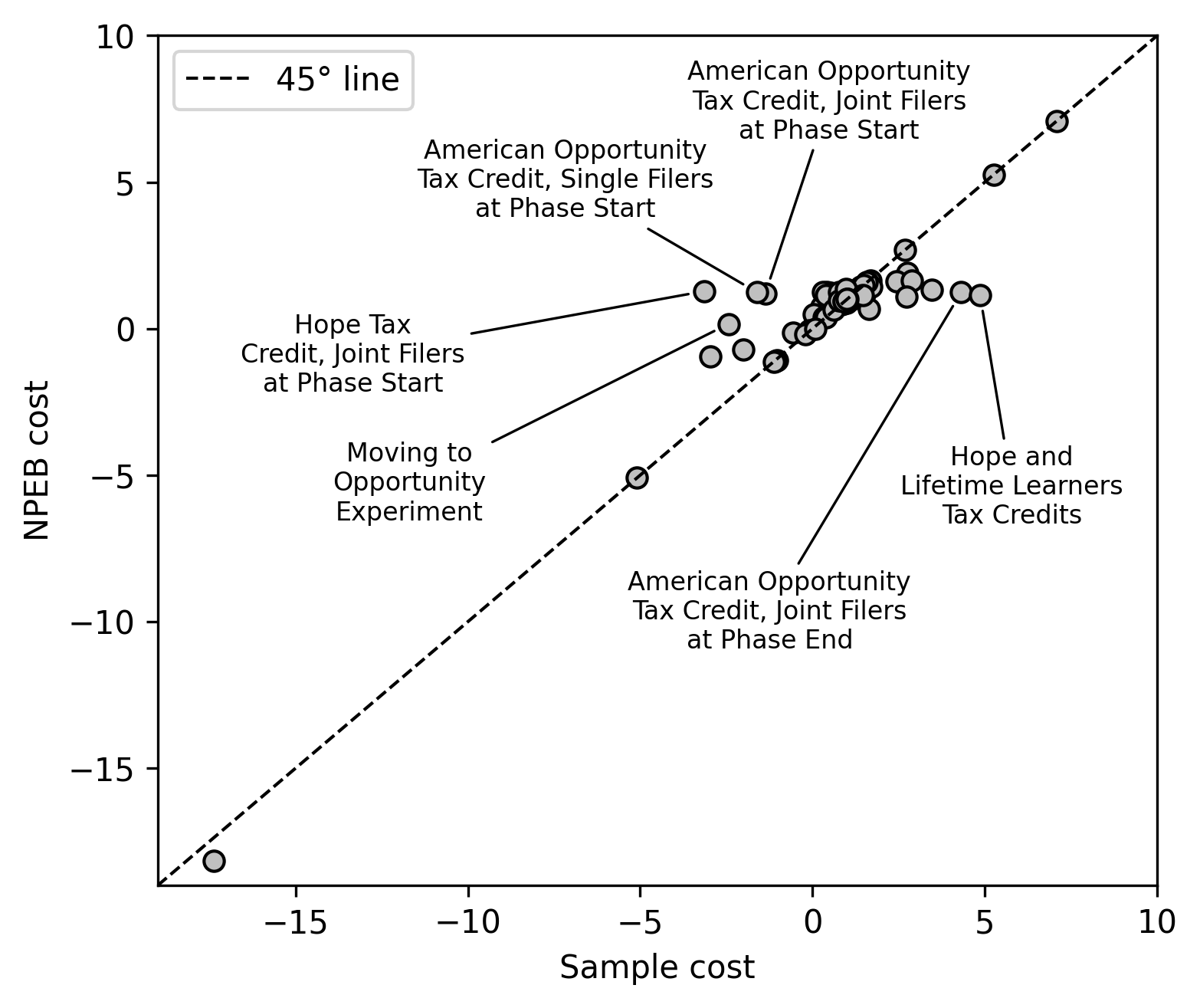}
\end{subfigure}
\hfill
\begin{subfigure}[t]{0.49\textwidth}
  \centering
  \caption{PEB WTP}
  \includegraphics[width=\linewidth]{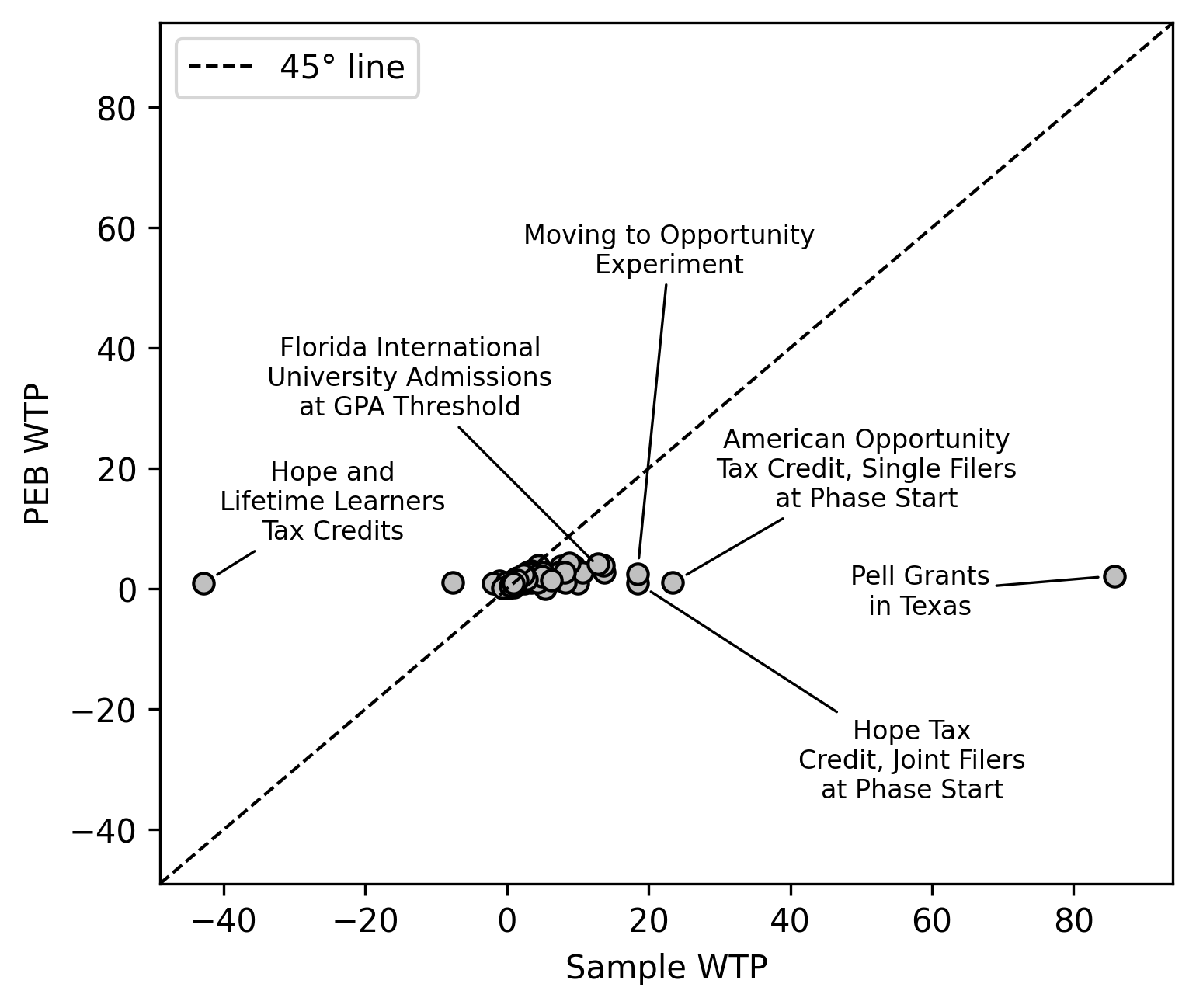}
\end{subfigure}
\begin{subfigure}[t]{0.49\textwidth}
  \centering
  \caption{PEB net cost}
  \includegraphics[width=\linewidth]{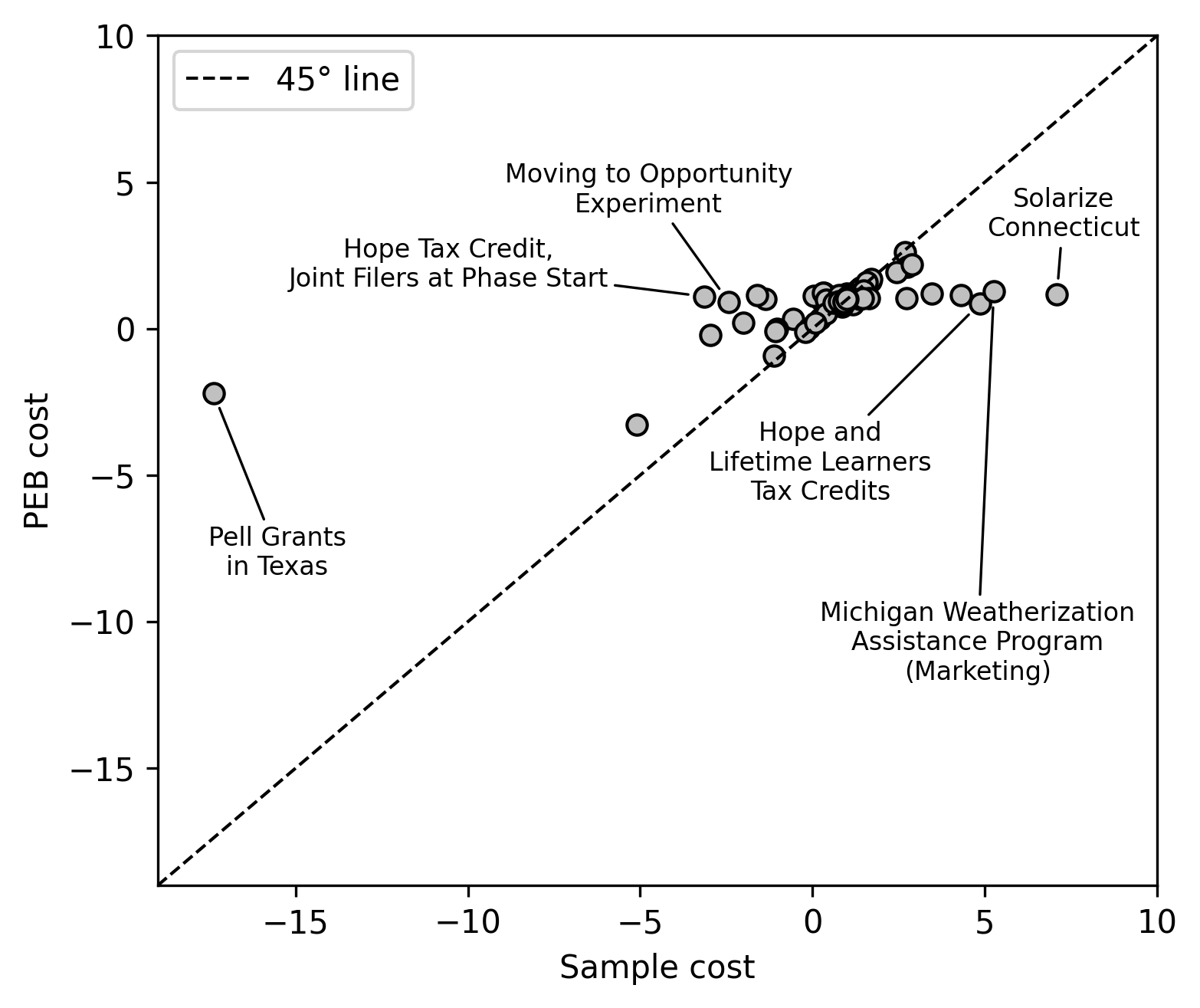}
\end{subfigure}
\begin{minipage}{\textwidth}
    \fontsize{9}{2}\linespread{1}\selectfont
    \footnotesize{
    \textit{Notes}: Each point represents a single policy. WTP and net cost are normalized with respect to program cost. Panels (a) and (b) plot nonparametric empirical Bayes posterior means versus sample estimates, with WTP in panel (a) and net cost in panel (b). Panels (c) and (d) plot parametric empirical Bayes posterior means versus sample estimates under a Gaussian prior, with WTP in panel (c) and net cost in panel (d). Policies that are labeled are the six policies farthest away from the 45-degree line in each panel.
    }
    \end{minipage}
\end{figure}

To highlight the consequences of these differences, I next show that the empirical Bayes and sample plug-in rules can result in welfare effects that are locally of the opposite sign. 
In Table \ref{tab:coupled_welfare_ests} I present estimates of the planner's local objective along the sample plug-in local spending rule, along the PEB local spending rule, and along the NPEB local spending rule, averaged over 1000 coupled bootstrap draws, following the coupled bootstrap method discussed in Supplemental Appendix \ref{app:coupled}. 
I derive estimates of the local objective for $\mu \in \{0.5, 3\}$ and set welfare weights $\eta_j = 1$ for all policies. I focus on consideration sets $V = \mathcal{B}_2$ and $\mathcal{B}_\infty$ (that is, $p \in \{2,\infty\}$), for which Section \ref{sec3} provides convergence guarantees for NPEB. I emphasize that the validity of these estimates relies only on the unbiased Gaussian sampling model in \eqref{eq:likelihood1}, and not on the Gaussian prior imposed for PEB or the location-scale prior structure imposed for NPEB.

\begin{table}[!htbp] 
    \caption{Local welfare estimates from coupled bootstrap}
    \label{tab:coupled_welfare_ests}
    \centering
    \resizebox{0.65\textwidth}{!}{
    \begin{tabular}{ccccc}
    \hline\hline 
    & & & & \\[-12pt]
    $V$ & $\mu$ & Sample plug-in rule & PEB rule & NPEB rule \\
    & & (1) & (2) & (3) \\
    \hline
    $\mathcal{B}_2$ & 0.5 & $-295.48$ & $46.19$ & $85.67$ \\[5pt]
    $\mathcal{B}_2$ & 3 & $-417.61$ & $55.84$ & $122.90$ \\[5pt]
    $\mathcal{B}_\infty$ & 0.5 & $-481.24$ & $361.61$ & $343.63$ \\[5pt]
    $\mathcal{B}_\infty$ & 3 & $-674.18$ & $337.72$ & $409.66$ \\
    \hline\hline\\[-10pt]
    \end{tabular}}
    \begin{minipage}{\textwidth}
    \fontsize{9}{2}\linespread{1}\selectfont
    \footnotesize{
    \textit{Notes}: Local welfare estimates are the average of 1000 coupled bootstrap estimates of the objective of the planner's local problem along the sample plug-in local spending rule (column 1), the parametric empirical Bayes local spending rule (column 2) and the nonparametric empirical Bayes local spending rule (column 3). $V$ denotes the consideration set of the planner's local problem. $\mu$ denotes the marginal welfare impact of closing the budget. I implement the coupled bootstrap procedure with an 80-20 train-test split ($\kappa = \frac{1}{4}$), using the training sample to construct the local spending rules and the test sample to evaluate the rules. See Supplemental Appendix \ref{app:coupled} for more details on the coupled bootstrap estimation procedure. All numbers are rounded to the nearest hundredth.
    }
    \end{minipage}
\end{table}

For all considered combinations of $V$ and $\mu$, the empirical Bayes local spending rules are estimated to produce a local \textit{increase} in net welfare, while the sample plug-in local spending rule is estimated to produce a local \textit{decrease} in net welfare. Because coupled bootstrap estimates are noisy, I interpret these results as descriptive evidence about the relative performance of the rules rather than as precise estimates of welfare gains. These differences appear to be driven in part by noisily estimated policies with relatively low estimated benefits or high estimated costs. In Figure \ref{fig:shrunk_vs_sample}, these are the policies that are left of the 45-degree line in panels (a) and (c) and right of the 45-degree line in panel (b). The sample plug-in rule treats these extreme estimates as signal and places substantial weight on cutting the corresponding policies. On the other hand, the empirical Bayes rules take into account the statistical uncertainty about those policies, reducing the influence of these noisy extremes on which policies are cut and by how much. The coupled bootstrap estimates suggest that this avoids welfare losses driven by estimation noise. This empirical finding---that the sample plug-in rule performs poorly in cases where the empirical Bayes rules perform well---is consistent with the theoretical results of Section \ref{sec3} and illustrates their practical importance.

\section{Conclusion} \label{sec6}

In this paper I study how to make optimal policy changes when policy impacts are estimated with statistical noise. I set up a statistically well-behaved decision problem where the planner makes local changes to upfront spending on a set of policies to maximize net welfare impact.
Under the assumption that the planner knows the distribution of sample estimates of policy impacts and the prior distribution of the true policy impacts, I characterize the optimal local spending rule, which locally maximizes posterior expected net welfare impact. A sufficient statistic for the optimal local spending rule is the posterior expected gradient of net welfare impact, which in turn depends on posterior mean benefits and net costs. A key implication is that a researcher who is unsure about the planner's specific parameters can report posterior mean benefits and net costs to the planner, who can use the posterior means to solve for the optimal local spending rule.

When the prior distribution of true impacts is unknown, I propose both a parametric approach and a nonparametric approach to estimating the prior. I then obtain an empirical Bayes local spending rule by solving the local problem with an estimated gradient that plugs in estimates of posterior mean benefit and net cost. I show that these empirical Bayes approaches perform well, unlike a sample plug-in approach, by deriving rates of convergence for two objects. %These rates show that the empirical Bayes estimate of the local objective converges uniformly and that the local welfare improvement from the empirical Bayes rule converges to the maximal local welfare improvement. 
Taken together, these results imply that these empirical Bayes approaches can asymptotically match the performance of the oracle planner as the number of policies grows, with the parametric approach providing faster rates than the nonparametric approach at the cost of additional parametric restrictions. Finally, in an application to a set of 127 policies studied by \cite{hendren2020unified} and \cite{hahn2024welfare}, I estimate that the plug-in approach can lead to welfare losses, while, in those same situations, careful incorporation of statistical uncertainty through empirical Bayes shrinkage yields welfare gains.

\FloatBarrier

%small
\singlespacing
{
\bibliographystyle{chicago}
\bibliography{ref}

\begin{thebibliography}{}

\bibitem[\protect\citeauthoryear{Chen}{Chen}{2026}]{chen2022empirical}
Chen, J. (2026).
\newblock Empirical bayes when estimation precision predicts parameters.
\newblock {\em Econometrica\/}~{\em 94\/}(2), 305--340.

\bibitem[\protect\citeauthoryear{Hahn, Hendren, Metcalfe, and
  Sprung-Keyser}{Hahn et~al.}{2026}]{hahn2024welfare}
Hahn, R.~W., N.~Hendren, R.~D. Metcalfe, and B.~Sprung-Keyser (2026).
\newblock A welfare analysis of policies impacting climate change.
\newblock {\em American Economic Review\/}~{\em 116\/}(7), 2368--2421.

\bibitem[\protect\citeauthoryear{Hendren and Sprung-Keyser}{Hendren and
  Sprung-Keyser}{2020}]{hendren2020unified}
Hendren, N. and B.~Sprung-Keyser (2020).
\newblock A unified welfare analysis of government policies.
\newblock {\em The Quarterly Journal of Economics\/}~{\em 135\/}(3),
  1209--1318.

\bibitem[\protect\citeauthoryear{Jiang}{Jiang}{2020}]{jiang2020general}
Jiang, W. (2020).
\newblock On general maximum likelihood empirical {B}ayes estimation of
  heteroscedastic {IID} normal means.
\newblock {\em Electronic Journal of Statistics\/}~{\em 14}, 2272--2297.

\bibitem[\protect\citeauthoryear{Leiner, Duan, Wasserman, and Ramdas}{Leiner
  et~al.}{2023}]{leiner2023data}
Leiner, J., B.~Duan, L.~Wasserman, and A.~Ramdas (2023).
\newblock Data fission: splitting a single data point.
\newblock {\em Journal of the American Statistical Association\/}, 1--12.

\bibitem[\protect\citeauthoryear{Oliveira, Lei, and Tibshirani}{Oliveira
  et~al.}{2024}]{oliveira2024unbiased}
Oliveira, N.~L., J.~Lei, and R.~J. Tibshirani (2024).
\newblock Unbiased risk estimation in the normal means problem via coupled
  bootstrap techniques.
\newblock {\em Electronic Journal of Statistics\/}~{\em 18\/}(2), 5405--5448.

\bibitem[\protect\citeauthoryear{Saha and Guntuboyina}{Saha and
  Guntuboyina}{2020}]{saha2020nonparametric}
Saha, S. and A.~Guntuboyina (2020).
\newblock On the nonparametric maximum likelihood estimator for {G}aussian
  location mixture densities with application to {G}aussian denoising.
\newblock {\em The Annals of Statistics\/}~{\em 48\/}(2), 738--762.

\bibitem[\protect\citeauthoryear{Soloff, Guntuboyina, and Sen}{Soloff
  et~al.}{2025}]{soloff2024multivariate}
Soloff, J.~A., A.~Guntuboyina, and B.~Sen (2025).
\newblock Multivariate, heteroscedastic empirical {B}ayes via nonparametric
  maximum likelihood.
\newblock {\em Journal of the Royal Statistical Society Series B: Statistical
  Methodology\/}~{\em 87\/}(1), 1--32.

\end{thebibliography}


\begin{thebibliography}{}

\bibitem[\protect\citeauthoryear{Andrews and Chen}{Andrews and
  Chen}{2025}]{andrews2025certified}
Andrews, I. and J.~Chen (2025).
\newblock Certified decisions.
\newblock {\em arXiv preprint arXiv:2502.17830\/}.

\bibitem[\protect\citeauthoryear{Andrews, Kitagawa, and McCloskey}{Andrews
  et~al.}{2024}]{andrews2024inference}
Andrews, I., T.~Kitagawa, and A.~McCloskey (2024).
\newblock Inference on winners.
\newblock {\em The Quarterly Journal of Economics\/}~{\em 139\/}(1), 305--358.

\bibitem[\protect\citeauthoryear{Athey and Wager}{Athey and
  Wager}{2021}]{athey2021policy}
Athey, S. and S.~Wager (2021).
\newblock Policy learning with observational data.
\newblock {\em Econometrica\/}~{\em 89\/}(1), 133--161.

\bibitem[\protect\citeauthoryear{Bergstrom, Dodds, and Rios}{Bergstrom
  et~al.}{2025}]{bergstrom2024optimal}
Bergstrom, K., W.~Dodds, and J.~Rios (2025).
\newblock Optimal policy reforms.

\bibitem[\protect\citeauthoryear{Bourguignon and Spadaro}{Bourguignon and
  Spadaro}{2012}]{bourguignon2012tax}
Bourguignon, F. and A.~Spadaro (2012).
\newblock Tax--benefit revealed social preferences.
\newblock {\em The Journal of Economic Inequality\/}~{\em 10\/}(1), 75--108.

\bibitem[\protect\citeauthoryear{Chen}{Chen}{2026}]{chen2022empirical}
Chen, J. (2026).
\newblock Empirical bayes when estimation precision predicts parameters.
\newblock {\em Econometrica\/}~{\em 94\/}(2), 305--340.

\bibitem[\protect\citeauthoryear{Chen, Deb, and Ignatiadis}{Chen
  et~al.}{2026}]{chen2026normal}
Chen, J., N.~Deb, and N.~Ignatiadis (2026).
\newblock Normal approximations in nonparametric empirical bayes.
\newblock {\em arXiv preprint arXiv:2605.31599\/}.

\bibitem[\protect\citeauthoryear{Chernozhukov, Lee, Rosen, and
  Sun}{Chernozhukov et~al.}{2025}]{chern2025policy}
Chernozhukov, V., S.~Lee, A.~M. Rosen, and L.~Sun (2025).
\newblock Policy learning with confidence.
\newblock {\em arXiv preprint arXiv:2502.10653\/}.

\bibitem[\protect\citeauthoryear{Chetty}{Chetty}{2008}]{chetty2008moral}
Chetty, R. (2008).
\newblock Moral hazard versus liquidity and optimal unemployment insurance.
\newblock {\em Journal of Political Economy\/}~{\em 116\/}(2), 173--234.

\bibitem[\protect\citeauthoryear{Chetty}{Chetty}{2009}]{chetty2009sufficient}
Chetty, R. (2009).
\newblock Sufficient statistics for welfare analysis: {A} bridge between
  structural and reduced-form methods.
\newblock {\em Annual Review of Economics\/}~{\em 1\/}(1), 451--488.

\bibitem[\protect\citeauthoryear{Efron}{Efron}{2012}]{efron2012large}
Efron, B. (2012).
\newblock {\em Large-scale inference: empirical {B}ayes methods for estimation,
  testing, and prediction}, Volume~1.
\newblock Cambridge University Press.

\bibitem[\protect\citeauthoryear{Feldstein}{Feldstein}{1999}]{feldstein1999tax}
Feldstein, M. (1999).
\newblock Tax avoidance and the deadweight loss of the income tax.
\newblock {\em Review of Economics and Statistics\/}~{\em 81\/}(4), 674--680.

\bibitem[\protect\citeauthoryear{Finkelstein and Hendren}{Finkelstein and
  Hendren}{2020}]{finkelstein2020welfare}
Finkelstein, A. and N.~Hendren (2020).
\newblock Welfare analysis meets causal inference.
\newblock {\em Journal of Economic Perspectives\/}~{\em 34\/}(4), 146--167.

\bibitem[\protect\citeauthoryear{Gruber}{Gruber}{1997}]{gruber1997consumption}
Gruber, J. (1997).
\newblock The consumption smoothing benefits of unemployment insurance.
\newblock {\em The American Economic Review\/}~{\em 87\/}(1), 192--205.

\bibitem[\protect\citeauthoryear{Gu and Koenker}{Gu and
  Koenker}{2023}]{gu2023invidious}
Gu, J. and R.~Koenker (2023).
\newblock Invidious comparisons: {R}anking and selection as compound decisions.
\newblock {\em Econometrica\/}~{\em 91\/}(1), 1--41.

\bibitem[\protect\citeauthoryear{Hahn, Hendren, Metcalfe, and
  Sprung-Keyser}{Hahn et~al.}{2026}]{hahn2024welfare}
Hahn, R.~W., N.~Hendren, R.~D. Metcalfe, and B.~Sprung-Keyser (2026).
\newblock A welfare analysis of policies impacting climate change.
\newblock {\em American Economic Review\/}~{\em 116\/}(7), 2368--2421.

\bibitem[\protect\citeauthoryear{Harsanyi}{Harsanyi}{1955}]{harsanyi1955cardinal}
Harsanyi, J.~C. (1955).
\newblock Cardinal welfare, individualistic ethics, and interpersonal
  comparisons of utility.
\newblock {\em Journal of Political Economy\/}~{\em 63\/}(4), 309--321.

\bibitem[\protect\citeauthoryear{Hendren}{Hendren}{2020}]{hendren2020measuring}
Hendren, N. (2020).
\newblock Measuring economic efficiency using inverse-optimum weights.
\newblock {\em Journal of Public Economics\/}~{\em 187}, 104198.

\bibitem[\protect\citeauthoryear{Hendren and Sprung-Keyser}{Hendren and
  Sprung-Keyser}{2020}]{hendren2020unified}
Hendren, N. and B.~Sprung-Keyser (2020).
\newblock A unified welfare analysis of government policies.
\newblock {\em The Quarterly Journal of Economics\/}~{\em 135\/}(3),
  1209--1318.

\bibitem[\protect\citeauthoryear{Jiang}{Jiang}{2020}]{jiang2020general}
Jiang, W. (2020).
\newblock On general maximum likelihood empirical {B}ayes estimation of
  heteroscedastic {IID} normal means.
\newblock {\em Electronic Journal of Statistics\/}~{\em 14}, 2272--2297.

\bibitem[\protect\citeauthoryear{Jiang and Zhang}{Jiang and
  Zhang}{2009}]{jiang2009general}
Jiang, W. and C.-H. Zhang (2009).
\newblock General maximum likelihood empirical {B}ayes estimation of normal
  means.
\newblock {\em The Annals of Statistics\/}, 1647--1684.

\bibitem[\protect\citeauthoryear{Kasy}{Kasy}{2018}]{kasy2018optimal}
Kasy, M. (2018).
\newblock Optimal taxation and insurance using machine learning—sufficient
  statistics and beyond.
\newblock {\em Journal of Public Economics\/}~{\em 167}, 205--219.

\bibitem[\protect\citeauthoryear{Kiefer and Wolfowitz}{Kiefer and
  Wolfowitz}{1956}]{kiefer1956consistency}
Kiefer, J. and J.~Wolfowitz (1956).
\newblock Consistency of the maximum likelihood estimator in the presence of
  infinitely many incidental parameters.
\newblock {\em The Annals of Mathematical Statistics\/}, 887--906.

\bibitem[\protect\citeauthoryear{Kitagawa and Tetenov}{Kitagawa and
  Tetenov}{2018}]{kitagawa2018should}
Kitagawa, T. and A.~Tetenov (2018).
\newblock Who should be treated? {E}mpirical welfare maximization methods for
  treatment choice.
\newblock {\em Econometrica\/}~{\em 86\/}(2), 591--616.

\bibitem[\protect\citeauthoryear{Kleven}{Kleven}{2021}]{kleven2021sufficient}
Kleven, H.~J. (2021).
\newblock Sufficient statistics revisited.
\newblock {\em Annual Review of Economics\/}~{\em 13\/}(1), 515--538.

\bibitem[\protect\citeauthoryear{Kline, Rose, and Walters}{Kline
  et~al.}{2024}]{kline2024discrimination}
Kline, P., E.~K. Rose, and C.~R. Walters (2024).
\newblock A discrimination report card.
\newblock {\em American Economic Review\/}~{\em 114\/}(8), 2472–2525.

\bibitem[\protect\citeauthoryear{Koenker and Mizera}{Koenker and
  Mizera}{2014}]{koenker2014convex}
Koenker, R. and I.~Mizera (2014).
\newblock Convex optimization, shape constraints, compound decisions, and
  empirical {B}ayes rules.
\newblock {\em Journal of the American Statistical Association\/}~{\em
  109\/}(506), 674--685.

\bibitem[\protect\citeauthoryear{Manski}{Manski}{2004}]{manski2004statistical}
Manski, C.~F. (2004).
\newblock Statistical treatment rules for heterogeneous populations.
\newblock {\em Econometrica\/}~{\em 72\/}(4), 1221--1246.

\bibitem[\protect\citeauthoryear{Mbakop and Tabord-Meehan}{Mbakop and
  Tabord-Meehan}{2021}]{mbakop2021model}
Mbakop, E. and M.~Tabord-Meehan (2021).
\newblock Model selection for treatment choice: {P}enalized welfare
  maximization.
\newblock {\em Econometrica\/}~{\em 89\/}(2), 825--848.

\bibitem[\protect\citeauthoryear{Mogstad, Romano, Shaikh, and Wilhelm}{Mogstad
  et~al.}{2024}]{mogstad2024inference}
Mogstad, M., J.~P. Romano, A.~M. Shaikh, and D.~Wilhelm (2024).
\newblock Inference for ranks with applications to mobility across
  neighbourhoods and academic achievement across countries.
\newblock {\em Review of Economic Studies\/}~{\em 91\/}(1), 476--518.

\bibitem[\protect\citeauthoryear{Robbins}{Robbins}{1956}]{robbins1956empirical}
Robbins, H. (1956).
\newblock An empirical {B}ayes approach to statistics.
\newblock In {\em Proceedings of the Third Berkeley Symposium on Mathematical
  Statistics and Probability}, Volume~1, pp.\  157--163.

\bibitem[\protect\citeauthoryear{Robert}{Robert}{2007}]{robert2007bayesian}
Robert, C.~P. (2007).
\newblock {\em The {B}ayesian choice: From decision-theoretic foundations to
  computational implementation}.
\newblock Springer.

\bibitem[\protect\citeauthoryear{Saez}{Saez}{2001}]{saez2001using}
Saez, E. (2001).
\newblock Using elasticities to derive optimal income tax rates.
\newblock {\em The Review of Economic Studies\/}~{\em 68\/}(1), 205--229.

\bibitem[\protect\citeauthoryear{Schmieder and Von~Wachter}{Schmieder and
  Von~Wachter}{2016}]{schmieder2016effects}
Schmieder, J.~F. and T.~Von~Wachter (2016).
\newblock The effects of unemployment insurance benefits: {N}ew evidence and
  interpretation.
\newblock {\em Annual Review of Economics\/}~{\em 8\/}(1), 547--581.

\bibitem[\protect\citeauthoryear{Soloff, Guntuboyina, and Sen}{Soloff
  et~al.}{2025}]{soloff2024multivariate}
Soloff, J.~A., A.~Guntuboyina, and B.~Sen (2025).
\newblock Multivariate, heteroscedastic empirical {B}ayes via nonparametric
  maximum likelihood.
\newblock {\em Journal of the Royal Statistical Society Series B: Statistical
  Methodology\/}~{\em 87\/}(1), 1--32.

\bibitem[\protect\citeauthoryear{Sun}{Sun}{2026}]{sun2024empirical}
Sun, L. (2026).
\newblock Empirical welfare maximization with constraints.
\newblock {\em Journal of Econometrics\/}~{\em 253}, 106169.

\bibitem[\protect\citeauthoryear{Vershynin}{Vershynin}{2018}]{vershynin2018high}
Vershynin, R. (2018).
\newblock {\em High-dimensional probability: {A}n introduction with
  applications in data science}, Volume~47.
\newblock Cambridge University Press.

\bibitem[\protect\citeauthoryear{Wald}{Wald}{1949}]{wald1949statistical}
Wald, A. (1949).
\newblock Statistical decision functions.
\newblock {\em The Annals of Mathematical Statistics\/}, 165--205.

\bibitem[\protect\citeauthoryear{Walters}{Walters}{2024}]{walters2024empirical}
Walters, C. (2024).
\newblock Empirical {B}ayes methods in labor economics.
\newblock In {\em Handbook of Labor Economics}, Volume~5, pp.\  183--260.
  Elsevier.

\bibitem[\protect\citeauthoryear{Yamin}{Yamin}{2025}]{yamin2025poverty}
Yamin, J.~C. (2025).
\newblock Poverty targeting with imperfect information.
\newblock {\em arXiv preprint arXiv:2506.18188\/}.

\end{thebibliography}
}

\clearpage
\newpage 
\appendix

% Restart Footnotes count
\newcommand{\AppendixPrefix}{A}
% Restart Figures, tables and equations and label with appendix section number
\renewcommand{\thetheorem}{\thesection.\arabic{theorem}}
%\counterwithin{theorem}{section}
\renewcommand{\thefigure}{\thesection.\arabic{figure}}
%\counterwithin{figure}{section}
\renewcommand{\thetable}{\thesection.\arabic{table}}
%\counterwithin{table}{section}
\renewcommand{\theequation}{\thesection.\arabic{equation}}
%\counterwithin{equation}{section}
\renewcommand{\thecorollary}{\thesection.\arabic{theorem}}
%\counterwithin{corollary}{section}
\renewcommand{\thelemma}{\thesection.\arabic{theorem}}
%\counterwithin{lemma}{section}
\renewcommand{\theproposition}{\thesection.\arabic{theorem}}
%\counterwithin{proposition}{section}
\setcounter{page}{1}
\renewcommand{\thepage}{S.\arabic{page}}
\setcounter{theorem}{0}
\setcounter{equation}{0}
\setcounter{figure}{0}
\setcounter{table}{0}

\begin{center}
\onehalfspacing 
\LARGE{Supplemental Appendix to ``Optimal Policy Choices Under Uncertainty''}
\end{center}

\onehalfspacing

\section{Appendix: Satisfying Assumptions \ref{ass:param} and \ref{ass:est_var}} 

\subsection{Distributions Satisfying Assumption \ref{ass:param}(5)}
\label{app:param}

Any regular distribution in an exponential family with parameter-independent support can be written as
\begin{align*}
    f_\beta(\theta \vert X_j) &= h(\theta, X_j) \exp\{\eta(\beta, X_j)'T(\theta, X_j) - A(\eta(\beta,X_j),X_j)\}.
\end{align*}
Note that
\begin{align*}
    \log f_\beta(\theta \vert X_j) &= \log(h(\theta, X_j)) + \eta(\beta,X_j)' T(\theta,X_j) - A(\eta(\beta,X_j),X_j) \\
    \Rightarrow s_\beta(\theta \vert X_j) &= \nabla_\beta \log f_\beta(\theta \vert X_j) = \nabla_\beta \eta(\beta,X_j) \left[ T(\theta,X_j) - \nabla_\eta A(\beta,X_j) \right] \\
    &= \nabla_\beta \eta(\beta,X_j) \left[T(\theta,X_j) - E_\beta[T(\theta,X_j) \vert X_j]\right]
\end{align*}
by properties of the log-partition function $A(\eta(\beta,X_j),X_j)$. 
Thus sufficient conditions for the prior score to satisfy Assumption \ref{ass:param}(5) are that 
\begin{enumerate}
    \item $\nabla_\beta \eta(\beta,X_j)$ has bounded norm uniformly over the support of $X_j$ and $\beta \in B$,
    \item There exists $k_s \in (0,2]$ and constant $\tilde C$ such that for all $p \geq 1$ and for all $\beta \in B$, 
    \begin{align*}
        \left(E_\beta \left[ \left\Vert T(\theta, X_j) - E_\beta[T(\theta,X_j) \vert X_j]\right] \right\Vert_2^p \right)^{1/p} \leq \tilde C p^{1/k_s}
    \end{align*}
    uniformly over the support of $X_j$.
\end{enumerate}

\subsection{Estimators Satisfying Assumption \ref{ass:est_var}(4)}
\label{app:est_rates}

In this appendix I show the estimators for $\alpha_0$ and $\Omega_0$ proposed in \eqref{eq:location_scale_ests} satisfy the estimation rate of Assumption \ref{ass:est_var}(4) under Assumptions \ref{ass:compact}, \ref{ass:bdd_sig}, \ref{ass:est_var}(3), and additional regularity assumptions. In particular, there exist constants $C_1, C_2$ such that for all $J$,
    \begin{align*}
        P \left( \Vert \widehat{\chi} - \chi_0 \Vert_J > C_1 \sqrt{\frac{\log J}{J}} \right) &\leq \frac{C_2}{J^2},
    \end{align*}
defining $\Vert \chi - \chi_0 \Vert_J \equiv \max_{1\leq j \leq J} \max \left\{\Vert a(X_j;\alpha) - a(X_j;\alpha_0) \Vert_\infty, \Vert B(X_j;\Omega)^{1/2} - B(X_j;\Omega_0)^{1/2} \Vert_{op} \right\}$ for $\chi = (\alpha, \Omega)$. Recall that I take $\Sigma_{1:J}$ and $X_{1:J}$ fixed throughout the paper.

Recall the estimators of \eqref{eq:location_scale_ests} are
\begin{equation*}
    \begin{gathered}
        (\widehat\alpha,\widehat\Omega) \in \underset{\alpha,\Omega}{\text{arg}\min} ~ \left( \frac{1}{J} \sum_{j=1}^J g_j(\alpha,\Omega)\right)^\prime W \left( \frac{1}{J} \sum_{j=1}^J g_j(\alpha,\Omega)\right) , \\
        g_j(\alpha,\Omega) \equiv q(X_j) \otimes \begin{pmatrix} Y_j - a(X_j;\alpha) \\ vec \left( \left( Y_j - a(X_j; \alpha)\right) \left( Y_j - a(X_j; \alpha)\right)' - \Sigma_j - B(X_j; \Omega)\right) \end{pmatrix}.
    \end{gathered}
\end{equation*}
For simplicity I will prove the result with $W=I$, although it is possible to extend the result to any $W$ with eigenvalues that are uniformly bounded away from $\infty$ and zero.
Let $\hat g_J(\alpha,\Omega) \equiv \frac{1}{J} \sum_{j=1}^J g_j(\alpha,\Omega)$ and $\bar g_J(\alpha,\Omega) \equiv \frac{1}{J} \sum_{j=1}^J E[g_j(\alpha,\Omega)]$. Note that $\bar g_J(\alpha_0,\Omega_0) = 0$.

The main regularity assumption I will impose is that functions $a(X_j;\alpha)$ and $B(X_j;\Omega)^{1/2}$ are uniformly Lipschitz in their parameters. That is, there exist constants $C_3, C_4$ such that
\begin{align*}
    \left\Vert a(X_j;\alpha) - a(X_j;\tilde\alpha) \right\Vert_{\infty} &\leq C_3 \Vert \alpha - \tilde\alpha \Vert_2 \\
    \left\Vert B(X_j;\Omega)^{1/2} - B(X_j;\tilde\Omega)^{1/2} \right\Vert_{op} &\leq C_4 \Vert \Omega - \tilde\Omega \Vert_F
\end{align*}
uniformly over the support of $X_j$.

I will also impose that the parameter is well-identified by the criterion function in the following sense: there exists some constant $\kappa > 0$ such that for all $J$, $\Vert \bar g_J(\chi) \Vert_2 \geq \kappa \Vert \chi - \chi_0 \Vert_2$.

Finally, I impose additional regularity assumptions:
\begin{enumerate}
    \item Parameter $\chi_0 = (\alpha_0,\Omega_0) \in \Theta$, for $\Theta$ a compact parameter space with finite dimension $d_\chi$, and $\widehat\chi \in \Theta$.
    \item Instrument $q(X_j)$ is uniformly bounded over the support of $X_j$.
    \item $a(X_j;\alpha)$ is uniformly bounded over the support of $X_j$ and all $\alpha$ in the projection of $\Theta$ onto $\alpha$.
    \item $B(X_j;\Omega)$ is uniformly bounded over the support of $X_j$ and all $\Omega$ in the projection of $\Theta$ onto $\Omega$.
\end{enumerate}

\begin{proof}
Throughout the proof I will use $C,C'$ to denote positive constants that may be different every time they are used.
 
Because $\widehat\chi$ minimizes the sample criterion function $\hat Q_J$, it follows that
\begin{align*}
    \left\Vert \hat g_J(\widehat\chi) \right\Vert_2 &\leq \left\Vert \hat g_J(\chi_0) \right\Vert_2.
\end{align*}

Below I will prove that 
\begin{align*}
    Pr \left( \sup_{\chi \in \Theta} \left\Vert \hat g_J(\chi) - \bar g_J(\chi) \right\Vert_2 > C \sqrt{\frac{\log J}{J}} \right) &\leq \frac{C'}{J^2}.
\end{align*}

Note that $\hat g_J(\chi) = \frac{1}{J} \sum_{j=1}^J g_j(\chi)$. The $g_j(\chi)$ are sub-exponential uniformly over $j$ (the location coordinates are uniformly sub-Gaussian because $Y_j$ is uniformly sub-Gaussian and $a(X_j;\alpha)$ is bounded, the scale coordinates are uniformly sub-exponential because $Y_jY_j'$ is uniformly sub-exponential and other terms are uniformly sub-Gaussian or bounded, and the dimension is fixed). 
Furthermore $\hat g_J(\chi) - \bar g_J(\chi)$ has zero mean. Then by Bernstein's inequality (see, for example, \cite{vershynin2018high} Exercise 5.23),
\begin{align*}
    Pr \left( \left\Vert \hat g_J(\chi) - \bar g_J(\chi) \right\Vert_2 > C \sqrt{\frac{\log J}{J}} \right) &\leq \frac{C'}{J^2}
\end{align*}
for each $\chi \in \Theta$.

Let $\chi, \tilde \chi \in \Theta$ be arbitrary. By assumption $\left\Vert a(X_j;\tilde\alpha) - a(X_j;\alpha)\right\Vert_2 \leq C_3 \left\Vert \alpha - \tilde\alpha\right\Vert_2$ and
\begin{align*}
    &\left\Vert (Y_j - a(X_j;\alpha))(Y_j - a(X_j;\alpha))' - (Y_j - a(X_j;\tilde\alpha))(Y_j - a(X_j;\tilde\alpha))' + B(X_j;\tilde\Omega) - B(X_j;\Omega)\right\Vert_F \\
    &=\bigg\Vert (Y_j - a(X_j;\alpha))(a(X_j;\tilde\alpha) - a(X_j;\alpha))' + (a(X_j;\tilde\alpha) - a(X_j;\alpha))(Y_j - a(X_j;\tilde\alpha))' \\
    &\qquad + B(X_j;\tilde\Omega) - B(X_j;\Omega) \bigg\Vert_F \\
    &\leq C(1+\left\Vert Y_j \right\Vert_2) \left\Vert \alpha - \tilde\alpha\right\Vert_2 + C \left\Vert \Omega - \tilde\Omega \right\Vert_F
\end{align*}
Together with boundedness of $q(X_j)$ this implies that
\begin{align*}
    \left\Vert g_j(\chi) - g_j(\tilde\chi) \right\Vert_2 &\leq C (1+\left\Vert Y_j \right\Vert_2) \left\Vert \chi - \tilde\chi \right\Vert_2.
\end{align*}

Let $\{\chi^1, \dots, \chi^{\mathcal N}\}$ be an $\varepsilon$-net of $\Theta$,  where by standard covering number results (see, for example, \cite{vershynin2018high} Proposition 4.2.10), $\mathcal N \leq \frac{C}{\epsilon^{d_\chi}}$.
At each $\chi^\ell$
\begin{align*}
    Pr \left( \left\Vert \hat g_J(\chi^\ell) - \bar g_J(\chi^\ell) \right\Vert_2 > A \sqrt{\frac{\log J}{J}} \right) &\leq \frac{C}{J^{C'A}},
\end{align*}
so taking a union bound over the net, 
\begin{align*}
    Pr \left( \max_{1 \leq \ell \leq \mathcal N} \left\Vert \hat g_J(\chi^\ell) - \bar g_J(\chi^\ell) \right\Vert_2 > A \sqrt{\frac{\log J}{J}} \right) &\leq \frac{C}{\epsilon^{d_\chi}} \frac{1}{J^{C'A}}
\end{align*}
for some constant $A$.

I set $\varepsilon = c_0 \sqrt{\frac{\log J}{J}}$ for some $c_0$ to be determined later. I then choose $A$ large enough (in particular $C'A > d_\chi/2 + 2$) so that for all $J \geq 3$,
\begin{align*}
    \frac{C}{\epsilon^{d_\chi}} \frac{1}{J^{C'A}} = \frac{C}{c_0^{d_\chi}} \left(\frac{J}{\log J}\right)^{d_\chi/2} \frac{1}{J^{C'A}} \leq \frac{C}{J^2}.
\end{align*}

For each $\chi \in \Theta$ there exists $\ell$ such that $\Vert \chi - \chi^\ell\Vert_2 \leq \varepsilon$. Then
\begin{align*}
    \left\Vert \hat g_J(\chi) - \bar g_J(\chi) \right\Vert_2 &\leq \left\Vert \hat g_J(\chi^\ell) - \bar g_J(\chi^\ell) \right\Vert_2 + \left\Vert \hat g_J(\chi) - \hat g_J(\chi^\ell)\right\Vert_2 + \left\Vert \bar g_J(\chi) - \bar g_J(\chi^\ell) \right\Vert_2 \\
    &\leq \left\Vert \hat g_J(\chi^\ell) - \bar g_J(\chi^\ell) \right\Vert_2 + \frac{1}{J} \sum_{j=1}^J C (1+\left\Vert Y_j \right\Vert_2) \left\Vert \chi - \chi^\ell \right\Vert_2 \\
    &\qquad + \frac{1}{J} \sum_{j=1}^J E\left[ C (1+\left\Vert Y_j \right\Vert_2) \left\Vert \chi - \chi^\ell \right\Vert_2 \right] \\
    &\leq \left\Vert \hat g_J(\chi^\ell) - \bar g_J(\chi^\ell) \right\Vert_2 + C\varepsilon \frac{1}{J} \sum_{j=1}^J (2+\left\Vert Y_j \right\Vert_2)
\end{align*}
because expectations of $\left\Vert Y_j \right\Vert_2$ are uniformly bounded. 
Note that $2+\left\Vert Y_j \right\Vert_2$ is sub-exponential, so I can apply Bernstein's inequality and add back the (bounded) mean to obtain
\begin{align*}
    Pr \left( \frac{1}{J} \sum_{j=1}^J (2+\left\Vert Y_j \right\Vert_2) > C \right) &\leq \frac{C'}{J^2}.
\end{align*}
Thus on the event that $\frac{1}{J} \sum_{j=1}^J (2+\left\Vert Y_j \right\Vert_2) \leq C$ and $\max_{1 \leq \ell \leq \mathcal N} \left\Vert \hat g_J(\chi^\ell) - \bar g_J(\chi^\ell) \right\Vert_2 \leq C\sqrt{\frac{\log J}{J}}$, it holds that
\begin{align*}
    \sup_{\chi \in \Theta} \left\Vert \hat g_J(\chi) - \bar g_J(\chi) \right\Vert_2 &\leq C\sqrt{\frac{\log J}{J}} + C\varepsilon.
\end{align*}
Choosing $c_0$ so that $Cc_0 \leq 1$ and recalling $\varepsilon = c_0 \sqrt{\frac{\log J}{J}}$,
\begin{align*}
    \sup_{\chi \in \Theta} \left\Vert \hat g_J(\chi) - \bar g_J(\chi) \right\Vert_2 &\leq C\sqrt{\frac{\log J}{J}}.
\end{align*}
Then applying union bound,
\begin{align*}
    Pr \left( \sup_{\chi \in \Theta} \left\Vert \hat g_J(\chi) - \bar g_J(\chi) \right\Vert_2 > C \sqrt{\frac{\log J}{J}} \right) &\leq \frac{C'}{J^2}.
\end{align*}

On the event that $\sup_{\chi \in \Theta} \left\Vert \hat g_J(\chi) - \bar g_J(\chi) \right\Vert_2 \leq C \sqrt{\frac{\log J}{J}}$, by the assumptions
\begin{align*}
    \Vert \widehat\chi - \chi_0 \Vert_2 &\leq \frac{1}{\kappa} \Vert \bar g_J(\widehat\chi) \Vert_2 \\
    &\leq \frac{1}{\kappa} \Vert \hat g_J(\widehat\chi) \Vert_2 + \frac{1}{\kappa} \Vert \bar g_J(\widehat\chi) - \hat g_J(\widehat\chi) \Vert_2 \\
    &\leq \frac{1}{\kappa} \Vert \hat g_J(\chi_0) \Vert_2 + \frac{C}{\kappa}  \sqrt{\frac{\log J}{J}} \\
    &= \frac{1}{\kappa} \Vert \hat g_J(\chi_0) - \bar g_J(\chi_0) \Vert_2 + \frac{C}{\kappa}  \sqrt{\frac{\log J}{J}} \\
    &\leq \frac{C}{\kappa} \sqrt{\frac{\log J}{J}}.
\end{align*}
By assumption,
\begin{align*}
    \left\Vert \widehat\chi - \chi_0\right\Vert_J \leq C \max\{ \left\Vert \widehat\alpha - \alpha_0 \right\Vert_2, \left\Vert \widehat\Omega - \Omega_0 \right\Vert_F \} \leq C \left\Vert \widehat\chi - \chi_0 \right\Vert_2.
\end{align*}
Thus 
\begin{align*}
    Pr \left( \left\Vert \widehat \chi - \chi_0 \right\Vert_J > C \sqrt{\frac{\log J}{J}} \right) &\leq \frac{C'}{J^2},
\end{align*}
as desired.

\end{proof}

\section{Appendix: Data and Estimation Details}
\label{app:data}

I obtain data from the replication packages of \citeappendix{hendren2020unified} and \citeappendix{hahn2024welfare}.
I restrict to the set of policies with non-missing point estimates and bootstrap uncertainty estimates for both WTP and net cost and with nonzero program costs. For policies from \citeappendix{hendren2020unified}, I restrict to the baseline policies defined in that paper. For policies from \citeappendix{hahn2024welfare}, I restrict to the policies in the baseline sample, as defined in that paper. I also exclude all policies classified by \citeappendix{hahn2024welfare} as international because the model proposed in this paper considers a single policymaker choosing among policies in a common policy environment. This results in 127 policies, 52 of which are climate policies from \citeappendix{hahn2024welfare}, 32 of which are child policies from \citeappendix{hendren2020unified}, and 43 of which are adult policies from \citeappendix{hendren2020unified}. I classify policies from \citeappendix{hendren2020unified} as child policies based on their program category or if the average age of beneficiaries is below 20; the remaining policies are classified as adult policies. For policies from \citeappendix{hahn2024welfare}, I present results taking the social cost of carbon to be \$193, including profits, excluding learning-by-doing gains, and excluding energy savings. 

Using the replication code for both \citeappendix{hendren2020unified} and \citeappendix{hahn2024welfare}, I construct 1000 bootstrap draws of WTP and net cost for each policy $j$. I normalize WTP and net cost by program cost, which I take to be fixed. I calculate each sampling uncertainty matrix $\Sigma_j$ with the estimated covariance matrix of the bootstrap draws, truncating eigenvalues of $\Sigma_j$ smaller than $10^{-4}$ at $10^{-4}$.

I allow the location and scale of the distribution of true policy impacts to differ across climate, child, and adult policies. For the parametric empirical Bayes approach, I assume that WTP and net cost are jointly Gaussian conditional on policy type. For each policy type $t$ I estimate the type-specific mean $\widehat a_t$ and covariance matrix $\widehat B_t$ by maximizing the marginal likelihood of the data. I parameterize the covariance matrix $\widehat B_t$ using its Cholesky decomposition. I truncate the eigenvalues of each $\widehat B_t$ close to zero, setting all eigenvalues smaller than $0.01$ equal to $0.01$. I then calculate posterior means with the Normal-Normal shrinkage formula using the type-specific location and scale estimates.

For the nonparametric empirical Bayes approach, I estimate the type-specific location and scale parameters using a minimum-distance approach, as described in Section \ref{sec_3.3.2}. I implement a feasible two-step weighted minimum-distance estimator separately for each policy type. In the first step, I use the unweighted sample moments. In the second step, I re-estimate the location and scale parameters using policy-specific inverse variance weights constructed from the preliminary estimates and the known sampling uncertainty $\Sigma_j$. Because the resulting scale estimates may not be positive semi-definite, I truncate the eigenvalues of each $\widehat B_t$ close to zero, setting all eigenvalues smaller than $0.01$ equal to $0.01$.
I standardize policy estimates using $\widehat a_t$ and the truncated $\widehat B_t$, then estimate the common prior using the NPMLE approach described in Section \ref{sec_3.3.2}. I calculate posterior means from the estimated prior, then transform the posterior means back to their original scale.

To construct local spending rules I assume $\eta_j = 1$ for all $j$ and present results for two different values $\mu = 0.5$ and $\mu = 3$. I construct decision rules for consideration sets $V = \mathcal B_2 = \{v : \Vert v \Vert_2 \leq 1\}$ and $V = \mathcal B_\infty = \{v : \Vert v \Vert_\infty \leq 1\}$.

\section{Appendix: Unbiased Estimation of Welfare with Coupled Bootstrap} \label{app:coupled}

Given any local spending rule, one may be interested in knowing the value of the planner's local objective for that local spending rule, that is, the rate of increase in net welfare impact along that local spending rule. The local objective along any local spending rule depends on the true benefits and net costs of policies, which are unknown. However, it is possible to obtain an unbiased estimate of the planner's local objective along any local spending rule. I follow \citeappendix{chen2022empirical} and describe how to implement the coupled bootstrap procedure of \citeappendix{oliveira2024unbiased} for unbiased estimation of the rate of increase in my setting.\footnote{This coupled bootstrap procedure falls into the broader category of data fission methods, which allow for distributions that are not Gaussian \citepappendix{leiner2023data}.}

As a reminder, I suppose that the sample estimates $(\widehat{WTP}_j, \widehat{G}_j)$ of benefits and net costs of each policy change $j$ are independent across policies and conditionally Gaussian, as summarized in \eqref{eq:likelihood1}:
\begin{equation*}
\begin{gathered}
    \begin{pmatrix} \widehat{WTP}_j \\ \widehat{G}_j \end{pmatrix} \Bigg\vert \begin{pmatrix} {WTP}_j \\ {G}_j \end{pmatrix}, X_j, \Sigma_j \overset{\text{ind.}}{\sim} N \left(\begin{pmatrix} {WTP}_j \\ {G}_j \end{pmatrix}, \Sigma_j \right), \qquad j = 1,\dots,J.
\end{gathered}
\end{equation*}
Let $\Xi_j \sim N(0, I_2) \in \R^2$ be i.i.d. draws from a standard multivariate Gaussian that are independent of $(\widehat{WTP}_j, \widehat{G}_j)$. Then for some $\kappa > 0$, construct random vectors $\widehat{Y}_j^{(1)}$ and $\widehat{Y}_j^{(2)}$ by adding and subtracting a scaled version of the Gaussian noise $\Xi_j$ from the sample estimates:
\begin{equation}
\begin{gathered} \label{eq:coupled}
    \widehat{Y}_j^{(1)} = \begin{pmatrix} \widehat{WTP}_j^{(1)} \\ \widehat{G}_j^{(1)} \end{pmatrix} \equiv \begin{pmatrix} \widehat{WTP}_j \\ \widehat{G}_j \end{pmatrix} + \sqrt{\kappa} \Sigma_j^{1/2} \Xi_j,\\ 
    \widehat{Y}_j^{(2)} = \begin{pmatrix} \widehat{WTP}_j^{(2)} \\ \widehat{G}_j^{(2)} \end{pmatrix} \equiv \begin{pmatrix} \widehat{WTP}_j \\ \widehat{G}_j \end{pmatrix} - \frac{1}{\sqrt{\kappa}} \Sigma_j^{1/2} \Xi_j. 
\end{gathered}
\end{equation}
The random vectors $(\widehat{Y}_j^{(1)}, \widehat{Y}_j^{(2)})$ are coupled bootstrap draws that are conditionally unbiased for the true benefit and net cost and are conditionally uncorrelated.
In particular, they satisfy
\begin{gather}
\begin{aligned}
    \widehat{Y}_j^{(1)} \Bigg\vert \begin{pmatrix} {WTP}_j \\ {G}_j \end{pmatrix}, X_j, \Sigma_j \overset{\text{ind.}}{\sim} N \left(\begin{pmatrix} {WTP}_j \\ {G}_j \end{pmatrix}, (1+\kappa) \Sigma_j \right), \\ 
    \widehat{Y}_j^{(2)} \Bigg\vert \begin{pmatrix} {WTP}_j \\ {G}_j \end{pmatrix}, X_j, \Sigma_j \overset{\text{ind.}}{\sim} N \left(\begin{pmatrix} {WTP}_j \\ {G}_j \end{pmatrix}, \left(1+\frac{1}{\kappa} \right) \Sigma_j \right), 
\end{aligned}
\label{eq:coupled_likelihood} \\
    Cov\left(\widehat{Y}_j^{(1)}, \widehat{Y}_j^{(2)} \bigg\vert {WTP}_j, G_j, X_j, \Sigma_j \right) = 0_{2 \times 2}, \label{eq:zero_cov_coupled}
\end{gather}
where I use $0_{n \times n}$ to denote the $n \times n$ zero matrix. As discussed in \citeappendix{chen2022empirical}, the coupled bootstrap approach is similar to a sample-splitting approach where $\widehat{Y}_j^{(1)}$ are training set estimates, $\widehat{Y}_j^{(2)}$ are testing set estimates, and $\kappa$ governs the train-test split.\footnote{To see why, \citeappendix{chen2022empirical} notes that if one splits i.i.d. micro-level data $\widehat{Y}_{jk}, k = 1,\dots,n_j$ into two sets (train and test) with proportions $\frac{1}{\kappa+1}$ and $\frac{\kappa}{\kappa+1}$ and takes $\widehat{Y}_j^{(1)}$ and $\widehat{Y}_j^{(2)}$ to be the sample means on those sets respectively, the central limit theorem suggests that \eqref{eq:coupled_likelihood} will approximately hold.}
Note that for the coupled bootstrap approach I only need the Gaussianity assumption \eqref{eq:likelihood1} and not the entire location-scale model of \eqref{eq:likelihood2}. The assumption of exact conditional Gaussianity of the sample estimates is already required for validity of the empirical Bayes estimation procedure. This means that the coupled bootstrap approach provides unbiased performance estimates of any local spending rule under much weaker assumptions than those needed to estimate the optimal local spending rule. 

These properties imply that if I construct the local spending rule using draws $\widehat{Y}_j^{(1)}$ and evaluate the planner's local objective along that local spending rule as if the true parameters were $\widehat{Y}_j^{(2)}$, the resulting estimate of the rate of increase will be unbiased.
Recall that the planner's local objective along a vector $v$ is given by
\begin{align*}
   \langle E_\pi \left[\nabla w \right], v \rangle. 
\end{align*}
Consider any local spending rule $v \left( \widehat{Y}_{1:J}^{(1)} \right)$ that depends only on the coupled bootstrap draw $(\widehat{Y}_{1}^{(1)}, \dots, \widehat{Y}_{J}^{(1)})$. For example, $v \left( \widehat{Y}_{1:J}^{(1)} \right)$ could be the empirical Bayes local spending rule where empirical Bayes estimates are constructed using only $(\widehat{Y}_{1}^{(1)}, \dots, \widehat{Y}_{J}^{(1)})$ under the model \eqref{eq:coupled_likelihood}.
An estimator for the planner's local objective along local spending rule $v \left( \widehat{Y}_{1:J}^{(1)} \right)$ is
\begin{align}\label{eq:coupled_estimator}
    \sum_{j=1}^J v_j \left( \widehat{Y}_{1:J}^{(1)} \right) \widehat{Y}_{j}^{(2)\prime}\begin{pmatrix} \eta_j \\ -\mu \end{pmatrix},
\end{align}
where $v_j \left( \widehat{Y}_{1:J}^{(1)} \right)$ denotes the $j$th component of $v\left( \widehat{Y}_{1:J}^{(1)} \right)$.
The following proposition, which is essentially Proposition 1 from \citeappendix{chen2022empirical}, shows that the above estimator is unbiased for the planner's local objective along local spending rule $v \left( \widehat{Y}_{1:J}^{(1)} \right)$.

\begin{proposition}\label{prop:coupled}
    Suppose the sample estimates $(\widehat{WTP}_j, \widehat{G}_j)$ are conditionally independent and Gaussian as in \eqref{eq:likelihood1}. For some $\kappa > 0$, construct coupled bootstrap draws $\widehat{Y}_{1:J}^{(1)}, \widehat{Y}_{1:J}^{(2)}$ as in \eqref{eq:coupled}. Let $\mathcal F \equiv \left\{WTP_{1:J}, G_{1:J}, \Sigma_{1:J}, X_{1:J}, \widehat{Y}_{1:J}^{(1)} \right\}$.
    Then for any local spending rule $v \left( \widehat{Y}_{1:J}^{(1)} \right)$ that is a function of $\widehat{Y}_{1:J}^{(1)}$, the estimator given by \eqref{eq:coupled_estimator} satisfies
    \begin{align*}
        E \left[ \sum_{j=1}^J v_j \left( \widehat{Y}_{1:J}^{(1)} \right) \widehat{Y}_{j}^{(2)\prime}\begin{pmatrix} \eta_j \\ -\mu \end{pmatrix} \bigg \vert \mathcal F \right] &= \left\langle \nabla w, v \left( \widehat{Y}_{1:J}^{(1)} \right) \right\rangle.
    \end{align*}
\end{proposition}

\begin{proof}
Observe that $\widehat{Y}_{j}^{(1)}$ and $\widehat{Y}_{j}^{(2)}$ are independent conditional on $(WTP_j, G_j, X_j, \Sigma_j)$ because \eqref{eq:coupled_likelihood} and \eqref{eq:zero_cov_coupled} hold, so
    \begin{align*}
        E \left[ \sum_{j=1}^J v_j \left( \widehat{Y}_{1:J}^{(1)} \right) \widehat{Y}_{j}^{(2)\prime}\begin{pmatrix} \eta_j \\ -\mu \end{pmatrix} \bigg \vert \mathcal F \right] &= \sum_{j=1}^J v_j \left( \widehat{Y}_{1:J}^{(1)} \right) E \left[\widehat{Y}_{j}^{(2)\prime}\bigg \vert \mathcal F \right]\begin{pmatrix} \eta_j \\ -\mu \end{pmatrix} \\
        &= \sum_{j=1}^J v_j \left( \widehat{Y}_{1:J}^{(1)} \right) \begin{pmatrix} WTP_j \\ G_j \end{pmatrix}' \begin{pmatrix} \eta_j \\ -\mu \end{pmatrix} \\
        &= \left\langle \nabla w, v \left( \widehat{Y}_{1:J}^{(1)} \right) \right\rangle.
    \end{align*}
\end{proof}

\section{Appendix: Local Spending Rules} \label{app:spending_rule_results}

\begin{figure}[!htbp]
\caption{Local spending rule results for $V = \mathcal{B}_2$}
\label{fig:size_and_mag_B2}
\centering
\begin{subfigure}[t]{0.49\textwidth}
  \centering
  \caption{Sample plug-in rule, $\mu = 0.5$}
  \includegraphics[width=\linewidth]{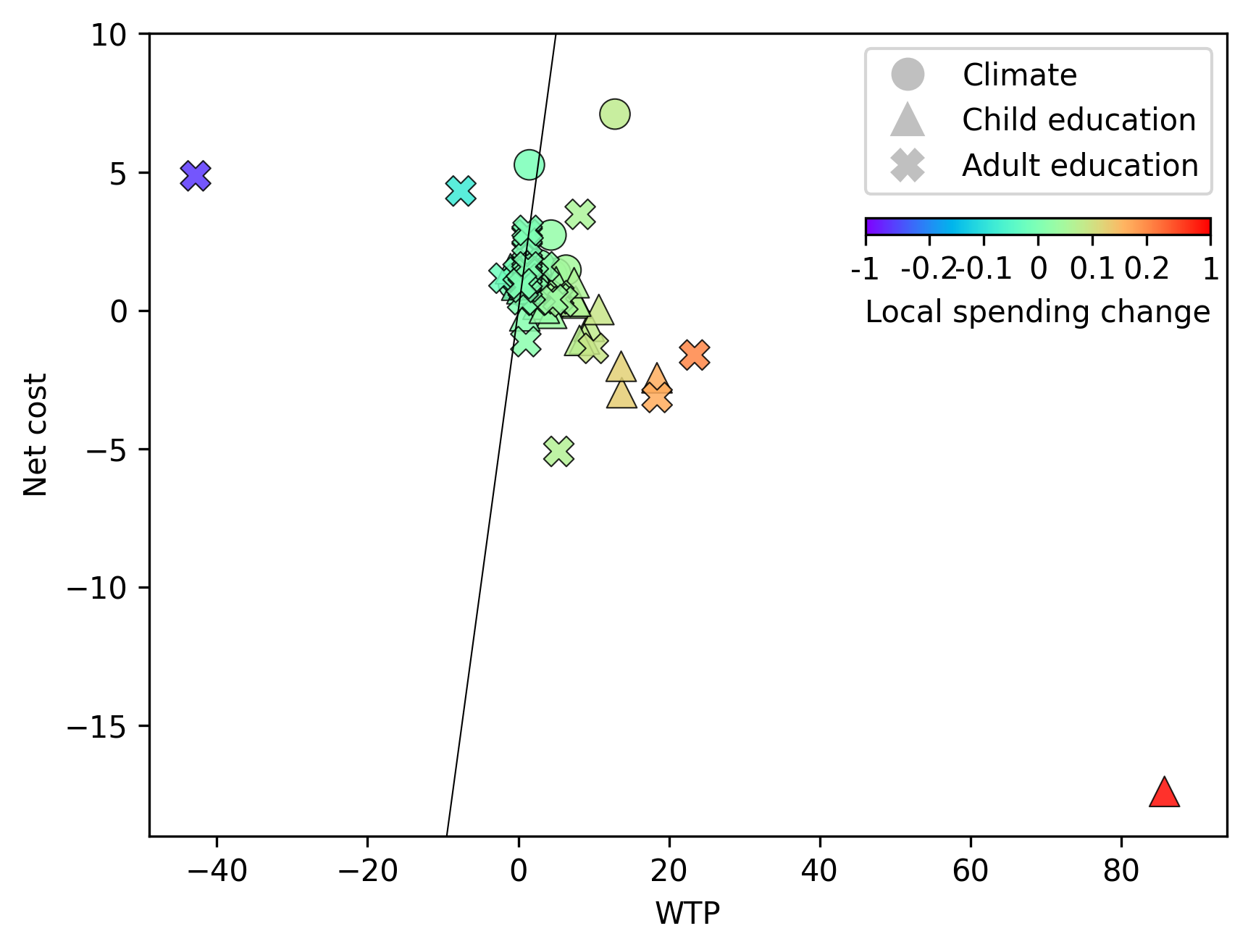}
\end{subfigure}
\begin{subfigure}[t]{0.49\textwidth}
  \centering
  \caption{Sample plug-in rule, $\mu = 3$}
  \includegraphics[width=\linewidth]{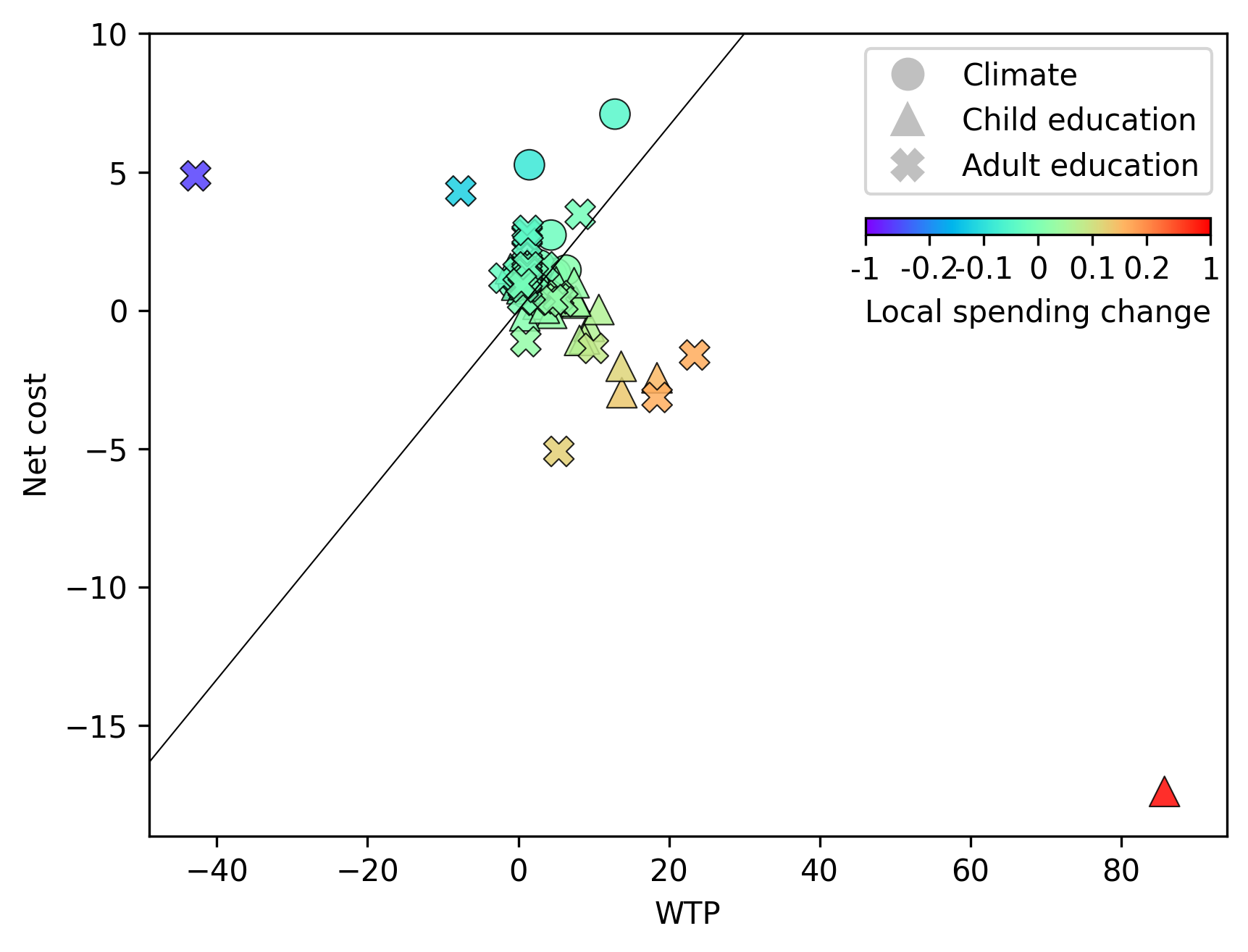}
\end{subfigure}
\begin{subfigure}[t]{0.49\textwidth}
  \centering
  \caption{NPEB rule, $\mu = 0.5$}
  \includegraphics[width=\linewidth]{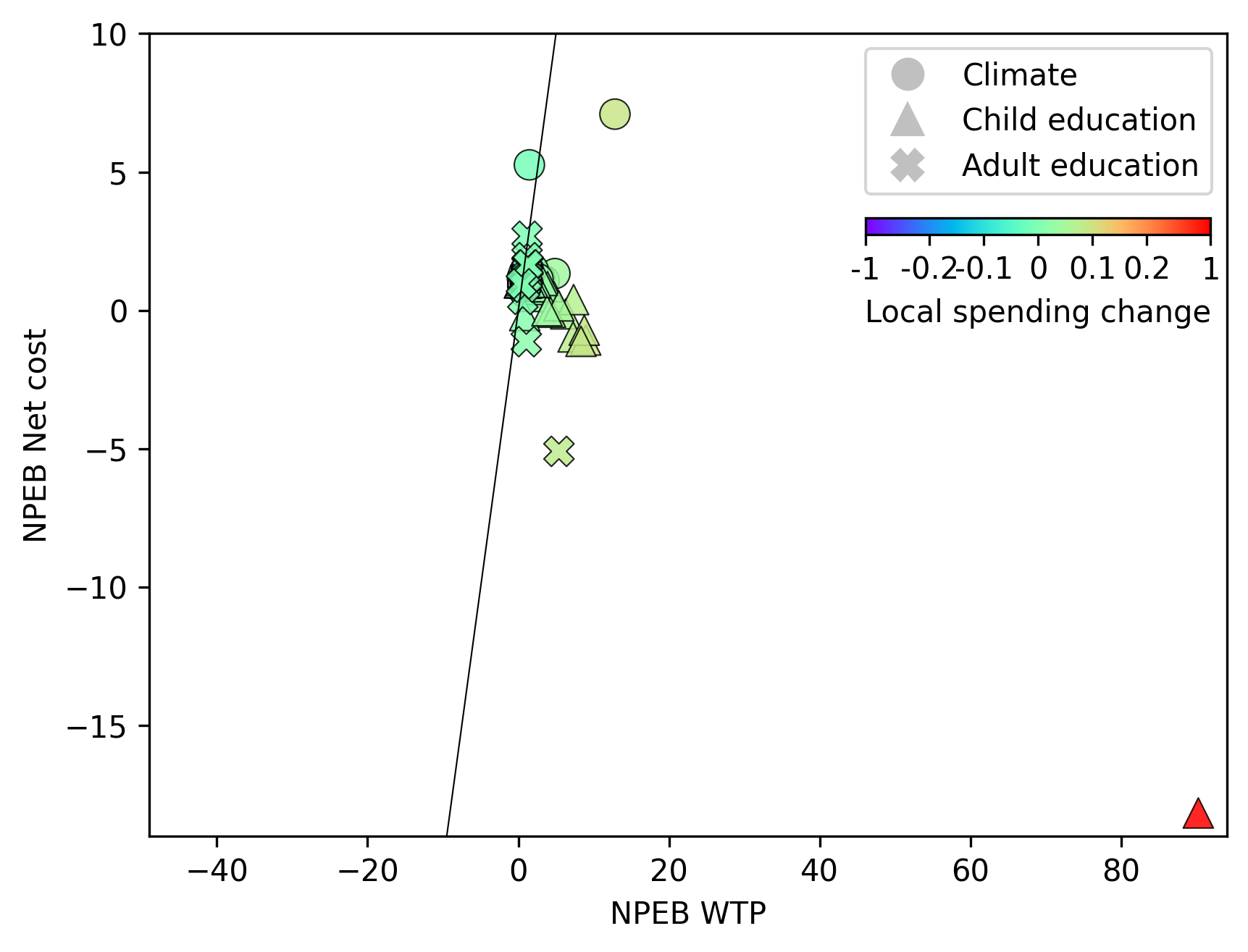}
\end{subfigure}
\begin{subfigure}[t]{0.49\textwidth}
  \centering
  \caption{NPEB rule, $\mu = 3$}
  \includegraphics[width=\linewidth]{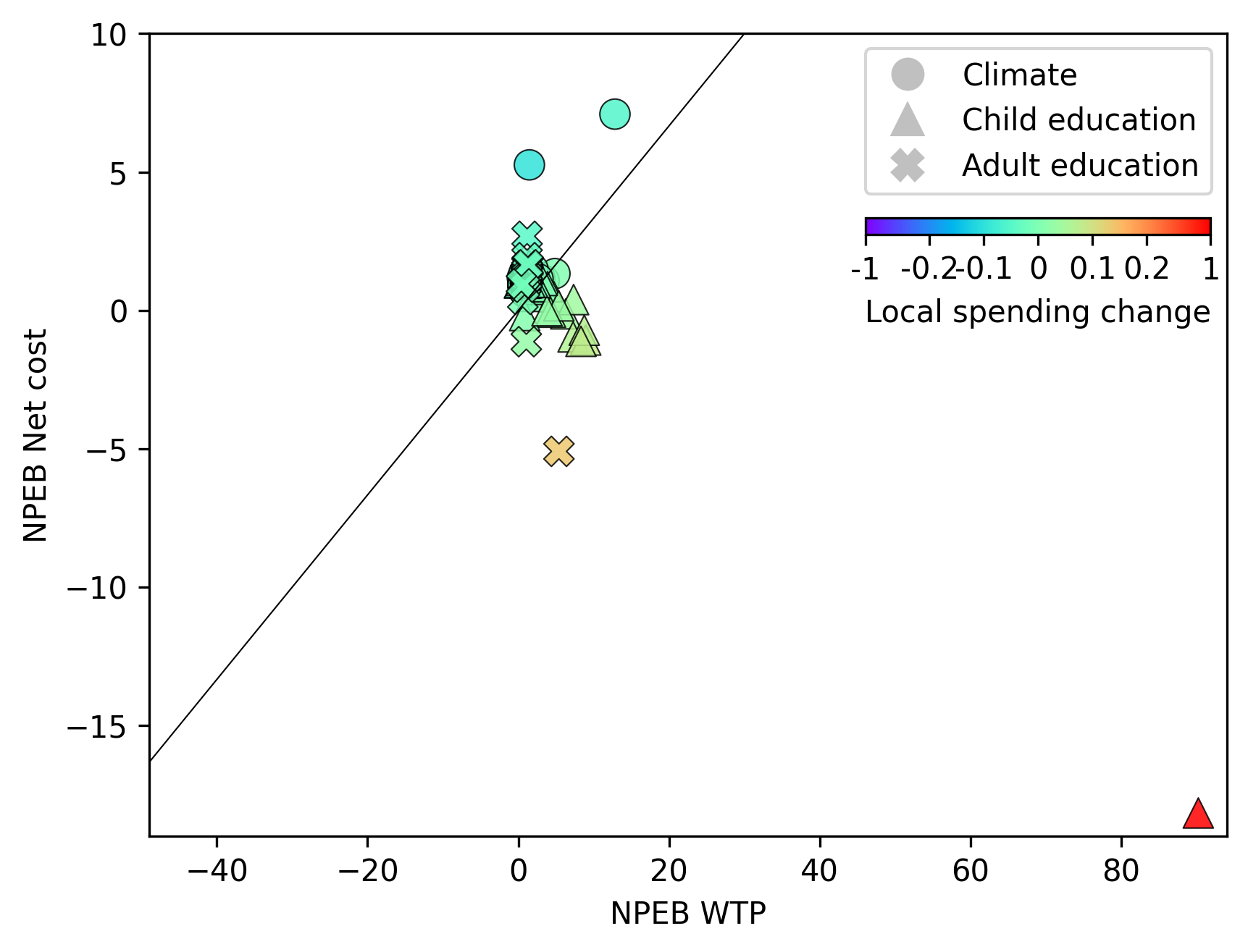}
\end{subfigure}
\begin{subfigure}[t]{0.49\textwidth}
  \centering
  \caption{PEB rule, $\mu = 0.5$}
  \includegraphics[width=\linewidth]{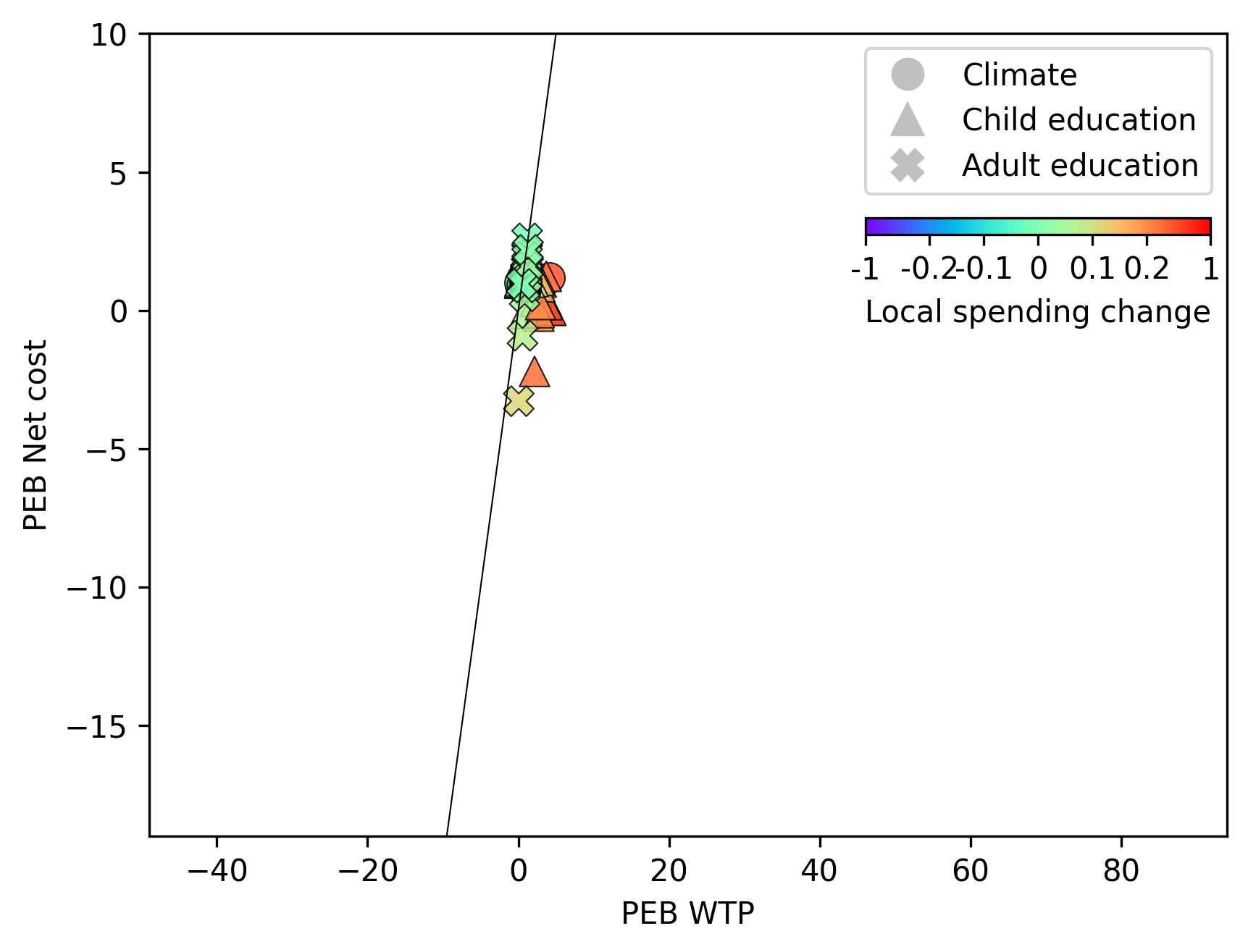}
\end{subfigure}
\begin{subfigure}[t]{0.49\textwidth}
  \centering
  \caption{PEB rule, $\mu = 3$}
  \includegraphics[width=\linewidth]{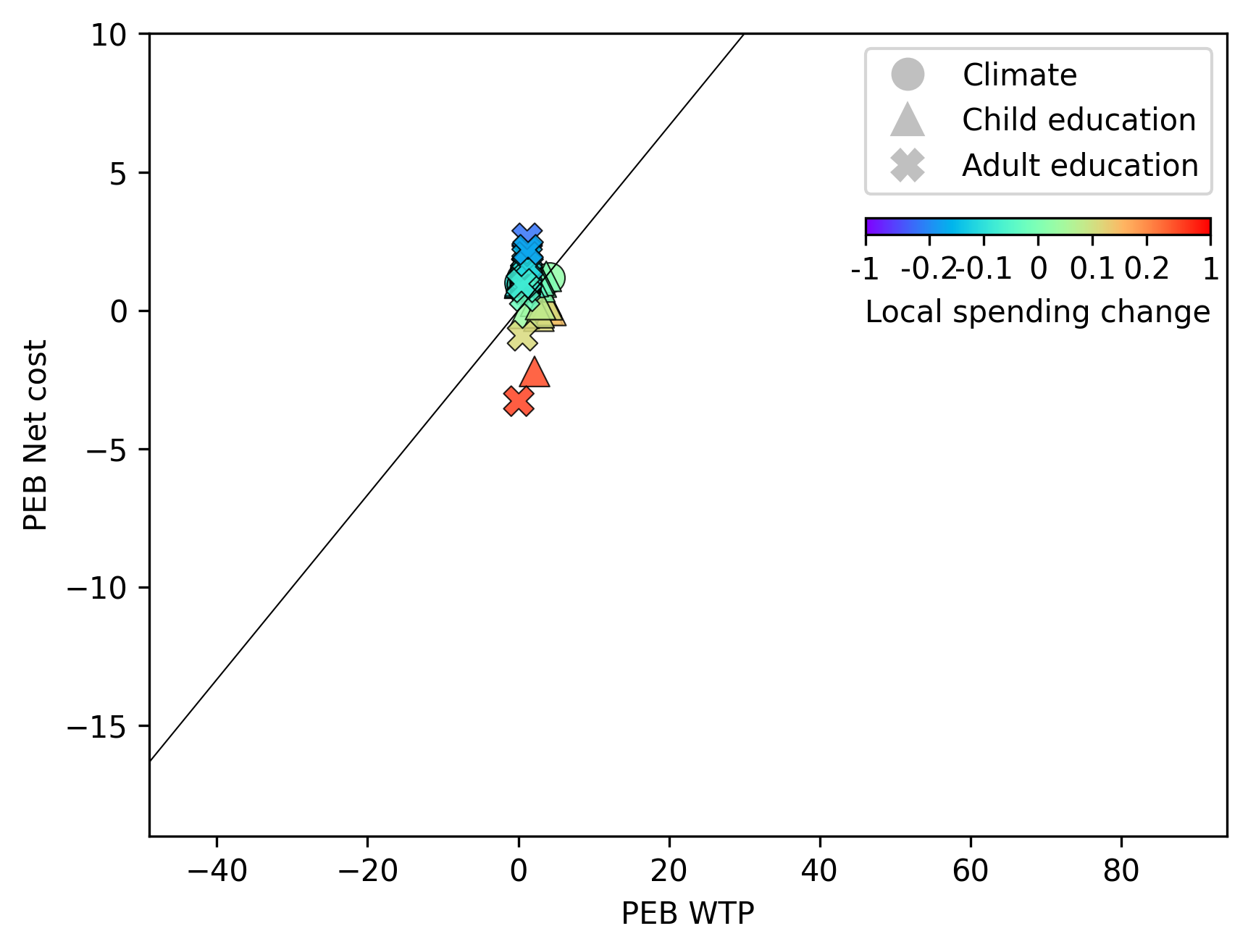}
\end{subfigure}
\begin{minipage}{\textwidth}
    \fontsize{9}{2}\linespread{1}\selectfont
    \footnotesize{
    \textit{Notes}: Each point represents the WTP and net cost for a single policy, for sample estimates in panels (a) and (b), NPEB estimates in panels (c) and (d), and PEB estimates under a Gaussian prior in panels (e) and (f). Point shape is determined by policy type; point color is determined by the local spending change on that corresponding policy. 
    The black line denotes the set of WTP and net costs for which the policy change from the local spending rule is exactly equal to zero. I take $\eta_j = 1$ for all $j$ for all results in this figure. 
    }
    \end{minipage}
\end{figure}

\begin{figure}[!htbp]
\caption{Local spending rule results for $V = \mathcal{B}_\infty$}
\label{fig:size_and_mag_Binf}
\centering
\begin{subfigure}[t]{0.49\textwidth}
  \centering
  \caption{Sample plug-in rule, $\mu = 0.5$}
  \includegraphics[width=\linewidth]{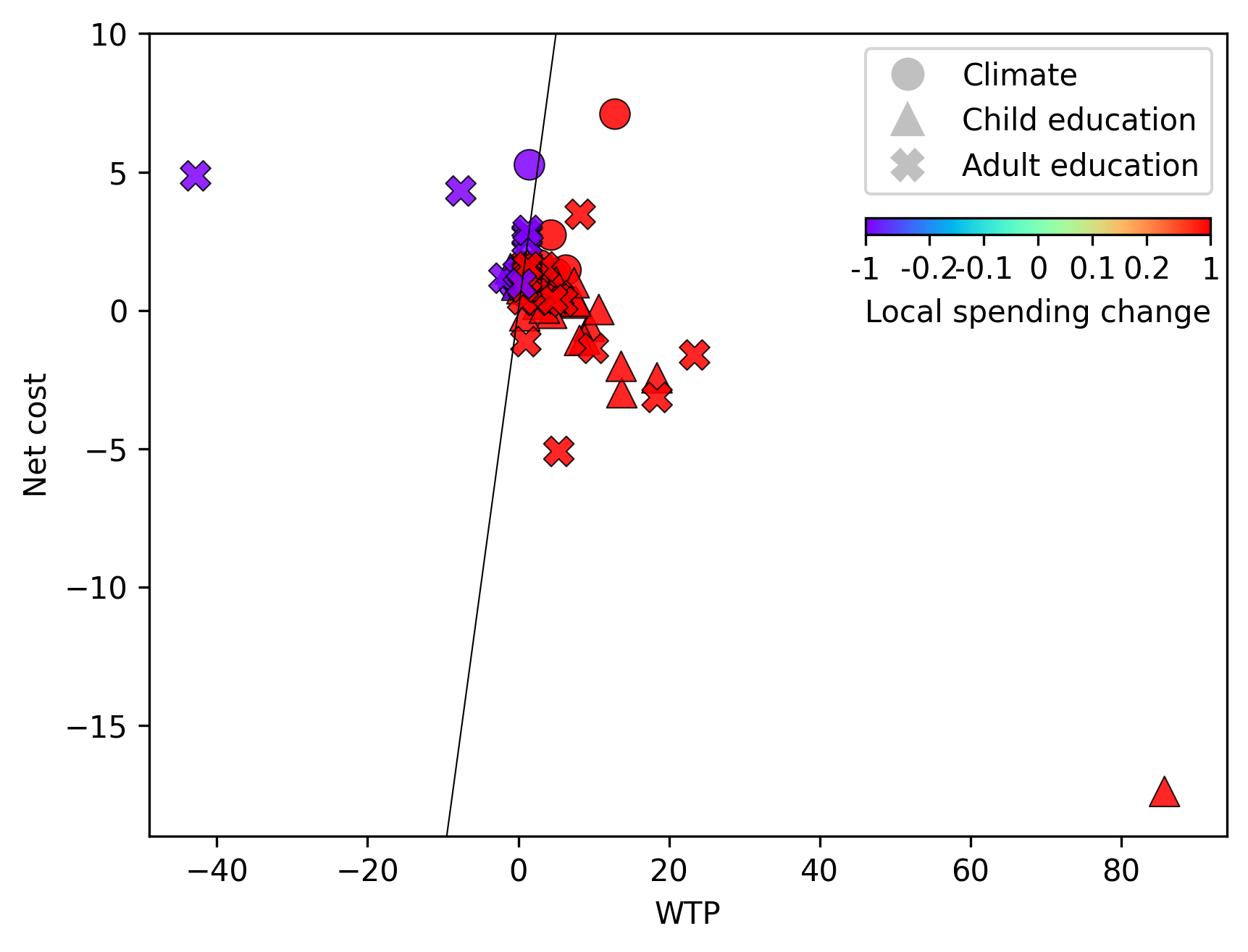}
\end{subfigure}
\begin{subfigure}[t]{0.49\textwidth}
  \centering
  \caption{Sample plug-in rule, $\mu = 3$}
  \includegraphics[width=\linewidth]{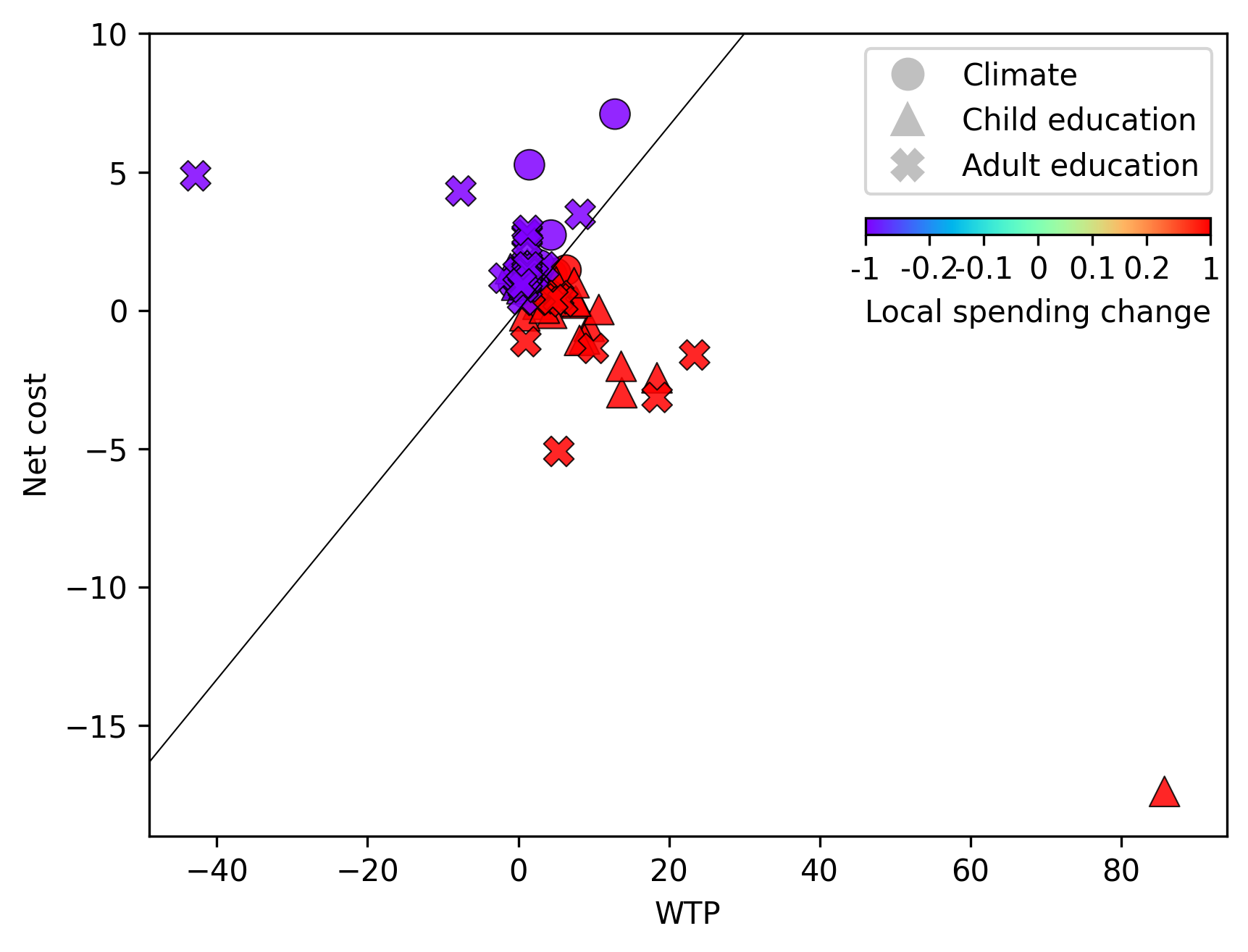}
\end{subfigure}
\begin{subfigure}[t]{0.49\textwidth}
  \centering
  \caption{NPEB rule, $\mu = 0.5$}
  \includegraphics[width=\linewidth]{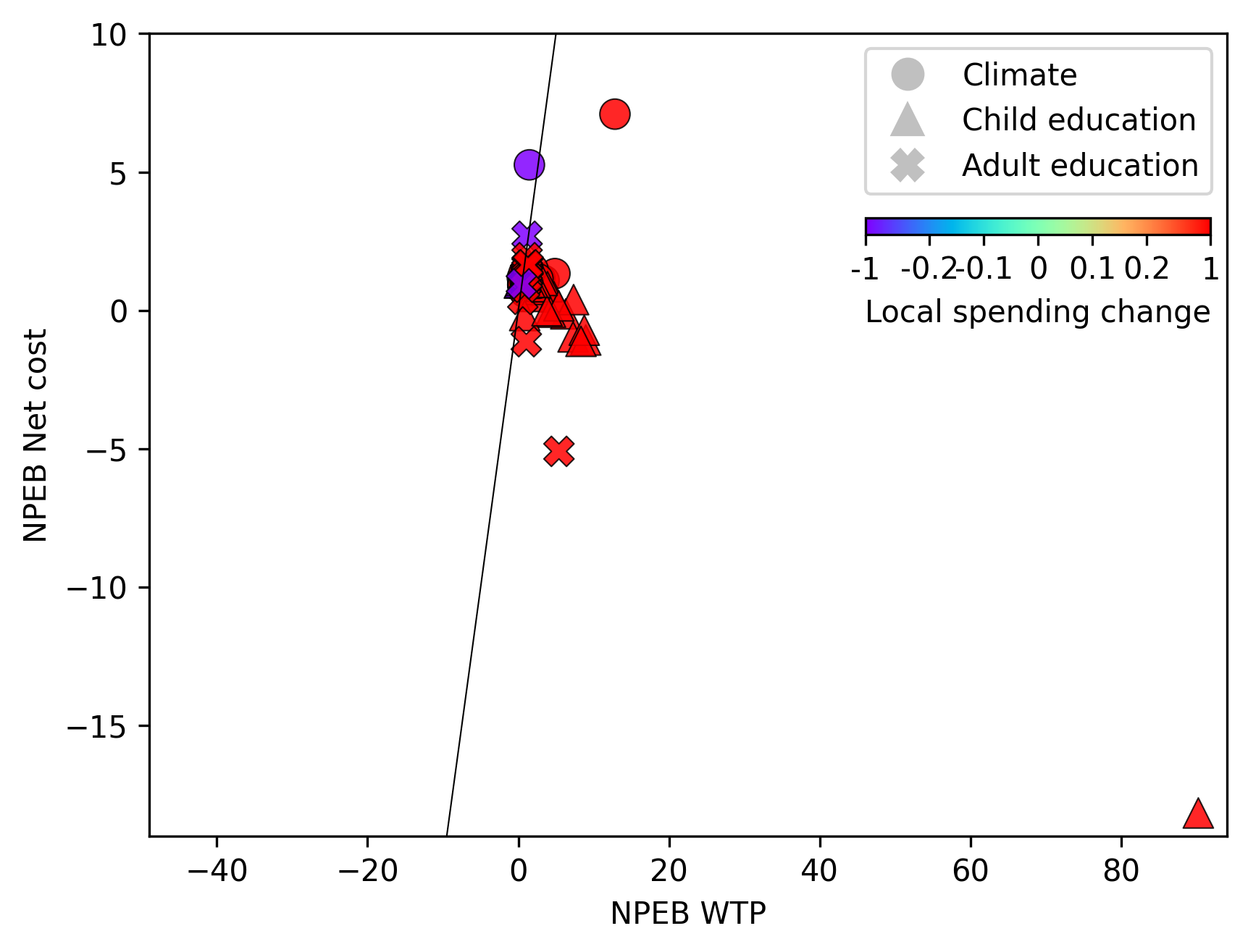}
\end{subfigure}
\begin{subfigure}[t]{0.49\textwidth}
  \centering
  \caption{NPEB rule, $\mu = 3$}
  \includegraphics[width=\linewidth]{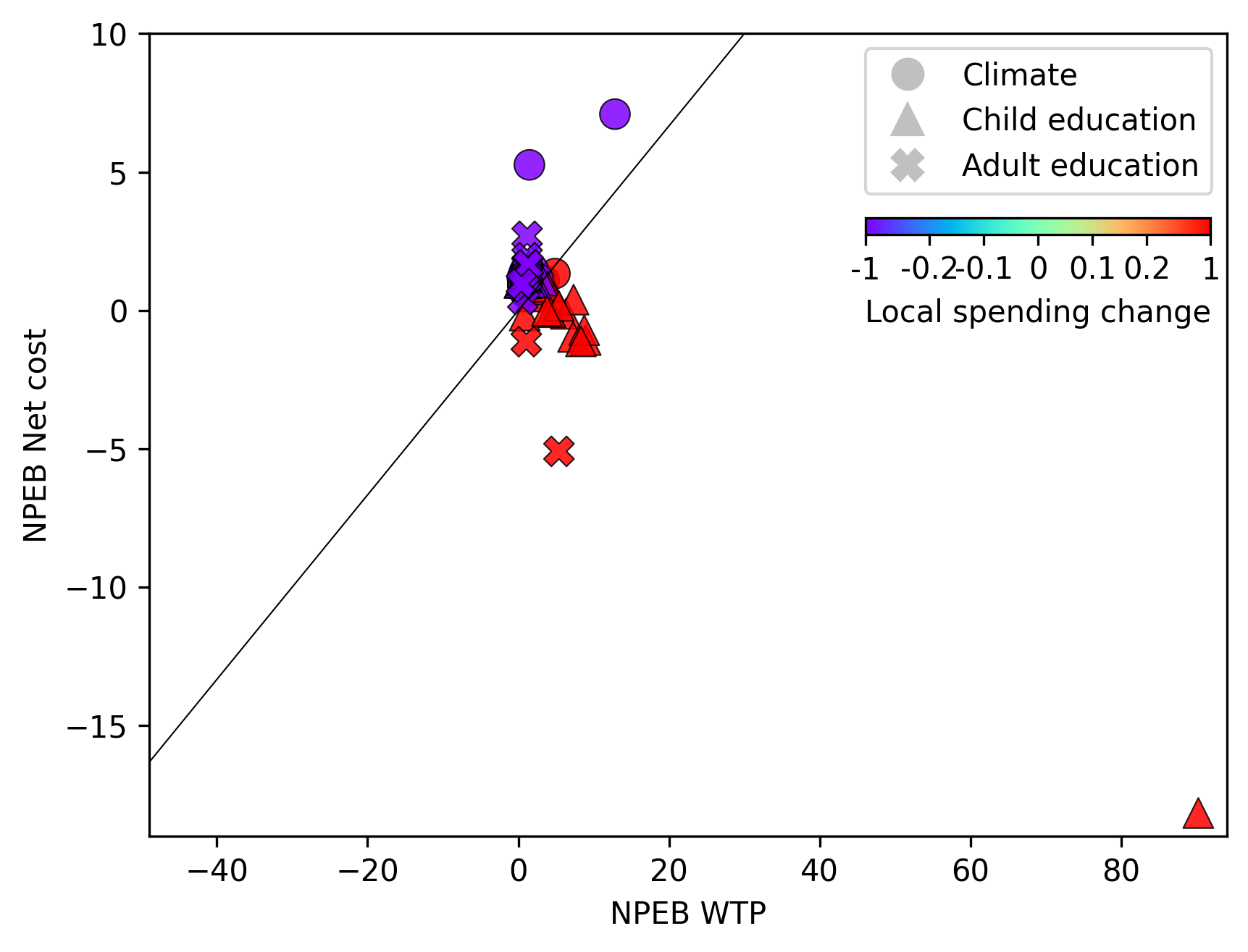}
\end{subfigure}
\begin{subfigure}[t]{0.49\textwidth}
  \centering
  \caption{PEB rule, $\mu = 0.5$}
  \includegraphics[width=\linewidth]{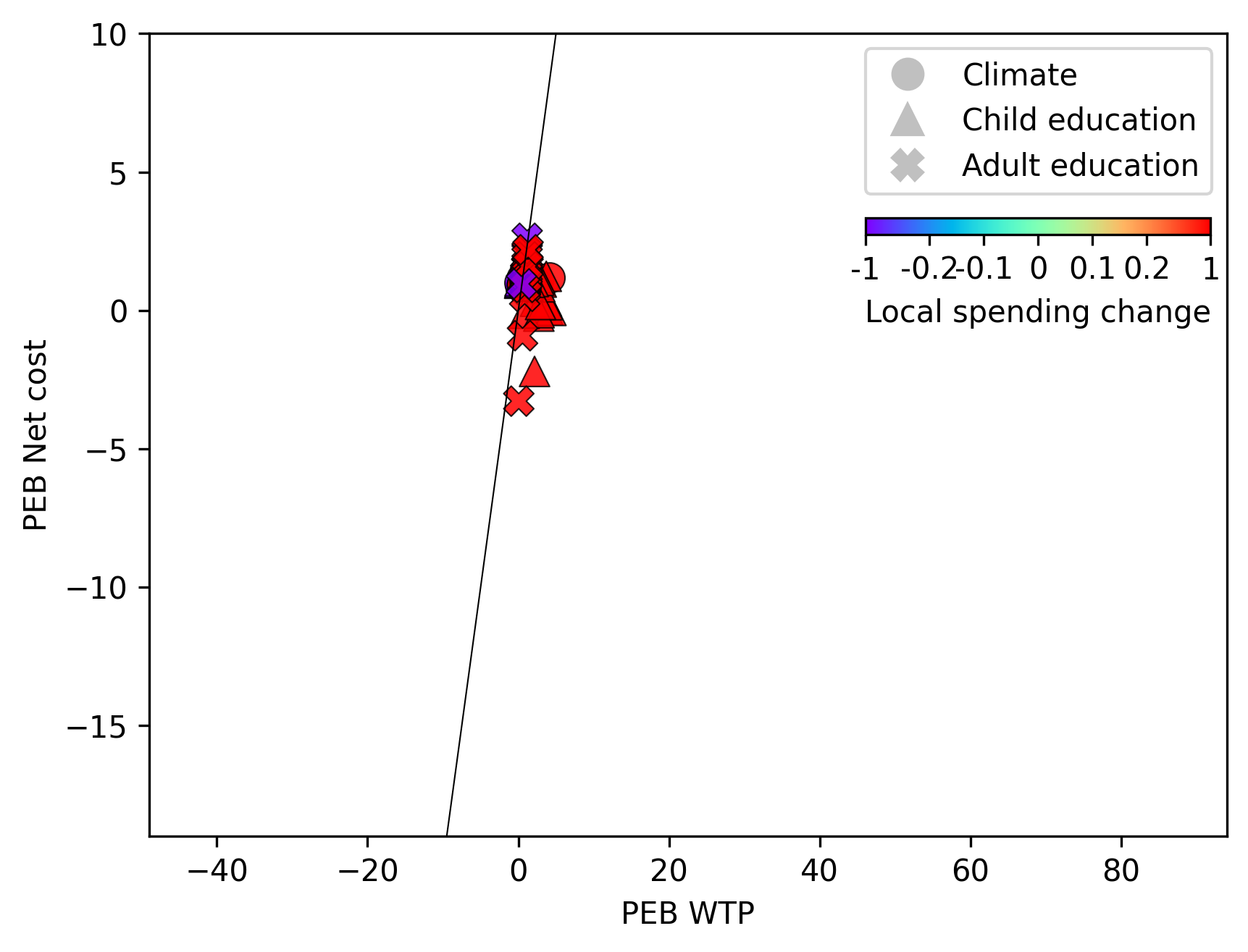}
\end{subfigure}
\begin{subfigure}[t]{0.49\textwidth}
  \centering
  \caption{PEB rule, $\mu = 3$}
  \includegraphics[width=\linewidth]{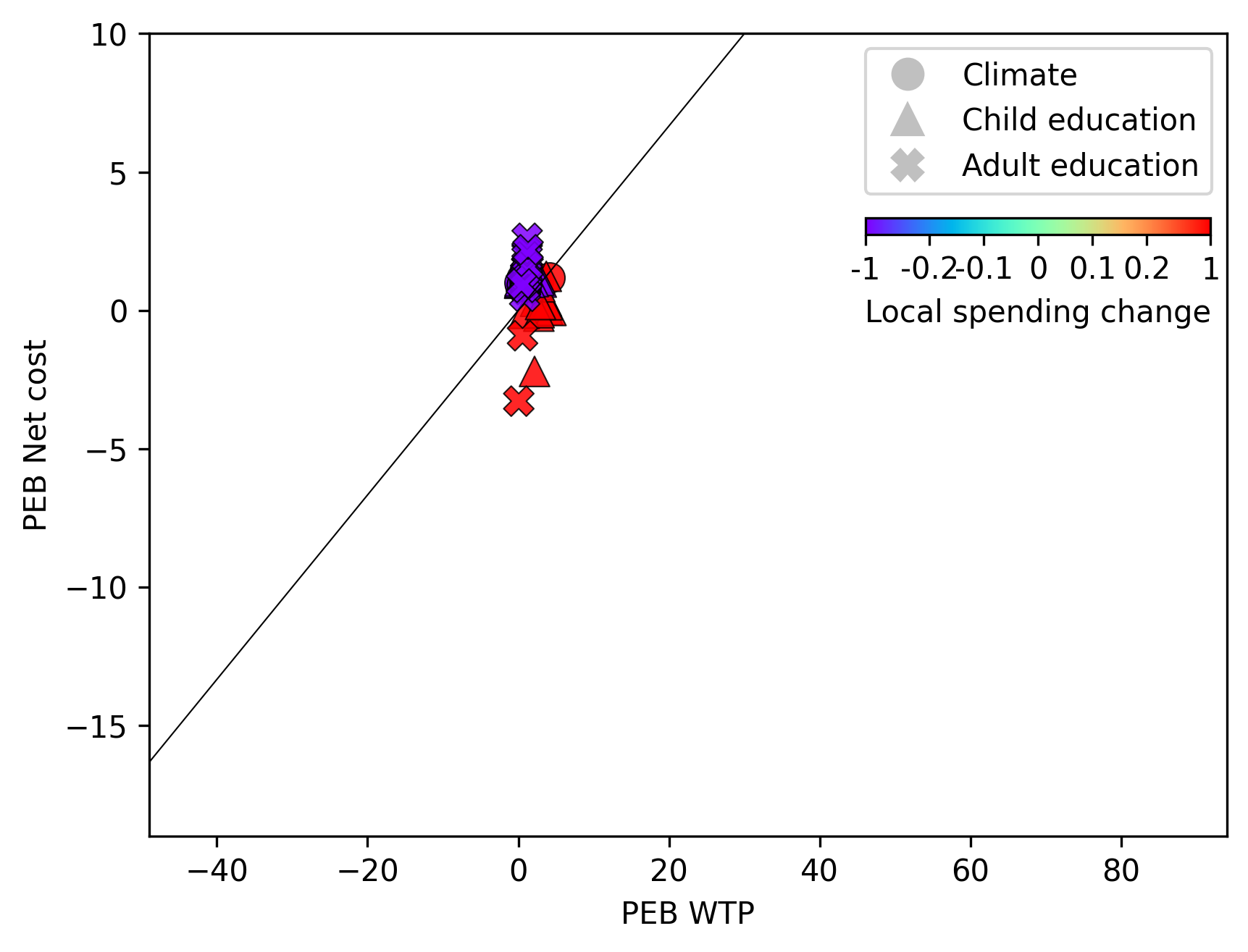}
\end{subfigure}
\begin{minipage}{\textwidth}
    \fontsize{9}{2}\linespread{1}\selectfont
    \footnotesize{
    \textit{Notes}: Each point represents the WTP and net cost for a single policy, for sample estimates in panels (a) and (b), NPEB estimates in panels (c) and (d), and PEB estimates under a Gaussian prior in panels (e) and (f). Point shape is determined by policy type; point color is determined by the local spending change on that corresponding policy. 
    The black line denotes the set of WTP and net costs for which the policy change from the local spending rule is exactly equal to zero. I take $\eta_j = 1$ for all $j$ for all results in this figure. 
    }
    \end{minipage}
\end{figure}

Figures \ref{fig:size_and_mag_B2} and \ref{fig:size_and_mag_Binf} illustrate what the empirical Bayes and sample plug-in local spending rules look like for particular choices of the planner's parameters. These figures plot estimates of net cost against estimates of WTP for sample estimates in panels (a) and (b), NPEB estimates in panels (c) and (d), and PEB estimates under a Gaussian prior in panels (e) and (f). Points are colored according to the sample plug-in local spending change on the corresponding policy in panels (a) and (b), the NPEB local spending change on the corresponding policy in panels (c) and (d), and the PEB local spending change on the corresponding policy in panels (e) and (f). Point shape corresponds to policy type. 
Figure \ref{fig:size_and_mag_B2} corresponds to consideration set $V = \mathcal{B}_2$ and Figure \ref{fig:size_and_mag_Binf} corresponds to consideration set $V = \mathcal{B}_{\infty}$. 
The black line plots the set of benefits and costs such that the local spending change is exactly zero, $WTP = \mu G$, and is meant to visually help distinguish policies with positive and negative local spending changes. I display results for $\mu = 0.5$ in panels (a), (c), and (e), and for $\mu = 3$ in panels (b), (d), and (f).

As the value of $\mu$ increases, all local spending rules make more negative and fewer positive local spending changes. Recall that $\mu$ is the marginal welfare impact from closing the budget. Thus, a larger $\mu$ implies a larger negative welfare impact from closing the budget, inducing the planner to make more negative local spending changes to reduce the size of the budget. 

Results under consideration set $\mathcal{B}_2$ in Figure \ref{fig:size_and_mag_B2} are quite different from those under consideration set $\mathcal{B}_\infty$ in Figure \ref{fig:size_and_mag_Binf}.
Under consideration set $\mathcal{B}_\infty$, for all $\mu$ the local spending rule takes the value $-1$ for all policies to the left of the black line and $1$ for all policies to the right. In contrast, under consideration set $\mathcal{B}_2$ the magnitude of the local spending change increases as the policy lies farther from the line,
while remaining between $-1$ and $1$ with negative changes to the left and positive changes to the right. More generally, under $V = \mathcal{B}_p$ the local spending rule becomes ``spikier'' as $p$ decreases from $\infty$ to $1$, with more homogeneous local spending changes near the black line and larger magnitudes farther away.

Differences between the sample plug-in local spending rule and the empirical Bayes local spending rules are driven by empirical Bayes shrinkage. For clarity I focus on consideration set $V = \mathcal{B}_2$, where these differences are visually more apparent.
Looking across policy types, child education policies are treated similarly by all rules. In contrast, the sample plug-in local spending rule makes several large positive and negative spending changes on adult education policies, while the empirical Bayes local spending rules do not. This difference arises from stronger empirical Bayes shrinkage for adult education policies, leading to greater divergence between the sample plug-in and empirical Bayes rules for adult education policies.

\section{Appendix: Proofs} \label{app:proofs}

Throughout this appendix $K$ denotes an arbitrary positive constant that does not depend on $J$ and may be different every time it is used. 

\subsection{Proof of Lemma \ref{lem:order}}
\begin{proof}
By definition of the dual norm, for any $p \geq 1$ $\max_{v \in \mathcal{B}_p} \langle \nabla w, v \rangle = \left\Vert \nabla w \right\Vert_{\frac{p}{p-1}}$. 

Case: $p=1$

Then 
\begin{align*}
    E \left[ \max_{v \in \mathcal{B}_1} \langle \nabla w, v \rangle \right] &= E\left[\left\Vert \nabla w \right\Vert_{\infty}\right] 
    = E\left[\left\Vert \begin{pmatrix}
        \eta_1 WTP_1 - \mu G_1 \\ \vdots \\ \eta_J WTP_J - \mu G_J
    \end{pmatrix} \right\Vert_{\infty} \right]
    \leq K \sqrt{\log J}
\end{align*}
by Assumptions \ref{ass:compact} and \ref{ass:eta}, using properties of sub-Gaussian variables (see \cite{vershynin2018high} Proposition 2.7.6).

Case: $p = \infty$

Then 
\begin{align*}
    E \left[ \max_{v \in \mathcal{B}_\infty} \langle \nabla w, v \rangle \right] &= E\left[\left\Vert \nabla w \right\Vert_{1}\right] 
    = E\left[\left\Vert \begin{pmatrix}
        \eta_1 WTP_1 - \mu G_1 \\ \vdots \\ \eta_J WTP_J - \mu G_J
    \end{pmatrix} \right\Vert_{1} \right] \\
    &= \sum_{j=1}^J E\left[ | \eta_j WTP_j - \mu G_j| \right]
    \leq KJ
\end{align*}
by Assumptions \ref{ass:compact} and \ref{ass:eta}, using properties of sub-Gaussian variables (see \cite{vershynin2018high} Proposition 2.6.1).

Case: $1 < p < \infty$

Then
\begin{align*}
    E \left[ \max_{v \in \mathcal{B}_p} \langle \nabla w, v \rangle \right] &= E\left[\left\Vert \nabla w \right\Vert_{\frac{p}{p-1}}\right] 
    = E\left[\left\Vert \begin{pmatrix}
        \eta_1 WTP_1 - \mu G_1 \\ \vdots \\ \eta_J WTP_J - \mu G_J
    \end{pmatrix} \right\Vert_{\frac{p}{p-1}} \right] \\
    &\leq \left( E\left[\left\Vert \begin{pmatrix}
        \eta_1 WTP_1 - \mu G_1 \\ \vdots \\ \eta_J WTP_J - \mu G_J
    \end{pmatrix} \right\Vert_{\frac{p}{p-1}}^{\frac{p}{p-1}} \right] \right)^{\frac{p-1}{p}} \\
    &= \left( \sum_{j=1}^J E\left[ |\eta_j WTP_j - \mu G_j|^{\frac{p}{p-1}} \right] \right)^{\frac{p-1}{p}} \\
    &\leq KJ^{\frac{p-1}{p}}\sqrt{\frac{p}{p-1}},
\end{align*}
where the first line follows from Jensen's inequality and the third line follows from Assumptions \ref{ass:compact} and \ref{ass:eta}, using properties of sub-Gaussian variables (see \cite{vershynin2018high} Proposition 2.6.1).

Furthermore, applying H\"older's inequality,
\begin{align*}
    E\left[\left\Vert \nabla w \right\Vert_{\frac{p}{p-1}}\right] &\leq J^{\frac{p-1}{p}} E\left[\left\Vert \nabla w \right\Vert_{\infty} \right] 
    \leq K J^{\frac{p-1}{p}} \sqrt{\log J}.
\end{align*}
Thus $E\left[\left\Vert \nabla w \right\Vert_{\frac{p}{p-1}}\right] \leq K J^{\frac{p-1}{p}} \min\left(\sqrt{\frac{p}{p-1}}, \sqrt{\log J}\right)$.

It can be seen that this bound is sharp by taking $\eta_j = \eta$ to be constant across $j$ and taking $(WTP_j, G_j)$ to be independent Gaussian random vectors (see, for example, \cite{vershynin2018high} Remark 2.7.7).
\end{proof}

\subsection{Proof of Proposition \ref{prop:plug_in}}
\begin{proof}
In what follows all expectations are conditional on $X_j$ and $\Sigma_j$, but the conditioning is omitted for notational simplicity. Throughout the proof I use the shorthand $E_\pi[ \cdot ]$ to denote the posterior expectation, taken conditional on the data $Y_{1:J}$.

First note that for $1 < p < \infty$, $1/N_p = J^{\frac{1-p}{p}} 1/\min\{ \sqrt{\frac{p}{p-1}}, \sqrt{\log J}\} \geq J^{\frac{1-p}{p}} 1/\sqrt{\frac{p}{p-1}}$.

Consider prior distribution $F_0$ that takes on value $(1,1)$ with probability $1$ and does not depend on $X_j$.
Finally suppose that $\eta_j = 1$ and $\mu = 1$ for all $j$.

This means that $\eta_j WTP_j - \mu G_j = WTP_j - G_j$ takes on value $0$ with probability $1$. Note this means that posterior mean $WTP_j^* - G_j^* = 0$ for all values of $\widehat{WTP}_j, \widehat{G}_j$.

Defining $\omega_j^2 = \Sigma_{j,11} - 2\Sigma_{j,12} + \Sigma_{j,22}$, this also means the unconditional distribution of $\eta_j \widehat{WTP}_j - \mu \widehat{G}_j = \widehat{WTP}_j - \widehat{G}_j$ is the Gaussian distribution $N(0, \omega_j^2)$. Note that $\omega_j^2$ is uniformly bounded, $0 < \tilde k_1 \leq \omega_j \leq \tilde k_2 < \infty$ for all $j$, by Assumption \ref{ass:bdd_sig}. Further note that $\widehat{WTP}_j - \widehat{G}_j$ is independent across $j$ by assumption.

Then the first objective, without normalization, is
\begin{align*}
    &\max_{v: \mathcal{Y} \to \mathcal{B}_p} E \left[ \left\vert E_\pi[\langle \nabla w, v(Y_{1:J}) \rangle - \langle \widehat{\nabla w},v(Y_{1:J}) \rangle] \right\vert \right] \\
    &\qquad = E \left[ \max_{v \in \mathcal{B}_p} \left\vert E_\pi[\langle \nabla w, v \rangle - \langle \widehat{\nabla w},v \rangle ] \right\vert \right] 
    = E \left[ \left\Vert E_\pi[\nabla w] - \widehat{\nabla w} \right\Vert_{\frac{p}{p-1}} \right] \\
    &\qquad \geq \left\Vert E \left[ \left\vert E_\pi[\nabla w] - \widehat{\nabla w} \right\vert \right] \right\Vert_{\frac{p}{p-1}} \\
    &\qquad = \left( \sum_{j=1}^J E \left[ \left\vert \widehat{WTP}_j - \widehat{G}_j \right\vert \right]^{\frac{p}{p-1}} \right)^{\frac{p-1}{p}}.
\end{align*}
Here I use the notation $\vert v \vert$ for a given vector $v$ to denote the vector with entries corresponding to the absolute value of the entries of $v$.

This means for $p = \infty$, the normalized first objective is
\begin{align*}
    &\frac{1}{N_\infty} \max_{v: \mathcal{Y} \to \mathcal{B}_\infty} E \left[ \left\vert E_\pi [ \langle \nabla w, v \rangle - \langle \widehat{\nabla w}, v \rangle ] \right\vert \right] 
    \geq \frac{1}{J} \sum_{j=1}^J E[ |\widehat{WTP}_j - \widehat{G}_j|] = \frac{1}{J} \sum_{j=1}^J \omega_j \sqrt{\frac{2}{\pi}} \geq K,
\end{align*}
for some constant $K$ that does not depend on $J$, because $\omega_j^2$ is uniformly bounded.

For $1 < p < \infty$, the normalized first objective is 
\begin{align*}
    &\frac{1}{N_p} \max_{v: \mathcal{Y} \to \mathcal{B}_p} E \left[ \left\vert E_\pi[\langle \nabla w, v(Y_{1:J}) \rangle - \langle \widehat{\nabla w},v(Y_{1:J}) \rangle] \right\vert \right] \\
    &\qquad \geq \frac{1}{\sqrt{\frac{p}{p-1}}} J^\frac{1-p}{p} \left( \sum_{j=1}^J E \left[ |\widehat{WTP}_j - \widehat{G}_j| \right]^{\frac{p}{p-1}} \right)^{\frac{p-1}{p}} \\
    &\qquad \geq K_p J^\frac{1-p}{p} \left( \sum_{j=1}^J \left(\omega_j \sqrt{\frac{2}{\pi}}\right)^{\frac{p}{p-1}} \right)^{\frac{p-1}{p}} \geq K_p
\end{align*}
because $\omega_j^2$ is uniformly bounded, where $K_p > 0$ is a constant that does depend on $p$ but does not depend on $J$.

For $p = 1$, the normalized first objective is 
\begin{align*}
    &\frac{1}{N_1} \max_{v: \mathcal{Y} \to \mathcal{B}_1} E \left[ \left\vert E_\pi[\langle \nabla w, v(Y_{1:J}) \rangle - \langle \widehat{\nabla w},v(Y_{1:J}) \rangle] \right\vert \right] \\
    &\qquad = \frac{1}{\sqrt{\log J}} E \left[ \left\Vert E_\pi[\nabla w] - \widehat{\nabla w} \right\Vert_{\infty} \right] \\
    &\qquad = \frac{1}{\sqrt{\log J}} E \left[ \max_{j=1,\dots,J} |\widehat{WTP}_j - \widehat{G}_j| \right] 
    \geq K
\end{align*}
because $\omega_j^2$ is uniformly bounded, using standard lower bounds for maxima of independent Gaussians (see, for example, \cite{vershynin2018high} Exercise 2.38).

To derive a result for the second objective, I split the cases into $p > 1$ and $p = 1$. For $p > 1$, consider the prior $F_0$ that takes on value $(1,-1)$ with probability $1$ and does not depend on $X_j$. I maintain the assumption that $\eta_j = 1$ for all $j$ and $\mu = 1$. 

This means that $\eta_j WTP_j - \mu G_j = WTP_j - G_j$ takes on value $2$ with probability $1$. Note this means that posterior mean $WTP_j^* - G_j^* = 2$ for all values of $\widehat{WTP}_j, \widehat{G}_j$.
This also means the unconditional distribution of $\eta_j \widehat{WTP}_j - \mu \widehat{G}_j = \widehat{WTP}_j - \widehat{G}_j$ is $N(2, \omega_j^2)$.

Let $p > 1$ be arbitrary. Note that 
\begin{align*}
    &\frac{\sum_{j=1}^J (WTP_j^*-G_j^*) \text{sign}(\widehat{WTP}_j-\widehat{G}_j)\vert\widehat{WTP}_j-\widehat{G}_j\vert^{\frac{1}{p-1}}}{\Vert \widehat{\nabla w} \Vert_{\frac{p}{p-1}}^{\frac{1}{p-1}} } \\
    &\qquad \leq 2\sum_{j=1}^J \mathbbm{1}\{\widehat{WTP}_j-\widehat{G}_j > 0\} \frac{\vert\widehat{WTP}_j-\widehat{G}_j\vert^{\frac{1}{p-1}}}{\Vert \widehat{\nabla w} \Vert_{\frac{p}{p-1}}^{\frac{1}{p-1}} } \\
    &\qquad \leq 2 \left(\sum_{j=1}^J \mathbbm{1}\{\widehat{WTP}_j-\widehat{G}_j > 0\} \right)^{\frac{p-1}{p}} \frac{\left( \sum_{j=1}^J \vert\widehat{WTP}_j-\widehat{G}_j\vert^{\frac{p-1}{p}} \right)^{1/p}}{\Vert \widehat{\nabla w} \Vert_{\frac{p}{p-1}}^{\frac{1}{p-1}} } \\
    &\qquad = 2 \left\Vert \mathbbm{1}\{\widehat{\nabla w} > 0 \} \right\Vert_{\frac{p}{p-1}}
\end{align*}
by H\"older's inequality.

Note too that for any $q \in [1,\infty]$, $\Vert E_\pi[\nabla w] \Vert_{q} = 2J^{1/q}$.

Then
\begin{align*}
    &\max_{v: \mathcal{Y} \to \mathcal{B}_p} E \left[ E_\pi [ \langle \nabla w, v \rangle - \langle \nabla w, \hat v_J \rangle ] \right] \\
    &\qquad =E \left[\left\Vert E_\pi[\nabla w] \right\Vert_{\frac{p}{p-1}} - 
    \frac{\sum_{j=1}^J (WTP_j^*-G_j^*) \text{sign}(\widehat{WTP}_j-\widehat{G}_j)\vert\widehat{WTP}_j-\widehat{G}_j\vert^{\frac{1}{p-1}}}{\Vert \widehat{\nabla w} \Vert_{\frac{p}{p-1}}^{\frac{1}{p-1}} } \right] \\
    &\qquad \geq 2 J^{\frac{p-1}{p}} - 2E \left[\left\Vert
    \mathbbm{1}\{\widehat{\nabla w} > 0\} \right\Vert_{\frac{p}{p-1}} \right].
\end{align*}

For $p = \infty$,
\begin{align*}
    \frac{1}{N_\infty} \max_{v: \mathcal{Y} \to \mathcal{B}_\infty} E \left[ E_\pi [ \langle \nabla w, v \rangle - \langle \nabla w, \hat v_J \rangle ] \right] &\geq 2 - \frac{1}{J} E \left[\left\Vert \mathbbm{1}\{\widehat{\nabla w} > 0\} \right\Vert_{1} \right] \\
    &= 2 - \frac{2}{J} \sum_{j=1}^J Pr \left(\widehat{WTP}_j - \widehat{G}_j > 0\right) \\
    &= 2 - \frac{2}{J} \sum_{j=1}^J \Phi \left(\frac{2}{\omega_j} \right) \geq K
\end{align*}
because $\omega_j$ is uniformly bounded and $\Phi(\cdot)$, the CDF of the standard Gaussian distribution, is strictly less than 1 for all inputs.

For $1 < p < \infty$,
\begin{align*}
    \frac{1}{N_p} \max_{v: \mathcal{Y} \to \mathcal{B}_p} E \left[ E_\pi [ \langle \nabla w, v \rangle - \langle \nabla w, \hat v_J \rangle ] \right] &\geq 2 \sqrt{\frac{p-1}{p}} \left(1 - J^{\frac{1-p}{p}} E \left[\left\Vert \mathbbm{1}\{\widehat{\nabla w} > 0\} \right\Vert_{\frac{p}{p-1}} \right]\right) \\
    &= 2 \sqrt{\frac{p-1}{p}} \left(1 - J^{\frac{1-p}{p}} E \left[ \left(\sum_{j=1}^J \mathbbm{1}\{\widehat{WTP}_j-\widehat{G}_j > 0\} \right)^{\frac{p-1}{p}} \right]\right) \\
    &\geq 2 \sqrt{\frac{p-1}{p}} \left(1 - J^{\frac{1-p}{p}} \left(\sum_{j=1}^J Pr \left( \widehat{WTP}_j-\widehat{G}_j > 0 \right)\right)^{\frac{p-1}{p}}\right) \\
    &\geq K_p,
\end{align*}
where the third line follows from Jensen's and the final line follows because $\omega_j$ is uniformly bounded and $\Phi(\cdot)$ is strictly less than 1 for all inputs.

For $p=1$, consider the prior $F_0 = N(0, \frac{\xi^2}{2} I_2)$ that does not depend on $X_j$.
I maintain the assumption that $\eta_j = 1$ for all $j$ and $\mu = 1$. I assume that $\Sigma_j = \frac{\sigma_j^2}{2} I_2$.
Suppose for $j \in L = \{j: j \text{ is odd}\}$ I have $\sigma_j^2 = \sigma_L^2$ and for $j \in H = \{j: j \text{ is even}\}$ I have $\sigma_j^2 = \sigma_H^2$, with $\sigma_H > \sigma_L$.

This means that $\eta_j WTP_j - \mu G_j = WTP_j - G_j \sim N(0,\xi^2)$. Then the unconditional distribution of $\eta_j \widehat{WTP}_j - \mu \widehat{G}_j = \widehat{WTP}_j - \widehat{G}_j$ is $N(0, \xi^2 + \sigma_j^2)$. Furthermore, posterior mean $\eta_j WTP^*_j - \mu G^*_j = WTP^*_j - G^*_j = \frac{\xi^2}{\xi^2 + \sigma_j^2} (\widehat{WTP}_j - \widehat{G}_j)$ is $N(0, \frac{\xi^4}{\xi^2+\sigma_j^2})$.

Note that 
\begin{align*}
    &\frac{1}{N_1} \max_{v: \mathcal{Y} \to \mathcal{B}_1} E[E_\pi[\langle \nabla w, v \rangle - \langle \nabla w, \hat v_J \rangle]] = \frac{1}{\sqrt{\log J}} \max_{v: \mathcal{Y} \to \mathcal{B}_1} E[E_\pi[\langle \nabla w, v \rangle - \langle \nabla w, \hat v_J \rangle]] \\
    &= \frac{1}{\sqrt{\log J}} E \bigg[ \Vert E_\pi[\nabla w] \Vert_\infty \\
    &\hspace{65pt} - \sum_{j=1}^J (WTP_j^* - G_j^*) \text{sign}(\widehat{WTP}_j - \widehat{G}_j) 1\{j = \text{arg}\max_k |\widehat{WTP}_k - \widehat{G}_k| \} \bigg] \\
    &= \frac{1}{\sqrt{\log J}} E \left[ \left\vert WTP^*_{j^*} - G^*_{j^*} \right\vert - \left\vert WTP^*_{\hat j} - G^*_{\hat j} \right\vert \right],
\end{align*}
where $j^* = \text{arg}\max_j |WTP_j^* - G_j^*|$ and $\hat j = \text{arg}\max_j |\widehat{WTP}_j - \widehat{G}_j|$.

Define
\begin{align*}
    s_g = \sqrt{\xi^2+\sigma_g^2} \qquad b_g = \frac{\xi^2}{\sqrt{\xi^2+\sigma_g^2}}, \qquad g \in \{L,H\}
\end{align*}
and note that $s_H > s_L$ and $b_H < b_L$. For $j \in L$ note that $WTP_j^* - G_j^* = b_L Z_j$ and $\widehat{WTP}_j-\widehat{G}_j = s_L Z_j$, while for $j \in H$ note that $WTP_j^* - G_j^* = b_H Z_j$ and $\widehat{WTP}_j-\widehat{G}_j = s_H Z_j$, where $Z_j$ are i.i.d. standard Gaussian random variables.

Defining $M_L = \max_{j \in L} |Z_j|$ and $M_H = \max_{j \in H} |Z_j|$, note that $j^* \in L$ if and only if $b_L M_L > b_H M_H$, while $\hat j \in L$ if and only if $s_L M_L > s_H M_H$.

Let $\varepsilon$ be small enough so that $s_H(1-\varepsilon) > s_L(1+\varepsilon)$ and $b_L(1-\varepsilon) > b_H(1+\varepsilon)$. Also define $n_J = \lfloor J/2 \rfloor$. Then on the event
\begin{align*}
    A_J = \left\{(1-\varepsilon) \sqrt{2 \log n_J} \leq M_L \leq (1+\varepsilon) \sqrt{2 \log n_J} \right\} \cap \left\{(1-\varepsilon) \sqrt{2 \log n_J} \leq M_H \leq (1+\varepsilon) \sqrt{2 \log n_J} \right\},
\end{align*}
it holds that $s_HM_H \geq s_H (1-\varepsilon)\sqrt{2 \log n_J} > s_L(1+\varepsilon)\sqrt{2 \log n_J} \geq s_L M_L$, so $\hat j \in H$. Meanwhile $|WTP_{j^*}^* - G_{j^*}^*| \geq b_LM_L$. So on event $A_J$,
\begin{align*}
    \left\vert WTP^*_{j^*} - G^*_{j^*} \right\vert - \left\vert WTP^*_{\hat j} - G^*_{\hat j} \right\vert &\geq b_LM_L - b_HM_H \\
    &\geq \left(b_L(1-\varepsilon) - b_H(1+\varepsilon)\right) \sqrt{2 \log n_J}.
\end{align*}
This means
\begin{align*}
    \frac{1}{\sqrt{\log J}} E \left[ \left\vert WTP^*_{j^*} - G^*_{j^*} \right\vert - \left\vert WTP^*_{\hat j} - G^*_{\hat j} \right\vert \right] &\geq \frac{1}{\sqrt{\log J}} K \sqrt{2 \log n_J} Pr(A_J) = K Pr(A_J).
\end{align*}

Note that $Pr(|Z_j| \leq x) = 2\Phi(x) - 1$, so for each $g \in \{L,H\}$, $Pr(M_g \leq x) = (2\Phi(x) - 1)^{|g|}$. Thus
\begin{align*}
    &Pr\left( (1-\varepsilon) \sqrt{2 \log n_J} \leq M_g \leq (1+\varepsilon)\sqrt{2 \log n_J} \right) \\
    &= \left( 2 \Phi\left((1+\varepsilon)\sqrt{2 \log n_J}\right)-1 \right)^{|g|} - \left( 2 \Phi\left((1-\varepsilon)\sqrt{2 \log n_J} \right)-1 \right)^{|g|},
\end{align*}
so $Pr(A_J) > 0$ for each $J \geq 4$. Furthermore $\frac{M_L}{\sqrt{2 \log n_J}} \to 1$ and $\frac{M_H}{\sqrt{2 \log n_J}} \to 1$ by standard results for independent Gaussian random variables (see, for example, \cite{vershynin2018high} Exercise 2.38). Thus there exists some $\underline K$ such that $Pr(A_J) \geq \underline{K}$ for all $J \geq 4$.

Finally, combining results I obtain 
\begin{align*}
    \frac{1}{N_1} \max_{v: \mathcal{Y} \to \mathcal{B}_1} E[E_\pi[\langle \nabla w, v \rangle - \langle \nabla w, \hat v_J \rangle]] \geq \frac{1}{\sqrt{\log J}} K \sqrt{2 \log n_J} = K.
\end{align*}

\end{proof}

\subsection{Proof of Theorem \ref{thm:regret_rates_param}}
Throughout this section I denote the expectation taken as if the prior were $f_\beta$ by $E_\beta[\cdot]$ (and the same for variances and covariances). I use $E[\cdot]$ to mean $E_{\beta_0}[\cdot]$ when unambiguous. Note that throughout I take $X_j$ to be fixed, so all expectations are also conditional on $X_j$ even when not specified.

I first state and prove several other results in order to prove Theorem \ref{thm:regret_rates_param}.

\begin{lemma} \label{lem:param_norm}
Suppose Assumptions \ref{ass:compact}, \ref{ass:eta}, \ref{ass:bdd_sig}, and \ref{ass:param} hold and suppose $J \geq 3$. Define $k \in (0,2]$ by $\frac{1}{k} = \max\{ \frac{1}{2} + \frac{1}{k_s}, \frac{1}{k_m}\}$. Then for any $r \in [1,\infty]$, the following rates hold:
\begin{align*}
    E \left[\left\Vert \widehat{\nabla w}^* - \nabla w^* \right\Vert_r \right] &\lesssim_{\mathcal H} \begin{cases} J^{1/r-1/2} r^{1/k} & r \in [1,\infty) \\ \sqrt{\frac{(\log J)^{2/k}}{J}} & r = \infty \end{cases}.
\end{align*}
\end{lemma}
\begin{proof}
    Let $m_j(y,\beta) = E_\beta[\theta_j | Y_j = y]$. Let $s_{\beta}(\theta_j) = \nabla_\beta \log f_{\beta}(\theta_j \vert X_j)$ denote the prior score. Note that $\nabla_\beta f_\beta(\theta_j \vert X_j) = f_\beta(\theta_j \vert X_j) s_\beta(\theta_j \vert X_j)'$.

    By definition, using $\varphi_{\Sigma_j}(x)$ to denote the density of the multivariate Gaussian with mean zero and variance matrix $\Sigma_j$,
    \begin{align*}
        m_{j}(y,\beta) = \frac{\int \theta \varphi_{\Sigma_j}(y-\theta) f_\beta(\theta \vert X_j) d\theta}{\int \varphi_{\Sigma_j}(y-\theta) f_{\beta}(\theta \vert X_j) d\theta}.
    \end{align*}
    Let
    \begin{align*}
        N_j(y,\beta) &\equiv \int \theta \varphi_{\Sigma_j}(y-\theta) f_\beta(\theta \vert X_j) d\theta, \qquad D_j(y,\beta) \equiv \int \varphi_{\Sigma_j}(y-\theta) f_{\beta}(\theta \vert X_j) d\theta.
    \end{align*}
    By Assumption \ref{ass:param}(2),
    \begin{align*}
        \nabla_\beta N_j(y,\beta) &= \int \theta \varphi_{\Sigma_j}(y-\theta) \nabla_\beta f_\beta(\theta \vert X_j) d\theta = \int \theta \varphi_{\Sigma_j}(y-\theta) f_\beta(\theta \vert X_j) s_\beta(\theta \vert X_j)' d\theta \\
        \nabla_\beta D_j(y,\beta) &= \int \varphi_{\Sigma_j}(y-\theta) \nabla_\beta f_{\beta}(\theta \vert X_j) d\theta = \int \varphi_{\Sigma_j}(y-\theta) f_{\beta}(\theta \vert X_j) s_\beta(\theta \vert X_j)' d\theta.
    \end{align*}

    Note that the posterior density is $\displaystyle \pi_{\beta,\Sigma_j}(\theta|y,X_j) = \frac{\varphi_{\Sigma_j}(y-\theta) f_\beta(\theta \vert X_j)}{D_j(y,\beta)}$. Thus
    \begin{align*}
        \frac{\nabla_\beta N_j(y,\beta)}{D_j(y,\beta)} &= \int \theta s_\beta(\theta \vert X_j)' \pi_{\beta,\Sigma_j}(\theta|y,X_j) d\theta = E_\beta[\theta_j s_\beta(\theta_j \vert X_j)' | X_j, Y_j=y] \\
        \frac{\nabla_\beta D_j(y,\beta)}{D_j(y,\beta)} &= \int s_\beta(\theta \vert X_j)' \pi_{\beta,\Sigma_j}(\theta|y,X_j) d\theta = E_\beta[s_\beta(\theta_j \vert X_j)'|X_j, Y_j = y].
    \end{align*}
    
    Then
    \begin{align*}
        \nabla_\beta m_j(y,\beta) &= \frac{\nabla_\beta N_j(y,\beta)}{D_j(y,\beta)} - \frac{N_j(y,\beta)}{D_j(y,\beta)} \frac{\nabla_\beta D_j(y,\beta)}{D_j(y,\beta)} \\
        &= E_\beta[\theta_j s_\beta(\theta_j \vert X_j)' | X_j, Y_j=y] - E_\beta[\theta_j|X_j, Y_j=y]E_\beta[s_\beta(\theta_j \vert X_j)'|X_j, Y_j = y] \\
        &= Cov_\beta(\theta_j, s_\beta(\theta_j \vert X_j) | X_j, Y_j =y).
    \end{align*}

    For $t \in [0,1]$ define $\beta_t = \beta_0 + t(\widehat\beta - \beta_0)$. Then 
    \begin{align*}
        \widehat\theta_j^* - \theta_j^* &= \left(\int_0^1 \nabla_\beta m_j(Y_j,\beta_t) dt \right) (\widehat\beta - \beta_0).
    \end{align*}
    Define $Z_j = \int_0^1 \left \Vert \nabla_\beta m_j(Y_j,\beta_t) \right\Vert_{op} dt$ and collect $Z = (Z_1, \dots, Z_J)$.
    Then
    \begin{align*}
        \left\Vert \widehat\theta_j^* - \theta_j^* \right\Vert_2 &\leq Z_j \left\Vert \widehat\beta - \beta_0 \right\Vert_2.
    \end{align*}
    Note that because $\eta_j$ and $\mu$ are uniformly bounded, it holds for each $j$ that
    \begin{align*}
        | \widehat{\nabla w}_j^* - \nabla w_j^* | &\lesssim_{\mathcal H} \left\Vert \widehat\theta_j^* - \theta_j^* \right\Vert_2 
        \Rightarrow \left\Vert \widehat{\nabla w}^* - \nabla w^* \right\Vert_r \lesssim_{\mathcal H} \Vert Z \Vert_r \left\Vert \widehat\beta - \beta_0 \right\Vert_2.
    \end{align*}
    Applying Cauchy-Schwarz and Assumption \ref{ass:param}(1),
    \begin{align*}
        E\left[\left\Vert \widehat{\nabla w}^* - \nabla w^* \right\Vert_r \right] &\lesssim_{\mathcal H} \frac{1}{\sqrt{J}} \sqrt{E[\Vert Z \Vert_r^2]}
    \end{align*}

In Lemma \ref{lem:param_moment} below I show that for any $p\geq 1$, $(E[|Z_j|^p])^{1/p} \lesssim_{\mathcal H} p^{1/k}$ for $k \in (0,2]$. Then the proof proceeds with the following three cases:

Case: $r \in [1,2]$:
\begin{align*}
    \sqrt{ E \left[ \left\Vert Z \right\Vert_r^2 \right]} &\leq J^{1/r-1/2} \sqrt{ E \left[ \left\Vert Z \right\Vert_2^2 \right]} = J^{1/r-1/2} \sqrt{ \sum_{j=1}^J E \left[ |Z_j|^2 \right]} \lesssim_{\mathcal H} J^{1/r}
\end{align*}
by H\"older's inequality and assumption.

Case: $r \in (2,\infty)$:
\begin{align*}
    \sqrt{ E \left[ \left\Vert Z \right\Vert_r^2 \right]} &\leq \left(E \left[ \left\Vert Z \right\Vert_r^r \right] \right)^{1/r} = \left(\sum_{j=1}^J E \left[ |Z_j|^r \right] \right)^{1/r} \lesssim_{\mathcal H} r^{1/k} J^{1/r}
\end{align*}
by assumption.

Case: $r = \infty$:
\begin{align*}
    \sqrt{ E \left[ \left\Vert Z \right\Vert_\infty^2 \right]} &= \left(E \left[ \left(\max_{1 \leq j \leq J} |Z_j| \right)^2 \right] \right)^{1/2} \leq \left(E \left[ \left(\max_{1 \leq j \leq J} |Z_j| \right)^{2 \log J} \right] \right)^{1/(2 \log J)} \\
    &\leq \left(\sum_{j=1}^J E \left[ |Z_j|^{2 \log J} \right] \right)^{1/(2 \log J)} \lesssim_{\mathcal H} \left(J (2\log J)^{2\log J/k} \right)^{1/(2 \log J)} \lesssim_{\mathcal H} (\log J)^{1/k}
\end{align*}
for $2 \log J \geq 2 \Rightarrow J \geq e$.
The result follows.
\end{proof}

\begin{lemma} \label{lem:param_moment}
    Suppose Assumptions \ref{ass:compact}, \ref{ass:eta}, \ref{ass:bdd_sig}, and \ref{ass:param} hold. Let $Z_j$ be defined as in the proof of Lemma \ref{lem:param_norm} and let $k$ be defined by $\frac{1}{k} = \max\{\frac{1}{2}+\frac{1}{k_s}, \frac{1}{k_m}\}$. Then $(E[|Z_j|^p])^{1/p} \lesssim_{\mathcal H} p^{1/k}$.
\end{lemma}
\begin{proof}
    Define $G_j(\beta) \equiv \nabla_\beta m_j(Y_j,\beta)$. Note that $Z_j \leq \sup_{\beta \in B} \Vert G_j(\beta) \Vert_{op}$. I first establish a moment bound for $G_j(\beta)$ for fixed $\beta$, then extend to a moment bound on $Z_j$.

    By Cauchy-Schwarz,
\begin{align*}
    \Vert G_j(\beta)\Vert_{op} &= \left\Vert Cov_{\beta}(\theta_j, s_{\beta}(\theta_j \vert X_j) | X_j, Y_j) \right\Vert_{op}
    \leq \left\Vert Var_{\beta}(\theta_j|X_j,Y_j) \right\Vert_{op}^{1/2} \left\Vert Var_{\beta}(s_{\beta}(\theta_j \vert X_j)|X_j,Y_j) \right\Vert_{op}^{1/2}.
\end{align*}

For any random vector $U$, $\Vert Var(U|Y_j,X_j) \Vert_{op} \leq E[\Vert U \Vert_2^2 \vert Y_j,X_j]$. Then by Jensen's inequality,
\begin{align*}
    E_{\beta_0} \left[ \Vert Var_\beta(\theta_j \vert Y_j,X_j) \Vert_{op}^p \right] &\leq E_{\beta_0} \left[ E_{\beta} \Vert \theta_j \Vert_2^{2p} \vert Y_j,X_j] \right].
\end{align*}

Recall that $p_{\beta,j}(y|X_j) = \int \varphi_{\Sigma_j}(y-\theta) f_\beta(\theta \vert X_j) d\theta$ denotes the marginal density of $Y_j$ given $X_j$. 
For fixed $\beta$, change of measure and Cauchy-Schwarz implies
\begin{align*}
    E_{\beta_0} \left[ E_{\beta} \Vert \theta_j \Vert_2^{2p} \vert Y_j,X_j] \right] &= E_{\beta} \left[ \Vert \theta_j \Vert_2^{2p} \frac{p_{\beta_0,j}(Y_j|X_j)}{p_{\beta,j}(Y_j|X_j)}\right] 
    \leq \left( E_\beta[\Vert \theta_j \Vert_2^{4p}] \right)^{1/2} \left( E_{\beta} \left[ \left(\frac{p_{\beta_0,j}(Y_j|X_j)}{p_{\beta,j}(Y_j|X_j)}\right)^2\right] \right)^{1/2}.
\end{align*}
By Assumptions \ref{ass:param}(3) and \ref{ass:param}(4),
\begin{align*}
    \sup_{\beta \in B} \left( E_{\beta_0} \left[ \Vert Var_\beta(\theta_j \vert Y_j,X_j) \Vert_{op}^p \right] \right)^{1/2p} &\lesssim_{\mathcal H} \sqrt{p}.
\end{align*}
The same argument using Assumptions \ref{ass:param}(4) and \ref{ass:param}(5) gives
\begin{align*}
    \sup_{\beta \in B} \left( E_{\beta_0} \left[ \Vert Var_\beta(s_\beta(\theta|X_j) \vert Y_j,X_j) \Vert_{op}^p \right] \right)^{1/2p} &\lesssim_{\mathcal H} p^{1/k_s}.
\end{align*}

Then by Cauchy-Schwarz,
\begin{align*}
    \sup_{\beta \in B} \left( E_{\beta_0}[\Vert G_j(\beta) \Vert_{op}^p] \right)^{1/p} &\lesssim_{\mathcal H} p^{1/2+1/k_s} \lesssim_{\mathcal H} p^{1/k}.
\end{align*}

To extend the moment bound to $Z_j$, note that by Assumption \ref{ass:param}(6), for every $\beta,\beta'\in B$
\begin{align*}
    \left( E_{\beta_0}[\Vert G_j(\beta) - G_j(\beta') \Vert_{op}^p] \right)^{1/p} &\lesssim_{\mathcal H} p^{1/k} \Vert \beta - \beta'\Vert_2.
\end{align*}
Since $B$ is compact, $\log N(\varepsilon, B, \Vert \cdot \Vert_2) \lesssim_{\mathcal H} \log\left( KD_B/\varepsilon \right)$ where $D_B$ is the diameter of $B$.
For $\ell \geq 1$ let $T_\ell$ be a minimal $2^{-\ell} D_B$-cover of $B$ and let $T_0 = \{\beta^\circ\}$ for some fixed $\beta^\circ \in B$. For every $\beta \in B$ let $\pi_\ell(\beta)$ denote a closest element of $T_\ell$ to $\beta$. 

By Assumption \ref{ass:param}(2) and dominated convergence, $G_j(\beta)$ is continuous in $\beta$, so
\begin{align*}
    G_j(\beta) = G_j(\beta^\circ) + \sum_{\ell = 1}^\infty \left( G_j(\pi_\ell(\beta)) - G_j(\pi_{\ell-1}(\beta)) \right).
\end{align*}
Note $\Vert \pi_\ell(\beta) - \pi_{\ell-1}(\beta) \Vert_2 \leq \Vert \pi_\ell(\beta) - \beta \Vert_2 + \Vert \pi_{\ell - 1}(\beta) - \beta \Vert_2 \leq 2^{-\ell} D_B + 2^{-(\ell-1)} D_B = 3 \cdot 2^{-\ell} D_B$. 
Also note that the number of pairs $(\pi_\ell(\beta), \pi_{\ell-1}(\beta))$ is at most $|T_\ell||T_{\ell-1}|$, which satisfies $\log(|T_\ell||T_{\ell-1}|) \lesssim_{\mathcal H} \ell$.

I also require the inequality that if random variables $W_1,\dots,W_N$ satisfy $\Vert W_i \Vert_{L^q} \leq K q^{1/k}$ for every $q \geq 1$ then $\left\Vert \max_{i \in [N]} |W_i| \right\Vert_{L^p} \lesssim K(p+\log N)^{1/k}$.

Then
\begin{align*}
    \left( E_{\beta_0} \left[ \sup_{\beta \in B} \Vert G_j(\pi_\ell(\beta)) - G_j(\pi_{\ell - 1}(\beta)) \Vert_{op}^p \right] \right)^{1/p} &\lesssim_{\mathcal H} 2^{-\ell} D_B(p+\ell)^{1/k}.
\end{align*}
Then the telescoping representation and Minkowski's inequality implies
\begin{align*}
    \left( E_{\beta_0} \left[ \sup_{\beta \in B} \Vert G_j(\beta) \Vert_{op}^p \right] \right)^{1/p} &\lesssim_{\mathcal H} p^{1/k} + D_B \sum_{\ell=1}^\infty 2^{-\ell} (p+\ell)^{1/k} \lesssim_{\mathcal H} p^{1/k}.
\end{align*}

The result follows.
\end{proof}

With these results, I can prove Theorem \ref{thm:regret_rates_param}.

\begin{proof}
For the first objective,
\begin{align*}
    \max_{v: \mathcal{Y} \to V_J} E \left[ \left\vert E_{\beta_0} [ \langle \nabla w, v(Y_{1:J}) \rangle - \langle \widehat{\nabla w}^*, v(Y_{1:J}) \rangle \vert Y_{1:J} ] \right\vert \right] &= E \left[ \max_{v \in V_J} \left\vert \langle \nabla w^*-\widehat{\nabla w}^*, v \rangle  \right\vert \right] \\
    &\leq E \left[ \max_{v \in \mathcal{B}_p} \left\vert \langle \nabla w^*-\widehat{\nabla w}^*, v \rangle  \right\vert \right] \\
    &= E \left[ \left\Vert \widehat{\nabla w}^* - \nabla w^* \right\Vert_{\frac{p}{p-1}} \right].
\end{align*}

For the second objective,
\begin{align*}
    &\max_{v: \mathcal{Y} \to V_J} E \left[E_{\beta_0} [ \langle \nabla w, v(Y_{1:J}) \rangle - \langle \nabla w, \hat v^*_J \rangle \vert Y_{1:J} ] \right] \\
    &\qquad = E \left[\max_{v \in V_J} \langle \nabla w^*, v \rangle - \max_{v \in V_J} \langle \widehat{\nabla w}^*, v \rangle + \max_{v \in V_J} \langle \widehat{\nabla w}^*, v \rangle - \langle \nabla w^*, \hat v^*_J \rangle \right] \\
    &\qquad \leq E \left[\max_{v \in V_J} \langle \nabla w^*- \widehat{\nabla w}^*, v \rangle + \langle \widehat{\nabla w}^*, \hat v^*_J \rangle - \langle \nabla w^*, \hat v^*_J \rangle \right] \\
    &\qquad \leq E \left[\max_{v \in \mathcal{B}_p} \langle \nabla w^* - \widehat{\nabla w}^*, v \rangle + \langle \widehat{\nabla w}^* - \nabla w^*, \hat v^*_J \rangle \right] \\
    &\qquad \leq 2E \left[ \left\Vert \widehat{\nabla w}^* - \nabla w^* \right\Vert_{\frac{p}{p-1}} \right].
\end{align*}

If $p = 1$: From Lemma \ref{lem:param_norm}, both the normalized first objective and the normalized second objective are upper bounded by:
\begin{align*}
    2\frac{1}{\sqrt{\log J}} \sqrt{\frac{(\log J)^{2/k}}{J}} &\lesssim_{\mathcal H} J^{-\frac{1}{2}} (\log J)^{\frac{1}{k} - \frac{1}{2}}.
\end{align*}

If $1 < p < \infty$: From Lemma \ref{lem:param_norm}, both the normalized first objective and the normalized second objective are upper bounded by:
\begin{align*}
    2 \frac{1}{J^{\frac{p-1}{p}} \min\left\{ \sqrt{\frac{p}{p-1}}, \sqrt{\log J} \right\}} J^{\frac{p-1}{p}-\frac{1}{2}} \left(\frac{p}{p-1}\right)^{\frac{1}{k}} &\lesssim_{\mathcal H} J^{-\frac{1}{2}} \left( \frac{p}{p-1}\right)^{\frac{1}{k}} \frac{1}{\sqrt{\min \left\{ \frac{p}{p-1}, \log J \right\}}}
\end{align*}

If $p = \infty$: From Lemma \ref{lem:param_norm}, both the normalized first objective and the normalized second objective are upper bounded by:
\begin{align*}
    2 \frac{1}{J} J^{\frac{1}{2}} &\lesssim_{\mathcal H} J^{-\frac{1}{2}}.
\end{align*}
\end{proof}

\subsection{Proof of Theorem \ref{thm:regret_rates}}

For notational simplicity denote the posterior expectation $E_{F_0, \alpha_0, \Omega_0} \left[ \cdot \vert \widehat{WTP}_{1:J}, \widehat{G}_{1:J} \right]$ by $E_\pi$, where $\pi$ represents expectation under the true posterior. Denote all oracle posterior means with an asterisk and all empirical Bayes posterior mean estimates with a hat and asterisk, so that the oracle posterior mean of $\tau_j$ is $\tau_j^*$ and the empirical Bayes posterior mean estimate for $\tau_j$ is $\widehat{\tau}_j^*$. 
In what follows I denote the type of policy $j$ by $t_j$, so that $X_j = t_j$ and the location and scale parameters associated with policy $j$ are $\alpha_{w,t_j}, \alpha_{g,t_j}$, and $\Omega_{t_j}$.

In order to prove Theorem \ref{thm:regret_rates}, I must first state and prove several other results.
I first state an upper bound on mean squared error regret, which is uniform for a given set of hyperparameters $\mathcal H$. I provide a proof of the result in Supplemental Appendix \ref{app:mse_proof}. This result is analogous to Theorem 1 in \citeappendix{chen2022empirical}, and the proof of the result is an extension of the proof in \citeappendix{chen2022empirical} to multivariate data, using additional results from \citeappendix{saha2020nonparametric}, \citeappendix{soloff2024multivariate}, and \citeappendix{jiang2020general}.

\begin{theorem}\label{thm:mse}
Suppose Assumptions \ref{ass:compact}, \ref{ass:bdd_sig}, \ref{ass:est_var}, and \ref{ass:good_approx} hold. Suppose also that $J \geq \max\{\frac{5}{\underline{k}}, 7\}$. Then
\begin{align*}
    \frac{1}{J} \sum_{j=1}^J E \left[ \left \Vert \begin{pmatrix} WTP_j^* \\ G_j^* \end{pmatrix} - \begin{pmatrix} \widehat{WTP}_j^* \\ \widehat{G}_j^* \end{pmatrix} \right \Vert_2^2 \right] \lesssim_{\mathcal H} \frac{(\log J)^6}{J}.
\end{align*}
\end{theorem}

From the upper bound on mean squared error regret one can obtain an upper bound on the squared norm difference between the true posterior expected and empirical Bayes gradients.
\begin{corollary}\label{cor:gradient}
    Under Assumptions \ref{ass:compact}, \ref{ass:eta}, \ref{ass:bdd_sig}, \ref{ass:est_var}, and \ref{ass:good_approx}, if $J \geq \max\{\frac{5}{\underline{k}}, 7\}$ it holds that 
    \begin{align*}
        E \left[ \left\Vert E_\pi [\nabla w ]-\widehat{\nabla w}^* \right\Vert_{2} \right] &\lesssim_{\mathcal H} (\log J)^3.
    \end{align*}
\end{corollary}
\begin{proof}
Note
\begin{align*}
    \frac{1}{\sqrt{J}} E \left[ \left\Vert E_\pi [\nabla w ]-\widehat{\nabla w}^* \right\Vert_{2} \right] &\leq \sqrt{ \frac{1}{J} E \left[ \left\Vert E_\pi [\nabla w ]-\widehat{\nabla w}^* \right\Vert_{2}^2 \right] } \\
    &\lesssim_{\mathcal H}  \sqrt{E \left[\frac{1}{J} \sum_{j=1}^J  \left \Vert \begin{pmatrix} WTP_j^* \\ G_j^* \end{pmatrix} - \begin{pmatrix} \widehat{WTP}_j^* \\ \widehat{G}_j^* \end{pmatrix} \right \Vert_2^2 \right]} \\
    &\lesssim_{\mathcal H} \frac{(\log J)^3}{\sqrt{J}},
\end{align*}
where the first line follows from Jensen's inequality, the second line follows from Assumption \ref{ass:eta} and $(a-b)^2 \leq 2a^2+2b^2$, and the third line follows from Theorem \ref{thm:mse}. The result follows.
\end{proof}

With these results, I can prove Theorem \ref{thm:regret_rates}.

\begin{proof}
For the first objective,
\begin{align*}
    \max_{v: \mathcal{Y} \to V_J} E \left[ \left\vert E_\pi [ \langle \nabla w, v(Y_{1:J}) \rangle - \langle \widehat{\nabla w}^*, v(Y_{1:J}) \rangle ] \right\vert \right] &= E \left[ \max_{v \in V_J} \left\vert \langle E_\pi [\nabla w ]-\widehat{\nabla w}^*, v \rangle  \right\vert \right] \\
    &\leq E \left[ \max_{v \in \mathcal{B}_p} \left\vert \langle E_\pi [\nabla w ]-\widehat{\nabla w}^*, v \rangle  \right\vert \right] \\
    &= E \left[ \left\Vert E_\pi [\nabla w ]-\widehat{\nabla w}^* \right\Vert_{\frac{p}{p-1}} \right].
\end{align*}

For the second objective,
\begin{align*}
    &\max_{v: \mathcal{Y} \to V_J} E \left[E_\pi [ \langle \nabla w, v(Y_{1:J}) \rangle - \langle \nabla w, \hat v^*_J \rangle ] \right] \\
    &\qquad = E \left[\max_{v \in V_J} \langle E_\pi [\nabla w ], v \rangle - \max_{v \in V_J} \langle \widehat{\nabla w}^*, v \rangle + \max_{v \in V_J} \langle \widehat{\nabla w}^*, v \rangle - \langle E_\pi [\nabla w ], \hat v^*_J \rangle \right] \\
    &\qquad \leq E \left[\max_{v \in V_J} \langle E_\pi [\nabla w ]- \widehat{\nabla w}^*, v \rangle + \langle \widehat{\nabla w}^*, \hat v^*_J \rangle - \langle E_\pi [\nabla w ], \hat v^*_J \rangle \right] \\
    &\qquad \leq E \left[\max_{v \in \mathcal{B}_p} \langle E_\pi [\nabla w ]- \widehat{\nabla w}^*, v \rangle + \langle \widehat{\nabla w}^* - E_\pi [\nabla w ], \hat v^*_J \rangle \right] \\
    &\qquad \leq 2E \left[ \left\Vert E_\pi [\nabla w ]-\widehat{\nabla w}^* \right\Vert_{\frac{p}{p-1}} \right].
\end{align*}

By H\"older's inequality, if $p \in [1,2)$ then $\left\Vert E_\pi [\nabla w ]-\widehat{\nabla w}^* \right\Vert_{\frac{p}{p-1}} \leq \left\Vert E_\pi [\nabla w ]-\widehat{\nabla w}^* \right\Vert_{2}$ and if $p \in [2,\infty]$ then $\left\Vert E_\pi [\nabla w ]-\widehat{\nabla w}^* \right\Vert_{\frac{p}{p-1}} \leq J^{\frac{1}{2} - \frac{1}{p}} \left\Vert E_\pi [\nabla w ]-\widehat{\nabla w}^* \right\Vert_{2}$.

So if $p \in [1,2)$ then the normalized first objective is
\begin{align*}
    \frac{1}{N_p} \max_{v: \mathcal{Y} \to V_J} E \left[ \left\vert E_\pi [ \langle \nabla w, v(Y_{1:J}) \rangle - \langle \widehat{\nabla w}^*, v(Y_{1:J}) \rangle ] \right\vert \right] &\leq J^{-\frac{p-1}{p}} E \left[ \left\Vert E_\pi [\nabla w ]-\widehat{\nabla w}^* \right\Vert_{\frac{p}{p-1}} \right] \\
    &\leq J^{-\frac{p-1}{p}} E \left[ \left\Vert E_\pi [\nabla w ]-\widehat{\nabla w}^* \right\Vert_{2} \right] \\
    &\lesssim_{\mathcal H} J^{\frac{1}{p} - 1} (\log J)^3 \qquad \text{from Corollary \ref{cor:gradient}},
\end{align*}
while if $p \in [2, \infty]$ the normalized first objective is 
\begin{align*}
    \frac{1}{N_p} \max_{v: \mathcal{Y} \to V_J} E \left[ \left\vert E_\pi [ \langle \nabla w, v(Y_{1:J}) \rangle - \langle \widehat{\nabla w}^*, v(Y_{1:J}) \rangle ] \right\vert \right] &\leq J^{-\frac{p-1}{p}} E \left[ \left\Vert E_\pi [\nabla w ]-\widehat{\nabla w}^* \right\Vert_{\frac{p}{p-1}} \right] \\
    &\leq J^{-\frac{1}{2}} E \left[ \left\Vert E_\pi [\nabla w ]-\widehat{\nabla w}^* \right\Vert_{2} \right] \\
    &\lesssim_{\mathcal H} J^{-\frac{1}{2}} (\log J)^3 \qquad \text{from Corollary \ref{cor:gradient}}.
\end{align*}
Similarly, if $p \in [1,2)$ then the normalized second objective is
\begin{align*}
    \frac{1}{N_p} \max_{v: \mathcal{Y} \to V_J} E \left[ E_\pi [ \langle \nabla w, v(Y_{1:J}) \rangle - \langle \nabla w, \hat v^*_J \rangle ] \right] &\lesssim_{\mathcal H} J^{\frac{1}{p} - 1} (\log J)^3.
\end{align*}
while if $p \in [2, \infty]$ the normalized second objective is 
\begin{align*}
    \frac{1}{N_p} \max_{v: \mathcal{Y} \to V_J} E \left[ E_\pi [ \langle \nabla w, v(Y_{1:J}) \rangle - \langle \nabla w, \hat v^*_J \rangle ] \right] &\lesssim_{\mathcal H} J^{-\frac{1}{2}} (\log J)^3.
\end{align*}

\end{proof}

\section{Appendix: Proof of Theorem \ref{thm:mse}} \label{app:mse_proof}

I restate the theorem for convenience.

\renewcommand{\thetheorem}{\ref{thm:mse}}
\begin{theorem}
Suppose Assumptions \ref{ass:compact}, \ref{ass:bdd_sig}, \ref{ass:est_var}, and \ref{ass:good_approx} hold. Suppose also that $J \geq \max \{\frac{5}{\underline{k}}, 7\}$. Then
\begin{align*}
    \frac{1}{J} \sum_{j=1}^J E \left[ \left \Vert \begin{pmatrix} WTP_j^* \\ G_j^* \end{pmatrix} - \begin{pmatrix} \widehat{WTP}_j^* \\ \widehat{G}_j^* \end{pmatrix} \right \Vert_2^2 \right] \lesssim_{\mathcal H} \frac{(\log J)^6}{J}.
\end{align*}
\end{theorem}
\renewcommand{\thetheorem}{\thesection.\arabic{theorem}}
\setcounter{theorem}{0}

The proof of the theorem will proceed entirely analogously to the proof of Theorem 1 in \citeappendix{chen2022empirical}.

\subsection{Notation}

I first review notation defined in the main text and introduce new notation.

Let $Y_j \equiv (\widehat{WTP}_j, \widehat{G}_j)^T, \theta_j \equiv (WTP_j, G_j)^T$. Then $Y_j \vert \theta_j, X_j, \Sigma_j \overset{\text{indep.}}{\sim} N(\theta_j, \Sigma_j)$ and $\theta_j = a_j(\alpha_0) + B_j(\Omega_0)^{1/2} \tau_j$, where $\tau_j \vert X_j, \Sigma_j \overset{\text{i.i.d.}}{\sim} F_0, a_j(\alpha_0) \equiv a(X_j; \alpha_0)$, and $B_j(\Omega_0) \equiv B(X_j; \Omega_0)$. 

Let $Z_j = B_j(\Omega_0)^{-1/2}(Y_j - a_j(\alpha_0))$ and $\Psi_j = B_j(\Omega_0)^{-1/2} \Sigma_j B_j(\Omega_0)^{-1/2}$. Then $Z_j \vert \tau_j, X_j, \Psi_j \overset{\text{indep.}}{\sim} N(\tau_j, \Psi_j)$ with $\tau_j \vert X_j,\Psi_j \overset{\text{i.i.d.}}{\sim} F_0$.

Let $\widehat\alpha$ and $\widehat\Omega$ be estimators of $\alpha_0$ and $\Omega_0$. 
Throughout the appendix I will occasionally use the shorthand $\chi = (\alpha, \Omega)$.
Let 
$$\Vert \widehat\chi - \chi_0 \Vert_J \equiv \max_{1\leq j \leq J} \max \left\{ \left\Vert a_j(\widehat\alpha) - a_j(\alpha_0) \right\Vert_\infty, \left\Vert B_j(\widehat\Omega)^{1/2} - B_j(\Omega_0)^{1/2} \right\Vert_{op} \right\},$$ 
where $\Vert \cdot \Vert_{op}$ denotes the Schatten $\infty$-norm.
For a given $\widehat\alpha, \widehat\Omega$ define
\begin{align*}
    \widehat Z_j &= \widehat{Z}_j(\widehat\alpha, \widehat\Omega) = B_j(\widehat \Omega)^{-1/2}(Y_j - a_j(\widehat\alpha)) = B_j(\widehat \Omega)^{-1/2}(B_{j}(\Omega_0)^{1/2}Z_j + a_j(\alpha_0) - a_j(\widehat\alpha)), \\
    \widehat\Psi_j &= \widehat\Psi_j(\widehat\alpha, \widehat\Omega) = B_j(\widehat \Omega)^{-1/2}\Sigma_jB_j(\widehat \Omega)^{-1/2}.
\end{align*}
Throughout this appendix I condition on $\Sigma_{1:J}$ and $X_{1:J}$ and thus take them as fixed. 

For any distribution $F$ and $\Psi$ define
\begin{align*}
    \varphi_{\Psi}(x) &= \exp \left(-\frac{1}{2}x^T \Psi^{-1} x \right) \\
    f_{F, \Psi}(x) &= \int \frac{1}{\sqrt{\det(2\pi \Psi)}} \varphi_{\Psi}(x-\tau) dF(\tau),
\end{align*}
and for any $F, \alpha, \Omega$ define
\begin{align*}
    \psi_j(z, \alpha,\Omega, F) &= \log \left(\int \varphi_{\widehat{\Psi}_j(\alpha,\Omega)} (\hat Z_j(\alpha,\Omega)-\tau) dF(\tau) \right).
\end{align*}
Denote the posterior mean of $\theta_j$ under any $F, \alpha, \Omega$ by
\begin{align*}
    \widehat\theta^*_{j, F, \alpha, \Omega} &= a_j(\alpha) + B_j(\Omega)^{1/2} \underbrace{E_{F, \widehat{\Psi}_j(\alpha,\Omega)}[\tau \vert \widehat{Z}_j(\alpha,\Omega)]}_{\widehat{\tau}^*_{j, F, \alpha, \Omega}},
\end{align*}
where $E_{F,\Psi}[h(\tau,Z)\vert z]$ is the posterior mean of $h(\tau, Z)$ when $Z = z$ under the model $\tau \sim F, Z \vert \tau \sim N(\tau, \Psi)$:
\begin{align*}
    E_{F,\Psi}[h(\tau,Z)\vert z] &= \frac{\int h(\tau,z) \varphi_\Psi(z-\tau) dF(\tau)}{\int \varphi_\Psi(z-\tau) dF(\tau)}.
\end{align*}
By Tweedie's formula
\begin{align*}
    E_{F, \Psi_j}[\tau_j \vert \widehat{Z}_j] &= \widehat{Z}_j + \Psi_j \frac{\nabla f_{F, \Psi_j}(\widehat{Z}_j)}{f_{F, \Psi_j}(\widehat{Z}_j)} \\
    \Rightarrow \widehat{\theta}^*_{j, F, \alpha, \Omega} &= Y_j + B_j(\Omega)^{1/2}\widehat{\Psi}_j(\alpha,\Omega) \frac{\nabla f_{F, \widehat{\Psi}_j(\alpha,\Omega)}(\widehat{Z}_j)}{f_{F, \widehat{\Psi}_j(\alpha,\Omega)}(\widehat{Z}_j)}.
\end{align*}
Denote by $\widehat{\theta}_{j,\widehat{F}_J, \widehat\alpha, \widehat\Omega}^*$ the empirical Bayes posterior mean of $\theta_j$ and by $\widehat{\theta}_{j,F_0, \alpha_0, \Omega_0}^* \equiv \theta_j^*$ the oracle posterior mean of $\theta_j$. 
I will collect $\theta_1, \dots, \theta_J$ into a $J \times 2$ matrix $\theta$. Similarly, $\widehat{\theta}^*_{F, \alpha, \Omega}$ collects $\widehat{\theta}^*_{1, F, \alpha, \Omega}, \dots, \widehat{\theta}^*_{J, F, \alpha, \Omega}$, and analogously for $\tau$ and $\widehat\tau^*_{F, \alpha, \Omega}$.

For some $\rho > 0$ define the regularized posterior mean as
\begin{align*}
    \widehat\theta_{j,F,\alpha,\Omega,\rho}^* = Y_j + B_j(\Omega)^{1/2} \widehat\Psi_j(\alpha,\Omega) \frac{\nabla f_{F,\widehat\Psi_j(\alpha,\Omega)}(\widehat{Z}_j(\alpha,\Omega))}{\max(f_{F,\widehat\Psi_j(\alpha,\Omega)}(\widehat{Z}_j(\alpha,\Omega)), \frac{\rho}{\sqrt{\det{\widehat\Psi_j(\alpha,\Omega)}}}) }
\end{align*}
and $\theta^*_{j, \rho} = \widehat\theta_{j,F_0,\alpha_0,\Omega_0,\rho}^*$. Similarly, define
\begin{align*}
    \widehat\tau_{j,F,\alpha,\Omega,\rho}^* = \widehat{Z}_j(\alpha,\Omega) + \widehat\Psi_j(\alpha,\Omega) \frac{\nabla f_{F,\widehat\Psi_j(\alpha,\Omega)}(\widehat{Z}_j(\alpha,\Omega))}{\max(f_{F,\widehat\Psi_j(\alpha,\Omega)}(\widehat{Z}_j(\alpha,\Omega)), \frac{\rho}{\sqrt{\det{\widehat\Psi_j(\alpha,\Omega)}}}) }
\end{align*}
and $\tau^*_{j, \rho} = \widehat\tau_{j,F_0,\alpha_0,\Omega_0,\rho}^*$.

Define
\begin{align*}
    \varphi_+(\rho) &= \sqrt{\log \frac{1}{(2\pi \rho)^2}}, \qquad \rho \in (0, (2\pi)^{-1})
\end{align*}
and observe that $\varphi_+(\rho) \lesssim \sqrt{\log(1/\rho)}$.

I will choose regularization parameter
\begin{align}\label{eq:oa3.5}
    \rho_J &= \min\left(\frac{1}{J^4} e^{-C_{\mathcal H,\rho}M_J^2 \Delta_J} , \frac{1}{2\pi e} \right),
\end{align}
where constant $C_{\mathcal H,\rho}$ will be chosen to satisfy Lemma \ref{lem:oa3.1}.

\subsection{Proof of main result}

Define the event 
\begin{align*}
    A_J &\equiv \left\{ \Vert \widehat\chi - \chi_0 \Vert_J \leq \Delta_J, \bar Z_J \equiv \max_{j \in [J]}( \max(\Vert Z_j \Vert_2, 1)) \leq M_J \right\}
\end{align*}
for constants $\Delta_J, M_J$ to be chosen. To prove the main result, I consider events $A_J$ and $A_J^C$ separately. Lemma \ref{lem:oa3.2} controls MSE regret on $A_J^C$, while Theorem \ref{thm:a.1} controls MSE regret on $A_J$.
While many of the results are true on $A_J$ for a broad class of $\Delta_J, M_J$, the ones I consider in this proof to obtain the rate of interest are
\begin{align} \label{eq:oa3.4}
    \Delta_J &= C_{\mathcal H} J^{-1/2}(\log J)^{1/2}, \qquad 
    M_J = (C_{\mathcal H}+1)(C_{2,\mathcal H}^{-1} \log J)^{1/2},
\end{align}
for constants $C_{\mathcal H}$ to be chosen and $C_{2,\mathcal H}$ determined by Theorem SM6.1.

\begin{lemma} \label{lem:oa3.2}
Under Assumptions \ref{ass:compact}, \ref{ass:bdd_sig}, \ref{ass:est_var}, and \ref{ass:good_approx}, suppose $\Delta_J$ and $M_J$ are of the form \eqref{eq:oa3.4} such that $Pr(\bar Z_J > M_J) \leq J^{-2}$. Then I can decompose
\begin{align*}
    \frac{1}{J} \sum_{j=1}^J E[\Vert \widehat{\theta}_{j,\widehat{F}_J, \widehat\alpha, \widehat\Omega}^* - \theta_j^* \Vert_2^2 \mathbbm{1}(\Vert \widehat\chi - \chi_0 \Vert_J > \Delta_J)] &\lesssim_{\mathcal H} Pr(\Vert \widehat\chi - \chi_0 \Vert_J > \Delta_J)^{1/2} (\log J) \\
    \frac{1}{J} \sum_{j=1}^J E[\Vert \widehat{\theta}_{j,\widehat{F}_J, \widehat\alpha, \widehat\Omega}^* - \theta_j^* \Vert_2^2 \mathbbm{1}(\bar Z_J > M_J)] &\lesssim_{\mathcal H} \frac{1}{J} (\log J).
\end{align*}
\end{lemma}
The proof of Lemma \ref{lem:oa3.2} is deferred to Supplemental Appendix \ref{app:lem:oa3.2}. 

\begin{theorem} \label{thm:a.1}
    Suppose Assumptions \ref{ass:compact}, \ref{ass:bdd_sig}, \ref{ass:est_var}, and \ref{ass:good_approx} hold. Fix some $C_1 > 0$, then there exists a constant $C_{\mathcal H,2}$ such that for $\Delta_J = C_{1} J^{-1/2}(\log J)^{1/2}$, $M_J = C_{\mathcal H,2} (\log J)^{1/2}$, and corresponding $A_J$, 
    \begin{align*}
        \frac{1}{J} \sum_{j=1}^J E \left[ \Vert \widehat{\theta}_{j,\widehat{F}_J, \widehat\alpha, \widehat\Omega}^* - \theta_j^* \Vert_2^2 \mathbbm{1}(A_J) \right] &\lesssim_{\mathcal H} J^{-1}(\log J)^{6}.
    \end{align*}
\end{theorem}
The proof of Theorem \ref{thm:a.1} is deferred to Supplemental Appendix \ref{app:thm:a.1}.

Combining these two results, I obtain the result of Theorem \ref{thm:mse}:
\begin{proof}[Proof of Theorem \ref{thm:mse}]
    Let $\Delta_J = C_{1,\mathcal H} J^{-1/2}(\log J)^{1/2}$ and $M_J = C\sqrt{\log J}$, where $C_{1,\mathcal H}$ is the constant in Assumption \ref{ass:est_var}(4) and $C$ is chosen in application of Theorem \ref{thm:a.1}. Then
    \begin{align*}
        \frac{1}{J} \sum_{j=1}^J E[\Vert \widehat{\theta}_{j,\widehat{F}_J, \widehat\alpha, \widehat\Omega}^* - \theta_j^* \Vert_2^2] &\leq \frac{1}{J} \sum_{j=1}^J E[\Vert \widehat{\theta}_{j,\widehat{F}_J, \widehat\alpha, \widehat\Omega}^* - \theta_j^* \Vert_2^2 \mathbbm{1}(A_J)] \\
        &\qquad + \frac{1}{J} \sum_{j=1}^J E[\Vert \widehat{\theta}_{j,\widehat{F}_J, \widehat\alpha, \widehat\Omega}^* - \theta_j^* \Vert_2^2 \mathbbm{1}(\Vert \widehat\chi - \chi_0 \Vert_J > \Delta_J)] \\
        &\qquad + \frac{1}{J} \sum_{j=1}^J E[\Vert \widehat{\theta}_{j,\widehat{F}_J, \widehat\alpha, \widehat\Omega}^* - \theta_j^* \Vert_2^2 \mathbbm{1}(\bar Z_J > M_J)] \\
        &\lesssim_{\mathcal H} J^{-1}(\log J)^{6} + J^{-1}(\log J) \\
        &\qquad \text{Theorem \ref{thm:a.1}, Lemma \ref{lem:oa3.2}, Assumption \ref{ass:est_var}(4)} \\
        &\lesssim_{\mathcal H} J^{-1}(\log J)^{6}.
\end{align*}
\end{proof}

\subsection{Proof of Lemma \ref{lem:oa3.2}} \label{app:lem:oa3.2}

\begin{proof}[Proof of Lemma \ref{lem:oa3.2}]
Observe that for any event $A$ on the data $Z_{1:J}$, by Cauchy-Schwarz
\begin{align*}
    E \left[ \frac{1}{J} \sum_{j=1}^J \left\Vert\widehat\theta_{j, \widehat{F}_J, \widehat\alpha, \widehat\Omega}^* - \theta_j^* \right\Vert_2^2 \mathbbm{1}(A) \right] &\leq E \left[ \left(\frac{1}{J} \sum_{j=1}^J \left\Vert\widehat\theta_{j, \widehat{F}_J, \widehat\alpha, \widehat\Omega}^* - \theta_j^* \right\Vert_2^2 \right)^2 \right]^{1/2} Pr(A)^{1/2}.
\end{align*}
Apply Lemma \ref{lem:oa3.6} to get
\begin{align*}
    \left(\frac{1}{J} \sum_{j=1}^J \left\Vert\widehat\theta_{j, \widehat{F}_J, \widehat\alpha, \widehat\Omega}^* - \theta_j^* \right\Vert_2^2 \right)^2 &\lesssim_{\mathcal H} \bar Z_J^4,
\end{align*}
since $\Vert \widehat\chi - \chi_0 \Vert_J \lesssim_{\mathcal H} 1$ under Assumption \ref{ass:est_var}.

Apply Lemma \ref{lem:oa3.7} to get $E[\bar Z_J^4] \lesssim_{\mathcal H} (\log J)^2$. This proves both claims.
\end{proof}

\begin{lemma} \label{lem:oa3.6}
    Assume that $\widehat F_J$ is supported within $[-C\bar M_J, C\bar M_J]^2$ for some constant $C$ where $\bar M_J = \max_j( \max (\Vert \widehat{Z}_j(\widehat\alpha, \widehat\Omega) \Vert_2, 1))$. Suppose $\Vert \widehat\chi - \chi_0 \Vert_J \lesssim_{\mathcal H} 1$, and Assumptions \ref{ass:compact},  \ref{ass:bdd_sig}, and \ref{ass:est_var} hold. Then letting $\bar Z_J = \max(\max_j \Vert Z_j \Vert_2, 1)$,
\begin{align*}
    \Vert \widehat\theta_{j, \widehat F, \widehat\alpha, \widehat\Omega}^* - \theta_j^* \Vert_2 \lesssim_{\mathcal H} \bar Z_J.
\end{align*}
\end{lemma}

\begin{proof}
    By Tweedie's formula, 
    \begin{align*}
        \Vert \widehat\theta_{j, \widehat F, \widehat\alpha, \widehat\Omega} - \theta_j^* \Vert_2 &= \left\Vert B_j(\widehat\Omega)^{1/2} \widehat\Psi_j \frac{\nabla f_{\widehat{F}_J,\widehat\Psi_j}(\widehat{Z}_j(\widehat\alpha,\widehat\Omega))}{f_{\widehat{F}_J,\widehat\Psi_j}(\widehat{Z}_j(\widehat\alpha,\widehat\Omega)) } - B_j(\Omega_0)^{1/2} \Psi_j \frac{\nabla f_{F_0,\Psi_j}(\widehat{Z}_j(\alpha_0,\Omega_0))}{f_{F_0,\Psi_j}(\widehat{Z}_j(\alpha_0,\Omega_0)) } \right\Vert_2 \\
        &= \left\Vert B_j(\widehat\Omega)^{1/2} E_{\widehat{F}_J,\widehat\Psi_j} \left[ \tau_j - \widehat{Z}_j \vert \widehat{Z}_j \right] - B_j(\Omega_0)^{1/2} E_{F_0, \Psi_j}\left[\tau_j - Z_j \vert Z_j \right] \right\Vert_2 \\
        &\lesssim_{\mathcal H} \bar M_J + \bar Z_J \lesssim_{\mathcal H} \max_j \max(\Vert\widehat{Z}_j\Vert_2, \Vert Z_j \Vert_2, 1)
    \end{align*}
    by the boundedness of $\widehat F_J$ and by Lemma \ref{lem:sm6.14}. 
    Note that $\Vert\widehat{Z}_j\Vert_2 = \Vert B_j(\widehat{\Omega})^{-1/2} B_j(\Omega_0)^{1/2} Z_j + B_j(\widehat{\Omega})^{-1/2} \left(a_j(\alpha_0) - a_j(\widehat\alpha) \right) \Vert_2 \lesssim_{\mathcal H} \max(\Vert Z_j \Vert_2, 1) \lesssim_{\mathcal H} \bar Z_J$. The result follows.
\end{proof}

\begin{lemma}\label{lem:oa3.7}
    Let $\bar Z_J = \max(\max_j \Vert Z_j\Vert_2, 1)$. Under Assumptions \ref{ass:compact} and \ref{ass:est_var}, for $t>1$
\begin{align*}
    Pr(\bar Z_J > t) \leq 2J \exp(-C_{\mathcal H} t^2) \quad \text{and} \quad E[\bar Z_J^p] \lesssim_{p,\mathcal H} (\log J)^{p/2}.
\end{align*}
Moreover, if $M_J = (C_{\mathcal{H}}+1)(C_{2,\mathcal H}^{-1} \log J)^{1/2}$ then for all sufficiently large choices of $C_{\mathcal H}$, $Pr(\bar Z_J > M_J) \leq J^{-2}$.
\end{lemma}
\begin{proof}
    The first claim is immediate under a union bound and noting that each $\Vert Z_j \Vert_2^2$ is a subexponential random variable, so $Pr(\Vert Z_j \Vert_2^2 > t) \leq 2\exp(-C_{\mathcal H} t) \Rightarrow Pr(\Vert Z_j \Vert_2 > t) \leq 2\exp(-C_{\mathcal H} t^2)$.
    
    The second claim follows from the observation that
    \begin{align*}
        E[\max_j(\max(\Vert Z_j \Vert_2,1))^p] &\leq \left( \sum_{j=1}^J E[(\max(\Vert Z_j \Vert_2,1))^{pc}] \right)^{1/c} \\
        &= \left( \sum_{j=1}^J E[(\max(\Vert Z_j \Vert_2^2,1))^{pc/2}] \right)^{1/c} \\
        &\leq J^{1/c} C_{\mathcal H}^{p} (pc)^{p/2}.
    \end{align*}
    where the last inequality follows from $\Vert Z_j \Vert_2^2$ being a subexponential random variable. Choose $c = \log J$ for $J^{1/\log J} = e$ to finish the proof. The moreover part follows exactly as in the proof of Lemma OA3.7 in \citeappendix{chen2022empirical}. 
\end{proof}

\subsection{Proof of Theorem \ref{thm:a.1}} \label{app:thm:a.1}

\begin{proof}[Proof of Theorem \ref{thm:a.1}]
    Choose $M_J$ to be of the form \eqref{eq:oa3.4}. By triangle inequality
    \begin{align*}
        &\Vert \widehat \theta_{\widehat{F}_J, \widehat\alpha, \widehat\Omega}^* - \theta^* \Vert_F \\
        &\leq \Vert \widehat \theta_{\widehat{F}_J, \widehat\alpha, \widehat\Omega}^* - \theta_{\widehat{F}_J, \alpha_0,\Omega_0}^* \Vert_F + \Vert \widehat \theta_{\widehat{F}_J, \alpha_0, \Omega_0}^* - \theta^*_{\widehat{F}_J, \alpha_0, \Omega_0, \rho_J} \Vert_F + \Vert \widehat \theta_{\widehat{F}_J, \alpha_0, \Omega_0, \rho_J}^* - \theta_{\rho_J}^* \Vert_F + \Vert \theta_{\rho_J}^* - \theta^* \Vert_F.
    \end{align*}
    Define
    \begin{align*}
        \xi_1 &\equiv \frac{\mathbbm{1}(A_J)}{J} \Vert \widehat \theta_{\widehat{F}_J, \widehat\alpha, \widehat\Omega}^* - \theta_{\widehat{F}_J, \alpha_0,\Omega_0}^* \Vert_F^2 \\
        \xi_2 &\equiv \frac{\mathbbm{1}(A_J)}{J} \Vert \widehat \theta_{\widehat{F}_J, \alpha_0, \Omega_0}^* - \theta^*_{\widehat{F}_J, \alpha_0, \Omega_0, \rho_J} \Vert_F^2 \\
        \xi_3 &\equiv \frac{\mathbbm{1}(A_J)}{J} \Vert \widehat \theta_{\widehat{F}_J, \alpha_0, \Omega_0, \rho_J}^* - \theta_{\rho_J}^* \Vert_F^2 \\
        \xi_4 &\equiv \frac{\mathbbm{1}(A_J)}{J} \Vert \theta_{\rho_J}^* - \theta^* \Vert_F^2.
    \end{align*}
    Then
\begin{align*}
    \frac{1}{J} E \left[ \Vert \widehat \theta_{\widehat{F}_J, \widehat\alpha, \widehat\Omega}^* - \theta^* \Vert_F^2 \mathbbm{1}(A_J) \right] \leq 4\left( E\xi_1 + E\xi_2 + E\xi_3 + E\xi_4 \right).
\end{align*}
By Lemma \ref{lem:oa3.3}, $\xi_1 \lesssim_{\mathcal H} M_J^2(\log J)^2 \Delta_J^2$ and thus $E\xi_1 \lesssim_{\mathcal H} M_J^2(\log J)^2 \Delta_J^2$.

By Lemma \ref{lem:oa3.1}, the truncation is not binding for the choice of $\rho_J$ in the lemma, so $\xi_2 = 0$.

By Lemma \ref{lem:oa3.5}, $E\xi_3 \lesssim_{\mathcal H} (\log J)^3 \delta_J^2$ for $\delta_J = J^{-1/2}(\log J)^{3/2}$, as in Corollary \ref{cor:oa3.1}.

By Lemma \ref{lem:oa3.4}, $E\xi_4 \lesssim_{\mathcal H} \frac{1}{J}$.

Thus the $E[\xi_3]$ rate is the dominating rate and the result follows from plugging in for $\delta_J$.
\end{proof}

\begin{lemma}\label{lem:oa3.3}
    Under the assumptions of Theorem \ref{thm:a.1}, $\xi_1 \lesssim_{\mathcal H} M_J^2 (\log J)^2 \Delta_J^2$.
\end{lemma}
\begin{proof}
Using Taylor's theorem and the equivalence of norms in fixed dimensions I can write
\begin{align*}
    &\left\Vert \widehat\theta^*_{j, \widehat{F}_J, \widehat\alpha, \widehat\Omega} - \theta^*_{j, \widehat{F}_J, \alpha_0, \Omega_0} \right\Vert_2 \\
    &\leq \left\Vert \Sigma_j \right\Vert_{op} \left\Vert B_j(\widehat\Omega)^{-1/2} \frac{\nabla f_{\widehat{F}_J, \widehat{\Psi}_j}(\widehat{Z}_j)}{f_{\widehat{F}_J, \widehat{\Psi}_j}(\widehat{Z}_j)} - B_j(\Omega_0)^{-1/2} \frac{\nabla f_{\widehat{F}_J, \Psi_j}(Z_j)}{f_{\widehat{F}_J, \Psi_j}(Z_j)} \right\Vert_2 \\
    &= \left\Vert \Sigma_j \right\Vert_{op} \left\Vert \frac{\partial \psi_j}{\partial a_j}\bigg\vert_{\widehat{F}_J, \widehat\alpha, \widehat\Omega} - \frac{\partial \psi_j}{\partial a_j} \bigg\vert_{\widehat{F}_J, \alpha_0, \Omega_0} \right\Vert_2 \\
    &\lesssim_{\mathcal H} \sup_{(\tilde\alpha, \tilde C) \in \mathcal{L}_j} \left\Vert \frac{\partial^2 \psi_j}{\partial a_j \partial a_j^T} \bigg\vert_{\widehat{F}_J, \tilde \alpha, \tilde C^2} \right\Vert_F \left\Vert a_j(\widehat\alpha) - a_j(\alpha_0) \right\Vert_\infty + \sup_{(\tilde\alpha, \tilde C) \in \mathcal{L}_j} \left\Vert \frac{\partial^2 \psi_j}{\partial a_j \partial \text{vec}(B_j^{1/2})^T} \bigg\vert_{\widehat{F}_J, \tilde \alpha, \tilde C^2} \right\Vert_F \left\Vert B_j(\widehat\Omega)^{1/2} - B_j(\Omega_0)^{1/2} \right\Vert_{op},
\end{align*}
where $\mathcal L_j$ denotes the line segment between $(a_j(\widehat\alpha), B_j(\widehat\Omega)^{1/2})$ and $(a_j(\alpha_0), B_j(\Omega_0)^{1/2})$. Using the bounds on derivatives obtained in Lemma \ref{lem:sm6.11}, 
\begin{align*}
    \mathbbm{1}(A_J) \left\Vert \widehat\theta^*_{j, \widehat{F}_J, \widehat\alpha, \widehat\Omega} - \theta^*_{j, \widehat{F}_J, \alpha_0, \Omega_0} \right\Vert_2 &\lesssim_{\mathcal H} M_J (\log J) \Delta_J \\
    \Rightarrow \xi_1 &\lesssim_{\mathcal H} M_J^2 (\log J)^2 \Delta_J^2.
\end{align*}
\end{proof}

\begin{lemma}\label{lem:oa3.1}
    Suppose $\bar Z_J = \max_{j \in [J]} \max(\Vert Z_j\Vert_2,1) \leq M_J, \Vert \widehat\chi - \chi_0 \Vert_J \leq \Delta_J$. Let $\widehat{F}_J$ satisfy Assumption \ref{ass:good_approx} and $\widehat\chi$ satisfy Assumption \ref{ass:est_var}. Then for $\Delta_J, M_J$ of the form \eqref{eq:oa3.4},
\begin{enumerate}
    \item $\max( \Vert \widehat Z_j \Vert_2, 1) \lesssim_{\mathcal H} M_J$
    \item There exists $C_{\mathcal H,\rho}$ such that with $\rho_J = \min\left(\frac{1}{J^4} \exp(-C_{\mathcal H,\rho} M_J^2 \Delta_J), \frac{1}{2\pi e} \right)$,  $f_{\widehat{F}_J, \Psi_j}(Z_j) \geq \frac{\rho_J}{\sqrt{\det(\Psi_j)}}$.
    \item The above $\rho_J$ satisfies $\log(1/\rho_J) \asymp_{\mathcal H} \log J, \varphi_+(\rho_J) \asymp \sqrt{\log(1/\rho_J)} \asymp_{\mathcal H} \sqrt{\log J}$, and $\rho_J \lesssim J^{-4}$.
\end{enumerate}
\end{lemma}
\begin{proof}
    For claim (1), recall that 
    \begin{align*}
        \widehat{Z}_j &= B_j(\widehat{\Omega})^{-1/2} B_j(\Omega_0)^{1/2} Z_j + B_j(\widehat{\Omega})^{-1/2} \left( a_j(\alpha_0) - a_j(\widehat\alpha) \right) \\
        \Rightarrow \max(\Vert \widehat{Z}_j \Vert_2, 1) &\leq \Vert B_j(\widehat{\Omega})^{-1/2} B_j(\Omega_0)^{1/2} Z_j \Vert_2 + \Vert B_j(\widehat{\Omega})^{-1/2} \left( a_j(\alpha_0) - a_j(\widehat\alpha) \right) \Vert_2 \\
        &\leq \Vert B_j(\widehat{\Omega})^{-1/2} \Vert_{op} \Vert B_j(\Omega_0)^{1/2} \Vert_{op} \Vert Z_j \Vert_2 + \Vert B_j(\widehat{\Omega})^{-1/2} \Vert_{op} \Vert a_j(\alpha_0) - a_j(\widehat\alpha) \Vert_2 \\
        &\lesssim_{\mathcal H} M_J + \Delta_J \lesssim_{\mathcal H} M_J.
    \end{align*}

    For claim (2), I can follow the proof of Theorem 5 in \citeappendix{jiang2020general} but for a multivariate distribution for a random vector in $\R^2$:
    Let $\widehat{F}_{J,j} = (1-\varepsilon) \widehat{F}_J + \varepsilon \delta_{\widehat Z_j}$. Then $f_{\widehat{F}_{J,j}, \widehat\Psi_i} (\widehat Z_i) \geq (1-\varepsilon) f_{\widehat{F}_{J}, \widehat\Psi_i} (\widehat Z_i)$ and $f_{\widehat{F}_{J,j}, \widehat\Psi_j} (\widehat Z_j) \geq \frac{\varepsilon}{\sqrt{\det(2\pi \widehat\Psi_j)}}$. So by Assumption \ref{ass:good_approx},
    \begin{align*}
        \prod_{i=1}^J f_{\widehat{F}_{J}, \widehat\Psi_i} (\widehat Z_i) &\geq \exp(-J\kappa_J) \prod_{i=1}^J f_{\widehat{F}_{J,j}, \widehat\Psi_i} (\widehat Z_i) \geq \exp(-J\kappa_J) (1-\varepsilon)^{J-1} \frac{\varepsilon}{\sqrt{\det(2\pi \widehat\Psi_j)}} \prod_{i \neq j} f_{\widehat{F}_{J,i}, \widehat\Psi_i} (\widehat Z_i).
    \end{align*}
    Thus taking $\varepsilon = 1/J$ and canceling terms, $f_{\widehat{F}_{J}, \widehat\Psi_j} (\widehat Z_j) \geq \frac{\exp(-J\kappa_J)}{2\pi e J \sqrt{\det(\widehat\Psi_j)}}$. Plugging in $\kappa_J = \frac{3}{J} \log \left( \frac{J}{(2\pi e)^{1/3}} \right)$ gives
    \begin{align*}
        f_{\widehat{F}_J, \widehat\Psi_j}(\widehat{Z}_j) &\geq \frac{1}{J^4 \sqrt{\det(\widehat\Psi_j)}},
    \end{align*}
    that is,
    \begin{align*}
        \int \frac{1}{2\pi} \exp \left(-\frac{1}{2}(\widehat{Z}_j-\tau)^T \widehat\Psi_j^{-1} (\widehat{Z}_j-\tau) \right) d\widehat{F}_J(\tau) &\geq \frac{1}{J^4}.
    \end{align*}
    Note that
    \begin{align*}
        (\widehat{Z}_j-\tau)^T \widehat\Psi_j^{-1} (\widehat{Z}_j-\tau) &= \left( \Sigma_j^{-1/2} B_j(\widehat\Omega)^{1/2}(\widehat{Z}_j - \tau) \right)^T \left( \Sigma_j^{-1/2} B_j(\widehat\Omega)^{1/2}(\widehat{Z}_j - \tau) \right)
    \end{align*}
    and one can verify that
    \begin{align*}
        &\Sigma_j^{-1/2} B_j(\widehat\Omega)^{1/2}(\widehat{Z}_j - \tau) \\
        &= \Sigma_j^{-1/2} B_j(\Omega_0)^{1/2}(Z_j - \tau) + \Sigma_j^{-1/2}(a_j(\alpha_0) - a_j(\widehat\alpha)) + \Sigma_j^{-1/2}(B_j(\Omega_0)^{1/2} - B_j(\widehat\Omega)^{1/2})\tau.
    \end{align*}
    Let
    \begin{align*}
        \xi(\tau) &\equiv \Sigma_j^{-1/2}(a_j(\alpha_0) - a_j(\widehat\alpha)) + \Sigma_j^{-1/2}(B_j(\Omega_0)^{1/2} - B_j(\widehat\Omega)^{1/2})\tau
    \end{align*}
    and note that $\Vert \xi(\tau) \Vert_2 \lesssim_{\mathcal H} \Delta_J M_J$ over the support of $\tau$ under $\widehat{F}_J$.
    Then 
    \begin{align*}
        &\exp \left(-\frac{1}{2}(\widehat{Z}_j-\tau)^T \widehat\Psi_j^{-1} (\widehat{Z}_j-\tau) \right) \\
        &= \exp \left(-\frac{1}{2}\left( \Sigma_j^{-1/2} B_j(\Omega_0)^{1/2}(Z_j - \tau) + \xi(\tau) \right)^T \left( \Sigma_j^{-1/2} B_j(\Omega_0)^{1/2}(Z_j - \tau) + \xi(\tau) \right) \right) \\
        &= \exp \left( -\frac{1}{2} (Z_j - \tau)^T \Psi_j^{-1} (Z_j - \tau) \right) \times \\
        &\qquad \exp \left(-\frac{1}{2} \xi(\tau)^T \xi(\tau) - \xi(\tau)^T \Sigma_j^{-1/2} B_j(\Omega_0)^{1/2}(Z_j - \tau) \right) \\
        &\leq \exp \left( -\frac{1}{2} (Z_j - \tau)^T \Psi_j^{-1} (Z_j - \tau) \right) \times \\
        &\qquad \exp \left( C_{\mathcal H, \rho} \Delta_J M_J \left\Vert \Sigma_j^{-1/2} B_j(\Omega_0)^{1/2}(Z_j - \tau) \right\Vert_2 \right) \\
        &\leq \exp \left( -\frac{1}{2} (Z_j - \tau)^T \Psi_j^{-1} (Z_j - \tau) \right) \exp \left( C_{\mathcal H, \rho} \Delta_J M_J^2 \right)
    \end{align*}
    where $C_{\mathcal H, \rho}$ is defined by optimizing the quadratic expression over $\Vert \xi(\tau)\Vert_2 \lesssim_{\mathcal H} \Delta_J M_J$ and the final line follows because $\left\Vert \Sigma_j^{-1/2} B_j(\Omega_0)^{1/2}(Z_j - \tau) \right\Vert_2 \lesssim_{\mathcal H} M_J$.
    Thus 
    \begin{align*}
        \int \frac{1}{2\pi} \exp \left(-\frac{1}{2}(Z_j-\tau)^T \Psi_j^{-1} (Z_j-\tau) \right) d\widehat{F}_J(\tau) &\geq \frac{1}{J^4} e^{-C_{\mathcal H, \rho} \Delta_J M_J^2} \\
        \Rightarrow f_{\widehat{F}_J, \Psi_j}(Z_j) &\geq \frac{1}{\sqrt{\det(\Psi_j)}} \frac{1}{J^4} e^{-C_{\mathcal H, \rho} \Delta_J M_J^2}.
    \end{align*}

    For claim (3), I calculate $\log(1/\rho_J) = \max(4 \log J + C_{\mathcal H, \rho}M_J^2\Delta_J, \log(2\pi e)) \asymp_{\mathcal H} \log J$, noting $M_J^2 \Delta_J \lesssim_{\mathcal H} J^{-1/2}(\log J)^{3/2} \lesssim_{\mathcal H} 1$.
\end{proof}

\begin{lemma}\label{lem:oa3.4}
    Under the assumptions of Theorem \ref{thm:a.1}, in the proof of Theorem \ref{thm:a.1} $E\xi_4 \lesssim_{\mathcal H} \frac{1}{J}$.
\end{lemma}

\begin{proof}
    Note that
    \begin{align*}
        E \left[ \Vert \theta^*_{j, \rho_J} - \theta^*_j \Vert_2^2 \right] &= E \left[ \left\Vert B_j(\Omega_0)^{1/2} \Psi_j \frac{\nabla f_{F_0,\Psi_j}(Z_j)}{\max(f_{F_0,\Psi_j}(Z_j), \frac{\rho_J}{\sqrt{\det{\Psi_j}}}) } - B_j(\Omega_0)^{1/2} \Psi_j \frac{\nabla f_{F_0,\Psi_j}(Z_j)}{f_{F_0,\Psi_j}(Z_j)} \right\Vert_2^2 \right] \\
        &\leq \Vert B_j(\Omega_0) \Vert_{op} E \left[ \left\Vert \Psi_j \frac{\nabla f_{F_0,\Psi_j}(Z_j)}{f_{F_0,\Psi_j}(Z_j)} \right\Vert_2^4 \right]^{1/2} E \left[ \left(1 - \frac{f_{F_0,\Psi_j}(Z_j)}{\max(f_{F_0,\Psi_j}(Z_j), \frac{\rho_J}{\sqrt{\det{\Psi_j}}}) } \right)^4 \right]^{1/2} \\
        &= \Vert B_j(\Omega_0) \Vert_{op} E \left[ \Vert E_{F_0, \alpha_0, \Omega_0} \left[ \tau_j - Z_j \vert Z_j \right] \Vert_2^4 \right]^{1/2} E \left[ \left(1 - \frac{f_{F_0,\Psi_j}(Z_j)}{\max(f_{F_0,\Psi_j}(Z_j), \frac{\rho_J}{\sqrt{\det{\Psi_j}}}) } \right)^4 \right]^{1/2} \\
        &\leq \Vert B_j(\Omega_0) \Vert_{op} E \left[ \Vert \tau_j - Z_j \Vert_2^4 \right]^{1/2} Pr \left( f_{F_0, \Psi_j}(Z_j) < \frac{\rho_J}{\sqrt{\det{\Psi_j}}} \right)^{1/2} \\
        &\lesssim_{\mathcal H} \Vert B_j(\Omega_0) \Vert_{op} E \left[ \Vert \tau_j - Z_j \Vert_2^4 \right]^{1/2} \rho_J^{1/4} \left(Var(Z_j^w) + Var(Z_j^g) \right)^{1/4} \\
        &\lesssim_{\mathcal H} \rho_J^{1/4} \lesssim_{\mathcal H} \frac{1}{J},
    \end{align*}
    where the second line follows from submultiplicativity and Cauchy-Schwarz, the third line from Tweedie's formula, the fourth line from Jensen's inequality, the fifth line from Lemma \ref{lem:sm6.9}, and the final line from Lemma \ref{lem:oa3.1}.
\end{proof}

\begin{corollary} \label{cor:oa3.1}
    Assume Assumptions \ref{ass:compact}, \ref{ass:bdd_sig}, \ref{ass:est_var}, and \ref{ass:good_approx} hold. Suppose $\Delta_J, M_J$ take the form \eqref{eq:oa3.4}. Define the rate function 
    \begin{align*}
        \delta_J = J^{-1/2}(\log J)^{3/2}.
    \end{align*}
    Then there exists some constant $B_{\mathcal H}$, depending solely on $C_{\mathcal H}^*$ in Corollary \ref{cor:sm6.1} and $\mathcal H$ such that
    \begin{align*}
        Pr \left( A_J, \bar h(f_{\hat F_J, \cdot}, f_{F_0, \cdot}) > B_{\mathcal H} \delta_J \right) \leq \left( \frac{\log\log J}{\log 2} + 10 \right)\frac{1}{J}.
    \end{align*} 
\end{corollary}
The proof of Corollary \ref{cor:oa3.1} is deferred to Supplemental Appendix \ref{app:cor:oa3.1}.

\begin{lemma} \label{lem:oa3.5}
    Under the assumptions of Theorem \ref{thm:a.1}, in the proof of Theorem \ref{thm:a.1}, $E\xi_3 \lesssim_{\mathcal H} (\log J)^3 \delta_J^2$, for $\delta_J = J^{-1/2}(\log J)^{3/2}$, as in Corollary \ref{cor:oa3.1}.
\end{lemma}
\begin{proof}
    Note that 
\begin{align*}
    \Vert \widehat \theta_{\widehat{F}_J, \alpha_0, \Omega_0, \rho_J}^* - \theta_{\rho_J}^* \Vert_F &= \Vert B_j(\Omega_0)^{1/2} (\widehat \tau_{\widehat{F}_J, \alpha_0, \Omega_0, \rho_J}^* - \tau_{\rho_J}^*) \Vert_F \\
    &\leq \Vert B_j(\Omega_0)^{1/2} \Vert_{op} \Vert \widehat \tau_{\widehat{F}_J, \alpha_0, \Omega_0, \rho_J}^* - \tau_{\rho_J}^* \Vert_F.
\end{align*}
Thus to control $\xi_3$ I will control the object $\frac{\mathbbm{1}(A_J)}{J} \Vert \widehat \tau_{\widehat{F}_J, \alpha_0, \Omega_0, \rho_J}^* - \tau_{\rho_J}^* \Vert_F^2$.

Let $B_J = \{\bar h(f_{\widehat{F}_J, \cdot}, f_{F_0, \cdot}) < B_{\mathcal H} \delta_J\}$ for constant $B_{\mathcal H}$ in Corollary \ref{cor:oa3.1}.
Let $F_1, \dots, F_N$ be a set of prior distributions that is a minimal $\omega$-covering of $\{F: \bar{h}(f_{F, \cdot}, f_{F_0, \cdot}) \leq B_{\mathcal H} \delta_J\}$ in the metric
\begin{align*}
    d_{M_J, \rho_J}(H_1, H_2) = \max_{j \in [J]} \sup_{z: \Vert z \Vert_2 \leq M_J} \left\Vert \frac{\Psi_j \nabla f_{H_1, \Psi_j}(z)}{\max \left(f_{H_1, \Psi_j}(z), \frac{\rho}{\sqrt{\det(\Psi_j)}} \right)} - \frac{\Psi_j \nabla f_{H_2, \Psi_j}(z)}{\max\left(f_{H_2, \Psi_j}(z), \frac{\rho}{\sqrt{\det(\Psi_j)}} \right)} \right\Vert_2,
\end{align*}
where $N \leq N(\omega/2, \mathcal{P}(\R^2), d_{M_J, \rho_J})$ by monotonicity relation of covering numbers, as in \citeappendix{chen2022empirical}.
Let $\tau_{\rho_J}^{(i)}$ be the posterior mean vector corresponding to prior $F_i$ with conditional moments $\alpha_0, \Omega_0$ and regularization parameter $\rho_J$.
Then
\begin{align*}
    \frac{\mathbbm{1}(A_J)}{J} \Vert \widehat \tau_{\widehat{F}_J, \alpha_0, \Omega_0, \rho_J}^* - \tau_{\rho_J}^* \Vert_F^2 &\leq \frac{4}{J} \left( \zeta_1^2 + \zeta_2^2 + \zeta_3^2 + \zeta_4^2 \right), \\
    \zeta_1^2 &\equiv \Vert \widehat \tau_{\widehat{F}_J, \alpha_0, \Omega_0, \rho_J}^* - \tau_{\rho_J}^* \Vert_F^2 \mathbbm{1}(A_J \cap B_J^C) \\
    \zeta_2^2 &\equiv \left( \Vert \widehat \tau_{\widehat{F}_J, \alpha_0, \Omega_0, \rho_J}^* - \tau_{\rho_J}^* \Vert_F - \max_{i \in [N]} \Vert \tau_{\rho_J}^{(i)} - \tau_{\rho_J}^* \Vert_F \right)^2_+ \mathbbm{1}(A_J \cap B_J) \\
    \zeta_3^2 &\equiv \max_{i \in [N]} \left(\Vert \tau_{\rho_J}^{(i)} - \tau_{\rho_J}^* \Vert_F - E [ \Vert \tau_{\rho_J}^{(i)} - \tau_{\rho_J}^* \Vert_F ] \right)^2_+ \\
    \zeta_4^2 &\equiv \max_{i \in [N]} \left(E [ \Vert \tau_{\rho_J}^{(i)} - \tau_{\rho_J}^* \Vert_F ] \right)^2.
\end{align*}

I will show that $\frac{1}{J} E[\zeta_1^2] \lesssim_{\mathcal H} \frac{\log J \log \log J}{J}$, $\frac{1}{J} E[\zeta_2^2+\zeta_3^2] \lesssim_{\mathcal H} \frac{(\log J)^4}{J}$, and $\frac{1}{J} E[\zeta_4^2] \lesssim_{\mathcal H} (\log J)^3 \delta_J^2$.
By definition of $\delta_J$, the dominating rate is $(\log J)^3 \delta_J^2 \lesssim_{\mathcal H} \frac{(\log J)^6}{J}$.
Thus $E[\xi_3] \lesssim_{\mathcal H} (\log J)^3 \delta_J^2$.

\textbf{To control $\zeta_1$}:
From section D.3.1 in \citeappendix{soloff2024multivariate},
\begin{align*}
    \frac{1}{J} E \zeta_1^2 &\lesssim_{\mathcal H} \varphi_+(\rho_J)^2 Pr(A_J \cap B_J^C) \lesssim_{\mathcal H} (\log J) Pr(A_J \cap B_J^C).
\end{align*}
By Corollary \ref{cor:oa3.1}, $Pr(A_J \cap B_J^C) \leq \left( \frac{\log \log J}{\log 2} + 10 \right) \frac{1}{J}$ and hence $\frac{1}{J} E \zeta_1^2 \lesssim_{\mathcal H} \frac{\log J \log\log J}{J}$.

\textbf{To control $\zeta_2$ and $\zeta_3$}:
As in section OA3.2.2 in \citeappendix{chen2022empirical} and section D.3.2 in \citeappendix{soloff2024multivariate}, on $A_J \cap B_J$ I can write
\begin{align*}
    &\frac{1}{J} \zeta_2^2 \\
    &\leq \mathbbm{1}(A_J \cap B_J) \min_{i \in [N]} \frac{1}{J} \sum_{j=1}^J \mathbbm{1}(\Vert Z_j \Vert_2 \leq M_J) \left\Vert \frac{\Psi_j \nabla f_{\widehat{F}_J, \Psi_j}(Z_j)}{\max(f_{\widehat{F}_J, \Psi_j}(Z_j), \frac{\rho}{\sqrt{\det(\Psi_j)}})} - \frac{\Psi_j \nabla f_{F_i, \Psi_j}(Z_j)}{\max(f_{F_i, \Psi_j}(Z_j), \frac{\rho}{\sqrt{\det(\Psi_j)}} )} \right\Vert_2^2 \\
    &\leq \omega^2.
\end{align*}

Section D.3.3 of \citeappendix{soloff2024multivariate} gives us
\begin{align*}
    E \zeta_3^2 &\lesssim_{\mathcal H} (\varphi_+(\rho_J))^2 \log(eN) \lesssim_{\mathcal H} \log J \log N
\end{align*}
using Lemma \ref{lem:oa3.1}.
Following section D.3.5 of \citeappendix{soloff2024multivariate}, I will choose $\omega = 2 \left(\bar k^{3/2} \varphi_+(\rho_J) + \bar k^2 \right) \frac{1}{J}$.
Note that by section D.3.5 of \citeappendix{soloff2024multivariate} and because $J \geq \frac{5}{\underline k}$, I can bound the metric entropy
\begin{align*}
    \log N \left( \left(\bar k^{3/2} \varphi_+(\rho_J) + \bar k^2 \right) \frac{1}{J}, \mathcal{P}(\mathbb R^2), d_{M_J, \rho_J} \right) &\lesssim_{\mathcal H} (\log J)^2 M_J^2.
\end{align*}
Also $\frac{1}{J} \zeta_2^2 \lesssim_{\mathcal H} J^{-1} \sqrt{\log J}$.
Thus $\frac{1}{J}E[\zeta_2^2+\zeta_3^2] \lesssim_{\mathcal H} \frac{(\log J)^3 M_J^2}{J} \lesssim_{\mathcal H} \frac{(\log J)^4}{J}$.

\textbf{To control $\zeta_4$}:
As in Section D.3.4 of \citeappendix{soloff2024multivariate}, using Lemma E.1 of \citeappendix{saha2020nonparametric} I can write
\begin{align*}
    \left( E \left\Vert \tau_{\rho_J}^{(i)} - \tau_{\rho_J}^* \right\Vert_F \right)^2 &\lesssim_{\mathcal H} \sum_{j=1}^J \max \left\{ (\varphi_+(\rho_J))^6, \left\vert \log h\left(f_{F_0, \Psi_j}, f_{F^{(i)}, \Psi_j} \right) \right\vert \right\} h^2\left(f_{F_0, \Psi_j}, f_{F^{(i)}, \Psi_j} \right).
\end{align*}
Then following the exact same argument in section OA3.2.4 of \citeappendix{chen2022empirical}, $\frac{1}{J} E \zeta_4^2 \lesssim_{\mathcal H} (\log J)^3 \delta_J^2 \lesssim_{\mathcal H} (\log J)^6 J^{-1}$.
\end{proof}

\subsection{Proof of Corollary \ref{cor:oa3.1}} \label{app:cor:oa3.1}

\subsubsection{Derivative computations} \label{app:sec:sm6.1}
I first compute derivatives of $\psi_j$ with respect to $a_j(\alpha)$ and $B_j(\Omega)^{1/2}$, which will be useful in later proofs.
Let $\nabla^2 f_{F, \Psi}(z)$ denote the Hessian matrix of $f_{F, \Psi}$ evaluated at $z$. Since all derivatives are evaluated at $F, \alpha, \Omega$, I denote $\hat{z} = \widehat{Z}_j(\alpha, \Omega)$ and $\hat \Psi_j = B_j(\Omega)^{-1/2}\Sigma_j B_j(\Omega)^{-1/2}$.
Then
\begin{small}
\begin{align*}
    \frac{\nabla f_{F, \Psi}(z)}{f_{F, \Psi}(z)} &= \Psi^{-1} E[\tau - Z \vert z] \\
    \frac{\nabla^2 f_{F, \Psi}(z)}{f_{F, \Psi}(z)} &= \Psi^{-1} E[(\tau-Z)(\tau-Z)^T \vert z] \Psi^{-1} - \Psi^{-1}
\end{align*}
\vspace{-25pt}
\begin{align*}
    \frac{\partial \psi_j}{\partial a_j^T} \bigg\vert_{F, \alpha, \Omega} &= -B_j(\Omega)^{-1/2} \frac{\nabla f_{F, \hat\Psi_j}(\hat z)}{f_{F, \hat\Psi_j}(\hat z)} \\
    &= \Sigma_j^{-1} B_j(\Omega)^{1/2} E[Z_j - \tau_j \vert \hat z]
\end{align*}
\vspace{-25pt}
\begin{align*}
    \frac{\partial^2 \psi_j}{\partial a_j \partial a_j^T} \bigg\vert_{F, \alpha, \Omega} &= B_j(\Omega)^{-1/2} \left( \frac{\nabla^2 f_{F, \hat\Psi_j}(\hat z)}{f_{F, \hat\Psi_j}(\hat z)} - \frac{\nabla f_{F, \hat\Psi_j}(\hat z) \nabla f_{F, \hat\Psi_j}(\hat z)^T}{f^2_{F, \hat\Psi_j}(\hat z)} \right) B_j(\Omega)^{-1/2}  \\
    &= B_j(\Omega)^{-1/2} \bigg( \hat\Psi_j^{-1} E[(\tau_j-Z_j)(\tau_j-Z_j)^T | \hat z] \hat\Psi_j^{-1} \\
    &\qquad - \hat\Psi_j^{-1} - \hat\Psi_j^{-1} E[\tau_j - Z_j|\hat z] E[\tau_j - Z_j|\hat z]^T \hat\Psi_j^{-1} \bigg) B_j(\Omega)^{-1/2}  
\end{align*}
\vspace{-25pt}
\begin{align*}
    \frac{\partial \psi_j}{\partial B_j^{1/2}} \bigg\vert_{F, \alpha, \Omega} &= \frac{\Sigma_j^{-1}B_j(\Omega)^{1/2}}{\sqrt{\det(2\pi \hat\Psi_j)}f_{F,\hat\Psi_j}(\hat z)} \int \varphi_{\hat\Psi_j}(\widehat{Z}_j - \tau) (\widehat{Z}_j - \tau)\tau^T dF(\tau) \\
    &= \Sigma_j^{-1}B_j(\Omega)^{1/2} E[(Z_j - \tau_j)\tau_j^T \vert \hat z] \\
    \Rightarrow \frac{\partial \psi_j}{\partial \text{vec}(B_j^{1/2})} \bigg\vert_{F, \alpha, \Omega} &= \frac{(I_2 \otimes \Sigma_j^{-1}B_j(\Omega)^{1/2})}{\sqrt{\det(2\pi \hat\Psi_j)}f_{F,\hat\Psi_j}(\hat z)} \underbrace{\int \varphi_{\hat\Psi_j}(\widehat{Z}_j - \tau) \text{vec}\left( (\widehat{Z}_j - \tau)\tau^T \right) dF(\tau)}_{Q_j(\widehat{Z}_j, F, \hat\Psi_j)} \\
    &= (I_2 \otimes \Sigma_j^{-1}B_j(\Omega)^{1/2}) E \left[\text{vec}\left((Z_j-\tau_j)\tau_j^T \right) \vert \hat z \right]
\end{align*}
\vspace{-25pt}

\begin{align*}
    \frac{\partial^2 \psi_j}{\partial \text{vec}(B_j^{1/2}) \partial a_j^T} \bigg\vert_{F, \alpha, \Omega} &= \frac{I_2 \otimes \Sigma_j^{-1}B_j(\Omega)^{1/2}}{\sqrt{\det(2\pi \hat\Psi_j)} f_{F,\hat\Psi_j}(\hat z)} \int \varphi_{\hat\Psi_j}(\widehat{Z}_j - \tau) \bigg\{\text{vec}\left((\widehat{Z}_j - \tau)\tau^T \right) (\widehat{Z}_j - \tau)^T \hat\Psi_j^{-1} \\[-15pt]
    &\hspace{200pt} - \tau \otimes I_2 \bigg\} B_j(\Omega)^{-1/2} dF(\tau)  \\
    &\qquad + (I_2 \otimes \Sigma_j^{-1}B_j(\Omega)^{1/2}) \frac{Q_j(Z_j, F, \hat\Psi_j)}{\sqrt{\det(2\pi \hat\Psi_j)}f_{F,\hat\Psi_j}(\hat z)} \frac{(\nabla f_{F,\hat\Psi_j}(\hat z))^T}{f_{F,\hat\Psi_j}(\hat z)} B_j(\Omega)^{-1/2} \\
    &= (I_2 \otimes \Sigma_j^{-1}B_j(\Omega)^{1/2}) E \left[\text{vec}\left((Z_j - \tau)\tau^T \right) (Z_j - \tau)^T \hat\Psi_j^{-1} - \tau \otimes I_2 \vert \hat z \right] B_j(\Omega)^{-1/2}  \\
    &\qquad + (I_2 \otimes \Sigma_j^{-1}B_j(\Omega)^{1/2}) E \left[\text{vec}\left((Z_j-\tau_j)\tau_j^T \right) \vert \hat z \right] E\left[(\tau_j - Z_j)^T \vert \hat z \right] \hat\Psi_j^{-1} B_j(\Omega)^{-1/2} 
\end{align*}
\end{small}

While calculation of the derivative $\displaystyle \frac{\partial^2 \psi_j}{\partial \text{vec}(B_j^{1/2}) \partial \text{vec}(B_j^{1/2})^T} \bigg\vert_{F, \alpha, \Omega}$ is tricky, one can verify that it is the weighted sum of posterior means 
$E\left[ \text{vec}\left((Z_j-\tau_j)\tau_j^T \right) \text{vec}\left((Z_j-\tau_j)\tau_j^T \right)^T \vert \hat z \right]$, \\ 
$E \left[\text{vec}\left((Z_j-\tau_j)\tau_j^T \right) \vert \hat z \right] E \left[\text{vec}\left((Z_j-\tau_j)\tau_j^T \right) \vert \hat z \right]^T$, 
$E[Z\tau^T \vert \hat z],$ 
and $E[\text{vec}(Z-\tau)\tau^T \vert \hat z]$, with weights that are simple functions of $\Sigma_j$ and $B_j(\Omega)$.

\subsubsection{Preliminary results}

Throughout this subsection I use the following high-level assumption on rates $\Delta_J, M_J$, which is exactly Assumption SM6.1 in \citeappendix{chen2022empirical}. Note that the assumption is satisfied for the choice \eqref{eq:oa3.4}.

\begin{assumption}\label{ass:sm6.1}
    Assume that 1) $\frac{1}{\sqrt{J}} \lesssim_{\mathcal H} \Delta_J \lesssim_{\mathcal H} M_J^{-3} \lesssim_{\mathcal H} 1$, and 2) $\sqrt{\log J} \lesssim_{\mathcal H} M_J$.
\end{assumption}

Much of this subsection will be focused on proving the following result.
\begin{theorem}\label{thm:sm6.1}
Under the assumptions of Theorem \ref{thm:a.1} and Assumption \ref{ass:sm6.1}, there exist constants $C_{1,\mathcal H}, C_{2,\mathcal H} > 0$ such that the following tail bound holds: Let
\begin{align*}
    \epsilon_J &= M_J \sqrt{\log J} \Delta_J \frac{1}{J} \sum_{j=1}^J h \left(f_{\widehat{F}_J, \Psi_j}, f_{F_0, \Psi_j} \right) + \Delta_J \log J e^{-C_{2,\mathcal H}M_J^2} + \Delta_J^2 M_J^2 \log J + \frac{M_J^2(\log J)^{3/2} \Delta_J}{\sqrt{J}}.
\end{align*}
Then
\begin{align*}
    Pr \left(\bar Z_J \leq M_J, \Vert \widehat\chi - \chi_0 \Vert_J \leq \Delta_J, \text{Sub}_J(\widehat{F}_J) >  C_{1,\mathcal H} \epsilon_J \right) &\leq \frac{9}{J}.
\end{align*}
\end{theorem}

Plugging in the rates \eqref{eq:oa3.4}, I obtain the following corollary:
\begin{corollary}\label{cor:sm6.1}
    Under the assumptions of Theorem \ref{thm:a.1}, suppose $\Delta_J, M_J$ are of the form \eqref{eq:oa3.4}. Then there exists a constant $C_{\mathcal H}^*$ such that the following tail bound holds: Let
\begin{align*}
    \varepsilon_J &= J^{-1/2} (\log J)^{3/2} \bar h \left( f_{\widehat{F}_J, \cdot}, f_{F_0, \cdot}\right) + J^{-1}(\log J)^3,
\end{align*}
then 
\begin{align*}
    Pr\left(A_J, \text{Sub}_J(\widehat{F}_J) > C_H^* \varepsilon_J \right) &\leq \frac{9}{J}.
\end{align*}
\end{corollary}
The proof follows exactly as the proof of Corollary SM6.1 of \citeappendix{chen2022empirical} but plugging in the rates for $\Delta_J$ and $M_J$ from \eqref{eq:oa3.4}.

\begin{proof}[Proof of Theorem \ref{thm:sm6.1}]
    As in section SM6.2.1 of \citeappendix{chen2022empirical}, if I construct random variables $a_J$ and $b_J$ such that on the event $A_J$,
    \begin{align*}
        \left\vert \frac{1}{J} \sum_{j=1}^J \psi_j(Z_j, \widehat\alpha, \widehat\Omega, \widehat{F}_J) - \frac{1}{J} \sum_{j=1}^J \psi_j(Z_j, \alpha_0, \Omega_0, \widehat{F}_J) \right\vert &\leq a_J, \\
        \left\vert \frac{1}{J} \sum_{j=1}^J \psi_j(Z_j, \widehat\alpha, \widehat\Omega, F_0) - \frac{1}{J} \sum_{j=1}^J \psi_j(Z_j, \alpha_0, \Omega_0, F_0) \right\vert &\leq b_J,
    \end{align*}
    then to prove the theorem it suffices to show that $Pr(\mathbbm{1}(A_J)(a_J + b_J + \kappa_J) \gtrsim_{\mathcal H} \epsilon_J) \leq \frac{9}{J}$.

    Let $\Delta_{\alpha,j} = a_j(\widehat{\alpha}) - a_j(\alpha_0)$, $\Delta_{\Omega,j} = \text{vec}(B_j(\widehat{\Omega})^{1/2}) - \text{vec}(B_j(\Omega_0)^{1/2})$, and $\Delta_j \equiv (\Delta_{\alpha,j}^T, \Delta_{\Omega,j}^T)^T$. I can take a second-order Taylor expansion of $\psi_j(Z_j, \widehat\alpha, \widehat\Omega, \widehat{F}_J) - \psi_j(Z_j, \alpha_0, \Omega_0, \widehat{F}_J)$ around $(a_j(\alpha), \text{vec}(B_j(\Omega)^{1/2}))$:
    \begin{align*}
        &\psi_j(Z_j, \widehat\alpha, \widehat\Omega, \widehat{F}_J) - \psi_j(Z_j, \alpha_0, \Omega_0, \widehat{F}_J) \\
        &= \frac{\partial \psi_j}{\partial a_j^T} \bigg\vert_{\widehat{F}_J, \chi_0} \Delta_{\alpha,j} + \frac{\partial \psi_j}{\partial \text{vec}(B_j^{1/2})^T} \bigg\vert_{\widehat{F}_J, \chi_0} \Delta_{\Omega,j} + \underbrace{\frac{1}{2} \Delta_j^T H_j(\tilde \alpha, \tilde \Omega, \widehat{F}_J) \Delta_j}_{R_{1j}},
    \end{align*}
    where $H_j(\tilde \alpha, \tilde \Omega, \widehat{F}_J)$ is the Hessian matrix with respect to $(a_j(\alpha), \text{vec}(B_j(\Omega))^{1/2})$.

    Truncate the denominators of the first derivatives by Lemma \ref{lem:oa3.1} for the choice of $\rho_J$ in \eqref{eq:oa3.5}, which does not bind, so that
    \begin{align*}
        D_{\alpha,j}(Z_j, \widehat{F}_J, \chi_0, \rho_J) &\equiv -B_j(\Omega_0)^{-1/2} \frac{\nabla f_{\widehat{F}_J, \Psi_j}(Z_j)}{\max\left(f_{\widehat{F}_J, \Psi_j}(Z_j), \frac{\rho_J}{\sqrt{\det(\Psi_j)}} \right)} = \frac{\partial \psi_j}{\partial a_j} \bigg\vert_{\widehat{F}_J, \chi_0} \\
        D_{\Omega,j}(Z_j, \widehat{F}_J, \chi_0, \rho_J) &\equiv \frac{(I_2 \otimes \Sigma_j^{-1}B_j(\Omega_0)^{1/2}) Q_j(Z_j, \widehat{F}_J, \Psi_j)}{\sqrt{\det(2\pi \Psi_j)} \max\left(f_{\widehat{F}_J,\Psi_j}(Z_j), \frac{\rho_J}{\sqrt{\det(\Psi_j)}} \right)} = \frac{\partial \psi_j}{\partial \text{vec}(B_j^{1/2})} \bigg\vert_{\widehat{F}_J, \chi_0}.
    \end{align*}
    Defining
    \begin{align*}
        \overline D_{k,j}(\widehat{F}_J, \chi_0, \rho_J) &= \int D_{k,j}(z, \widehat{F}_J, \chi_0, \rho_J) f_{F_0, \Psi_j}(z) dz \qquad \text{for } k \in \{\alpha, \Omega\},
    \end{align*}
    as in section SM6.2.2 of \citeappendix{chen2022empirical} I define for each $k \in \{\alpha, \Omega\}$
    \begin{align*}
        U_{1k} &= \frac{1}{J} \sum_{j=1}^J \overline{D}_{k,j}(\widehat{F}_J, \chi_0, \rho_J)^T \Delta_{k,j} \\
        U_{2k} &= \frac{1}{J} \sum_{j=1}^J \left[D_{k,j}(Z_j,\widehat{F}_J, \chi_0, \rho_J) - \overline{D}_{k,j}(\widehat{F}_J, \chi_0, \rho_J) \right]^T \Delta_{k,j} \\
        R_1 &= \frac{1}{J} \sum_{j=1}^J R_{1j}
    \end{align*}
    and let 
    \begin{align*}
        a_J &= |R_1| + \sum_{k \in \{\alpha,\Omega\}} |U_{1k}| + |U_{2k}|.
    \end{align*}

    Similarly take a Taylor expansion
    \begin{align*}
        \psi_j(Z_j, \widehat\alpha, \widehat\Omega, F_0) - \psi_j(Z_j, \alpha_0, \Omega_0, F_0) &= \frac{\partial \psi_j}{\partial a_j^T} \bigg\vert_{F_0, \chi_0} \Delta_{\alpha,j} + \frac{\partial \psi_j}{\partial \text{vec}(B_j^{1/2})^T} \bigg\vert_{F_0, \chi_0} \Delta_{\Omega,j} \\
        &\qquad + \underbrace{\frac{1}{2} \Delta_j^T H_j(\tilde \alpha, \tilde \Omega, F_0) \Delta_j}_{R_{2j}} \\
        &= \sum_{k \in \{\alpha,\Omega\}} D_{k,j}(Z_j, F_0, \chi_0, 0)^T \Delta_{k,j} + R_{2j} \\
        &\equiv U_{3\alpha j} + U_{3 \Omega j} + R_{2j}.
    \end{align*}
    Defining $U_{3k} = \frac{1}{J} \sum_{j=1}^J U_{3kj}$ for $k \in \{\alpha,\Omega\}$ and $R_2 = \frac{1}{J} \sum_{j=1}^J R_{2j}$, let
    \begin{align*}
        b_J &= |R_2| + \sum_{k \in \{\alpha,\Omega\}} |U_{3k}|.
    \end{align*}
    Note that 
    \begin{align*}
        a_J + b_J + \kappa_J &\leq \kappa_J + |R_1| + |R_2| + \sum_{k \in \{\alpha,\Omega\}} |U_{1k}| + |U_{2k}| + |U_{3k}|.
    \end{align*}
    I now bound each term individually, following \citeappendix{chen2022empirical}.

    \textbf{Bounding $U_{1\alpha}$}:
    
    I will follow the proof of Lemma SM6.1 in \citeappendix{chen2022empirical} to show that 
    \begin{align*}
        |U_{1\alpha}| &\equiv \left\vert \frac{1}{J}\sum_{j=1}^J \overline{D}_{\alpha,j}(\widehat{F}_J, \chi_0, \rho_J)^T \Delta_{\alpha,j} \right\vert \lesssim_{\mathcal H} \Delta_J \left[ \frac{\sqrt{\log J}}{J} \sum_{j=1}^J h\left( f_{F_0, \Psi_j}, f_{\widehat{F}_J, \Psi_j} \right) + \frac{M_J^{1/2}}{J} \right].
    \end{align*}
    Note that
    \begin{align}
        \left\Vert \overline D_{\alpha,j}(\widehat{F}_J, \chi_0, \rho_J) \right\Vert_2 &\lesssim_{\mathcal H} \left\Vert \int \frac{\nabla f_{\widehat{F}_J, \Psi_j}}{\max \left(f_{\widehat{F}_J, \Psi_j}, \frac{\rho_J}{\sqrt{\det(\Psi_j)}} \right)} f_{F_0, \Psi_j}(z) dz \right\Vert_2 \nonumber \\
        &\leq \left\Vert \int \frac{\nabla f_{\widehat{F}_J, \Psi_j}}{\max \left(f_{\widehat{F}_J, \Psi_j}, \frac{\rho_J}{\sqrt{\det(\Psi_j)}} \right)} \left[f_{F_0, \Psi_j}(z) - f_{\widehat{F}_J, \Psi_j}(z) \right] dz \right\Vert_2 \label{eq:sm6.29} \\
        &\qquad + \left\Vert \int \frac{\nabla f_{\widehat{F}_J, \Psi_j}}{\max \left(f_{\widehat{F}_J, \Psi_j}, \frac{\rho_J}{\sqrt{\det(\Psi_j)}} \right)} f_{\widehat{F}_J, \Psi_j}(z) dz \right\Vert_2. \label{eq:sm6.30}
    \end{align}
    Following section SM6.3.1 of \citeappendix{chen2022empirical}, 
    \begin{align*}
        [\eqref{eq:sm6.29}]^2 &\lesssim h^2\left(f_{F_0, \Psi_j}, f_{\widehat{F}_J, \Psi_j} \right) \int \frac{ \left\Vert \nabla f_{\widehat{F}_J, \Psi_j} \right\Vert_2^2 }{\left(\max \left(f_{\widehat{F}_J, \Psi_j}, \frac{\rho_J}{\sqrt{\det(\Psi_j)}} \right) \right)^2} \left( f_{F_0,\Psi_j}(z) + f_{\widehat{F}_J, \Psi_j}(z) \right) dz.
    \end{align*}
    By Lemmas \ref{lem:oa3.1} and \ref{lem:sm6.8},
    \begin{align*}
        \frac{ \left\Vert \nabla f_{\widehat{F}_J, \Psi_j} \right\Vert_2^2 }{\left(\max \left(f_{\widehat{F}_J, \Psi_j}, \frac{\rho_J}{\sqrt{\det(\Psi_j)}} \right) \right)^2} &\lesssim \left\Vert \Psi_j^{-1} \right\Vert_{F} \varphi_+^2(\rho_J) \lesssim_{\mathcal H} \log J \\
        \Rightarrow \eqref{eq:sm6.29} &\lesssim_{\mathcal H} h\left(f_{F_0, \Psi_j}, f_{\widehat{F}_J, \Psi_j} \right) \sqrt{\log J}.
    \end{align*}
    As in section SM6.3.2 of \citeappendix{chen2022empirical}, by Cauchy-Schwarz
    \begin{align*}
        \eqref{eq:sm6.30} &\leq \int \left\Vert \frac{\nabla f_{\widehat{F}_J, \Psi_j}}{f_{\widehat{F}_J, \Psi_j}} \right\Vert_2 \mathbbm{1}\left( f_{\widehat{F}_J, \Psi_j}(z) \leq \frac{\rho_J}{\sqrt{\det(\Psi_j)}} \right) f_{\widehat{F}_J, \Psi_j}(z) dz \\
        &\leq \sqrt{E_{Z \sim f_{\widehat{F}_J, \Psi_j}} \left[ \left\Vert \Psi_j^{-1} E_{\widehat{F}_J, \Psi_j}\left[\tau - Z \vert Z \right] \right\Vert_2^2 \right]} \sqrt{Pr_{f_{\widehat{F}_J, \Psi_j}} \left(f_{\widehat{F}_J, \Psi_j}(Z) \leq \frac{\rho_J}{\sqrt{\det(\Psi_j)}} \right) }.
    \end{align*}
    By Jensen's inequality and law of iterated expectations, the first term is
    \begin{align*}
        \sqrt{E_{Z \sim f_{\widehat{F}_J, \Psi_j}} \left[ \left\Vert \Psi_j^{-1} E_{\widehat{F}_J, \Psi_j}\left[\tau - Z \vert Z \right] \right\Vert_2^2 \right]} &\leq \left\Vert \Psi_j^{-1} \right\Vert_F \sqrt{E_{\tau \sim \widehat{F}_J, Z \sim N(\tau,\Psi_j)} \left[ \left\Vert \tau - Z \right\Vert_2^2 \vert Z \right]} \\
        &= \left\Vert \Psi_j^{-1} \right\Vert_F \sqrt{\text{tr}(\Psi_j)}.
    \end{align*}
    By Lemma \ref{lem:sm6.9}, the second term is bounded by a constant times $\rho_J^{1/4}\left(\text{tr}\left(Var_{Z \sim f_{\widehat{F}_J, \Psi_j}}(Z) \right) \right)^{1/4}$ and $\text{tr}\left(Var_{Z \sim f_{\widehat{F}_J, \Psi_j}}(Z) \right) \lesssim_{\mathcal H} M_J^2$, so by Lemma \ref{lem:oa3.1}, $\eqref{eq:sm6.30} \lesssim_{\mathcal H} \rho_J^{1/4} M_J^{1/2} \lesssim_{\mathcal H} M_J^{1/2} J^{-1}$.
    The result follows by using these results to bound $|U_{1\alpha}|$.

    \textbf{Bounding $U_{1\Omega}$}:
    
    I will follow the proof of Lemma SM6.2 in \citeappendix{chen2022empirical} to show that 
    \begin{align*}
        |U_{1\Omega}| &\lesssim_{\mathcal H} \Delta_J \left[ \frac{M_J \sqrt{\log J}}{J} \sum_{j=1}^J h\left( f_{\widehat{F}_J, \Psi_j}, f_{F_0, \Psi_j} \right) + \frac{M_J^{3/2}}{J} \right].
    \end{align*}
    As in the proof to bound $U_{1\alpha}$, decompose
    \begin{align}
        \left\Vert \overline D_{\Omega,j}(\widehat{F}_J, \chi_0, \rho_J) \right\Vert_2 &\lesssim_{\mathcal H} \left\Vert \int \frac{Q_j(z, \widehat{F}_J,\Psi_j)}{\max \left(f_{\widehat{F}_J, \Psi_j}, \frac{\rho_J}{\sqrt{\det(\Psi_j)}} \right)} \left[f_{F_0, \Psi_j}(z) - f_{\widehat{F}_J, \Psi_j}(z) \right] dz \right\Vert_2 \label{eq:sm6.33} \\
        &\qquad + \left\Vert \int \frac{Q_j(z, \widehat{F}_J,\Psi_j)}{\max \left(f_{\widehat{F}_J, \Psi_j}, \frac{\rho_J}{\sqrt{\det(\Psi_j)}} \right)} f_{\widehat{F}_J, \Psi_j}(z) dz \right\Vert_2. \label{eq:sm6.34}
    \end{align}
    Following section SM6.4.1 of \citeappendix{chen2022empirical}, from Lemma \ref{lem:sm6.10}
    \begin{align*}
        [\eqref{eq:sm6.33}]^2 &\lesssim h^2\left(f_{F_0, \Psi_j}, f_{\widehat{F}_J, \Psi_j} \right) \int \frac{ \left\Vert Q_j(z, \widehat{F}_J,\Psi_j) \right\Vert_2^2 }{\left(\max \left(f_{\widehat{F}_J, \Psi_j}, \frac{\rho_J}{\sqrt{\det(\Psi_j)}} \right) \right)^2} \left( f_{F_0,\Psi_j}(z) + f_{\widehat{F}_J, \Psi_j}(z) \right) dz \\
        &\lesssim_{\mathcal H} M_J^2 h^2\left(f_{F_0, \Psi_j}, f_{\widehat{F}_J, \Psi_j} \right) \log J \\
        \Rightarrow \eqref{eq:sm6.33} &\lesssim_{\mathcal H} M_J h\left(f_{F_0, \Psi_j}, f_{\widehat{F}_J, \Psi_j} \right) \sqrt{\log J}.
    \end{align*}
    
    As in section SM6.4.2 of \citeappendix{chen2022empirical}, by Cauchy-Schwarz
    \begin{align*}
        \eqref{eq:sm6.34} &\leq \sqrt{E_{Z \sim f_{\widehat{F}_J, \Psi_j}} \left[ \left\Vert E_{\widehat{F}_J, \Psi_j} \left[(Z-\tau)\tau^T \vert Z \right] \right\Vert_F^2 \right]} \sqrt{Pr_{f_{\widehat{F}_J, \Psi_j}} \left(f_{\widehat{F}_J, \Psi_j}(Z) \leq \frac{\rho_J}{\sqrt{\det(\Psi_j)}} \right) } \\
        &\lesssim_{\mathcal H} M_J \rho_J^{1/4} M_J^{1/2} \lesssim_{\mathcal H} M_J^{3/2} J^{-1}.
    \end{align*}
    
    \textbf{Bounding $U_{2\alpha}, U_{2\Omega}$}:

    I will follow the proof of Lemma SM6.3 in \citeappendix{chen2022empirical} to show that for $k \in \{\alpha,\Omega\}$,
    \begin{align*}
        Pr \left(\Vert \widehat\chi - \chi_0 \Vert_J \leq \Delta_J, \overline{Z}_J \leq M_J, |U_{2k}| \gtrsim_{\mathcal H} r_J \right) \leq \frac{2}{J}
    \end{align*}
    for $r_J = \Delta_J e^{-C_{\mathcal H}M_J^2}\log J + \frac{M_J^2(\log J)^{3/2}}{\sqrt{J}} \Delta_J$.

    I will choose some $\overline{U}_{2k}$ such that if $\Vert \widehat\chi - \chi_0 \Vert_J \leq \Delta_J$ and $\overline{Z}_J \leq M_J$ then $|U_{2k}| \leq \overline{U}_{2k}$. Thus a bound on $Pr(\overline{U}_{2k} > t)$ suffices.

    Define
    \begin{align*}
        D_{k,j,M_J}(Z_j, \widehat{F}_J, \chi_0, \rho_J) &= D_{k,j}(Z_j, \widehat{F}_J, \chi_0, \rho_J) \mathbbm{1}(\Vert Z_j \Vert_2 \leq M_J) \\
        \overline{D}_{k,j,M_J}(\widehat{F}_J, \chi_0, \rho_J) &= \int D_{k,j,M_J}(z, \widehat{F}_J, \chi_0, \rho_J) f_{F_0, \Psi_j}(z) dz.
    \end{align*}
    On $\overline{Z}_J \leq M_J$ note that
    \begin{align}
        |U_{2k}| &\leq \left\vert \frac{1}{J} \sum_{j=1}^J \left\{ D_{k,j,M_J}(Z_j, \widehat{F}_J, \chi_0, \rho_J) - \overline{D}_{k,j,M_J}(\widehat{F}_J, \chi_0, \rho_J) \right\}^T \Delta_{k,j} \right\vert \label{eq:sm6.36} \\
        &\qquad + \left\vert \frac{1}{J} \sum_{j=1}^J \left\{ \overline{D}_{k,j}(\widehat{F}_J, \chi_0, \rho_J) - \overline{D}_{k,j,M_J}(\widehat{F}_J, \chi_0, \rho_J)\right\}^T \Delta_{k,j} \right\vert. \label{eq:sm6.37}
    \end{align}
    By Lemmas \ref{lem:sm6.10} and \ref{lem:oa3.1}, uniformly over all $F$,
    \begin{align}
        \left\Vert D_{k,j}(z, F, \chi_0, \rho_J) \right\Vert_2 &\lesssim_{\mathcal H} \Vert z \Vert_2 \sqrt{\log J} + \log J. \label{eq:sm6.38}
    \end{align}
    Thus
    \begin{align*}
        \eqref{eq:sm6.37} &\lesssim_{\mathcal H} \Delta_J \left( \sqrt{\log J} \max_{j \in [J]} \int_{\Vert z \Vert_2 > M_J} \Vert z \Vert_2 f_{F_0, \Psi_j}(z) dz + \log J \max_{j \in [J]} Pr_{F_0, \Psi_j}(\Vert Z_j \Vert_2 > M_J) \right).
    \end{align*}
    By Cauchy-Schwarz,
    \begin{align*}
        \int_{\Vert z \Vert_2 > M_J} \Vert z \Vert_2 f_{F_0, \Psi_j}(z) dz &\leq \sqrt{E[\Vert Z_j \Vert_2^2]Pr(\Vert Z_j \Vert_2 > M_J)} \lesssim_{\mathcal H} \sqrt{Pr(\Vert Z_j \Vert_2 > M_J)}.
    \end{align*}
    Because each $Z_j$ is such that $Z_j \vert \tau_j \sim N(\tau_j, \Psi_j)$ and $\tau_j \sim F_0$ which is mean zero, each $Z_j$ is sub-Gaussian so that $Pr_{F_0, \Psi_j}(\Vert Z_j \Vert_2 > M_J) \leq \exp(-C_{\mathcal H} M_J^2)$ for some constant $C_{\mathcal H}$.
    Thus $\eqref{eq:sm6.37} \lesssim_{\mathcal H} \Delta_J e^{-C_{\mathcal H} M_J^2} \log J$.

    To bound \eqref{eq:sm6.36}, let $F_1, \dots, F_N$ be a minimal $\omega$-covering of distributions on $\R^2$, $\mathcal P(\R^2)$, under the pseudometric
    \begin{align} \label{eq:d_k_infty_M}
        d_{k,\infty,M_J}(F_1, F_2) &= \max_{j \in [J]} \sup_{\Vert z \Vert_2 \leq M_J} \left\Vert D_{k,j}(z, F_1, \chi_0, \rho_J) - D_{k,j}(z, F_2, \chi_0, \rho_J) \right\Vert_2,
    \end{align}
    taking $N = N(\omega, \mathcal P(\R^2), d_{k,\infty,M_J})$. Project $\widehat{F}_J$ to the $\omega$-covering to obtain
    \begin{align*}
        \eqref{eq:sm6.36} &\leq 2\omega \Delta_J + \max_{i \in [N]} \left\vert \frac{1}{J} \sum_{j=1}^J \left\{ D_{k,j,M_J}(Z_j, F_i, \chi_0, \rho_J) - \overline{D}_{k,j,M_J}(F_i, \chi_0, \rho_J) \right\}^T \Delta_{k,j} \right\vert.
    \end{align*}
    Define
    \begin{align*}
        v_{i,j}(\chi) &\equiv \left\{ D_{k,j,M_J}(Z_j, F_i, \chi_0, \rho_J) - \overline{D}_{k,j,M_J}(F_i, \chi_0, \rho_J) \right\}^T \Delta_{k,j}(\chi), \qquad V_{J,i}(\chi) \equiv \frac{1}{J} \sum_{j=1}^J v_{i,j}(\chi)
    \end{align*}
    for $\Delta_{\alpha,j}(\tilde\chi) = a_j(\tilde\alpha) - a_j(\alpha_0)$ and $\Delta_{\Omega,j}(\tilde\chi) = \text{vec}(B_j(\tilde\Omega)^{1/2}) - \text{vec}(B_j(\Omega_{0})^{1/2})$. 
    Then it follows that $\eqref{eq:sm6.36} \lesssim \omega \Delta_J + \max_{i \in [N]} \sup_{\chi \in S} \vert V_{J,i}(\chi) \vert$ where 
    \begin{align}
        S &= \left\{ \chi = (\alpha,\Omega): \Vert \chi - \chi_0 \Vert_J \leq \Delta_J, \left( a(\cdot;\alpha), B(\cdot,\Omega)^{1/2} \right) \in \mathcal C \right\} \label{eq:sm6.40}
    \end{align}
    for $\mathcal C$ as defined in Assumption \ref{ass:est_var}.
    Thus for some $\omega$ to be chosen, let
    \begin{align*}
        \overline{U}_{2k} = C_{\mathcal H} \left\{ \Delta_J (\log J) e^{-C_{\mathcal H} M_J^2} + \omega \Delta_J + \max_{i \in [N]} \sup_{\chi \in S} \vert V_{J,i}(\chi) \vert \right\}.
    \end{align*}
    To bound $Pr(\overline{U}_{2k} > t)$ I first look at the empirical process $\max_{i \in [N]} \sup_{\chi \in S} \vert V_{J,i}(\chi) \vert$. 
    For fixed $\chi, \chi_1, \chi_2 \in S$, because the $v_{i,j}(\chi)$ are independent across $j$ and bounded, and because
    \begin{align*}
        \left\Vert D_{k,j,M_J}(Z_j, F_i,\chi_0, \rho_J) \right\Vert_2 \lesssim_{\mathcal H} M_J \sqrt{\log J}
    \end{align*}
    as shown in \eqref{eq:sm6.38}, it follows that 
    \begin{align*}
        \left\Vert V_{J,i}(\chi) \right\Vert_{\psi_2} &\lesssim_{\mathcal H} \frac{M_J \sqrt{\log J}}{\sqrt{J}} \Delta_J, \\
        \left\Vert V_{J,i}(\chi_1) - V_{J,i}(\chi_2) \right\Vert_{\psi_2} &\lesssim_{\mathcal H} \frac{M_J \sqrt{\log J}}{\sqrt{J}} \left\Vert \chi_1 - \chi_2 \right\Vert_J.
    \end{align*}
    Thus $\chi \mapsto V_{J,i(\chi)}$ is a process with sub-Gaussian increments under $\Vert \cdot \Vert_J$. By the same chaining argument as in the proof of Lemma SM6.3 of \citeappendix{chen2022empirical}, for every $u > 0$,
    \begin{align*}
        \sup_{\chi \in S} \vert V_{J,i}(\chi) \vert \lesssim_{\mathcal H} \frac{M_J \sqrt{\log J}}{\sqrt{J}} \left[(1+u) \Delta_J + \int_0^\infty \sqrt{\log N(\varepsilon, S, \Vert \cdot \Vert_J)} d \varepsilon \right]
    \end{align*}
    with probability at least $1-2e^{-u^2}$.

    I know $N(\varepsilon,S,\Vert \cdot \Vert_J) \leq N(\varepsilon, \mathcal C, \Vert \cdot \Vert_\infty)$ and $S$ has diameter at most $2 \Delta_J$ under $\Vert \cdot \Vert_J$. Thus
    \begin{align*}
        \int_0^\infty \sqrt{\log N(\varepsilon, S, \Vert \cdot \Vert_J)} d \varepsilon &\lesssim_{\mathcal H} \int_0^{2\Delta_J} \sqrt{\log \left(\frac{C_{\mathcal C}}{\varepsilon} \right)} d \varepsilon \lesssim_{\mathcal H} \Delta_J \sqrt{\log \left(\frac{C_{\mathcal C}}{\Delta_J} \right)}.
    \end{align*}
    So by union bound,
    \begin{align*}
        Pr \left(\max_{i \in [N]} \sup_{\chi \in S} \vert V_{J,i}(\chi) \vert \gtrsim_{\mathcal H} \frac{M_J \sqrt{\log J}}{\sqrt{J}}\left[ (1+u) \Delta_J + \Delta_J \sqrt{\log \left(\frac{C_{\mathcal C}}{\Delta_J} \right)} \right] \right) &\leq 2Ne^{-u^2},
    \end{align*}
    choosing $u = \sqrt{\log N} + \sqrt{\log J}$ so that the right hand side is bounded by $2/J$.
    Taking $\omega = M_J \frac{1+\sqrt{\log(1/\rho_J)}}{\rho_J}\frac{\rho_J}{\sqrt{J}} \geq \frac{1+\sqrt{\log(1/\rho_J)}}{\rho_J}\frac{\rho_J}{\sqrt{J}}$, by Lemma \ref{prop:sm6.2}, Lemma \ref{lem:oa3.1}, and Assumption \ref{ass:sm6.1},
    \begin{align*}
        \log N(\omega, \mathcal{P}(\R^2), d_{\alpha, \infty, M_J}) &\lesssim_{\mathcal H} (\log J)^2 M_J^2 \\
        \log N(\omega, \mathcal{P}(\R^2), d_{\Omega, \infty, M_J}) &\lesssim_{\mathcal H} (\log J)^2 M_J^2.
    \end{align*}
    Note that this means $\omega \lesssim_{\mathcal H} \frac{1}{\sqrt{J}} (\log J)^{1/2} M_J$
    and $(1+u) \lesssim_{\mathcal H} M_J \log J$. In order to do this I require $\frac{\rho_J}{\sqrt{J}} \leq \min(1/e, 4/(2\pi \underline{k}), e^{-1/(2\bar k)})$. By Lemma \ref{lem:oa3.1} this holds because $\frac{\rho_J}{\sqrt{J}} \leq \frac{1}{2\pi e \sqrt{J}}$ and I can trivially decrease $\underline{k}$ and increase $\bar k$ until the conditions hold.

    Then since $V_{J,i}(\chi)$ is the only random expression in $\overline{U}_{2k}$,
    \begin{align*}
        Pr \left(\overline{U}_{2k} \gtrsim_{\mathcal H} \Delta_J e^{-C_{\mathcal H}M_J^2}\log J + \frac{M_J^2(\log J)^{3/2}}{\sqrt{J}} \Delta_J \right) \leq \frac{2}{J},
    \end{align*}
    using that $\Delta_J \gtrsim_{\mathcal H} \frac{1}{\sqrt{J}}$ so $\sqrt{\log\left( \frac{C_{\mathcal C}}{\Delta_J} \right)} \lesssim_{\mathcal H} \sqrt{\log J}$.

    \textbf{Bounding $U_{3\alpha}, U_{3\Omega}$}:
    
    I will follow the proof of Lemma SM6.4 in \citeappendix{chen2022empirical} to show that for $k \in \{\alpha,\Omega\}$,
    \begin{align*}
        Pr \left(\Vert \widehat\chi - \chi_0 \Vert_J \leq \Delta_J, \overline{Z}_J \leq M_J, |U_{3k}| \gtrsim_{\mathcal H} \Delta_J \left\{ e^{-C_{\mathcal H}M_J^2} + \frac{M_J^2}{\sqrt{J}} \left(1 + \sqrt{\log J} \right) \right\} \right) \leq \frac{2}{J}.
    \end{align*}
    On the event $\overline{Z}_J \leq M_J$,
    \begin{align*}
        U_{3k} &= \underbrace{\frac{1}{J} \sum_{j=1}^J \left\{ D_{k,j,M_J}(Z_j, F_0, \chi_0, 0) - \overline{D}_{k,j,M_J}(F_0, \chi_0, 0) \right\}^T \Delta_{k,j}}_{V_J(\widehat\chi)} + \frac{1}{J} \sum_{j=1}^J \overline{D}_{k,j,M_J}(F_0, \chi_0, 0)^T \Delta_{k,j}.
    \end{align*}
    Because $E_{F_0,\Psi_j}[D_{k,j}(Z_j,F_0,\chi_0,0)] = 0$, it follows that \begin{align*}
        \overline{D}_{k,j,M_J}(F_0, \chi_0,0) = -E[D_{k,j}(Z_j,F_0,\chi_0,0) \mathbbm{1}(\Vert Z_j\Vert_2 > M_J)].
    \end{align*}
    Then by Cauchy-Schwarz,
    \begin{align*}
        \Vert \overline{D}_{k,j,M_J}(F_0, \chi_0, 0) \Vert_2 &\leq Pr(\Vert Z_j \Vert_2 > M_J)^{1/2} \left( E[\Vert D_{k,j}(Z_j,F_0,\chi_0,0)\Vert_2^2] \right)^{1/2},
    \end{align*}
    and since 
    \begin{align*}
        \Vert D_{k,j}(Z_j,F_0,\chi_0,0)\Vert_2 &\lesssim_{\mathcal H} T_{k,j}(Z_j)
    \end{align*}
    where $T_{\alpha,j}(z) = \frac{\Vert \nabla f_{F_0,\Psi_j}(z) \Vert_2}{ f_{F_0,\Psi_j}(z)}$ and $T_{\Omega,j}(z) = \frac{\Vert Q_j(z,F_0,\Psi_j) \Vert_2}{ f_{F_0,\Psi_j}(z)}$, it follows that
    \begin{align*}
        \Vert \overline{D}_{k,j,M_J}(F_0, \chi_0, 0) \Vert_2 &\lesssim_{\mathcal H} Pr(\Vert Z_j \Vert_2 > M_J)^{1/2} \left( E[T_{k,j}^2(Z_j)] \right)^{1/2}.
    \end{align*}
    Because both $T_{k,j}$ are of the form $\Vert E[ f(\tau,Z)| Z]\Vert_2$, by Jensen's inequality $E[T_k^2] \leq E[\Vert f(\tau,Z_j) \Vert_2^2] \lesssim_{\mathcal H} 1$. Then because $Z_j$ is sub-Gaussian,
    \begin{align*}
        \Vert \overline{D}_{k,j,M_J}(F_0, \chi_0, 0) \Vert_2 &\lesssim_{\mathcal H} e^{-C_{\mathcal H}M_J^2}.
    \end{align*}

    Because of the truncation to $\Vert z \Vert_2 \leq M_J$,
    \begin{align*}
        \Vert D_{k,j,M_J}(Z_j, F_0,\chi_0,0) - \overline{D}_{k,j,M_J}(F_0, \chi_0, 0) \Vert_2 \lesssim_{\mathcal H} M_J^2
    \end{align*}
    so for fixed $\chi, \chi_1, \chi_2 \in S$ as defined in \eqref{eq:sm6.40},
    \begin{align*}
        \left\Vert V_J(\chi_1) - V_J(\chi_2) \right\Vert_{\psi_2} &\lesssim_{\mathcal H} \frac{M_J^2}{\sqrt{J}} \Vert \chi_1 - \chi_2 \Vert_J, \\
        \Vert V_J(\chi) \Vert_{\psi_2} &\lesssim_{\mathcal H} \frac{\Delta_J M_J^2}{\sqrt{J}}.
    \end{align*}
    
    Then by the same chaining argument as for bounding $U_{2k}$ but taking $u = \sqrt{\log J}$, 
    \begin{align*}
        \sup_{\chi \in S} \vert V_J(\chi) \vert &\lesssim_{\mathcal H} \frac{M_J^2}{\sqrt{J}} \left[ \left( 1+\sqrt{\log J}\right) \Delta_J + \Delta_J \sqrt{\log \left(\frac{C_{\mathcal C}}{\Delta_J} \right)} \right] \\
        &\lesssim_{\mathcal H} \frac{M_J^2}{\sqrt{J}} \left( 1+\sqrt{\log J}\right) \Delta_J
    \end{align*}
    with probability at least $1-2/J$, because $\Delta_J \gtrsim_{\mathcal H} \frac{1}{\sqrt{J}}$. Thus letting
    \begin{align*}
        \overline{U}_{3k} = C_{\mathcal H} \left( \sup_{\chi \in S} \vert V_J(\chi) \vert + \Delta_J e^{-C_{\mathcal H}M_J^2} \right),
    \end{align*}
    the tail bound for $\overline{U}_{3k}$ gives the result.

    \textbf{Bounding $R_1, R_2$}:
    
    I follow the proofs of Lemmas SM6.5 and SM6.6 in \citeappendix{chen2022empirical}.
    
    To bound $R_1$, it can be shown that each $R_{1j}$ can be upper bounded by a constant times
    \begin{align*}
        \max\bigg( &\Vert \Delta_{\alpha,j}\Vert_\infty^2 \left\Vert \frac{\partial^2 \psi_j}{\partial a_j \partial a_j^T} \bigg\vert_{\widehat{F}_J, \tilde \alpha, \tilde\Omega} \right\Vert_F, 
        \Vert \Delta_{\alpha,j}\Vert_\infty \Vert B_j(\widehat{\Omega})^{1/2} - B_j(\Omega_0)^{1/2} \Vert_{op} \left\Vert \frac{\partial^2 \psi_j}{\partial a_j \partial \text{vec}(B_j^{1/2})^T} \bigg\vert_{\widehat{F}_J, \tilde \alpha, \tilde\Omega} \right\Vert_F, \\ 
        &\qquad \Vert B_j(\widehat{\Omega})^{1/2} - B_j(\Omega_0)^{1/2} \Vert_{op}^2 \left\Vert \frac{\partial^2 \psi_j}{\partial \text{vec}(B_j^{1/2}) \partial \text{vec}(B_j^{1/2})^T} \bigg\vert_{\widehat{F}_J, \tilde \alpha, \tilde\Omega} \right\Vert_F \bigg) 
    \end{align*}
    
    By assumption and Lemma \ref{lem:sm6.11}, it follows that $|R_{1j}| \lesssim_{\mathcal H} \Delta_J^2M_J^2 \log J \Rightarrow |R_{1}| \lesssim_{\mathcal H} \Delta_J^2M_J^2 \log J$.

    To bound $R_2$, I will show that
    \begin{align*}
        Pr \left(\Vert \widehat\chi - \chi_0 \Vert_J \leq \Delta_J, \overline{Z}_J \leq M_J, |R_2| \gtrsim_{\mathcal H} \Delta_J^2 \right) \leq \frac{1}{J}.
    \end{align*}
    By the same logic as above, $\mathbbm{1}(A_J) |R_2| \lesssim_{\mathcal H} \Delta_J^2 \frac{1}{J} \sum_{j=1}^J \mathbbm{1}(A_J) D_j$, where
    \begin{align*}
        D_j &\equiv \max\bigg( \left\Vert \frac{\partial^2 \psi_j}{\partial a_j \partial a_j^T} \bigg\vert_{F_0, \tilde \alpha, \tilde\Omega} \right\Vert_F, 
        \left\Vert \frac{\partial^2 \psi_j}{\partial a_j \partial \text{vec}(B_j^{1/2})^T} \bigg\vert_{F_0, \tilde \alpha, \tilde\Omega} \right\Vert_F, \left\Vert \frac{\partial^2 \psi_j}{\partial \text{vec}(B_j^{1/2}) \partial \text{vec}(B_j^{1/2})^T} \bigg\vert_{F_0, \tilde \alpha, \tilde\Omega} \right\Vert_F \bigg) 
    \end{align*}
    By the derivative calculations in Section \ref{app:sec:sm6.1}, these derivatives are functions of posterior moments under $F_0$, evaluated at transformed observation $\widehat{Z}_j(\tilde \alpha, \tilde\Omega)$. Then Lemma \ref{lem:sm6.14} implies that on $A_J$, 
    \begin{align*}
        \mathbbm{1}(A_J) D_j \lesssim_{\mathcal H} \mathbbm{1}(A_J) \max(\Vert \widehat{Z}_j(\tilde \alpha, \tilde\Omega) \Vert_2, 1)^4 \lesssim_{\mathcal H} \mathbbm{1}(A_J) \max(\Vert Z_j \Vert_2, 1)^4.
    \end{align*}
    By Chebyshev's inequality, there exists some $C_{\mathcal H}$ such that
    \begin{align*}
        Pr \left(\frac{1}{J} \sum_{j=1}^J \max(\Vert Z_j \Vert_2, 1)^4 \geq C_{\mathcal H} \right) &\leq \frac{1}{J}
    \end{align*}
    because $Var(\frac{1}{J} \sum_{j=1}^J \max(\Vert Z_j \Vert_2, 1)^4) \lesssim_{\mathcal H} \frac{1}{J}$. Thus $Pr \left(A_J, |R_2| \gtrsim_{\mathcal H} \Delta_J^2 \right) \leq \frac{1}{J}$.

    \vspace{10pt}

    To conclude the proof of the theorem, I apply a union bound (as in Lemma SM6.13 in \citeappendix{chen2022empirical}) to the above rates to obtain the result, following Appendix SM6 of \citeappendix{chen2022empirical}. In the rate $\epsilon_J$, the first term comes from $U_{1\Omega}$, the second and fourth terms from $U_{2k}$, and the third term from $R_1$. The other rates derived are dominated. The leading terms in $\epsilon_J$ dominate $\kappa_J$.
\end{proof}

\subsubsection{Proof of Corollary \ref{cor:oa3.1}}

I first state a result, which is a multivariate analogue of Theorem SM7.1 in \citeappendix{chen2022empirical}, that will be used in the proof of the corollary.

\begin{theorem}\label{thm:sm7.1}
    Suppose $J \geq 7$.
    Let $\tau_j \vert \Psi_{1:J} \sim F_0$, where $F_0$ satisfies Assumption \ref{ass:compact}.
    Fix positive sequences $\gamma_J, \lambda_J \to 0$ with $\gamma_J, \lambda_J \leq 1$, constants $\epsilon, C^* > 0$. Consider the set of distributions that approximately maximize the likelihood
    \begin{align*}
        A(\gamma_J, \lambda_J) &= \{H \in \mathcal{P}(\R^2) : \text{Sub}_J(H) \leq C^*(\gamma_J^2 + \bar{h}(f_{H,\cdot}, f_{F_0,\cdot})\lambda_J)\}
    \end{align*}
    and consider the set of distributions that are far from $F_0$ in $\bar h$
    \begin{align*}
        B(t,\lambda_J,\epsilon) &= \{H \in \mathcal{P}(\R^2): \bar{h}(f_{H,\cdot}, f_{F_0, \cdot}) \geq t B \lambda_J^{1-\epsilon}\}
    \end{align*}
    for some constant $B$ to be chosen. Assume that for some $C_\lambda$,
    \begin{align*}
        \lambda_J^2 &\geq \gamma_J^2 \geq \frac{C_\lambda}{J}(\log J)^{3}.
    \end{align*}
    Then the probability that $A \cap B$ is nonempty is bounded for $t > 1$, that is, there exists a choice of $B$ that depends on $\mathcal H, C^*,$ and $C_\lambda$ such that
    \begin{align*}
        Pr \left( A(\gamma_J,\lambda_J) \cap B(t,\lambda_J,\epsilon) \neq \emptyset \right) &\leq (\log_2(1/\epsilon)+1)J^{-t^2}.
    \end{align*}
\end{theorem}
\begin{proof}
    The proof closely follows the proof of Theorem SM7.1 in \citeappendix{chen2022empirical}.
    Decompose $B(t,\lambda_J,\epsilon) \subseteq \cup_{k=1}^K B_k(t,\lambda_J)$ where for some $B > 1$ to be chosen and $K = \lceil |\log_2(1/\epsilon)| \rceil$,
    \begin{align*}
        B_k &= \left\{H: \bar{h}(f_{H,\cdot}, f_{F_0, \cdot}) \in \left( tB\lambda_J^{1-2^{-k}}, tB\lambda_J^{1-2^{-k+1}} \right] \right\}.
    \end{align*}
    If $Pr(A(\gamma_J, \lambda_J) \cap B_k(t,\lambda_J) \neq \emptyset) \leq J^{-t^2}$ the result follows from a union bound.

    Let $\mu_{J,k} = B \lambda_J^{1-2^{-k+1}}$, so that $B_k = \left\{H: \bar{h}(f_{H,\cdot}, f_{F_0, \cdot}) \in \left( t\mu_{J,k+1}, t\mu_{J,k} \right] \right\}$. Fix a $k \in [K]$. 
    
    For $\omega = \frac{1}{J^2}$ consider an $\omega$-net for $\mathcal{P}(\R^2)$ under $\Vert \cdot \Vert_{\infty,M}$ (recall the definition of $\Vert \cdot \Vert_{\infty,M}$ from \eqref{eq:inf_M_norm}). Letting $N = N(\omega, \mathbb{F}, \Vert \cdot \Vert_{\infty,M})$ for $\mathbb{F}$ the space of $f_{F,\cdot}$ induced by $F \in \mathcal{P}(\R^2)$, let $H_1, \dots, H_N$ denote the distributions making up the $\omega$-net. And for each $i \in [N]$ let $H_{k,i}$ be a distribution, if it exists, with $\Vert f_{H_{k,i}} - f_{H_i} \Vert_{\infty,M} \leq \omega$ and $\bar{h}(f_{H_{k,i}, \cdot}, f_{F_0, \cdot}) \geq t \mu_{J,k+1}$. Finally let $I_k$ collect the indices $i$ for which $H_{k,i}$ exists.

    For any fixed distribution $H \in B_k(t,\lambda_J)$ there exists some $H_i$ in the covering such that $\Vert f_{H} - f_{H_i} \Vert_{\infty,M} \leq \omega$. Furthermore $H$ serves as witness that $H_{k,i}$ exists with $\Vert f_H - f_{H_{k,i}} \Vert_{\infty,M} \leq 2\omega$.

    Note that an upper bound for $f_{H,\Psi_j}(z)$ is given by
    \begin{align*}
        f_{H,\Psi_j}(z) &\leq \begin{cases} f_{H_{k,i}, \Psi_j}(z) + 2\omega & \Vert z \Vert_2 \leq M \\ \frac{1}{\sqrt{\det(2\pi \Psi_j)}} & \Vert z \Vert_2 > M. \end{cases}
    \end{align*}
    Defining $v(z) = \omega \mathbbm{1}(\Vert z \Vert_2 \leq M) + \omega \left(\frac{M}{\Vert z \Vert_2} \right)^{3} \mathbbm{1}(\Vert z \Vert_2 > M)$, 
    \begin{align*}
        f_{H,\Psi_j}(z) &\leq \begin{cases} f_{H_{k,i}, \Psi_j}(z) + 2v(z) & \Vert z \Vert_2 \leq M \\ \frac{f_{H_{k,i},\Psi_j}(z) + 2v(z)}{\sqrt{\det(2\pi \Psi_j)} 2v(z)} & \Vert z \Vert_2 > M. \end{cases}
    \end{align*}
    This means the likelihood ratio between $F_0$ and $H$ is upper bounded:
    \begin{align*}
        \prod_{j=1}^J \frac{f_{H,\Psi_j}(Z_j)}{f_{F_0,\Psi_j}(Z_j)} &\leq \left( \max_{i \in I_k} \prod_{j=1}^J \frac{f_{H_{k,i},\Psi_j}(Z_j) + 2v(Z_j)}{f_{F_0,\Psi_j}(Z_j)} \right) \prod_{j: \Vert Z_j \Vert_2 > M} \frac{1}{\sqrt{\det(2\pi \Psi_j)} 2v(Z_j)}.
    \end{align*}
    If $H \in A(\gamma_J,\lambda_J)$ the likelihood ratio is also lower bounded as in the proof of Theorem SM7.1 in \citeappendix{chen2022empirical}:
    \begin{align*}
        \prod_{j=1}^J \frac{f_{H,\Psi_j}(Z_j)}{f_{F_0,\Psi_j}(Z_j)} &\geq \exp \left(-J C^*(t^2 \lambda_J^2 + t^2 \mu_{J,k}\lambda_J) \right),
    \end{align*}
    so it follows that, choosing some $a > 1$,
    \begin{small}
    \begin{align}
        &Pr \left(A(\gamma_J,\lambda_J) \cap B_k(t,\lambda_J) \neq \emptyset\right) \nonumber \\
        &\leq Pr \left( \left( \max_{i \in I_k} \prod_{j=1}^J \frac{f_{H_{k,i},\Psi_j}(Z_j) + 2v(Z_j)}{f_{F_0,\Psi_j}(Z_j)} \right) \prod_{j: \Vert Z_j \Vert_2 > M} \frac{1}{\sqrt{\det(2\pi \Psi_j)} 2v(Z_j)} \geq \exp \left(-J C^*(t^2 \lambda_J^2 + t^2 \mu_{J,k}\lambda_J) \right) \right) \nonumber \\
        &\leq Pr\left( \max_{i \in I_k} \prod_{j=1}^J \frac{f_{H_{k,i},\Psi_j}(Z_j) + 2v(Z_j)}{f_{F_0,\Psi_j}(Z_j)} \geq e^{-Jt^2aC^*(\gamma_J^2+\mu_{J,k}\lambda_J)} \right) \label{eq:sm7.4} \\
        &\qquad + Pr \left( \prod_{j: \Vert Z_j \Vert_2 > M} \frac{1}{\sqrt{\det(2\pi \Psi_j)} 2v(Z_j)} \geq e^{Jt^2(a-1)C^*(\gamma_J^2+\mu_{J,k}\lambda_J)} \right) \label{eq:sm7.5}.
    \end{align}
    \end{small}

    By union bound, Markov's inequality, and independence over $j$,
    \begin{align*}
        \eqref{eq:sm7.4} &\leq \sum_{i \in I_k} e^{Jt^2aC^*(\gamma_J^2+\mu_{J,k}\lambda_J)/2} \prod_{j=1}^J E \left[\sqrt{\frac{f_{H_{k,i},\Psi_j}(Z_j) + 2v(Z_j)}{f_{F_0,\Psi_j}(Z_j)}} \right].
    \end{align*}
    Note
    \begin{align*}
        E \left[\sqrt{\frac{f_{H_{k,i},\Psi_j}(Z_j) + 2v(Z_j)}{f_{F_0,\Psi_j}(Z_j)}} \right] &= \int \sqrt{f_{H_{k,i},\Psi_j}(z) + 2v(z)} \sqrt{f_{F_0,\Psi_j}(z)} dz \\
        &\leq 1 - h^2(f_{H_{k,i},\Psi_j}, f_{F_0,\Psi_j}) + \int \sqrt{2v(z)f_{F_0,\Psi_j}(z)}dz \\
        &\leq 1 - h^2(f_{H_{k,i},\Psi_j}, f_{F_0,\Psi_j}) + \sqrt {2\int v(z)dz} \qquad \text{Jensen's} \\
        &= 1 - h^2(f_{H_{k,i},\Psi_j}, f_{F_0,\Psi_j}) + M\sqrt{6\pi \omega} \\
        \Rightarrow \prod_{j=1}^J E \left[\sqrt{\frac{f_{H_{k,i},\Psi_j}(Z_j) + 2v(Z_j)}{f_{F_0,\Psi_j}(Z_j)}} \right] &\leq \exp\left(-J \bar{h}^2(f_{H_{k,i},\cdot}, f_{F_0, \cdot}) + J M \sqrt{6\pi \omega} \right)
    \end{align*}
    using $\prod_{j=1}^J t_j \leq \exp(\sum_{j=1}^J (t_j-1))$ for $t_j > 0$.
    Then
    \begin{align*}
        \eqref{eq:sm7.4} &\leq \sum_{i \in I_k} \exp\left(\frac{Jt^2aC^*}{2}(\gamma_J^2+\mu_{J,k}\lambda_J) -J \bar{h}^2(f_{H_{k,i},\cdot}, f_{F_0, \cdot}) + JM \sqrt{6\pi \omega} \right) \\
        &\leq \exp\left(\frac{Jt^2aC^*}{2}(\gamma_J^2+\mu_{J,k}\lambda_J) -J t^2 \mu_{J,k+1}^2 + M\sqrt{6\pi} + C(\log J)^3 \max \left(1, \frac{M}{\sqrt{\log J}}, \frac{M^2}{\log J} \right) \right) 
    \end{align*}
    because $\bar{h}(f_{H_{k,i},\cdot}, f_{F_0,\cdot}) \geq t \mu_{J,k+1}, |I_k| \leq N$, $\omega = \frac{1}{J^2}$, and $\log N \lesssim_{\mathcal H} (\log(1/\omega))^3 \max\left(1, \frac{M}{\sqrt{\log(1/\omega)}}, \frac{M^2}{\log(1/\omega)} \right)$ by Suppl. Lemma 5 of \citeappendix{soloff2024multivariate} and Suppl. Lemma F.6 of \citeappendix{saha2020nonparametric}.

    By Markov's inequality, taking $x \mapsto x^{1/(2\log J)}$,
    \begin{align*}
        \eqref{eq:sm7.5} &\leq E \left[ \prod_{j=1}^J \left( \frac{1}{(\det(2\pi \Psi_j))^{1/6}} \frac{\Vert Z_j \Vert_2}{(2\omega)^{1/3}M} \right)^{\frac{3}{2\log J} \mathbbm{1}(\Vert Z_j \Vert_2 > M)} \right] \exp \left( -\frac{J(a-1)t^2 C^*(\gamma_J^2 + \mu_{J,k} \lambda_J)}{2\log J} \right).
    \end{align*}
    Define
    \begin{align*}
        a_j = \frac{1}{(\det(2\pi \Psi_j))^{1/6}(2\omega)^{1/3}M} \leq \frac{C_{\underline{k}}J^{2/3}}{M}, \qquad \lambda = \frac{3}{2\log J}
    \end{align*}
    for some constant $C_{\underline{k}}$ depending only on $\underline{k},$
    then as in the proof of Theorem 7 in \citeappendix{soloff2024multivariate}, using Suppl. Lemma 2 in \citeappendix{soloff2024multivariate},
    \begin{align*}
        E \left[ \prod_{j=1}^J \left( \frac{1}{(\det(2\pi \Psi_j))^{1/6}} \frac{\Vert Z_j \Vert_2}{(2\omega)^{1/3}M} \right)^{\frac{3}{2\log J} \mathbbm{1}(\Vert Z_j \Vert_2 > M)} \right] &= E \left[ \left\{\prod_j (a_j \Vert Z_j \Vert_2)^{\mathbbm{1}(\Vert Z_j \Vert_2 > M) } \right\}^{\lambda} \right] \\
        &\leq \exp \left(\sum_{j=1}^J a_j^\lambda E \left[ \Vert Z_j \Vert_2^\lambda \mathbbm{1}(\Vert Z_j \Vert_2 > M)\right] \right) \\
        &\leq \exp \left(\sum_{j=1}^J a_j^\lambda M^{\lambda} \left(C e^{-M^2/(8\bar k)} + \left( \frac{2\mu_q}{M} \right)^q \right)\right) \\
        &\leq \exp \left(C_{\underline{k}}e \left(C + \frac{J2^q\mu_q^q}{M^q} \right) \right)
    \end{align*}
    for some constant $C$, taking $M \geq \sqrt{8 \bar{k} \log J}$ and $\frac{3}{2\min(1,q)} \leq \log J$ so that $\lambda \in (0, \min(1,q))$ and defining $\mu_q \equiv (E[\Vert \tau \Vert_2^q])^{1/q}$ for $\tau \sim F_0$.
    Thus
    \begin{align*}
        \eqref{eq:sm7.5} &\leq \exp \left(C_{\underline{k}}C e + C_{\underline{k}} e\frac{J2^q\mu_q^q}{M^q} -\frac{J(a-1)t^2 C^*(\gamma_J^2 + \mu_{J,k} \lambda_J)}{2\log J} \right).
    \end{align*}
    Note that under Assumption \ref{ass:compact}, $\mu_q \leq S_0 \sqrt{q}$ for all $q$, where $S_0 \geq 1$ without loss of generality. 
    So taking $a=2, q = c_m \log J, M = 2eS_0 \sqrt{c_m \log J}$ where constant $c_m \geq 2$ is chosen sufficiently large that $M \geq \sqrt{8 \overline{k} \log J}$, then $\left( \frac{2\mu_q}{M} \right)^q \leq e^{-q} = J^{-c_m}$, so $J(2\mu_q/M)^q \leq 1/J$. Using $\lambda_J^2 \geq \gamma_J^2 \geq \frac{C_\lambda}{J}(\log J)^{3}$,
    \begin{align*}
        \eqref{eq:sm7.5} &\leq \exp \left(C_{\underline{k}}C e + \frac{C_{\underline{k}} e}{J} - t^2 \frac{C^*C_\lambda(1 + B)}{2}(\log J)^2 \right) \\
        \eqref{eq:sm7.4} &\leq \exp\left(-t^2 (\log J)^3 \left(C_\lambda\left(-C^* - C^*B + B^2 \right) - (C+5)4e^2S_0^2c_m \right) \right) .
    \end{align*}
    There exists large enough $B$ such that $\eqref{eq:sm7.4} \leq 0.5\exp(-t^2 \log J)$ and $\eqref{eq:sm7.5} \leq 0.5\exp(-t^2 \log J)$, so $\eqref{eq:sm7.4} + \eqref{eq:sm7.5} \leq J^{-t^2}$, which concludes the proof.
\end{proof}

Finally, the proof of Corollary \ref{cor:oa3.1} follows as in Appendix SM7 of \citeappendix{chen2022empirical}, which uses Corollary \ref{cor:sm6.1} and Theorem \ref{thm:sm7.1}, but replacing the constants $\alpha, \beta,$ and $-p/(2p+1)$ with $2, \frac{1}{2},$ and $-\frac{1}{2}$ respectively to match the rates $\Delta_J, M_J, \delta_J$ chosen here.

\subsection{Auxiliary lemmas}

\begin{lemma}\label{lem:sm6.8}
    Fix a probability measure $F$ on $\R^2$ and any $z \in \R^2$. Then
\begin{align*}
    \left(\frac{\left\Vert \nabla f_{F, \Psi_j}(z) \right\Vert_2}{f_{F, \Psi_j}(z)} \right)^2 \leq \left\Vert \Psi_j^{-1} \right\Vert_{F} \log \left(\frac{1}{(2\pi)^2 \det(\Psi_j) f^2_{F, \Psi_j}(z)} \right)
\end{align*}
and 
\begin{align*}
    \left\Vert \Psi_j + \Psi_j \frac{\nabla^2 f_{F, \Psi_j}(z)}{f_{F, \Psi_j}(z)} \Psi_j \right\Vert_F \leq \left\Vert \Psi_j \right\Vert_{F} \log \left( \frac{1}{(2\pi)^2\det(\Psi_j) f_{F,  \Psi_j}^2(z)} \right).
\end{align*}
Also for every $z \in \R^2$ and all $\rho \in (0, 1/2\pi \sqrt{e})$,
\begin{align*}
    \frac{\left\Vert \nabla f_{F, \Psi_j}(z) \right\Vert_2}{\max \left( f_{F, \Psi_j}(z), \frac{\rho}{\sqrt{\det(\Psi_j)}} \right)} \leq \left\Vert \Psi_j^{-1} \right\Vert_{F}^{1/2} \varphi_+(\rho)
\end{align*}
while for every $z \in \R^2$ and all $\rho \in (0, 1/2\pi e)$,
\begin{align*}
    \left(\frac{\left\Vert \nabla f_{F, \Psi_j}(z) \right\Vert_2}{f_{F, \Psi_j}(z)} \right)^2 \frac{f_{F, \Psi_j}(z)}{\max \left( f_{F, \Psi_j}(z), \frac{\rho}{\sqrt{\det(\Psi_j)}} \right)} \leq \left\Vert \Psi_j^{-1} \right\Vert_{F} \varphi_+^2(\rho)
\end{align*}
and 
\begin{align*}
    \left\Vert \Psi_j + \Psi_j \frac{\nabla^2 f_{F, \Psi_j}(z)}{f_{F, \Psi_j}(z)} \Psi_j \right\Vert_F \frac{f_{F, \Psi_j}(z)}{\max \left( f_{F, \Psi_j}(z), \frac{\rho}{\sqrt{\det(\Psi_j)}} \right)} \leq \left\Vert \Psi_j \right\Vert_{F} \varphi_+^2(\rho).
\end{align*}
\end{lemma}

\begin{proof}
    This lemma extends Suppl. Lemma F.1 of \citeappendix{saha2020nonparametric} to a heteroscedastic setting, using the approach of \citeappendix{soloff2024multivariate}.

    As in section D.2 of \citeappendix{soloff2024multivariate}, for any fixed $j$ let $F_j$ denote the distribution of $\xi_j = \Psi_j^{-1/2} \tau_j$ where $\tau_j \sim F$. Then for $\check{z}_j = \Psi_j^{-1/2} z$ one can verify
    \begin{align*}
        f_{F, \Psi_j}(z) &= \frac{1}{\sqrt{\det(\Psi_j)}} f_{F_j, I_2}(\check{z}_j) \\
        \nabla f_{F, \Psi_j}(z) &= \frac{1}{\sqrt{\det(\Psi_j)}} \Psi_j^{-1/2} \nabla f_{F_j, I_2}(\check{z}_j) \\
        \nabla^2 f_{F, \Psi_j}(z) &= \frac{1}{\sqrt{\det(\Psi_j)}} \Psi_j^{-1/2} \nabla^2 f_{F_j, I_2}(\check{z}_j) \Psi_j^{-1/2}.
    \end{align*}
    Then using (F.1) in Suppl. Lemma F.1 of \citeappendix{saha2020nonparametric} 
    \begin{align*}
        \left(\frac{\left\Vert \nabla f_{F, \Psi_j}(z) \right\Vert_2}{f_{F, \Psi_j}(z)}\right)^2 &= \left(\frac{\left\Vert \Psi_j^{-1/2} \nabla f_{F_j, I_2}(\check{z}_j) \right\Vert_2}{f_{F_j, I_2}(\check{z}_j)}\right)^2 \leq \left\Vert \Psi_j^{-1} \right\Vert_{F} \log \left(\frac{1}{(2\pi)^2 \det(\Psi_j) f^2_{F, \Psi_j}(z)} \right).
    \end{align*}
    By inspection of the proof, I can replace the trace with a Frobenius norm in equation (F.1) in Suppl. Lemma F.1 of \citeappendix{saha2020nonparametric} to obtain 
    \begin{align*}
        \left\Vert \Psi_j + \Psi_j \frac{\nabla^2 f_{F, \Psi_j}(z)}{f_{F, \Psi_j}(z)} \Psi_j \right\Vert_F &= \left\Vert \Psi_j^{1/2} \left(I_2 + \frac{\nabla^2 f_{F_j, I_2}(\check{z}_j)}{f_{F_j, I_2}(\check{z}_j)} \right) \Psi_j^{1/2} \right\Vert_F \\
        &\leq \left\Vert \Psi_j \right\Vert_{F} \log \left( \frac{1}{(2\pi)^2\det(\Psi_j) f_{F,  \Psi_j}^2(z)} \right).
    \end{align*}
    Similarly, from (F.2) in Suppl. Lemma F.1 of \citeappendix{saha2020nonparametric}
    \begin{align*}
        \frac{\left\Vert \nabla f_{F, \Psi_j}(z) \right\Vert_2}{\max \left( f_{F, \Psi_j}(z), \frac{\rho}{\sqrt{\det(\Psi_j)}} \right)} &= \frac{\left\Vert \Psi_j^{-1/2} \nabla f_{F_j, I_2}(\check{z}_j) \right\Vert_2}{\max \left( f_{F_j, I_2}(\check{z}_j), \rho \right)}
        \leq \left\Vert \Psi_j^{-1} \right\Vert_{F}^{1/2} \varphi_+(\rho)
    \end{align*}
    and from (F.3) in Suppl. Lemma F.1 of \citeappendix{saha2020nonparametric}
    \begin{align*}
        \left(\frac{\left\Vert \nabla f_{F, \Psi_j}(z) \right\Vert_2}{f_{F, \Psi_j}(z)} \right)^2 \frac{f_{F, \Psi_j}(z)}{\max \left( f_{F, \Psi_j}(z), \frac{\rho}{\sqrt{\det(\Psi_j)}} \right)} &= \left(\frac{\left\Vert \Psi_j^{-1/2} \nabla f_{F_j, I_2}(\check{z}_j) \right\Vert_2}{f_{F_j, I_2}(\check{z}_j)}\right)^2 \frac{f_{F_j, I_2}(\check{z}_j)}{\max \left( f_{F_j, I_2}(\check{z}_j), \rho \right)} \\
        &\leq \left\Vert \Psi_j^{-1} \right\Vert_{F} \varphi_+^2(\rho).
    \end{align*}
    Finally I follow the proof of Lemma SM6.8 in \citeappendix{chen2022empirical} and look at cases:

    1) $f_{F, \Psi_j}(z) \leq \frac{\rho}{\sqrt{\det(\Psi_j)}}$. Then because $t \log(1/(2\pi t)^2)$ is increasing over $t \in (0, 1/2\pi e)$, using the result from above
    \begin{align*}
        &\left\Vert \Psi_j + \Psi_j \frac{\nabla^2 f_{F, \Psi_j}(z)}{f_{F, \Psi_j}(z)} \Psi_j \right\Vert_F \sqrt{\det(\Psi_j)} f_{F, \Psi_j}(z) \\
        &\leq \left\Vert \Psi_j \right\Vert_{F} \sqrt{\det(\Psi_j)} f_{F, \Psi_j}(z) \log \left( \frac{1}{(2\pi)^2\det(\Psi_j) f_{F, \Psi_j}^2(z)} \right) \\
        &\leq \left\Vert \Psi_j \right\Vert_{F} \rho \log \left( \frac{1}{(2\pi \rho)^2} \right) = \left\Vert \Psi_j \right\Vert_{F} \rho \varphi_+^2(\rho).
    \end{align*}
    The result follows from dividing by $\max \left( \sqrt{\det(\Psi_j)} f_{F, \Psi_j}(z), \rho \right) = \rho$.
    
    2) $f_{F, \Psi_j}(z) > \frac{\rho}{\sqrt{\det(\Psi_j)}}$. Then because $\log(1/(2\pi t)^2)$ is decreasing in $t$, using the result from above
    \begin{align*}
        \left\Vert \Psi_j + \Psi_j \frac{\nabla^2 f_{F, \Psi_j}(z)}{f_{F, \Psi_j}(z)} \Psi_j \right\Vert_F &\leq \left\Vert \Psi_j \right\Vert_{F} \log \left( \frac{1}{(2\pi)^2\det(\Psi_j) f_{F,  \Psi_j}^2(z)} \right) \\
        &\leq \left\Vert \Psi_j \right\Vert_{F} \log \left( \frac{1}{(2\pi\rho)^2} \right) = \left\Vert \Psi_j \right\Vert_{F} \varphi_+^2(\rho).
    \end{align*}
\end{proof}

\begin{lemma}\label{lem:sm6.9}
    Let $f$ be a density for random vector $Z \in \R^n$. Then for any $M,t>0$,
\begin{align*}
    \int_{\R^n} \mathbbm{1}(f(z) \leq t) f(z) dz &\leq (2M)^nt + \frac{\sum_{i=1}^n Var(Z_i)}{M^2}.
\end{align*}
In particular, for $n=2$, choosing $M = t^{-1/4}\left( Var(Z_1)+Var(Z_2) \right)^{1/4}$ gives
\begin{align*}
    \int_{\R^2} \mathbbm{1}(f(z) \leq t) f(z) dz &\leq 5t^{1/2}\left( Var(Z_1)+Var(Z_2) \right)^{1/2}.
\end{align*}
\end{lemma}

\begin{proof}
    As in the proof in \citeappendix{chen2022empirical}, assume without loss of generality that $E_f[Z] = 0$.
    \begin{align*}
        \int_{\R^n} \mathbbm{1}(f(z) \leq t) f(z) dz &\leq \int_{\R^n} \mathbbm{1}(f(z) \leq t, \Vert z \Vert_2 < M) f(z) dz + \int_{\R^n} \mathbbm{1}(f(z) \leq t, \Vert z \Vert_2 \geq M ) f(z) dz \\
        &\leq \int_{\Vert z \Vert_2 < M} t dz + Pr(\Vert Z \Vert_2 > M) \\
        &\leq (2M)^n t + \frac{\sum_{i=1}^n Var(Z_i)}{M^2} \qquad \text{multivariate Chebyshev}.
    \end{align*}
\end{proof}

\begin{lemma}\label{lem:sm6.10}
    Recall that $Q_j(z, F, \Psi_j) = \int \varphi_{\Psi_j}(z - \tau) \text{vec}((z-\tau)\tau^T) dF(\tau)$. For any $F, z,$ and $\rho_J \in (0, e^{-1}/2\pi)$, 
\begin{align*}
    \frac{\left\Vert Q_j(z, F, \Psi_j)\right\Vert_2}{\max\left(f_{F, \Psi_j}(z), \frac{\rho_J}{\sqrt{\det(\Psi_j)}} \right)} &\leq \sqrt{\det(2\pi \Psi_j)} \left\Vert \Psi_j \right\Vert_F \left( \Vert z \Vert_2 \left\Vert \Psi_j^{-1} \right\Vert_F^{1/2} \varphi_+(\rho_J) + \varphi_+^2(\rho_J) \right).
\end{align*}
Then under the choices of $\rho_J$ and $M_J$ in Lemma \ref{lem:oa3.1} and Assumption \ref{ass:sm6.1} respectively, uniformly over $F, j \in [J]$, and $\Vert z \Vert_2 \leq M_J$,
\begin{align*}
    \frac{\left\Vert Q_j(z, F, \Psi_j)\right\Vert_2}{\max\left(f_{F, \Psi_j}(z), \frac{\rho_J}{\sqrt{\det(\Psi_j)}} \right)} &\lesssim_{\mathcal H} M_J \sqrt{\log J}.
\end{align*}
\end{lemma}

\begin{proof}
    Note 
    \begin{align*}
        \left\Vert Q_j(z, F, \Psi_j) \right\Vert_2 &\leq \sqrt{\det(2\pi \Psi_j)} f_{F, \Psi_j}(z) \left\Vert E_{F, \Psi_j} [(z-\tau) | z] \right\Vert_2 \left\Vert z \right\Vert_2 \\
        &\qquad + \sqrt{\det(2\pi \Psi_j)} f_{F, \Psi_j}(z) \left\Vert E_{F, \Psi_j} [(z-\tau)(z-\tau)^T | z] \right\Vert_F.
    \end{align*}
    Then from Lemma \ref{lem:sm6.8},
    \begin{align*}
        \frac{f_{F, \Psi_j}(z)}{\max\left(f_{F, \Psi_j}(z), \frac{\rho_J}{\sqrt{\det(\Psi_j)}} \right)} \left\Vert E_{F, \Psi_j} [(z-\tau) | z] \right\Vert_2 &\leq \left\Vert \Psi_j \right\Vert_F \left\Vert \Psi_j^{-1} \right\Vert_F^{1/2} \varphi_+(\rho_J)
    \end{align*}
    and
    \begin{align*}
        \frac{f_{F, \Psi_j}(z)}{\max\left(f_{F, \Psi_j}(z), \frac{\rho_J}{\sqrt{\det(\Psi_j)}} \right)} \left\Vert E_{F, \Psi_j} [(z-\tau)(z-\tau)^T | z] \right\Vert_F &\leq \left\Vert \Psi_j \right\Vert_F \varphi_+^2(\rho_J).
    \end{align*}
    Thus
    \begin{align*}
        \frac{\left\Vert Q_j(z, F, \Psi_j)\right\Vert_2}{\max\left(f_{F, \Psi_j}(z), \frac{\rho_J}{\sqrt{\det(\Psi_j)}} \right)} &\leq \sqrt{\det(2\pi \Psi_j)} \left\Vert \Psi_j \right\Vert_F \left( \Vert z \Vert_2 \left\Vert \Psi_j^{-1} \right\Vert_F^{1/2} \varphi_+(\rho_J) + \varphi_+^2(\rho_J) \right).
    \end{align*}
\end{proof}

\begin{lemma}\label{lem:sm6.11}
    Under the assumptions in Lemma \ref{lem:oa3.1} and Assumption \ref{ass:good_approx}, suppose $(a_j(\tilde \alpha), B_j(\tilde \Omega)^{1/2})$ are on the line segment between $(a_j(\widehat \alpha), B_j(\widehat \Omega)^{1/2})$ and $(a_j(\alpha_0), B_j(\Omega_0)^{1/2})$, and define $\tilde\Psi_j, \tilde Z_j$ accordingly. Then, second derivatives evaluated at $\widehat{F}_J, \tilde \alpha, \tilde \Omega, \tilde Z_j$ satisfy 
    \begin{align*}
        \left\Vert \frac{\partial^2 \psi_j}{\partial a_j \partial a_j^T} \bigg\vert_{\widehat{F}_J, \tilde\alpha, \tilde\Omega} \right\Vert_{F} &\lesssim_{\mathcal H} \log J \\
        \left\Vert \frac{\partial^2 \psi_j}{\partial \text{vec}(B_j^{1/2}) \partial a_j^T} \bigg\vert_{\widehat{F}_J, \tilde\alpha, \tilde\Omega} \right\Vert_{F} &\lesssim_{\mathcal H} M_J \log J \\
        \left\Vert \frac{\partial^2 \psi_j}{\partial \text{vec}(B_j^{1/2}) \partial \text{vec}(B_j^{1/2})^T} \bigg\vert_{\widehat{F}_J, \tilde\alpha, \tilde\Omega} \right\Vert_{F} &\lesssim_{\mathcal H} M_J^2 \log J.
    \end{align*}
\end{lemma}

\begin{proof}
    Note that as in the proof of Lemma \ref{lem:oa3.1}, 
    $\widehat{Z}_j = B_j(\widehat\Omega)^{-1/2}(B_j(\tilde \Omega)^{1/2}\tilde Z_j + a_j(\tilde \alpha) - a_j(\widehat\alpha))$, where $\Vert B_j(\tilde \Omega)^{1/2} - B_j(\widehat\Omega)^{1/2} \Vert_\infty \leq \Delta_J$ and $\Vert a_j(\tilde \alpha) - a_j(\widehat\alpha) \Vert_\infty \leq \Delta_J$. Thus $\Vert \tilde Z_j \Vert_2 \lesssim_{\mathcal H} M_J$. Furthermore by the same argument as in Lemma \ref{lem:oa3.1}, $f_{\widehat{F}_J, \tilde \Psi_j}(\tilde Z_j) \sqrt{\det(\tilde \Psi_j)} \geq \frac{1}{J^4} e^{-C_H \Delta_J M_J^2}$. Thus as in \citeappendix{chen2022empirical}, $|\log(f_{\widehat{F}_J, \tilde \Psi_j}(\tilde Z_j) \sqrt{\det(\tilde \Psi_j)})| \lesssim_{\mathcal H} \log J$.

    Using Lemma \ref{lem:sm6.8} and properties of logarithms
    \begin{align*}
        \left\Vert E_{\widehat{F}_J, \tilde \Psi_j} [\tau_j - Z_j \vert \tilde Z_j] \right\Vert_2 &= \left\Vert \tilde\Psi_j \frac{\nabla f_{\widehat{F}_J, \tilde \Psi_j}(\tilde Z_j)}{f_{\widehat{F}_J, \tilde \Psi_j}(\tilde Z_j)} \right\Vert_2 \\
        &\lesssim_{\mathcal H} \sqrt{\log \left(\frac{1}{f_{\widehat{F}_J, \tilde \Psi_j}(\tilde Z_j)} \right)} \lesssim_{\mathcal H} \sqrt{\log J}
    \end{align*}
    and 
    \begin{align*}
        \left\Vert E_{\widehat{F}_J, \tilde \Psi_j} [(\tau_j - Z_j)(\tau_j-Z_j)^T \vert \tilde Z_j] \right\Vert_F &= \left\Vert \tilde\Psi_j + \tilde\Psi_j \frac{\nabla^2 f_{\widehat{F}_J, \tilde \Psi_j}}{f_{\widehat{F}_J, \tilde \Psi_j}} \tilde \Psi_j \right\Vert_F \\
        &\lesssim_{\mathcal H} \log \left(\frac{1}{f_{\widehat{F}_J, \tilde \Psi_j}(\tilde Z_j)} \right) \lesssim_{\mathcal H} \log J.
    \end{align*}
    And note that because $\Vert \tilde{Z}_j \Vert_2 \lesssim_{\mathcal H} M_J$ and $\Vert \tau \Vert_2 \lesssim_{\mathcal H} M_J$ under the support of $\widehat{F}_J$,
    \begin{align*}
        \left\Vert E_{\widehat{F}_J, \tilde \Psi_j}\left[ \text{vec}\left( (Z_j - \tau_j) \tau_j^T \right) \vert \tilde Z_j \right] \right\Vert_2 &\leq \left\Vert E_{\widehat{F}_J, \tilde \Psi_j}[(\tau_j - Z_j)(\tau_j - Z_j)^T \vert \tilde Z_j] \right\Vert_F \\
        &\qquad + \left\Vert E_{\widehat{F}_J, \tilde \Psi_j}[\tau_j - Z_j \vert \tilde Z_j] \right\Vert_2 \left\Vert \tilde Z_j \right\Vert_2 \\
        &\lesssim_{\mathcal H} \log J + M_J \sqrt{\log J} \lesssim_{\mathcal H} M_J \sqrt{\log J} \\
        \left\Vert E_{\widehat{F}_J, \tilde \Psi_j}\left[ \text{vec}\left( (Z_j - \tau_j) \tau_j^T \right)(Z_j - \tau)^T \vert \tilde Z_j \right] \right\Vert_F &= \left\Vert E_{\widehat{F}_J, \tilde \Psi_j}\left[ ( \tau_j \otimes I_2 )(Z_j - \tau_j)(Z_j - \tau)^T \vert \tilde Z_j \right] \right\Vert_F \\
        &\leq E_{\widehat{F}_J, \tilde \Psi_j} [ \Vert \tau \Vert_2 \Vert (Z_j - \tau_j)(Z_j - \tau)^T \Vert_F \vert \tilde Z_j] \lesssim_{\mathcal H} M_J \log J.
    \end{align*}
    Similarly, one can check
    \begin{align*}
        &\left\Vert E_{\widehat{F}_J, \tilde \Psi_j}[ \text{vec}\left((Z_j-\tau_j)\tau_j^T \right) \text{vec}\left((Z_j-\tau_j)\tau_j^T \right)^T \vert \tilde Z_j] \right\Vert_F \\
        &= \left\Vert E_{\widehat{F}_J, \tilde \Psi_j}[ (\tau_j \otimes I_2)(Z_j-\tau_j)(Z_j-\tau_j)^T (\tau_j^T \otimes I_2) \vert \tilde Z_j] \right\Vert_F \\
        &\leq E_{\widehat{F}_J, \tilde \Psi_j} [ \Vert \tau \Vert_2^2 \Vert (Z_j-\tau_j)(Z_j-\tau_j)^T \Vert_F \vert \tilde Z_j] \lesssim_{\mathcal H} M_J^2 \log J.
    \end{align*}

    Then plugging the above results into the derivative expressions derived in Section \ref{app:sec:sm6.1}, 
    \begin{align*}
        \left\Vert \frac{\partial^2 \psi_j}{\partial a_j \partial a_j^T} \bigg\vert_{\widehat{F}_J, \tilde\alpha, \tilde\Omega} \right\Vert_F &\lesssim_{\mathcal H} \log J \\
        \left\Vert \frac{\partial^2 \psi_j}{\partial a_j \partial \text{vec}(B_j^{1/2})^T} \bigg\vert_{\widehat{F}_J, \tilde\alpha, \tilde\Omega} \right\Vert_F &\lesssim_{\mathcal H} M_J \log J \\
        \left\Vert \frac{\partial^2 \psi_j}{\partial \text{vec}(B_j^{1/2}) \partial \text{vec}(B_j^{1/2})^T} \bigg\vert_{\widehat{F}_J, \tilde\alpha, \tilde\Omega} \right\Vert_F &\lesssim_{\mathcal H} M_J^2 \log J.
    \end{align*}
    where the final line follows because the derivative is a sum of the above derived terms times functions of $\Omega$ and $\Sigma$.
\end{proof}

\begin{lemma}\label{lem:sm6.14}
    Suppose $\tau \sim F_0$, $Z \vert \tau \sim N_2(\tau, \Psi)$, and $\underline \psi I_2 \preceq \Psi \preceq \overline \psi I_2$. Suppose also that Assumptions \ref{ass:compact} and \ref{ass:est_var} hold. Then for every $p > 0$, $E_{F_0,\Psi}[\Vert \tau \Vert_2^p \vert Z = z] \lesssim_{p,\mathcal H} \left(\max\{\Vert z \Vert_2, 1\}\right)^p$.
\end{lemma}
\begin{proof}
    Define $\Vert x \Vert_{\Psi^{-1}}^2 \equiv x^T \Psi^{-1} x$ and let $M = \max\{\Vert z \Vert_{\Psi^{-1}}, M_0\}$, where $M_0$ is a sufficiently large constant such that, uniformly over $\Psi$, $F_0(\Vert \tau \Vert_{\Psi^{-1}} \leq M_0) \geq \frac{3}{4}$. $M_0$ exists by Markov's inequality, noting $E[\Vert \tau \Vert_{\Psi^{-1}}^2] \leq \underline \psi^{-1} E[\Vert \tau \Vert_2^2] \lesssim_{\mathcal H} 1$.

    I can write 
    \begin{align*}
        &\int \Vert \tau \Vert_2^p \exp \left(-\frac{1}{2} \Vert z-\tau\Vert_{\Psi^{-1}}^2 \right) dF_0(\tau) \\
        &\leq \left(3 \sqrt{\overline\psi} M \right)^p \int \exp \left(-\frac{1}{2} \Vert z-\tau\Vert_{\Psi^{-1}}^2 \right) dF_0(\tau) + e^{-2M^2} \int_{\Vert \tau \Vert_{\Psi^{-1}} > 3M} \Vert \tau \Vert_2^p dF_0(\tau) \\
        &\leq \left(3 \sqrt{\overline\psi} M \right)^p \int \exp \left(-\frac{1}{2} \Vert z-\tau\Vert_{\Psi^{-1}}^2 \right) dF_0(\tau) + e^{-2M^2} E_{F_0}[\Vert \tau \Vert_2^p],
    \end{align*}
    using that $\Vert \tau \Vert_{\Psi^{-1}} > 3M$ implies $\Vert z - \tau\Vert_{\Psi^{-1}} \geq \Vert \tau \Vert_{\Psi^{-1}} - \Vert z \Vert_{\Psi^{-1}} > 2M$.
    
    Furthermore, since $\Vert z \Vert_{\Psi^{-1}} \leq M$,
    \begin{align*}
        \int \exp \left(-\frac{1}{2} \Vert z-\tau\Vert_{\Psi^{-1}}^2 \right) dF_0(\tau) &\geq \int_{\Vert \tau \Vert_{\Psi^{-1}} \leq M} \exp \left(-\frac{1}{2} \Vert z-\tau\Vert_{\Psi^{-1}}^2 \right) dF_0(\tau) \\
        &\geq e^{-2M^2} F_0(\Vert \tau \Vert_{\Psi^{-1}} \leq M) \geq \frac{3}{4} e^{-2M^2},
    \end{align*}
    using that $\Vert \tau \Vert_{\Psi^{-1}} \leq M$ implies $\Vert z - \tau\Vert_{\Psi^{-1}} \leq \Vert z \Vert_{\Psi^{-1}} + \Vert \tau \Vert_{\Psi^{-1}} \leq 2M$.
    Thus
    \begin{align*}
        E_{F_0,\Psi}[\Vert \tau \Vert_2^p \vert Z = z] &\leq \left(3 \sqrt{\overline\psi} M \right)^p + \frac{4}{3} E_{F_0}[\Vert \tau \Vert_2^p].
    \end{align*}
    Assumption \ref{ass:compact} implies $E_{F_0}[\Vert \tau \Vert_2^p] \lesssim_{p, \mathcal H} 1$. Furthermore, $M \leq \max\{\underline{\psi}^{-1/2} \Vert z \Vert_2, M_0\} \lesssim_{\mathcal H} \max\{\Vert z \Vert_2, 1\}$. Thus
    \begin{align*}
        E_{F_0,\Psi}[\Vert \tau \Vert_2^p \vert Z = z] &\lesssim_{p,\mathcal H} \left( \max\{\Vert z \Vert_2, 1\} \right)^p.
    \end{align*}
\end{proof}

\begin{lemma}\label{prop:sm6.2}
    Recalling $d_{\alpha,\infty,M}$ and $d_{\Omega,\infty,M}$ from \eqref{eq:d_k_infty_M}, the following bounds hold for $\eta \leq \min \left(1/e, 4/(2\pi \underline{k}), e^{-1/(2\bar k)} \right)$ under the assumptions of Theorem \ref{thm:sm6.1}:
    \begin{small}
    \begin{align*}
        \log N\left(\frac{1+\sqrt{\log(1/\rho_J)}}{\rho_J} \eta, \mathcal{P}(\R^2), d_{\alpha,\infty,M} \right) &\lesssim_{\mathcal H} (\log(1/\eta))^3 \max \left\{1, \frac{M}{\sqrt{\log(1/\eta)}}, \frac{M^2}{\log(1/\eta)} \right\} \\
        \log N\left(\frac{1+M\sqrt{\log(1/\rho_J)}+\log(1/\rho_J)}{\rho_J} \eta, \mathcal{P}(\R^2), d_{\Omega,\infty,M} \right) &\lesssim_{\mathcal H} (\log(1/\eta))^3 \max \left\{1, \frac{M}{\sqrt{\log(1/\eta)}}, \frac{M^2}{\log(1/\eta)} \right\}.
    \end{align*}
    \end{small}
\end{lemma}
\begin{proof}
    Fix some $\Vert z \Vert_2 \leq M$. Let $T_{\alpha,j} = \nabla f_{F,\Psi_j}(z)$ and $T_{\Omega,j} = Q_j(z,F,\Psi_j)$.
    As in the proof of Proposition SM6.2 in \citeappendix{chen2022empirical},
    \begin{align*}
        &\Vert D_{k,j}(z, F_1, \chi_0, \rho_J) - D_{k,j}(z, F_2, \chi_0, \rho_J) \Vert_2 \\
        &\lesssim_{\mathcal H} \frac{1}{\rho_J} \Vert T_{k,j}(z,F_1,\chi_0) - T_{k,j}(z,F_2,\chi_0) \Vert_2 \\
        &\qquad + \frac{\Vert T_{k,j}(z,F_2,\chi_0) \Vert_2}{\rho_J \max(f_{F_2,\Psi_j}(z), \rho_J/\sqrt{\det(\Psi_j)})} \vert f_{F_1,\Psi_j}(z) - f_{F_2,\Psi_j}(z) \vert.
    \end{align*}
    Then by Lemmas \ref{lem:sm6.8} and \ref{lem:sm6.10},
    \begin{align*}
        \Vert D_{\alpha,j}(z, F_1, \chi_0, \rho_J) - D_{\alpha,j}(z, F_2, \chi_0, \rho_J) \Vert_2 &\lesssim_{\mathcal H} \frac{1}{\rho_J} \Vert \nabla f_{F_1, \Psi_j}(z) - \nabla f_{F_2, \Psi_j}(z) \Vert_2 \\
        &\qquad + \frac{\sqrt{\log (1/\rho_J)}}{\rho_J} \vert f_{F_1,\Psi_j}(z) - f_{F_2,\Psi_j}(z) \vert \\
        \Vert D_{\Omega,j}(z, F_1, \chi_0, \rho_J) - D_{\Omega,j}(z, F_2, \chi_0, \rho_J) \Vert_2 &\lesssim_{\mathcal H} \frac{1}{\rho_J} \Vert Q_j(z,F_1,\Psi_j) - Q_j(z,F_2,\Psi_j) \Vert_2 \\
        &\qquad + \frac{M\sqrt{\log(1/\rho_J)} + \log(1/\rho_J)}{\rho_J} \vert f_{F_1,\Psi_j}(z) - f_{F_2,\Psi_j}(z) \vert.
    \end{align*}
    Define 
    \begin{align}
        K_{F,\Psi_j}(z) \equiv \int (z-\tau)(z-\tau)^T \varphi_{\Psi_j}(z-\tau) dF(\tau) \label{eq:k_function}.
    \end{align}
    Note that I can write
    \begin{align*}
        Q_j(z,F,\Psi_j) &= \text{vec} \left(-\sqrt{\det(2\pi \Psi_j)}\Psi_j \nabla f_{F,\Psi_j}(z)z' - K_{F,\Psi_j}(z) \right).
    \end{align*}
    Then it follows that
    \begin{align*}
        \Vert Q_j(z,F_1,\Psi_j) - Q_j(z,F_2,\Psi_j) \Vert_2 &\leq \sqrt{\det(2\pi \Psi_j)}\left\Vert \Psi_j \right\Vert_{op} \left\Vert z \right\Vert_2 \left\Vert \nabla f_{F_1,\Psi_j}(z) - \nabla f_{F_2,\Psi_j}(z) \right\Vert_2 \\
        &\qquad + \left\Vert K_{F_1,\Psi_j}(z) - K_{F_2,\Psi_j}(z) \right\Vert_F \\
        &\lesssim_{\mathcal H} M \left\Vert \nabla f_{F_1,\Psi_j}(z) - \nabla f_{F_2,\Psi_j}(z) \right\Vert_2 + \left\Vert K_{F_1,\Psi_j}(z) - K_{F_2,\Psi_j}(z) \right\Vert_F.
    \end{align*}
    Similar to Appendix C of \citeappendix{soloff2024multivariate} define
    \begin{align}
        \Vert f_{F_1, \cdot} - f_{F_2, \cdot} \Vert_{\infty,M} &\equiv \max_{j \in [J]} \sup_{\Vert z \Vert_2 \leq M} \vert f_{F_1,\Psi_j}(z) - f_{F_2,\Psi_j}(z) \vert \label{eq:inf_M_norm} \\
        \Vert f_{F_1, \cdot} - f_{F_2, \cdot} \Vert_{\nabla,M} &\equiv \max_{j \in [J]} \sup_{\Vert z \Vert_2 \leq M} \Vert \nabla f_{F_1,\Psi_j}(z) - \nabla f_{F_2,\Psi_j}(z) \Vert_2 \label{eq:nabla_M_norm} \\
        \Vert K_{F_1, \cdot} - K_{F_2, \cdot} \Vert_{\nabla^2,M} &\equiv \max_{j \in [J]} \sup_{\Vert z \Vert_2 \leq M} \Vert K_{F_1,\Psi_j}(z) - K_{F_2, \Psi_j}(z) \Vert_F. \label{eq:nabla_sq_M_norm}
    \end{align}
    Then
    \begin{align*}
        d_{\alpha,\infty,M}(F_1,F_2) &\lesssim_{\mathcal H} \frac{1}{\rho_J} \Vert f_{F_1, \cdot} - f_{F_2, \cdot} \Vert_{\nabla,M} + \frac{\sqrt{\log (1/\rho_J)}}{\rho_J} \Vert f_{F_1, \cdot} - f_{F_2, \cdot} \Vert_{\infty,M} \\
        d_{\Omega,\infty,M}(F_1,F_2) &\lesssim_{\mathcal H} \frac{1}{\rho_J} \Vert K_{F_1, \cdot} - K_{F_2, \cdot} \Vert_{\nabla^2,M} + \frac{M}{\rho_J} \Vert f_{F_1, \cdot} - f_{F_2, \cdot} \Vert_{\nabla,M} \\
        &\qquad + \frac{M\sqrt{\log (1/\rho_J)} + \log(1/\rho_J)}{\rho_J} \Vert f_{F_1, \cdot} - f_{F_2, \cdot} \Vert_{\infty,M}
    \end{align*}

    Let $\mathbb{F}$ be the space of functions $f_{F,\cdot}$ induced by the space of distributions $F \in \mathcal{P}(\R^2)$.
    By Suppl. Lemma 5 of \citeappendix{soloff2024multivariate}, Suppl. Lemma 6 of \citeappendix{soloff2024multivariate}, and Suppl. Lemma F.6 of \citeappendix{saha2020nonparametric}, if $\eta \leq \min \left(1/e, 4/(2\pi \underline{k}) \right)$
    \begin{align*}
        \log N(\eta, \mathbb{F}, \Vert \cdot \Vert_{\infty,M}) &\lesssim_{\mathcal H} \log(1/\eta)^3 \max \left(1, \frac{M}{\sqrt{\log(1/\eta)}}, \frac{M^2}{\log(1/\eta)} \right) \\
        \log N(\eta, \mathbb{F}, \Vert \cdot \Vert_{\nabla,M}) &\lesssim_{\mathcal H} \log(1/\eta)^3 \max \left(1, \frac{M}{\sqrt{\log(1/\eta)}}, \frac{M^2}{\log(1/\eta)} \right).
    \end{align*}
    Let $\mathbb{K}$ be the space of functions $K_{F,\cdot}$ induced by the space of distributions $F \in \mathcal{P}(\R^2)$. By Lemma \ref{prop:sm6.1} if $\eta \leq e^{-1/(2\bar k)}$
    \begin{align*}
        \log N(\eta, \mathbb{K}, \Vert \cdot \Vert_{\nabla^2,M}) &\lesssim_{\mathcal H} (\log 1/\eta)^3 \max \left(1, \frac{M}{\sqrt{\log(1/\eta)}}, \frac{M^2}{\log(1/\eta)} \right).
    \end{align*}
    Thus
    \begin{small}
    \begin{align*}
        \log N\left(\frac{1+\sqrt{\log(1/\rho_J)}}{\rho_J} \eta, \mathcal{P}(\R^2), d_{\alpha,\infty,M} \right) &\lesssim_{\mathcal H} \log(1/\eta)^3 \max \left(1, \frac{M}{\sqrt{\log(1/\eta)}}, \frac{M^2}{\log(1/\eta)} \right) \\
        \log N\left(\frac{1+M\sqrt{\log(1/\rho_J)}+\log(1/\rho_J)}{\rho_J} \eta, \mathcal{P}(\R^2), d_{\Omega,\infty,M} \right) &\lesssim_{\mathcal H} \log(1/\eta)^3 \max \left(1, \frac{M}{\sqrt{\log(1/\eta)}}, \frac{M^2}{\log(1/\eta)} \right).
    \end{align*}
    \end{small}
\end{proof}

\begin{lemma}\label{prop:sm6.1}
    Recall $\Vert \cdot \Vert_{\nabla^2, M}$ from \eqref{eq:nabla_sq_M_norm} and $\mathbb{K}$ the space of $K_{F,\Psi_j}$ induced by distributions $F \in \mathcal{P}(\R^2)$, for $K_{F,\Psi_j}$ defined by \eqref{eq:k_function}.
    Then for $\eta \leq e^{-1/(2\bar k)}$,  under the assumptions of Theorem \ref{thm:sm6.1}, $\log N(\eta, \mathbb{K}, \Vert \cdot \Vert_{\nabla^2, M}) \lesssim_{\mathcal H} (\log (1/\eta))^3 \max \left\{1, \frac{M}{\sqrt{\log(1/\eta)}}, \frac{M^2}{\log(1/\eta)} \right\}.$
\end{lemma}
\begin{proof}
    The proof follows the proof of Suppl. Lemmas 3-6 in \citeappendix{soloff2024multivariate}. Define the second-derivative metric
    \begin{align*}
        \left\Vert f_{F_1,\cdot} - f_{F_2,\cdot} \right\Vert_{\nabla^2_f, M} &\equiv \max_{j \in [J]} \sup_{\left\Vert z \right\Vert_2 \leq M} \left\Vert \nabla^2 f_{F_1,\Psi_j}(z) - \nabla^2 f_{F_2,\Psi_j}(z) \right\Vert_F.
    \end{align*}
    Note that
    \begin{align*}
        K_{F,\Psi_j}(z) = \sqrt{\det(2\pi \Psi_j)} \left(\Psi_j \nabla^2 f_{F,\Psi_j}(z) \Psi_j + \Psi_j f_{F,\Psi_j}(z) \right),
    \end{align*}
    so there exists some constant $C_{\mathcal H}$ such that 
    \begin{align*}
        \left\Vert K_{F_1,\cdot} - K_{F_2,\cdot} \right\Vert_{\nabla^2,M} &\leq C_{\mathcal H} \left(\left\Vert f_{F_1,\cdot} - f_{F_2,\cdot} \right\Vert_{\nabla^2_f,M} + \left\Vert f_{F_1,\cdot} - f_{F_2,\cdot} \right\Vert_{\infty,M}\right).
    \end{align*} 
    
    Fix a distribution $F \in \mathcal{P}(\R^2)$ and fix some $a \geq 1$ to be chosen. Define the sets $S_M \equiv \{z \in \R^2: \Vert z \Vert_2 \leq M\}$ and $S_M^a \equiv \{z: d(z,S_M) \leq a\}$, and let $L = N(a, S_M^a, \Vert \cdot \Vert_2)$. Using the multivariate moment-matching construction of Suppl. Lemmas 3 and 4 of \citeappendix{soloff2024multivariate}, with the degree of the matched moments increased by two, there exists a discrete distribution $H$ supported on $S_M^a$ with at most
    \begin{align*}
        l \lesssim_{\mathcal H} a^4 L + 1
    \end{align*}
    atoms such that 
    \begin{align*}
        \left\Vert f_{F,\cdot} - f_{H,\cdot}\right\Vert_{\infty,M} \lesssim_{\mathcal H} e^{-a^2/(2\bar k)}
    \end{align*}
    and 
    \begin{align*}
        \left\Vert f_{F,\cdot} - f_{H,\cdot}\right\Vert_{\nabla_f^2,M} \lesssim_{\mathcal H} (1+a^2)e^{-a^2/(2\bar k)},
    \end{align*}
    where the second bound follows because
    \begin{align*}
        \nabla^2 \varphi_{\Psi_j}(z) = \left( \Psi_j^{-1} zz^T \Psi_j^{-1} - \Psi_j^{-1} \right) \varphi_{\Psi_j}(z)
    \end{align*}
    so it is sufficient to match multivariate moments with two additional moments relative to Suppl. Lemma 3 of \citeappendix{soloff2024multivariate}, and uniformly over $j$,
    \begin{align*}
        \sup_{\Vert z \Vert_2 \geq a} \left\Vert \nabla^2 \varphi_{\Psi_j}(z) \right\Vert_F &\lesssim_{\mathcal H} (1+a^2) e^{-a^2/(2\bar k)}
    \end{align*}
    so the same argument used in Suppl. Lemmas 3 and 4 of \citeappendix{soloff2024multivariate} applies.

    Let $\mathcal C$ be a minimal $\alpha$-net of $S^a_M$ and let $H'$ approximate each atom of $H$ with its closest element from $\mathcal C$, so that $H = \sum_i w_i \delta_{a_i}$ and $H' = \sum_i w_i \delta_{b_i}$ where $w_i$ are convex weights. Since the first and third derivatives of the Gaussian density are uniformly bounded under eigenvalue bounds on $\Psi_j$, the mean value theorem gives
    \begin{align*}
        \left\Vert f_{H,\cdot} - f_{H',\cdot} \right\Vert_{\nabla^2_f,M} + \left\Vert f_{H,\cdot} - f_{H',\cdot} \right\Vert_{\infty,M} &\lesssim_{\mathcal H} \alpha.
    \end{align*}
    
    Let $\Delta_{l-1}$ be the $(l-1)$-simplex of probability vectors in $l$ dimensions and let $\mathcal D$ be a minimal $\beta$-net of $\Delta_{l-1}$ in the $\Vert \cdot \Vert_1$ norm. Let $H''$ be the distribution that approximates the weights $w$ by their closest element $v \in \mathcal D$, so $H'' = \sum_i v_i \delta_{b_i}$. Then because the Gaussian density and its Hessian are uniformly bounded,
    \begin{align*}
        \left\Vert f_{H',\cdot} - f_{H'',\cdot} \right\Vert_{\nabla^2_f,M} + \left\Vert f_{H',\cdot} - f_{H'',\cdot} \right\Vert_{\infty,M} &\lesssim_{\mathcal H} \beta.
    \end{align*}
    
    So combining the previous bounds,
    \begin{align*}
        \left\Vert f_{F,\cdot} - f_{H'',\cdot} \right\Vert_{\nabla^2_f,M} + \left\Vert f_{F,\cdot} - f_{H'',\cdot} \right\Vert_{\infty,M} &\lesssim_{\mathcal H} (1+a^2)e^{-a^2/(2\bar k)} + \alpha + \beta.
    \end{align*}
    For any $\eta \leq e^{-1/(2\bar k)}$ choose $\delta = c_0 \eta^2$, $\alpha = \beta = \delta$ and $a = \sqrt{2\bar k \log(1/\alpha)}$, where $c_0 \in (0,1]$ is a sufficiently small constant to be chosen later. Note that $a \geq 1$. Furthermore because $0 < \eta < 1$,
    \begin{align*}
        (1+a^2)e^{-a^2/(2\bar k)} + \alpha + \beta &= c_0 \eta^2 \left(3 + 2\bar k \log(1/c_0 \eta^2) \right) = c_0 \eta^2 \left(3 + 2\bar k \log(1/c_0) + 4\bar k \log(1/\eta) \right) \\
        &\leq c_0 \eta \left(3 + 2\bar k \log(1/c_0) + \frac{4\bar k}{e} \right).
    \end{align*}
    I can choose $c_0$ sufficiently small that
    \begin{align*}
        \left\Vert f_{F,\cdot} - f_{H'',\cdot} \right\Vert_{\nabla^2_f,M} + \left\Vert f_{F,\cdot} - f_{H'',\cdot} \right\Vert_{\infty,M} &\leq C_{\mathcal H}^{-1} \eta
    \end{align*}
    for the constant $C_{\mathcal H}$ defined above, so that
    \begin{align*}
        \left\Vert K_{F,\cdot} - K_{H'',\cdot} \right\Vert_{\nabla^2,M} &\leq \eta.
    \end{align*}

    Following the argument of Suppl. Lemma 5 in \citeappendix{soloff2024multivariate} and the proof of Theorem 4.1 in \citeappendix{saha2020nonparametric}, the number of possible $H''$ is bounded above, up to a constant that depends on $\mathcal H$, by $\left(1 + \frac{2}{\beta}\right)^l\left(1+\frac{N(\alpha, S_M^a, \Vert \cdot \Vert_2)}{l}\right)^l \equiv A^l$. 
    Thus $\log N(\eta, \mathbb{K}, \Vert \cdot \Vert_{\nabla^2, M}) \lesssim_{\mathcal H} l \log A$. 
    
    Note that $\log A \lesssim_{\mathcal H} \log \left[\left(1+\frac{2}{\beta}\right) \left(1 + \frac{a}{\alpha} \right)^2\right] \lesssim_{\mathcal H} \log(1/\delta)$. Finally note that $L = N(a,S_M^a, \Vert \cdot \Vert_2) \lesssim_{\mathcal H} \left(1+\frac{M}{a}\right)^2$. Together with the expressions for $l$ and $L$,
    \begin{align*}
        \log N(\eta, \mathbb{K}, \Vert \cdot \Vert_{\nabla^2, M}) &\lesssim_{\mathcal H} a^4\left(1+\frac{M}{a}\right)^2 \log(1/\delta) \\
        &\lesssim_{\mathcal H} (\log 1/\delta)^3 \max \left(1, \frac{M}{\sqrt{\log(1/\delta)}}, \frac{M^2}{\log(1/\delta)} \right).
    \end{align*}
    Since $\log(1/\delta) = 2\log(1/\eta) + \log(1/c_0) \asymp_{\mathcal H} \log(1/\eta)$,
    \begin{align*}
        \log N(\eta, \mathbb{K}, \Vert \cdot \Vert_{\nabla^2, M}) &\lesssim_{\mathcal H} (\log 1/\eta)^3 \max \left(1, \frac{M}{\sqrt{\log(1/\eta)}}, \frac{M^2}{\log(1/\eta)} \right).
    \end{align*}
    
\end{proof}

\FloatBarrier

%small
\singlespacing
{
\bibliographystyleappendix{chicago}
\bibliographyappendix{ref}
}

\end{document}